\documentclass[journal,twocolumn,twoside]{IEEEtran}

\usepackage[cmex10]{amsmath}
\usepackage{amsfonts}
\usepackage{amssymb}
\usepackage{amsthm}
\usepackage{microtype}
\usepackage{graphicx}
\usepackage[usenames,dvipsnames]{color}   %
\usepackage{tikz}                         %

\usepackage{hyperref}
\hypersetup{colorlinks = true, linkcolor = red, citecolor = blue, urlcolor = blue}
\usepackage[all]{hypcap}   %

\usepackage[capitalise,noabbrev]{cleveref}   %
\usepackage{autonum}                         %
\crefformat{equation}{(#2#1#3)}
\crefmultiformat{equation}{(#2#1#3)}{ and~(#2#1#3)}{, (#2#1#3)}{, and~(#2#1#3)}%
\crefrangeformat{equation}{(#3#1#4) to (#5#2#6)}

\usepackage{bbm}                     %
\usepackage{cite}                    %
\usepackage{comment}                 %
\usepackage[shortlabels]{enumitem}   %
\usepackage[T3,T1]{fontenc}          %
\usepackage[utf8]{inputenc}
\usepackage{mathrsfs}                %
\usepackage{mathtools}               %
\usepackage{mleftright}              %
\usepackage{romannum}                %
\usepackage{stmaryrd}                %
\usepackage{upgreek}                 %
\usepackage{balance}

\newtheorem{theorem}{Theorem}
\newtheorem{corollary}{Corollary}
\newtheorem{definition}{Definition}
\newtheorem{lemma}{Lemma}
\newtheorem{proposition}{Proposition}

\allowdisplaybreaks   %
\makeatletter
\newcommand\gobblepars{%
    \@ifnextchar\par%
        { \expandafter\gobblepars\@gobble}%
        {}}
\makeatother

\newcommand\ifrac[2]{{#1}/{#2}}
\newcommand\ifraca[2]{{(#1)}/{#2}}
\newcommand\ifracb[2]{{#1}/{(#2)}}
\newcommand\ifracab[2]{{(#1)}/{(#2)}}

\newcommand{\abs}[1]{\mleft\lvert #1 \mright\rvert}
\newcommand{\ang}[1]{\mleft\langle #1 \mright\rangle}
\newcommand{\bkt}[1]{\mleft[ #1 \mright]}
\newcommand{\brc}[1]{\mleft\{ #1 \mright\}}
\newcommand{\frnorm}[1]{\mleft\lVert #1 \mright\rVert_\mathrm{F}}
\newcommand{\norm}[1]{\mleft\lVert #1 \mright\rVert}
\newcommand{\prn}[1]{\mleft( #1 \mright)}

\newcommand{\bigprn}[1]{\bigl( #1 \bigr)}
\newcommand{\Bigprn}[1]{\Bigl( #1 \Bigr)}
\newcommand{\biggprn}[1]{\biggl( #1 \biggr)}
\newcommand{\Biggprn}[1]{\Biggl( #1 \Biggr)}
\newcommand{\bigbkt}[1]{\bigl[ #1 \bigr]}
\newcommand{\Bigbkt}[1]{\Bigl[ #1 \Bigr]}
\newcommand{\biggbkt}[1]{\biggl[ #1 \biggr]}

\newcommand{\bigbrc}[1]{\bigl\{ #1 \bigr\}}
\newcommand{\biggbrc}[1]{\biggl\{ #1 \biggr\}}
\newcommand{\Bignorm}[1]{\Bigl\Vert #1 \Bigr\Vert}

\newcommand{\mbkt}[1]{\mleft[ #1 \mright]}
\newcommand{\mbrc}[1]{\mleft\{ #1 \mright\}}
\newcommand{\mprn}[1]{\mleft( #1 \mright)}

\newcommand{\ibrc}[1]{\{ #1 \}}
\newcommand{\inorm}[1]{\lVert #1 \rVert}
\newcommand{\iprn}[1]{( #1 )}
\newcommand{\iabs}[1]{\lvert #1 \rvert}
\newcommand{\ibkt}[1]{[ #1 ]}

\DeclareMathOperator*{\cl}{cl}
\DeclareMathOperator*{\proj}{proj}
\DeclareMathOperator*{\rad}{rad}

\newcommand{\E}[1]{\mathbb{E} \mbkt{#1}}
\newcommand{\iE}[1]{\mathbb{E} \bigbkt{#1}}

\newcommand{\Ewrt}[2]{\mathop{\mathbb{E}}_{#1} \mbkt{#2}}
\newcommand{\Iv}[1]{\mathbbm{1} \mbrc{#1}}
\renewcommand{\P}[1]{\mathbb{P} \mprn{#1}}
\newcommand{\iP}[1]{\mathbb{P} \iprn{#1}}
\newcommand{\Pwrt}[2]{\mathop{\mathbb{P}}_{#1} \mprn{#2}}

\newcommand{\tallinf}{\mathop{\mathrm{inf}\vphantom{\mathrm{sup}}}}

\renewcommand{\complement}{\mathrm{c}}

\newcommand{\given}{\:\middle\vert\:}

\DeclareMathOperator*{\latinf}{\tilde{\bigwedge}}
\DeclareMathOperator*{\ilatinf}{\tilde{\wedge}}
\newcommand{\stackclap}[2]{\stackrel{\mathclap{#1}}{#2}}   %
\renewcommand{\top}{\mathrm{T}}

\newcommand{\circled}[1]{
    \tikz[baseline=(char.base)]{\node[shape=circle, draw, inner sep=0.5pt] (char) {#1};}
}

\DeclareSymbolFont{tipa}{T3}{cmr}{m}{n}
\DeclareMathAccent{\invbreve}{\mathalpha}{tipa}{16}

\makeatletter
\def\@tvsp{\mathchoice{{}\mkern-4.5mu}{{}\mkern-4.5mu}{{}\mkern-2.5mu}{}}
\def\ltVert{\mleft|\@tvsp\left|\@tvsp\left|}
\def\rtVert{\right|\@tvsp\right|\@tvsp\mright|}
\newcommand{\tnorm}[1]{\ltVert #1 \rtVert}

\newcommand{\itnorm}[1]{|\@tvsp\mathinner{|\@tvsp\mathinner{|#1|}\@tvsp|}\@tvsp|}
\makeatother

\newcommand{\tdots}{\mathinner{\ldotp\mkern-2mu\ldotp\mkern-2mu\ldotp}}
\newcommand{\tcdots}{\mathinner{\cdotp\mkern-2mu\cdotp\mkern-2mu\cdotp}}

\newcommand{\wildcard}{{\mkern2mu\cdot\mkern2mu}}
\newcommand{\lwildcard}{{\mkern2mu\cdot}}
\newcommand{\rwildcard}{{\cdot\mkern2mu}}
\newcommand{\diff}{\mathop{}\!\mathrm{d}}
\newcommand{\idiff}{\mathrm{d}}
\newcommand{\ibigwedge}{\wedge}

\begin{document}

\pagenumbering{arabic}
\bstctlcite{IEEEexample:BSTcontrol}

\title{Doeblin Curves}
\author{Dongmin~Lee,~William~Lu,~Anuran~Makur, and~Japneet~Singh
    \thanks{The author ordering is alphabetical. This work was supported by the National Science Foundation (NSF) CAREER Award under Grant CCF-2337808. An earlier version of this paper was presented in part at the 2024 IEEE International Symposium on Information Theory (ISIT) \cite{ConferenceVersion}.}
    \thanks{Dongmin Lee is with the Department of Computer Science, Purdue University, West Lafayette, IN 47907, USA (e-mail: lee4818@purdue.edu).}
    \thanks{William Lu is with the Department of Computer Science, Purdue University, West Lafayette, IN 47907, USA (e-mail: lu909@purdue.edu).}
    \thanks{Anuran Makur is with the Department of Computer Science and the Elmore Family School of Electrical and Computer Engineering, Purdue University, West Lafayette, IN 47907, USA (e-mail: amakur@purdue.edu).}
    \thanks{Japneet Singh is with the Elmore Family School of Electrical and Computer Engineering, Purdue University, West Lafayette, IN 47907, USA (e-mail: sing1041@purdue.edu).}
}

\maketitle

\begin{abstract}
    Recent research on Doeblin coefficients has shed light on their usefulness as a multi-way generalization of the Dobrushin contraction coefficient for TV distance, in a separate vein from their classic role in the theory of Markov chain ergodicity. However, strong conditions, such as being bounded away from $0$, are typically necessary for Doeblin coefficients to establish the existence of information contraction. Building on recently formulated concepts of nonlinear information contraction, we aim to propose a finer-grained Doeblin-based characterization of multi-way contraction behavior which yields non-vacuous contraction guarantees even for channels whose Doeblin coefficient is $0$. To this end, we introduce the notion of a \emph{Doeblin curve}\textemdash a nonlinear function which quantifies the contraction behavior of a Markov kernel on collections of input distributions at specific levels of divergence and power. Through the course of our analysis, we develop a new variational characterization of Doeblin coefficients, present several properties of Doeblin curves, define several versions of power-constrained Doeblin curves, and derive upper and lower bounds using our aforementioned variational characterization. We then utilize these results in diverse areas, including generalization bounds for noisy iterative optimization, error bounds for reliable computation with noisy circuits, and differential privacy guarantees for online iterative algorithms. In particular, we extend results in these areas to broader domains or group settings, leveraging Doeblin curves to reveal finer-grained contraction phenomena than Doeblin coefficients.
\end{abstract}


\section{Introduction}

Analyzing the contraction properties of channels or Markov kernels is a fundamental problem in information theory stemming back to the celebrated \emph{data processing inequality}, which states that the $f$-divergence between two distributions does not increase after they are passed through a Markov kernel \cite{Csiszar1967, Csiszar1972}. Data processing inequalities can be strengthened using contraction coefficients \cite{Hirschfeld_1935, Dobrushin1956, Renyi1959, AhlswedeGacs1976, CohenKempermanZbaganu1998, evans1999signal}, which quantify the degree to which two distributions contract after being pushed forward through a Markov kernel (also see more recent work \cite{Anantharametal2014, MakurZheng2015, Raginsky2016, PolyanskiyWu2017, CalmonPolyanskiyWu2018, Makur2019, MakurZheng2020, Huangetal2023, LeeMakur2024} and the references therein). In some cases, the quantification of this contraction behavior can be extended to an arbitrary collection of distributions represented by another Markov kernel, e.g., using Doeblin coefficients \cite{MakurSingh2023b}.

Historically, such contraction analyses were closely tied to investigating the convergence rate of a Markov chain to its steady-state distribution. Initial work in this vein by Doeblin \cite{Doeblin1937,Doeblin1938} established exponential convergence rates for all Markov matrices with strictly positive Doeblin coefficients, while more general weak ergodicity results for inhomogeneous Markov chains have been obtained using similar ideas, e.g., \cite[Lemma 3]{ChestnutContinous}. Furthermore, \cite{MakurSingh2023b} recently showed that Doeblin coefficients may be perceived as an $n$-way generalization of total variation (TV) distance, a view in which the submultiplicativity of Doeblin coefficients corresponds to the data processing inequality for $n$ input distributions.

However, the utility of Doeblin coefficients remains limited in the regime of channels whose Doeblin coefficient is $0$, for which information contraction properties cannot be easily established using existing Doeblin-based techniques. For example, prior work has demonstrated that Markov chains exhibit convergence to stationarity even under much weaker conditions such as drift and local minorization, cf. \cite{Rosenthal2004}, and a strictly positive Doeblin coefficient may be considered a strong assumption. Consequently, contemporary results in Markov process analysis often shed the dependence on Doeblin coefficients altogether, instead being rooted in alternative methods such as spectral gap (or Poincar\'{e}) and logarithmic Sobolev inequalities \cite{MontenegroTetali2006}.

The issue of contraction coefficients taking on trivial values exists in broader information-theoretic settings too \cite{PolyanskiyWu2016,CalmonPolyanskiyWu2018}. For example, Dobrushin's contraction coefficient for TV distance often takes the value $1$ which does not demonstrate any meaningful information contraction \cite{PolyanskiyWu2016}. To circumvent this issue, \cite{PolyanskiyWu2016} develops nonlinear functions called \emph{Dobrushin curves} that capture TV contraction even if the contraction coefficient is $1$. Propelled by this line of inquiry, we seek to understand Doeblin-based information contraction properties of Markov kernels under general conditions, including the case of zero Doeblin coefficient. To this end, we introduce the notion of a \emph{Doeblin curve} in this work (also see \cite{ConferenceVersion}). This nonlinear function quantifies the degree of contraction exhibited by its associated Markov kernel $K$, when applied to an arbitrary collection of input distributions represented by another kernel $W$. In contrast to conventional Doeblin coefficients, the nonlinearity of the Doeblin curve captures a more nuanced view of the contraction behavior of $K$.

\subsection{Main Contributions and Outline}

We briefly outline the structure of our paper and list our main contributions. Following a discussion of preliminaries, in \cref{subsection:variational-characterization-of-doeblin-coefficients}, we present a new variational characterization of Doeblin coefficients as an infimum over arbitrary partitions of the output space. Next, we introduce and formally define the Doeblin curve of a general Markov kernel and enumerate basic properties in \cref{subsection:doeblin-curves}. We discuss power-constrained versions in \cref{subsection:power-constrained-doeblin-curves}, derive upper and lower bounds using the aforementioned variational characterization in \cref{subsection:bounds-on-power-constrained-doeblin-curves}, and present examples of kernels with closed-form Doeblin curves. Notably, in contrast to much of the literature on information contraction and Doeblin coefficients, this work is developed in the context of Markov kernels over general Polish spaces (as opposed to Markov matrices over finite sets).

Building on these theoretical results, we discuss several applications of Doeblin curves in \cref{section:main-results-on-applications}. In \cref{subsection:bounds-on-generalization-error}, we employ Doeblin curves to derive generalization error bounds for noisy iterative algorithms operating on feasible sets with infinite diameter, a situation in which Doeblin coefficients fail to produce non-trivial contraction bounds. In \cref{subsection:reliable-computation-using-noisy-q-ary-gates}, we discuss lower bounds for reliable computation using circuits of noisy $q$-ary gates. Lastly, in \cref{subsection:relation-to-differential-privacy-and-privacy-amplification}, we utilize our variational characterization of Doeblin coefficients to extend the definition of $(\epsilon, \delta)$-local differential privacy (LDP) to a group setting, and provide differential privacy guarantees for noisy iterative algorithms in terms of the Doeblin curve of the privacy mechanism.

We provide proofs of our main results in \cref{section:proofs-of-main-results-on-doeblin-curves} and proofs for the aforementioned applications in \cref{section:proofs-of-main-results-on-applications}. We defer further technical miscellany to \cref{appendix:technical-lemmata,appendix:variational-characterization-under-equicontinuity,appendix:markov-kernel-examples} and relate Doeblin curves back to the classic setting of Markov chain ergodicity in \cref{appendix:doeblin-curves-and-markov-chains}.

\subsection{Related Literature}

We summarize the prior literature on Doeblin coefficients, multi-way divergence metrics, and strong data processing inequalities (SDPIs). At the outset, Doeblin's seminal work \cite{Doeblin1937,Doeblin1938} introduced coupling as a technique for establishing uniform exponential convergence rates of Markov chains with respect to TV distance (also see \cite{Kimbleton1969,BhattacharyaWaymire2001,Chestnut2010}), and minorization as a technique to prove the weak ergodicity of inhomogeneous Markov chains (also see \cite{ChestnutContinous}). Notably, Doeblin coefficients were extremal minorization constants in such techniques. Subsequent results \cite[Theorem 3.1]{BhattacharyaWaymire2001}, \cite[Lemma 5]{GohariGunluKramer2020}, \cite[Section IV-D]{Makur2020} established the equivalence between Doeblin minorization and degradation by erasure channels, akin to the view of contraction coefficients for Kullback-Leibler (KL) divergence as domination by erasure channels under the less noisy preorder \cite{PolyanskiyWu2017}. Doeblin coefficients also find relevance in various machine learning and applied probability problems, such as change detection algorithms \cite{chen2022change}, regret bounds in multi-armed bandit problems \cite{Vrettos2020}, Markov chain Monte Carlo methods \cite{rosenthal1995}, analysis of mixing times \cite{doeblin1940elements}, and estimation of entropy rates \cite{kontoyiannis1994}. Previous work on strong data processing inequalities \cite[Remark 3.2]{Raginsky2016}, \cite[Section I-D]{MakurPolyanskiy2018} has also shown how the Doeblin minorization condition yields upper bounds on contraction coefficients for $f$-divergences.

Extending the notion of $f$-divergence to compare three or more distributions simultaneously is another key avenue of research. For example, \cite{Gushchin2008,Williamson2022} develop some general approaches for such extensions. Moreover, \cite{MakurSingh2023b} established several properties of Doeblin coefficients (including geometric aspects, simultaneous and maximal coupling characterizations, and contraction properties over Bayesian networks), and noted that Doeblin coefficients may be perceived as a multi-way generalization of TV distance. (Interestingly, max-Doeblin coefficients \cite{MakurSingh2023b}, or maximal leakage \cite{IssaWagnerKamath2020}, also share many of these properties \cite{MakurSingh2025}.)

Much like Doeblin coefficients, the Dobrushin contraction coefficient for TV distance also plays a seminal role in the analysis of Markov kernels. Both the Doeblin and Dobrushin coefficients share key properties such as submultiplicativity \cite{PolyanskiyWu2017,MakurSingh2023b}, and the latter enjoys utility in applications such as generalization error bounds \cite{CalmonGeneralizationError}, differential privacy guarantees \cite{asoodeh2021local}, and analysis of reinforcement learning algorithms \cite{RLDobrushinApplication}. As noted earlier, \cite{PolyanskiyWu2016} introduced Dobrushin curves to quantify information dissipation in channels with average input cost constraint even when their Dobrushin coefficient is trivially $1$. Akin to this nonlinear perspective, other studies \cite{PolyanskiyWu2017,CalmonPolyanskiyWu2018} developed strong data processing inequalities which capture similar fine-grained contraction behavior. In this paper, we extend these ideas to the realm of Doeblin coefficients, demonstrating that analogous nonlinear enhancements can be devised to broaden the efficacy of Doeblin-based approaches for characterizing contraction behavior in general settings.

\subsection{Preliminaries}

We define relevant notation, review technical preliminaries, and provide several pertinent characterizations of Doeblin coefficients used in our paper. Let $\mathbb{N} \triangleq \ibrc{1, 2, \tdots}$ denote the natural numbers starting from $1$. Let $\mathbb{R}_+$ denote the non-negative real numbers. Let $[n] \triangleq \ibrc{1, \tdots, n}$ denote integer intervals. Given a function $f\colon \mathcal{X} \rightarrow \mathbb{R}$ defined on a convex domain $\mathcal{X} \subseteq \mathbb{R}$, let $\invbreve{f}\colon \mathcal{X} \rightarrow \mathbb{R}$ denote its upper concave envelope (i.e., the pointwise infimum of all its concave upper bounds, or the ``convex hull of $f$'' \cite[Chapter B, Proposition 2.5.1, Definition 2.5.3]{HiriartUrrutyLemarechal2001}).   %
Given a Banach space $(\mathcal{X}, \inorm{\wildcard})$, denote the \emph{diameter} and \emph{Chebyshev radius} of a (multi-)subset $\mathcal{S} \subseteq \mathcal{X}$ as   \begin{equation}
    \tnorm{\mathcal{S}}_\infty \smash{\triangleq} \sup_{x, y \in \mathcal{S}} \norm{x - y} , \quad
    \rad(\mathcal{S}) \triangleq \tallinf_{a \in \mathcal{X}} \sup_{x \in \mathcal{S}} \norm{x - a} . \label{eq:diameter}
\end{equation}\par
All measures considered in this paper are finite.   %
Let $(\mathcal{X}, \mathcal{F}_\mathcal{X})$ and $(\mathcal{Y}, \mathcal{F}_\mathcal{Y})$ be (measurable) Polish spaces equipped with Borel $\sigma$-algebras, and let $\mathscr{P}$ denote the set of probability measures on $(\mathcal{Y}, \mathcal{F}_\mathcal{Y})$. (Throughout this paper, all Polish spaces considered are assumed to have cardinality at least $2$.) A \emph{Markov kernel} (or channel) $K$ acting from $\mathcal{X}$ to $\mathcal{Y}$ is a mapping $K\colon \mathcal{X} \times \mathcal{F}_\mathcal{Y} \rightarrow [0, 1]$ such that ${K(\lwildcard \mid x)}$ is a probability measure for each $x \in \mathcal{X}$ and $K(A \mid \rwildcard)$ is a measurable function for each $A \in \mathcal{F}_\mathcal{Y}$. For any scalar $\alpha \in \mathbb{R}$, kernel $K$, and measure $\pi$, we write $K \geq \alpha \pi$ to denote that
\begin{align}
    \forall x \in \mathcal{X},\, \forall A \in \mathcal{F}_\mathcal{Y}, \quad K(A \mid x) \geq \alpha \cdot \pi(A).
\end{align}
Given probability distributions $P$ and $Q$, let $P \otimes Q$ denote their product distribution and let $P * Q$ denote their convolution. Given probability measures $P$ and $Q$ both on $(\mathcal{X}, \mathcal{F}_{\mathcal{X}})$, let the TV distance be denoted $\inorm{P-Q}_{\mathsf{TV}}\triangleq \sup_{A \in \mathcal{F}_\mathcal{X}} |P(A)-Q(A)|$.  Given $\sigma$-algebras $\mathcal{F}_\mathcal{X}$ and $\mathcal{F}_\mathcal{Y}$, let $\mathcal{F}_\mathcal{X} \otimes \mathcal{F}_\mathcal{Y}$ denote their product $\sigma$-algebra. Given measures $P_i: \mathcal{F}_\mathcal{Y} \rightarrow \mathbb{R}_+$ for $i \in [n]$, let $[P_1, \tdots, P_n]\colon [n] \times \mathcal{F}_\mathcal{Y} \rightarrow \mathbb{R}_+$ denote the kernel formed by collecting these measures, i.e.,
\begin{align}
    \forall i \in [n], \, \forall A \in \mathcal{F}_\mathcal{Y}, \quad \bkt{P_1, \tdots, P_n}(A \mid i) \triangleq P_i(A).
\end{align}
Let $\updelta_x$ denote the Dirac measure (i.e., point mass) concentrated at $x$.   %
Let $\Xi$ denote the identity kernel, i.e.,
\begin{align}
    \forall x \in \mathcal{X}, \quad \Xi(\lwildcard \mid x) \triangleq \updelta_x(\wildcard) .
\end{align}

Given measures $\pi, \mu\colon \mathcal{F}_\mathcal{Y} \rightarrow \mathbb{R}_+$, we write $\pi \ll \mu$ to indicate that $\pi$ is dominated by (or \emph{absolutely continuous with respect to}) $\mu$, i.e.,   \begin{math}
    \mu(A) = 0 \Rightarrow \pi(A) = 0
\end{math} for all $A \in \mathcal{F}_\mathcal{Y}$.
When $\pi \ll \mu$ and $\mu$ is $\sigma$-finite, we let $\frac{\idiff\pi}{\idiff\mu}\colon \mathcal{Y} \rightarrow \mathbb{R}_+$ denote the Radon-Nikodym derivative of $\pi$ with respect to $\mu$. We say that a Markov kernel $K\colon \mathcal{X} \times \mathcal{F}_\mathcal{Y} \rightarrow [0, 1]$ is \emph{absolutely continuous} if there exists a $\sigma$-finite measure $\mu\colon \mathcal{F}_\mathcal{Y} \rightarrow \mathbb{R}_+$ such that $K(\lwildcard \mid x) \ll \mu$ for all $x \in \mathcal{X}$.\footnote{Note that such a common dominating measure $\mu$ always exists for any $K$, but we additionally require $\sigma$-finiteness of $\mu$ so that \cite[Chapter V, Theorem 4.44]{Cinlar2011} can be invoked.} For example, note that if $\mathcal{X} = \ibrc{x_1, x_2, \tdots}$ is countably infinite, one such dominating finite measure $\mu$ is
\begin{equation}
    \mu(\wildcard) = \sum_{i=1}^\infty \frac{K(\lwildcard \mid x_i)}{2^i} . \label{eq:common-dominating-measure}
\end{equation}
By Doob's variant of the Radon-Nikodym theorem \cite[Chapter V, Theorem 4.44]{Cinlar2011}, absolute continuity of $K$ implies the existence of a kernel function of Radon-Nikodym derivatives $\frac{\idiff K}{\idiff \mu}(y \mid x)\colon \mathcal{Y} \times \mathcal{X} \rightarrow \mathbb{R}_+$, such that
\begin{align}
    \forall x \in \mathcal{X}, \, \forall A \in \mathcal{F}_\mathcal{Y}, \quad K(A \mid x) = \int_A \frac{\diff K}{\diff \mu}(\lwildcard \mid x) \diff \mu
\end{align}
and $\frac{\idiff K}{\idiff \mu}$ is jointly measurable with respect to the product $\sigma$-algebra $\mathcal{F}_\mathcal{Y} \otimes \mathcal{F}_\mathcal{X}$. We impose this mild regularity condition on $K$ as a precondition for several of our results; it is standard in measure-theoretic treatments of information theory (cf. \cite[Theorem 2.12]{PolyanskiyWu2025}).

Given a collection of measures $\ibrc{\pi_i}_{i \in \mathcal{I}}$ on a Polish space $(\mathcal{Y}, \mathcal{F}_\mathcal{Y})$, where $\mathcal{I}$ is an arbitrary index set, define their \emph{greatest common component} $\ibigwedge_{i \in \mathcal{I}} \pi_i$ \cite[Chapter 3, Section 7.1]{mainbook} as the unique measure on $(\mathcal{Y}, \mathcal{F}_\mathcal{Y})$ given by
\begin{equation}
    \forall A \in \mathcal{F}_\mathcal{Y}, \quad \bigwedge_{i \in \mathcal{I}} \pi_i(A) \triangleq \sup \brc{\nu(A): \forall i \in \mathcal{I}, \, \pi_i \geq \nu}, \label{eq:greatest-common-component}
\end{equation}
where the supremum is taken over all measures $\nu\colon \mathcal{F}_\mathcal{Y} \rightarrow \mathbb{R}_+$ such that $\pi_i \geq \nu$ for each $i \in \mathcal{I}$. We remark that $\ibigwedge_{i \in \mathcal{I}} \pi_i(A) \neq \inf_{i \in \mathcal{I}} \pi_i(A)$ in general, because the pointwise infimum of measures is not a measure in general (as it may fail to satisfy countable additivity). We refer readers to \cite[Chapter 3, Theorem 7.1]{mainbook} for a proof that the pointwise supremum in \cref{eq:greatest-common-component} defines a valid measure.

When we want to emphasize the discrete, finite-dimensional nature of a setting, we denote vectors and matrices with lowercase and uppercase bold letters, respectively. Let $\mathscr{P}_{d-1} \subset \mathbb{R}^d$ be the $(d - 1)$-dimensional probability simplex of \emph{row} pmf vectors, and let $\mathbb{R}_\mathsf{sto}^{d \times d}$ be the set of all $d \times d$ row stochastic matrices. Let $\mathbf{e}_i \in \mathscr{P}_{d-1}$ be the $i$th standard basis \emph{row} vector. Given a matrix $\mathbf{A}$, let $[\mathbf{A}]_{i,j}$ denote its entry at row $i$ and column $j$, and let $[\mathbf{A}]_{\ang{i}} = \mathbf{e}_i \mathbf{A}$ denote its $i$th row, represented as a \emph{row} vector. If $\mathbf{A}$ is square, let $\lambda_i(\mathbf{A})$ denote its $i$th eigenvalue (counting algebraic multiplicity and ordered by descending magnitude). Let $\inorm{\wildcard}_p$ denote the $\ell^p$-norm in Euclidean space.

Given a Markov kernel $K$, define its \emph{Doeblin coefficient} $\tau(K) \in [0, 1]$ and \emph{complementary Doeblin coefficient} $\rho(K) \in [0, 1]$ as
\begin{align}
    \tau(K) &\triangleq \sup \brc{\alpha \in \mathbb{R}: \exists \pi \in \mathscr{P}, \, K \geq \alpha \pi} , \label{eq:doeblin-coefficient} \\
    \rho(K) &\triangleq 1 - \tau(K)  .
\end{align}
The Doeblin coefficient of $K$ may be characterized in terms of its greatest common component as
\begin{equation}
    \tau(K) = \bigwedge_{x \in \mathcal{X}} K(\mathcal{Y} \mid x)  , \label{eq:doeblin-coefficient-gcc}
\end{equation}
because the scalar $\alpha = \ibigwedge_{x \in \mathcal{X}} K(\mathcal{Y} \mid x)$ and the measure $\pi(\wildcard) = \frac{1}{\alpha} \ibigwedge_{x \in \mathcal{X}} K(\lwildcard \mid x)$ achieve the supremum in \cref{eq:doeblin-coefficient} by \cite[Theorem 1]{ChestnutContinous}.\footnote{If $\alpha = 0$, we may take $\pi$ to be any probability measure on $\mathcal{Y}$.} For finite spaces (i.e., row stochastic matrices $\mathbf{K} \in \mathbb{R}_\mathsf{sto}^{m \times n}$), $\tau(\mathbf{K})$ reduces to \emph{Doeblin's coefficient of ergodicity} \cite{MakurSingh2023a}, \cite[Definition 1]{MakurSingh2023b}, \cite[Eq. 11]{ChestnutContinous}, defined as
\begin{equation}
    \tau(\mathbf{K}) = \sum_{j=1}^n \min_{i \in [m]} \bkt{\mathbf{K}}_{i,j} . \label{eq:doeblin-coefficient-matrix}
\end{equation}
This implies that the complementary Doeblin coefficient $\rho(K)$ may be interpreted as a ``generalized TV distance'' among a collection of probability distributions $\{K(\lwildcard \mid x): x \in \mathcal{X}\}$ \cite{MakurSingh2023b}.

\section{Main Results on Doeblin Curves} \label{section:main-results-on-doeblin-curves}

We begin by recalling the \emph{maximal coupling characterization} of Doeblin coefficients \cite[Theorem 2]{MakurSingh2023b} (also see \cite{mainbook}) defined as follows.
Consider a collection of probability measures $\ibrc{P_i}_{i\in\mathcal{I}}$ where $\mathcal{I}$ is a Polish index set and each probability measure $P_i$ is on a respective Polish space $(\mathcal{Y}_i, \mathcal{F}_{\mathcal{Y}_i})$. A $\emph{coupling}$ of probability measures $\ibrc{P_i}_{i\in\mathcal{I}}$ is a collection $\ibrc{Y_i}_{i\in\mathcal{I}}$ of random variables all defined on the common measure space $\otimes_{i \in \mathcal{I}} (\mathcal{Y}_i, \mathcal{F}_{\mathcal{Y}_i})$, equipped with a (joint) probability measure $\mathbb{P}$ such that for all $i \in \mathcal{I}$, the marginal probability law of $Y_i$ on $(\mathcal{Y}_i, \mathcal{F}_{\mathcal{Y}_i})$ is equal to $P_i$. (Since each $(\mathcal{Y}_i, \mathcal{F}_{\mathcal{Y}_i})$ is Polish, results such as the Kolmogorov extension theorem guarantee the existence of measures $\mathbb{P}$ on the infinite collection of random variables $\ibrc{Y_i}_{i \in \mathcal{I}}$ which are consistent with finite dimensional distributions \cite{Kallenberg2021}.) With this groundwork established, the following proposition states the desired maximal coupling characterization for Polish spaces.

\begin{proposition}[Maximal Coupling Characterization of Doeblin Coefficients] \label{proposition:maximal-coupling-characterization-of-doeblin-coefficients}
    Let $K\colon \mathcal{X} \times \mathcal{F}_\mathcal{Y} \rightarrow [0, 1]$ be a Markov kernel defined over Polish spaces $(\mathcal{X}, \mathcal{F}_\mathcal{X})$ and $(\mathcal{Y}, \mathcal{F}_\mathcal{Y})$. Then,
    \begin{equation}
        \tau(K) = \sup_{\mathbb{P}: Y_x \sim K(\cdot \mid x)} \P{\forall x, x' \in \mathcal{X}, \, Y_x = Y_{x'}} , \label{eq:doeblin-coefficient-max-coupling}
    \end{equation}
    where the supremum is taken over all couplings $\{Y_x\}_{x\in\mathcal{X}}$ of the probability measures $\{K(\lwildcard \mid x)\}_{x\in\mathcal{X}}$.
\end{proposition}

\cref{proposition:maximal-coupling-characterization-of-doeblin-coefficients} is proved in \cref{subsection:proofs-of-doeblin-characterizations} for completeness. We remark that for any Markov kernel $K$, the supremum in \cref{eq:doeblin-coefficient-max-coupling} is achieved by the construction in \cref{eq:maximal-coupling}, which we refer to as the ``maximal coupling'' for the remainder of this paper. Next, we present our main results, beginning with a new variational characterization of Doeblin coefficients.

\subsection{Variational Characterization of Doeblin Coefficients} \label{subsection:variational-characterization-of-doeblin-coefficients}

Our first main result is a new variational characterization of Doeblin coefficients, expressed as an infimum over arbitrary partitions of the underlying space. In this sense, our result is comparable to the Gel'fand-Yaglom-Peres variational characterization of KL divergence \cite{gelfand1956} or the extensions to $f$-divergences in \cite{Gilardoni2009}, \cite[Theorem 7.6]{PolyanskiyWu2025}.

\begin{theorem}[Variational Characterization of Doeblin Coefficient] \label{theorem:variational-characterization-of-doeblin-coefficient}
    Let $K\colon \mathcal{X} \times \mathcal{F}_\mathcal{Y} \rightarrow [0, 1]$ be an absolutely continuous Markov kernel defined over Polish spaces $(\mathcal{X}, \mathcal{F}_\mathcal{X})$ and $(\mathcal{Y}, \mathcal{F}_\mathcal{Y})$. Then,
    \begin{equation}
        \tau(K) = \inf_{n \in \mathbb{N}} \inf_{x_1, \dots, x_n \in \mathcal{X}} \inf_{\substack{\text{$n$-partition of $\mathcal{Y}$} \\ A_1, \dots, A_n}} \sum_{i=1}^n K(A_i \mid x_i) ,
    \end{equation}
    where the innermost infimum is taken with respect to all measurable $n$-partitions $A_1, \tdots, A_n \in \mathcal{F}_\mathcal{Y}$, i.e., $A_i \cap A_j = \varnothing$ for all $i \neq j$ and $\bigcup_{i=1}^n A_i = \mathcal{Y}$.
\end{theorem}

\cref{theorem:variational-characterization-of-doeblin-coefficient} is proved in \cref{subsection:proofs-of-doeblin-characterizations}. The proof uses the characterization of Doeblin coefficients in \cref{eq:doeblin-coefficient-gcc} and the Radon-Nikodym theorem to express $\tau(K)$ in terms of the density of the greatest common component of $K$. To relate this quantity to the individual densities $\frac{\idiff K}{\diff \mu}(\lwildcard \mid x)$ for $x \in \mathcal{X}$, we utilize the notion of \emph{lattice infimum} as a measurable analogue of pointwise infimum, so that the lattice infimum of uncountably many measurable functions $\frac{\idiff K}{\idiff \mu}(\lwildcard \mid x)$ (indexed by $x \in \mathcal{X}$) remains measurable. Approximating this uncountable lattice infimum with a countable sequence gives rise to the infima over $n \in \mathbb{N}$ and $x_1, \tdots, x_n \in \mathcal{X}$ in \cref{theorem:variational-characterization-of-doeblin-coefficient}, and we complete the proof by leveraging the following proposition to convert the integral of the pointwise minimum of finitely many densities into an infimum over partitions.

\begin{proposition}[Integral Characterization of Infimum] \label{proposition:integral-characterization-of-infimum}
    Let $K\colon \mathcal{X} \times \mathcal{F}_\mathcal{Y} \rightarrow [0, 1]$ be an absolutely continuous Markov kernel defined over Polish spaces $(\mathcal{X}, \mathcal{F}_\mathcal{X})$ and $(\mathcal{Y}, \mathcal{F}_\mathcal{Y})$. For any fixed $x_1, \tdots, x_n \in \mathcal{X}$ and any fixed constants $\gamma_1, \tdots, \gamma_n \geq 0$, we have
    \begin{equation}
        \inf_{\substack{\text{$n$-partition of $\mathcal{Y}$} \\ A_1, \dots, A_n}} \sum_{i=1}^n \gamma_i K(A_i \mid x_i) = \int_{\mathcal{Y}} \prn{\min_{i \in [n]} \gamma_i \frac{\diff K}{\diff \mu} (\lwildcard \mid x_i)} \diff \mu  .
    \end{equation}
\end{proposition}

\cref{proposition:integral-characterization-of-infimum} is proved in \cref{subsection:proofs-of-doeblin-characterizations} and is also used in the proofs of subsequent results. We remark that \cref{proposition:integral-characterization-of-infimum} may be interpreted as a special case of \cref{theorem:variational-characterization-of-doeblin-coefficient} where $\mathcal{X}$ is finite. To see this, observe that whereas the infima over $n \in \mathbb{N}$ and $x_1, \tdots, x_n \in \mathcal{X}$ in \cref{theorem:variational-characterization-of-doeblin-coefficient} define the finite subset of $\mathcal{X}$ included in the sum (and thus the infinite subset of $\mathcal{X}$ omitted from the sum), they are unnecessary when $\mathcal{X}$ is finite, as any $x_i \in \mathcal{X}$ may be ``omitted'' from the sum by simply taking the respective $A_i$ to be $\emptyset$. We note that by imposing stronger regularity conditions on $K$, we may prove \cref{theorem:variational-characterization-of-doeblin-coefficient} without the use of lattice infima; we refer interested readers to \cref{appendix:variational-characterization-under-equicontinuity} for an analytically simpler argument under equicontinuity assumptions on $\frac{\idiff K}{\idiff \mu}$.

Throughout the rest of the paper, we find it helpful to define a notion of \emph{constrained Doeblin coefficient}\footnote{The constrained Doeblin coefficient $\tau_\mathcal{S}(K)$ is essentially the (unconstrained) Doeblin coefficient $\tau(K')$ of the kernel $K'$ obtained by restricting the input space of $K$ to $\mathcal{S}$. Hence, \cref{theorem:variational-characterization-of-doeblin-coefficient} also holds for constrained Doeblin coefficients by taking the middle infimum over $x_1, \tdots, x_n \in \mathcal{S}$ instead of $\mathcal{X}$.} for a Markov kernel $K\colon \mathcal{X} \times \mathcal{F}_\mathcal{Y} \rightarrow [0, 1]$ and a measurable set $\mathcal{S} \in \mathcal{F}_\mathcal{X}$ as   \begin{align}
    \tau_\mathcal{S}(K) &\triangleq \sup \bigbrc{\alpha \in \mathbb{R}: \exists \pi \in \mathscr{P}, \, \forall x \in \mathcal{S}, \, K(\lwildcard \mid x) \geq \alpha \pi} , \\\rho_\mathcal{S}(K) &\triangleq 1 - \tau_\mathcal{S}(K) .
\end{align}
We remark that $\rho_\mathcal{S}(K)$ is monotonically non-decreasing in $\mathcal{S}$, i.e., for any two sets $\mathcal{R}, \mathcal{S} \in \mathcal{F}_\mathcal{X}$ such that $\mathcal{R} \subseteq \mathcal{S}$, we have $\rho_\mathcal{R}(K) \leq \rho_\mathcal{S}(K)$.

\subsection{Doeblin Curves} \label{subsection:doeblin-curves}

As our centerpiece contribution, we introduce the notion of a Doeblin curve to capture contraction properties of Markov kernels whose Doeblin coefficient is $0$. For any Polish space $(\mathcal{U}, \mathcal{F}_\mathcal{U})$, define the composition of two kernels $W\colon \mathcal{U} \times \mathcal{F}_\mathcal{X} \rightarrow \mathbb{R}_+$ and $K\colon \mathcal{X} \times \mathcal{F}_\mathcal{Y} \rightarrow \mathbb{R}_+$ as the kernel $WK\colon \mathcal{U} \times \mathcal{F}_\mathcal{Y} \rightarrow \mathbb{R}_+$ given by
\begin{equation}
    \forall u \in \mathcal{U}, \, \forall A \in \mathcal{F}_\mathcal{Y}, \quad WK(A \mid u) \triangleq\! \int_\mathcal{X}\!K(A \mid x) \, W(\diff x \mid u) . \label{eq:kernel-composition}
\end{equation}
Now, we define the \emph{Doeblin curve} of a Markov kernel $K\colon \mathcal{X} \times \mathcal{F}_\mathcal{Y} \rightarrow [0, 1]$ as the function $\mathrm{F}_K(\wildcard ; \mathcal{G})\colon [0, 1] \rightarrow [0, 1]$ given by
\begin{align}   %
    \mathrm{F}_K(t; \mathcal{G}) \triangleq \sup \brc{\rho(WK): \rho(W) \leq t, \, W \in \mathcal{G}}, \label{eq:doeblin-curve}
\end{align}
where the supremum is taken over all Markov kernels $W\colon \mathcal{U} \times \mathcal{F}_\mathcal{X} \rightarrow [0, 1]$ from all Polish spaces $(\mathcal{U}, \mathcal{F}_\mathcal{U})$, such that $W$ belongs to some non-empty constraint set $\mathcal{G}$. We remark that Doeblin curves generalize the notion of Dobrushin curves in \cite{PolyanskiyWu2016}, just as Doeblin coefficients generalize contraction coefficients for TV distance. When $\mathcal{G}$ only contains Markov kernels $W: \mathcal{U} \times \mathcal{F}_\mathcal{X} \rightarrow [0, 1]$ where $|\mathcal{U}| = 2$, the Doeblin curve $\mathrm{F}_K(\wildcard ; \mathcal{G})$ reduces to the Dobrushin curve of $K$, $F_{\mathsf{TV}}(t)\triangleq \sup\{\inorm{K(\lwildcard \mid P)- K(\lwildcard \mid Q)}_{\mathsf{TV}}: \inorm{P-Q}_{\mathsf{TV}} \leq t, (P, Q) \in \mathcal{G}\}$ \cite[Equation (32)]{PolyanskiyWu2016} (note that a Markov kernel $W\colon \mathcal{U} \times \mathcal{F}_\mathcal{X} \rightarrow [0, 1]$ is an $U$-indexed family of probability distributions with cardinality $|\mathcal{U}|$, thus a pair of probability distributions $(P, Q)$ each supported on $\mathcal{X}$ is equivalent to a binary-source Markov kernel $W: \{0,1\}\times \mathcal{F}_{\mathcal{X}}\to [0,1]$ with $\rho(W)=\inorm{P-Q}_{\mathsf{TV}}$).

We define the joint range of the input and output complementary Doeblin coefficients as   $
    \mathfrak{F}(K; \mathcal{G}) \triangleq \cl \ibrc{(\rho(W), \rho(WK)):\allowbreak W \in \mathcal{G}} \subseteq [0, 1]^2
$,
where $\cl$ denotes closure. Specifically, the joint range $\mathfrak{F}(K; \mathcal{G})$ is contained within the set $\ibrc{(t, y) \in [0, 1]^2: y \leq \mathrm{F}_K(t; \mathcal{G})}$. For any kernel $K$, the Doeblin curve $\mathrm{F}_K(t; \mathcal{G})$ is monotonically non-decreasing in $t$ for any fixed constraint set $\mathcal{G}$. Also, the Doeblin curve $\mathrm{F}_K(t; \mathcal{G})$ and joint range $\mathfrak{F}(K; \mathcal{G})$ are monotonically non-decreasing in $\mathcal{G}$, i.e., for any two sets of kernels $\mathcal{G}, \mathcal{H}$ such that $\mathcal{G} \subseteq \mathcal{H}$, we have $\mathrm{F}_K(\wildcard; \mathcal{G}) \leq \mathrm{F}_K(\wildcard ; \mathcal{H})$ and $\mathfrak{F}(K; \mathcal{G}) \subseteq \mathfrak{F}(K; \mathcal{H})$.

By submultiplicativity of complementary Doeblin coefficients \cite[Section 4]{ChestnutContinous}, \cite[Theorem 1]{MakurSingh2023b}, we have $\rho(WK) \leq \rho(W) \rho(K)$. Hence, the Doeblin curve and joint range satisfy the outer bounds   
\begin{align}
    \mathrm{F}_K(t; \mathcal{G}) &\leq \rho(K) \, t , \\
    \mathfrak{F}(K; \mathcal{G}) &\subseteq \ibrc{(t, y) \in [0, 1]^2: y \leq \rho(K) \, t} , \label{eq:doeblin-curve-submultiplicativity-bound}
\end{align}
and since $\rho(K) \leq 1$, we have   \begin{equation}
    \mathrm{F}_K(t; \mathcal{G}) \leq t  , \quad
    \mathfrak{F}(K; \mathcal{G}) \subseteq \ibrc{(t, y) \in [0, 1]^2: y \leq t} . \label{eq:doeblin-curve-trivial-bound}
\end{equation}
\cref{figure:markov-matrix-joint-range} presents the numerically simulated joint range for a discrete memoryless channel $\mathbf{K} \in \mathbb{R}_\mathsf{sto}^{5 \times 5}$, comprising 1 million instances of $\mathbf{W}$ sampled uniformly from the set of $5 \times 5$ row stochastic matrices. The blue dashed line depicts the Doeblin curve of $\mathbf{K}$ obtained analytically from \cref{proposition:properties-of-doeblin-curves}, Part 2. The blue dots represent numerically sampled points $(\rho(\mathbf{W}), \rho(\mathbf{W} \mathbf{K})) \in \mathfrak{F}(\mathbf{K}; \mathbb{R}_\mathsf{sto}^{5 \times 5})$ in the joint range of $\mathbf{K}$. The figure shows that when $\mathbf{W}$ is uniformly sampled, such as in settings where channel inputs do not inherently tend towards degenerate regions of the probability simplex, the corresponding points in the joint range indicate noticeably more information contraction than the worst-case bound given by the Doeblin curve would suggest.

\begin{figure}
    \centering           \includegraphics[trim=1.75cm 0.25cm 2cm 0.25cm, clip, width=0.9\columnwidth]{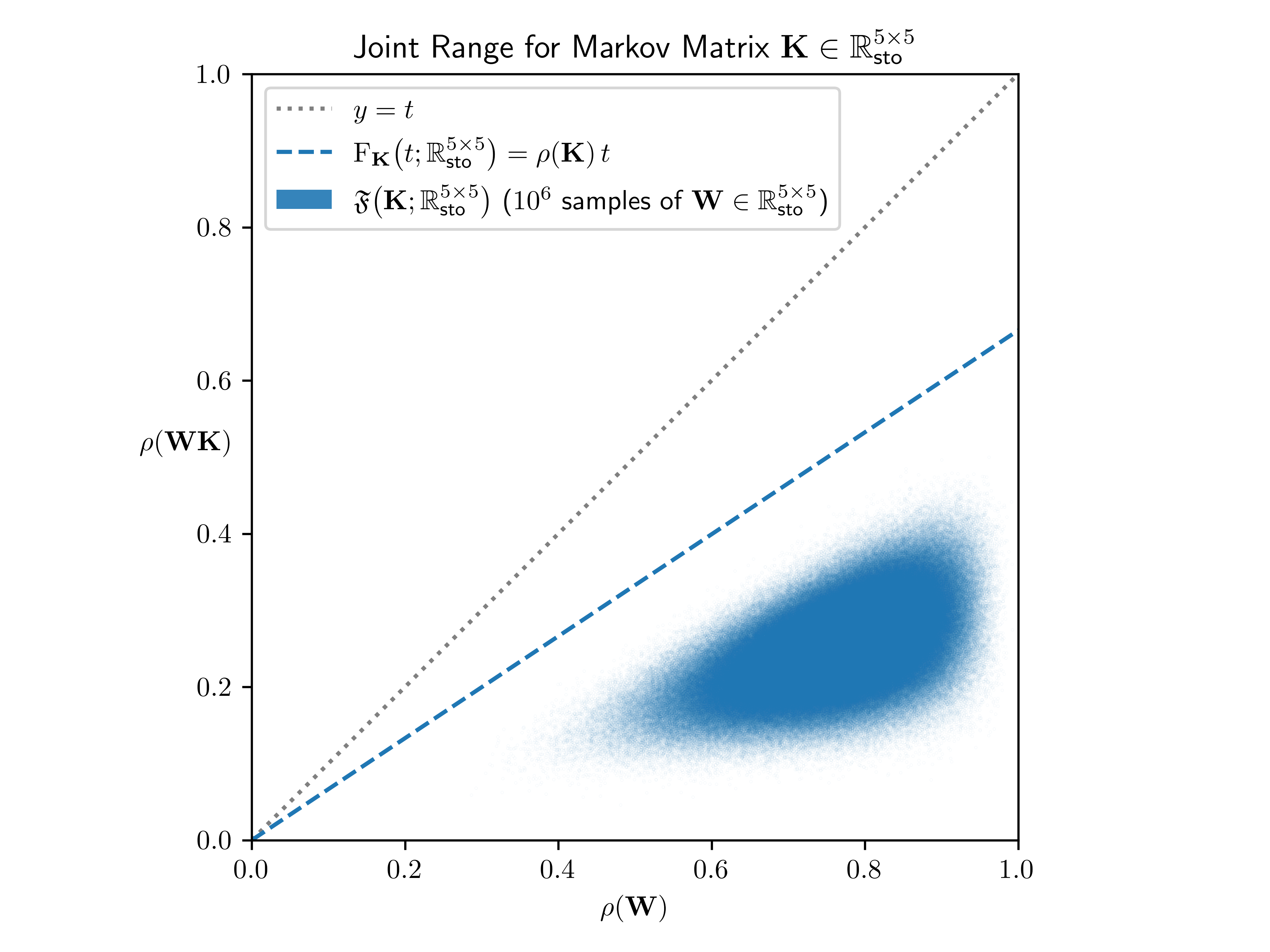}
    \caption{Randomly sampled joint range $\mathfrak{F}(\mathbf{K}; \mathbb{R}_\mathsf{sto}^{5 \times 5})$ for a $5 \times 5$ Markov matrix $\mathbf{K}$, shaded in blue.}
    \label{figure:markov-matrix-joint-range}
\end{figure}

We enumerate several additional properties of Doeblin curves in the following proposition.\footnote{This proposition supersedes \cite[Proposition 3]{ConferenceVersion}.} For brevity, we call a Markov kernel $W\colon \mathcal{U} \times \mathcal{F}_\mathcal{X} \rightarrow [0, 1]$ a \emph{constant kernel} if there exists a fixed probability measure $\pi\colon \mathcal{F}_\mathcal{X} \rightarrow [0, 1]$ such that $W(\lwildcard \mid u) = \pi(\wildcard)$ for all $u \in \mathcal{U}$.

\begin{proposition}[Properties of Doeblin Curves] \label{proposition:properties-of-doeblin-curves}
    The Doeblin curve defined in \cref{eq:doeblin-curve} satisfies the following properties:
    \begin{enumerate}
        \item \emph{(Data processing property)} Given Markov kernels $K_1$ and $K_2$ and constraint sets $\mathcal{G}_1$ and $\mathcal{G}_2$, the Doeblin curve of the composite channel represented by the diagram   \begin{math}
            \mathcal{X}_1 \overset{K_1}{\longrightarrow} \mathcal{Y}_1 = \mathcal{X}_2 \overset{K_2}{\longrightarrow} \mathcal{Y}_2 \,
        \end{math},
        with constraint set $\mathcal{G}$ consisting of all kernels $W$ such that $W \in \mathcal{G}_1$ and $WK_1 \in \mathcal{G}_2$, satisfies   \begin{math}
            \mathrm{F}_{K_1 K_2}(t; \mathcal{G}) \leq \mathrm{F}_{K_2} \mprn{\mathrm{F}_{K_1} \mprn{t; \mathcal{G}_1}; \mathcal{G}_2}
        \end{math}.
        \item \emph{(Sharpness)} For any Markov kernel $K\colon \mathcal{X} \times \mathcal{F}_\mathcal{Y} \rightarrow [0, 1]$ and any convex constraint set $\mathcal{G}$ containing the identity kernel $\Xi\colon \mathcal{X} \times \mathcal{F}_\mathcal{X} \rightarrow [0, 1]$ from $\mathcal{X}$ to $\mathcal{X}$, i.e.,   %
          \begin{math}
            \Xi(\lwildcard \mid u) = \updelta_u(\wildcard)
        \end{math} for all $u \in \mathcal{X}$,
        and a constant kernel $\bar{W}\colon \mathcal{X} \times \mathcal{F}_\mathcal{X} \rightarrow [0, 1]$, i.e.,   \begin{math}
            \bar{W}(\lwildcard \mid u) = \pi(\wildcard)
        \end{math} for all $u \in \mathcal{X}$
        for some fixed probability measure $\pi\colon \mathcal{F}_\mathcal{X} \rightarrow [0, 1]$, the Doeblin curve $\mathrm{F}_K(t; \mathcal{G})$ achieves the bound in \cref{eq:doeblin-curve-submultiplicativity-bound} with equality, i.e.,   \begin{math}
            \mathrm{F}_K(t; \mathcal{G}) = \rho(K) \, t
        \end{math}.
        \item \emph{(Super-homogeneity)} For any Markov kernel $K\colon \mathcal{X} \times \mathcal{F}_\mathcal{Y} \rightarrow [0, 1]$ and any convex constraint set $\mathcal{G}$ containing a constant kernel $\bar{W}\colon \mathcal{U} \times \mathcal{F}_\mathcal{X} \rightarrow [0, 1]$, the Doeblin curve $\mathrm{F}_K(t; \mathcal{G})$ satisfies   \begin{math}
            \mathrm{F}_K(\lambda t; \mathcal{G}) \geq \lambda \mathrm{F}_K(t; \mathcal{G})
        \end{math} for all $\lambda \in [0, 1]$ and $t \in [0, 1]$,
        or equivalently, the function $t \mapsto \mathrm{F}_K(t; \mathcal{G}) / t$ is non-increasing.   %
        \item \emph{(Lipschitz continuity)} For any Markov kernel $K\colon \mathcal{X} \times \mathcal{F}_\mathcal{Y} \rightarrow [0, 1]$ and any convex constraint set $\mathcal{G}$ containing a constant kernel $\bar{W}\colon \mathcal{U} \times \mathcal{F}_\mathcal{X} \rightarrow [0, 1]$, the Doeblin curve $\mathrm{F}_K(t; \mathcal{G})$ is $1$-Lipschitz continuous in $t$.
    \end{enumerate}
\end{proposition}

\cref{proposition:properties-of-doeblin-curves} is proved in \cref{subsection:proofs-of-doeblin-curve-properties}, primarily by using the fact that convex combinations of any Markov kernel with a constant kernel linearly interpolate $\rho$. We emphasize that the super-homogeneity property (\cref{proposition:properties-of-doeblin-curves}, Part 3) does \emph{not} necessarily imply that $\mathrm{F}_K$ is concave, and the concavity of $\mathrm{F}_K$ is also unknown in the Dobrushin case \cite[Remark 1]{PolyanskiyWu2016}.   %

Although the joint range $\mathfrak{F}(K; \mathcal{G})$ is contained within the area under the Doeblin curve $\mathrm{F}_K(\wildcard ; \mathcal{G})$ due to the definition of $\mathrm{F}_K$ as a supremum, it is \emph{not} necessarily the entire area under $\mathrm{F}_K$, as shown by the following counterexample.

\begin{proposition}[Joint Range of $2 \times n$ Discrete Channels] \label{proposition:joint-range-of-2-by-n-discrete-channels}
    The joint range of any discrete memoryless channel $\mathbf{K} \in \mathbb{R}_\mathsf{sto}^{2 \times n}$, considering inputs $\mathbf{W} \in \mathbb{R}_\mathsf{sto}^{m \times 2}$ with any number of rows $m \geq 2$, is the line   \begin{math}
        \mathfrak{F} \bigprn{\mathbf{K}; \bigcup_{m \geq 2} \mathbb{R}_\mathsf{sto}^{m \times 2}} = \brc{(t, \rho(\mathbf{K}) \, t): t \in [0, 1]}
    \end{math}.
\end{proposition}

\cref{proposition:joint-range-of-2-by-n-discrete-channels} is proved in \cref{subsection:proofs-of-doeblin-curve-properties}.

\subsection{Power-Constrained Doeblin Curves} \label{subsection:power-constrained-doeblin-curves}

Inspired by earlier developments of nonlinear information contraction \cite{PolyanskiyWu2016}, we note that the amount by which input distributions contract after passing through a channel depends on the power level of the input distributions. Hence, power constraints on the input kernel $W$ admit natural and useful classes of constraint sets $\mathcal{G}$ which we consider in our subsequent analysis.
While \cite{PolyanskiyWu2016} (which only considers the $|\mathcal{U}|=2$ case) uses the \emph{average power} $\itnorm{W}_\mathsf{A} \triangleq \E{\int_\mathcal{X} M(\norm{x}) \, W(\idiff x \mid U)}$ (where the expectation is taken over $U\sim\mathsf{Bernoulli}(1/2)$) to measure the power level of a kernel $W\colon \{0,1\}\times \mathcal{F}_{\mathcal{X}} \to [0, 1]$, this definition implicitly assumes the marginal distribution of $U$ and thus does not generalize to kernels with arbitrary $\mathcal{U}$ (especially when $\mathcal{U}$ is not compact, e.g., when $\mathcal{U}=\mathbb{R}$). Instead, we define two notions of power that extremize over $\mathcal{U}$ to avoid depending on the marginal distribution of $U$.

\begin{definition}[Power Level] \label{definition:power-level}
    Let $(\mathcal{X}, \inorm{\wildcard})$ be a separable Banach space equipped with the Borel $\sigma$-algebra $\mathcal{F}_\mathcal{X}$ induced by the norm topology. Given some convex and strictly increasing power function $M\colon \mathbb{R}_+ \rightarrow \mathbb{R}_+$ with $M(0) = 0$, we define the following notions of power level for a Markov kernel $W\colon \mathcal{U} \times \mathcal{F}_\mathcal{X} \rightarrow [0, 1]$.
    \begin{itemize}
        \item \emph{Uniform average power:}
        \begin{equation}
            \tnorm{W}_\mathsf{UA} \triangleq \sup_{u \in \mathcal{U}} \int_\mathcal{X} M(\norm{x}) \, W(\diff x \mid u) . \label{eq:uniform-average-power}
        \end{equation}
        \item \emph{Average extremal power:} Let $\mathbb{P}$ be the maximal coupling of random variables $\ibrc{X_u}_{u \in \mathcal{U}}$ with $X_u \sim W(\lwildcard \mid u)$ for each $u \in \mathcal{U}$, as defined in \cref{eq:maximal-coupling}. Then,
        \begin{equation}
            \tnorm{W}_\mathsf{AE} \triangleq \E{\sup_{u \in \mathcal{U}} M(\norm{X_u})}  , \label{eq:average-extremal-power}
        \end{equation}
        where the expectation is taken with respect to the maximal coupling $\ibrc{X_u}_{u \in \mathcal{U}} \sim \mathbb{P}$, and we only define this notion of power for kernels $W$ such that $\sup_{u \in \mathcal{U}} M(\inorm{X_u})$ is measurable.
    \end{itemize}
\end{definition}

The uniform average power and average extremal power have several desirable properties. For example, both notions upper bound the average power and satisfy the relation $\itnorm{W}_\mathsf{A} \leq \itnorm{W}_\mathsf{UA} \leq \itnorm{W}_\mathsf{AE}$. In particular, in the setting of \cite{PolyanskiyWu2016} (i.e., when $U\sim \mathsf{Bernoulli}(1/2)$), $\itnorm{W}_{\textsf{A}}\leq \itnorm{W}_{\textsf{UA}}\leq\itnorm{W}_{\textsf{AE}}\leq 2\itnorm{W}_{\textsf{A}}$, i.e., both notions approximate $\itnorm{W}_{\textsf{A}}$ by a factor of $2$. Furthermore, the difference between the two definitions, $\itnorm{W}_{\textsf{AE}}-\itnorm{W}_{\textsf{UA}}$, can be bounded under certain conditions, which is central to proving \cref{proposition:upper-bound-on-uniform-average-doeblin-curve-finite,proposition:upper-bound-on-uniform-average-doeblin-curve-totally-bounded}.
Note that the canonical notion of power (cf. \cite[Section 20.1]{PolyanskiyWu2025}) corresponds to taking $\tnorm{W}_\mathsf{UA}$ with $M(z) = z^2$. We provide examples of non-trivial kernels satisfying the average extremal power constraint in \cref{appendix:markov-kernel-examples}.

We write $\mathrm{F}_K^\mathsf{UA}(t; p)$ and $\mathrm{F}_K^\mathsf{AE}(t; p)$ to denote the Doeblin curves where the set $\mathcal{G}$ in \cref{eq:doeblin-curve} consists of all kernels satisfying the respective power constraints $\itnorm{W}_\mathsf{UA} \leq p$ and $\itnorm{W}_\mathsf{AE} \leq p$ for some $p \in [0, \infty)$. Since Doeblin curves are monotonically non-decreasing in the constraint set, we have $\mathrm{F}_K^\mathsf{UA}(t; p) \geq \mathrm{F}_K^\mathsf{AE}(t; p)$. Furthermore, we note that each of the power-constrained Doeblin curves satisfies the following homogeneity property (akin to Dobrushin curves). %

\begin{proposition}[Homogeneity of Power-Constrained Doeblin Curves] \label{proposition:homogeneity-of-power-constrained-doeblin-curves}
    For any Markov kernel $K\colon \mathcal{X} \times \mathcal{F}_\mathcal{Y} \rightarrow [0, 1]$, the power-constrained Doeblin curves $\mathrm{F}_K^\mathsf{UA}$ and $\mathrm{F}_K^\mathsf{AE}$ satisfy   $\mathrm{F}_K^\mathsf{UA}(\lambda t; \lambda p) = \lambda \mathrm{F}_K^\mathsf{UA}(t; p)$ and $\mathrm{F}_K^\mathsf{AE}(\lambda t; \lambda p) = \lambda \mathrm{F}_K^\mathsf{AE}(t; p)$
    for all $\lambda \in \mathbb{R}_+$ such that $\lambda t \leq 1$.
\end{proposition}

\cref{proposition:homogeneity-of-power-constrained-doeblin-curves} is proved in \cref{subsection:proofs-of-doeblin-curve-properties} using similar arguments based on convex combinations with constant kernels as those used in the proof of \cref{proposition:properties-of-doeblin-curves}. In addition, the Doeblin curves of certain classes of kernels such as additive noise channels are invariant under scaling (when the power constraint is scaled accordingly) as the following proposition shows.

\begin{proposition}[Scale Invariance of Power-Constrained Doeblin Curves for Additive Noise Channels]\label{proposition:scale-invariant-doeblin-curves-of-additive-noise-channels}
    Let $\mathcal{X}=\mathcal{Y}=\mathbb{R}^d$. Suppose $K_\sigma\colon\mathcal{X}\times\mathcal{F}_\mathcal{Y}\to[0,1]$ is an additive noise channel parameterized by $\sigma$ such that \begin{equation}K_\sigma(A\mid\mathbf{x}) = \int_A \frac{1}{\sigma^d} f\prn{\frac{\mathbf{y}-\mathbf{x}}{\sigma}} \diff \mathbf{y}\end{equation} where $f\colon\mathbb{R}^d\to\mathbb{R}_+$ is a probability density function. Then, if the power function is $M(z)=z^2$, $\mathrm{F}_{K_{\sigma}}^{\mathsf{UA}}(t; p) = \mathrm{F}_{K_1}^{\mathsf{UA}}(t;p/\sigma^2)$ for all $\sigma>0$ and $p\geq 0$.
\end{proposition}

\cref{proposition:scale-invariant-doeblin-curves-of-additive-noise-channels} is proved in \cref{subsection:proofs-of-doeblin-curve-properties}.

\subsection{Bounds on Power-Constrained Doeblin Curves} \label{subsection:bounds-on-power-constrained-doeblin-curves}

In this subsection, we present upper and lower bounds on the power-constrained Doeblin curves for a general kernel and provide examples of kernels whose Doeblin curves may be computed in closed form. Throughout this analysis, let $(\mathcal{X}, \inorm{\wildcard})$ be a separable Banach space equipped with the Borel $\sigma$-algebra $\mathcal{F}_\mathcal{X}$ induced by the norm topology, let $(\mathcal{Y}, \mathcal{F}_\mathcal{Y})$ be an arbitrary Polish space, and let $K\colon \mathcal{X} \times \mathcal{F}_\mathcal{Y} \rightarrow [0, 1]$ be an absolutely continuous Markov kernel. Let $\mathcal{B}(a, r) \subset \mathcal{X}$ denote the closed ball centered at $a \in \mathcal{X}$ with radius $r$ in the $\|\wildcard\|$-norm. Define functions $\theta\colon \mathcal{X} \times \mathbb{R}_+ \rightarrow [0, 1]$ and $\Theta\colon \mathbb{R}_+ \rightarrow [0, 1]$ as \begin{equation}
    \theta(a, r) \triangleq \rho_{\mathcal{B}(a, r)}(K), \quad \Theta(r) \triangleq \sup_{a \in \mathcal{X}} \theta(a, r) . \label{eq:theta-def}
\end{equation}
We remark that $\Theta$ (and hence the upper concave envelope $\invbreve{\Theta}$, by \cref{lemma:properties-of-upper-concave-envelope}, Part 2) are non-decreasing, because $\rho_\mathcal{S}$ is non-decreasing in $\mathcal{S}$ as previously mentioned. Since $\mathcal{X}$ is a Banach space, there exists an element $\mathbf{0} \in \mathcal{X}$ such that $\inorm{\mathbf{0}} = 0$; we denote this element in bold font to distinguish it from the scalar $0 \in \mathbb{R}$. Let
\begin{equation}
    \gamma \triangleq \sup \brc{\frac{\rad(\mathcal{S})}{\tnorm{\mathcal{S}}_\infty}: \mathcal{S} \subset \mathcal{X}, \, 0 < \tnorm{\mathcal{S}}_\infty < \infty} \label{eq:jung-constant}
\end{equation}
denote the \emph{Jung constant} of $\mathcal{X}$ \cite{Jung1901,Yan2004}, i.e., the tightest multiplicative factor between the Chebyshev radius and diameter of any bounded subset of $\mathcal{X}$.   %
In Euclidean space $\mathcal{X} = \mathbb{R}^d$, we have $\gamma \leq d/(d + 1)$ for a general norm $\inorm{\wildcard}$ \cite[Theorem 6]{bohnenblust1938convex}, and $\gamma = 1/2$ for the $\ell^\infty$-norm \cite[Section 5.3]{PolyanskiyWu2025}.

First, we present upper and lower bounds on the average extremal power-constrained Doeblin curve of $K$ (similar to \cite{PolyanskiyWu2016}).\footnote{The results in \cref{subsection:bounds-on-power-constrained-doeblin-curves} supersede \cite[Theorem 2 and Corollaries 1 and 2]{ConferenceVersion}.}

\begin{theorem}[Bounds on Average Extremal Doeblin Curve] \label{theorem:bounds-on-average-extremal-doeblin-curve}
    The power-constrained Doeblin curve $\mathrm{F}_K^\mathsf{AE}$ satisfies the upper and lower bounds   \begin{math}
        t \, \theta \iprn{\mathbf{0}, M^{-1} \iprn{\ifrac{p}{t}}} \leq \mathrm{F}_K^\mathsf{AE}(t; p) \leq t \, \invbreve{\Theta} \iprn{2 \gamma \, M^{-1} \iprn{\ifrac{p}{t}}}
    \end{math}
    for all $t \in (0, 1]$, where $\theta$ and $\Theta$ are defined in \cref{eq:theta-def}, the Jung constant $\gamma$ is defined in \cref{eq:jung-constant}, and $M$ is the power function.
\end{theorem}

The proof is provided in \cref{subsection:proofs-of-bounds-on-power-constrained-doeblin-curves}, utilizing the variational characterization of Doeblin coefficients from \cref{theorem:variational-characterization-of-doeblin-coefficient} and the maximal coupling characterization from \cref{proposition:maximal-coupling-characterization-of-doeblin-coefficients}. %

If $K$ acts as a convolution operator on $\mathbb{R}^d$ (e.g., additive noise), then $\theta(a, s)$ is independent of $a$, namely $\theta(a, s) = \Theta(s)$ for all $a \in \mathcal{X}$. This immediately leads to the following counterpart to \cite[Corollary 5]{PolyanskiyWu2016}.

\begin{corollary}[Average Extremal Doeblin Curves of Convolution Kernels] \label{corollary:average-extremal-doeblin-curves-of-convolution-kernels}
    For any kernel $K$ which acts as a convolution operator on $\mathcal{X} = \mathcal{Y} = \mathbb{R}^d$, and for which $\Theta$ (defined in \cref{eq:theta-def}) is concave, we have  \begin{math}
        t \, \Theta \iprn{M^{-1} \iprn{\ifrac{p}{t}}} \leq \mathrm{F}_K^\mathsf{AE}(t; p) \leq t \, \Theta \iprn{(\ifracb{2d}{d+1}) \, M^{-1} \iprn{\ifrac{p}{t}}}
    \end{math} for all $t\in(0,1]$, where $M$ is the power function.
    Furthermore, if $d = 1$ or $\inorm{\wildcard}$ is the $\ell^\infty$-norm, then  \begin{math}
        \mathrm{F}_K^\mathsf{AE}(t; p) = t \, \Theta \iprn{M^{-1} \iprn{\ifrac{p}{t}}}
    \end{math}.
\end{corollary}

The equality in the $d = 1$ or $\ell^\infty$-norm case trivially follows, since the upper and lower bounds from \cref{theorem:bounds-on-average-extremal-doeblin-curve} match. The next corollary provides three examples of closed-form Doeblin curves for convolution kernels on $\mathbb{R}$.

\begin{corollary}[Examples of Average Extremal Doeblin Curves] \label{corollary:examples-of-average-extremal-doeblin-curves}
    Consider the Gaussian, Laplace, and $q$-Gaussian \cite[Section 2.3]{UmarovTsallis2008} ($q = 2$) additive noise kernels on $\mathcal{X} = \mathcal{Y} = \mathbb{R}$, given by 
    \begin{align}
        K_1 \iprn{A \,|\, x\,;\: \sigma^2} &= \int_A \frac{1}{\sigma \sqrt{2 \pi}} \exp \prn{-\frac{(y - x)^2}{2 \sigma^2}} \diff y, \\
        K_2 \iprn{A \,|\, x\,;\: b} &= \int_A \frac{1}{2b} \exp \prn{-\frac{\abs{y - x}}{b}} \diff y, \\
        K_3 \iprn{A\,|\,x\,;\: \beta} &= \int_A \frac{\sqrt{\beta}}{\pi} \prn{\frac{1}{1 + \beta (y - x)^2}} \diff y,
    \end{align}
    respectively. Then, under the norm $\inorm{\wildcard} = |\wildcard|$ and power function $M(z) = z^2$, the average extremal power-constrained Doeblin curves for these kernels are
    \begin{align}
        \mathrm{F}_{K_1}^\mathsf{AE} \iprn{t; p, \sigma^2} &= t \prn{1 - 2 \Phi \prn{-\frac{1}{\sigma} \sqrt{\frac{p}{t}}}}, \\
        \mathrm{F}_{K_2}^\mathsf{AE} \iprn{t; p, b} &= t \prn{1 - \exp \prn{-\frac{1}{b} \sqrt{\frac{p}{t}}}}, \\
        \mathrm{F}_{K_3}^\mathsf{AE} \iprn{t; p, \beta} &= \frac{2t}{\pi} \arctan \prn{\sqrt{\frac{\beta p}{t}}},
    \end{align}
    where $\Phi\colon \mathbb{R} \rightarrow (0, 1)$ denotes the standard Gaussian CDF   \begin{math}
        \Phi(y) = \int_{-\infty}^y (\ifrac{1}{\sqrt{2 \pi}}) \exp \iprn{-\ifrac{t^2}{2}} \diff t
    \end{math}. Furthermore, the joint ranges $\mathfrak{F}(K_1; \mathcal{G}_p)$, $\mathfrak{F}(K_2; \mathcal{G}_p)$, and $\mathfrak{F}(K_3; \mathcal{G}_p)$ comprise the entire areas under the corresponding Doeblin curves, where $\mathcal{G}_p$ is the set of all Markov kernels $W$ that satisfy $\itnorm{W}_\mathsf{AE} \leq p$.
\end{corollary}

\cref{corollary:examples-of-average-extremal-doeblin-curves} is proved in \cref{subsection:proofs-of-bounds-on-power-constrained-doeblin-curves}. Note that \cite{PolyanskiyWu2016} presents a similar result on the power-constrained Dobrushin curves of Gaussian additive noise kernels. By following the steps in this proof up to \cref{eq:examples-of-average-extremal-doeblin-curves-concave}, we obtain that $\Theta$ is concave for any convolution kernel on $\mathbb{R}$ whose density function $g(z) = g(x - y) = \frac{\idiff K}{\idiff y}(y \mid x)$ is symmetric about $z = 0$ and non-increasing on $z \in [0, \infty)$. For such convolution kernels, the Doeblin curve $\mathrm{F}_K^\mathsf{AE}$ is concave in $t$, since \cref{corollary:average-extremal-doeblin-curves-of-convolution-kernels} establishes that $\mathrm{F}_K^\mathsf{AE}$ is the \emph{perspective} of the composition of a concave function $\Theta$ and a concave non-decreasing function $M^{-1}$ (cf. \cite[Remark 4]{PolyanskiyWu2016}).

\cref{figure:convolution-kernels-doeblin-curves} presents Doeblin curves for the convolution kernels in \cref{corollary:examples-of-average-extremal-doeblin-curves} with $\sigma^2 = 1$, $b = 1$, and $\beta = 1$. The trivial upper bound \cref{eq:doeblin-curve-trivial-bound} due to submultiplicativity of $\rho$ is depicted by all three curves remaining below the gray dotted line $y = t$. Since the derivatives of all three curves approach $1$ as $t \rightarrow 0$, the complementary Doeblin coefficients for all three kernels are $\rho(K_1) = \rho(K_2) = \rho(K_3) = 1$. Hence, the contraction behavior of these kernels is only captured by their Doeblin curves through the analysis in \cref{theorem:bounds-on-average-extremal-doeblin-curve} and the ensuing corollaries.

\begin{figure}
    \centering   \includegraphics[trim=1.75cm 0.25cm 2cm 0.25cm, clip, width=0.9\columnwidth]{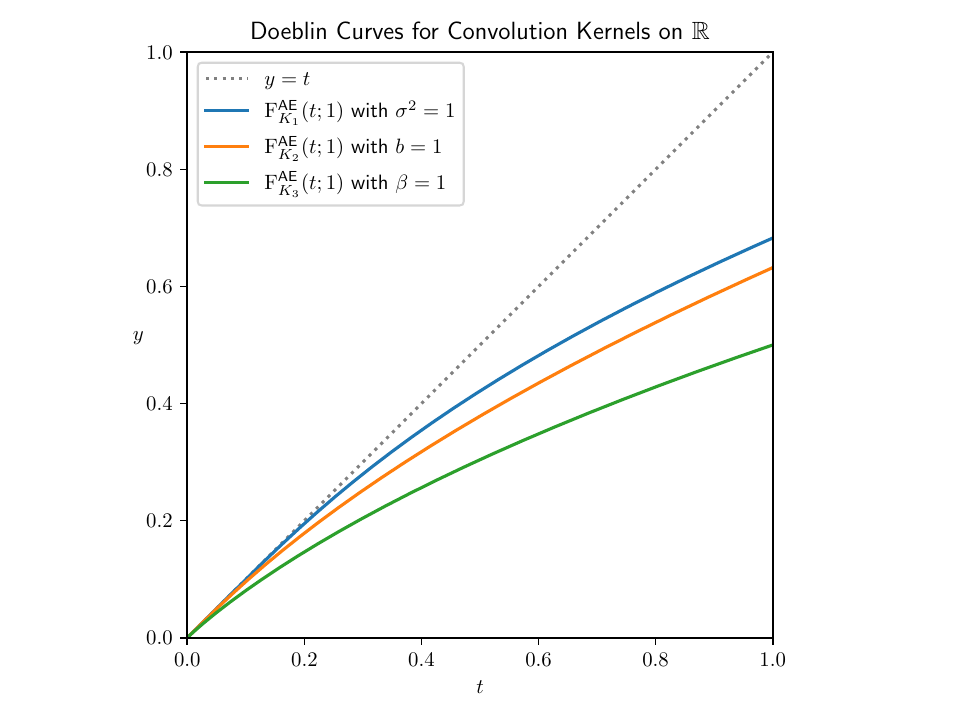}
    \caption{Average extremal power-constrained Doeblin curves for the Gaussian, Laplace, and $q$-Gaussian ($q = 2$) additive noise kernels on $\mathbb{R}$.}
    \label{figure:convolution-kernels-doeblin-curves}
\end{figure}

Lastly, we present bounds on the uniform average power-constrained Doeblin curve of a general kernel $K$. To do so, we impose additional regularity structure by requiring all input kernels $W$ to have a fixed source space $\mathcal{U}$ with finiteness or boundedness assumptions. We also impose sub-Gaussianity assumptions on $W$ as required.

\begin{proposition}[Upper Bound on Uniform Average Doeblin Curve] \label{proposition:upper-bound-on-uniform-average-doeblin-curve-finite}
    Let $\mathcal{U}$ be a finite set. Let $\mathcal{G}$ be the set of all Markov kernels $W\colon \mathcal{U} \times \mathcal{F}_\mathcal{X} \rightarrow [0, 1]$ satisfying the following assumptions:
    \begin{enumerate}[a)]
        \item (Uniform average power constraint) $\itnorm{W}_\mathsf{UA} \leq p$.
        \item (Sub-Gaussianity) Consider the set of random variables $\ibrc{Z_u}_{u \in \mathcal{U}}$ given by $Z_u = M(\inorm{X_u})$, $X_u \sim W(\lwildcard \mid u)$ for each $u \in \mathcal{U}$. Then, there exists $\sigma > 0$ such that for any $u \in \mathcal{U}$,   \begin{math}
            \iE{\exp(\lambda \prn{Z_u - \E{Z_u}})} \leq \exp \iprn{\ifrac{\sigma^2 \lambda^2}{2}}
        \end{math} for all $\lambda \in \mathbb{R}$.
    \end{enumerate}
    Then, the constrained Doeblin curve $\mathrm{F}_K(t; \mathcal{G})$ satisfies the upper bound   \begin{math}
        \mathrm{F}_K(t; \mathcal{G}) \leq t \, \invbreve{\Theta} \iprn{2 \gamma \, M^{-1} \iprn{\ifraca{p + \sigma \sqrt{2 \log_e \abs{\mathcal{U}}}}{t}}}
    \end{math}
    for all $t \in (0, 1]$ where $\Theta$ is defined in \cref{eq:theta-def}, the Jung constant $\gamma$ is defined in \cref{eq:jung-constant}, and $M$ is the power function.
\end{proposition}

\begin{proposition}[Upper Bound on Uniform Average Doeblin Curve] \label{proposition:upper-bound-on-uniform-average-doeblin-curve-totally-bounded}
    Let $(\mathcal{U}, d_\mathcal{U})$ be a Polish space endowed with a metric $d_\mathcal{U}\colon \mathcal{U} \times \mathcal{U} \rightarrow \mathbb{R}_+$ such that $(\mathcal{U}, d_\mathcal{U})$ is totally bounded, i.e., the $\epsilon$-covering number   \begin{math}
        N(\epsilon, \mathcal{U}, d_\mathcal{U}) \triangleq \min \{ n \in \mathbb{N}: \exists \ibrc{u^*_1, \tdots, u^*_n} \subseteq \mathcal{U}, \forall u \in \mathcal{U}, \, \exists i \in [n], \, d_\mathcal{U}(u, u^*_i) \leq \epsilon \}
    \end{math}
    is finite for all $\epsilon > 0$. Let $\mathcal{G}$ be the set of all Markov kernels $W\colon \mathcal{U} \times \mathcal{F}_\mathcal{X} \rightarrow [0, 1]$ satisfying the following assumptions:
    \begin{enumerate}[a)]
        \item (Uniform average power constraint) $\itnorm{W}_\mathsf{UA} \leq p$.
        \item (Sub-Gaussian increments) Consider the random process $\ibrc{Z_u}_{u \in \mathcal{U}}$ given by $Z_u = M(\inorm{X_u})$, where $\ibrc{X_u}_{u \in \mathcal{U}} \sim \mathbb{P}$ is the maximal coupling of random variables with $X_u \sim W(\lwildcard \mid u)$ for each $u \in \mathcal{U}$, as defined in \cref{eq:maximal-coupling}. Then, there exists $\sigma > 0$ such that for any $u, v \in \mathcal{U}$,   \begin{math}
            \iE{\exp(\lambda \iprn{(Z_u - Z_v) - \E{Z_u - Z_v}})} \leq \exp \iprn{\ifrac{\sigma^2 d_\mathcal{U}(u, v)^2 \lambda^2}{2}}
        \end{math} for all $\lambda \in \mathbb{R}$.
        \item (Measurability) Under the above definition of $\ibrc{Z_u}_{u \in \mathcal{U}}$, the quantity $\sup_{u \in \mathcal{U}} \ibrc{Z_u - \mathbb{E}[Z_u]}$ is measurable.
    \end{enumerate}
    Then, the constrained Doeblin curve $\mathrm{F}_K(t; \mathcal{G})$ satisfies the upper bound   \begin{equation}
        \mathrm{F}_K(t; \mathcal{G}) \leq t \, \invbreve{\Theta} \Biggprn{\!2 \gamma \, M^{-1} \!\Biggprn{\frac{p}{t} + \frac{32 \sigma}{t} \hspace{-1em}\int\limits_0^{\tnorm{\mathcal{U}}_\infty}\hspace{-1em} \sqrt{\log_e N(\epsilon, \mathcal{U}, d_\mathcal{U})} \diff \epsilon} \!\! }
    \end{equation}
    for all $t \in (0, 1]$, where   \begin{math}
        \itnorm{\mathcal{U}}_\infty \triangleq \sup_{u, v \in \mathcal{U}} d_\mathcal{U}(u, v)
    \end{math}
    denotes the diameter of $(\mathcal{U}, d_\mathcal{U})$, $\Theta$ is defined in \cref{eq:theta-def}, the Jung constant $\gamma$ is defined in \cref{eq:jung-constant}, and $M$ is the power function.
\end{proposition}

\begin{proposition}[Lower Bound on Uniform Average Doeblin Curve] \label{proposition:lower-bound-on-uniform-average-doeblin-curve}
    The power-constrained Doeblin curve $\mathrm{F}_K^\mathsf{UA}$ satisfies the lower bound   \begin{math}
        \mathrm{F}_K^\mathsf{UA}(t; p) \geq t \, \theta \iprn{\mathbf{0}, M^{-1} \iprn{\ifrac{p}{t}}}
    \end{math}
    for all $t \in (0, 1]$ where $\theta$ is defined in \cref{eq:theta-def} and $M$ is the power function.
\end{proposition}

\cref{proposition:upper-bound-on-uniform-average-doeblin-curve-finite,proposition:upper-bound-on-uniform-average-doeblin-curve-totally-bounded,proposition:lower-bound-on-uniform-average-doeblin-curve} are proved in \cref{subsection:proofs-of-bounds-on-power-constrained-doeblin-curves} by adapting the arguments from the proof of \cref{theorem:bounds-on-average-extremal-doeblin-curve}. We emphasize that \cref{proposition:lower-bound-on-uniform-average-doeblin-curve} does \emph{not} apply to the constrained Doeblin curves $\mathrm{F}_K(t; \mathcal{G})$ from \cref{proposition:upper-bound-on-uniform-average-doeblin-curve-finite,proposition:upper-bound-on-uniform-average-doeblin-curve-totally-bounded}, which are upper-bounded by $\mathrm{F}_K^\mathsf{UA}$ in general due to the stricter constraints imposed on the input kernels $W$ when defining the former. We provide examples of non-trivial kernels satisfying the preconditions of \cref{proposition:upper-bound-on-uniform-average-doeblin-curve-finite,proposition:upper-bound-on-uniform-average-doeblin-curve-totally-bounded} in \cref{appendix:markov-kernel-examples}.

\section{Main Results on Applications} \label{section:main-results-on-applications}

In this section, we present three illustrative applications of Doeblin curves to other areas of information theory and learning theory. Firstly, we derive generalization error bounds for iterative optimization algorithms operating over unbounded feasible sets. Secondly, we establish lower bounds for reliable computation using circuits of noisy $q$-ary gates. Lastly, we introduce a new definition of $(\epsilon, \delta)$-differential privacy based on our variational characterization of Doeblin coefficients, and provide improved privacy bounds for online iterative optimization. The applications showcase how Doeblin curves can provide nontrivial guarantees even when strong data processing constants such as the Dobrushin coefficient or the mutual information contraction coefficient degenerate to $1$ and do not provide any useful information.
While our results on generalization error can be specialized to use Dobrushin curves (see footnote 6 below), the latter two applications rely on Doeblin curves to trace the contraction between a collection of probability distributions, which cannot be done with Dobrushin curves (which can only analyze the contraction between two distributions).

\subsection{Bounds on Generalization Error} \label{subsection:bounds-on-generalization-error}

In this subsection, we use Doeblin curves to extend existing information-theoretic generalization error bounds for noisy iterative optimization algorithms in \cite{CalmonGeneralizationError} to settings where the feasible set has infinite diameter or the optimization problem is unconstrained altogether (also see \cite{RaginskyEnergyBound}). We utilize the information-theoretic framework introduced in \cite{CalmonGeneralizationError}, which provides generalization error bounds for compact feasible sets by leveraging the Dobrushin contraction coefficient of additive noise channels. Their results indicate that data points used in earlier iterations have a decaying contribution to the generalization error due to the cumulative effect of noise injection, with the rate of decay governed by the Dobrushin coefficient of the underlying noise channel. This analysis yields a non-trivial bound only when the feasible set has finite diameter, since the Dobrushin coefficient trivially equals $1$ otherwise. To address feasible sets with infinite diameter, we utilize Doeblin curves as a finer-grained tool to capture information decay for data points used in earlier iterations even when Dobrushin coefficients fail to do so, thereby yielding non-trivial generalization error bounds despite the absence of coarser coefficient-based contraction.\footnote{In all applications of Doeblin curves, suitable constraints may be placed on the Doeblin curve to obtain tighter bounds by incorporating prior knowledge about the input distributions. For example, when deriving our results on generalization error, we only utilize the Doeblin curve $\mathrm{F}_\Phi^\mathsf{UA}$ of the additive noise mechanism to quantify the contraction between \emph{two} distributions pushed through the mechanism. Hence, we may instead consider the constrained curve which only includes input kernels with $|\mathcal{U}| = 2$, which reduces to the \emph{Dobrushin curve} of the noise mechanism. Nonetheless, we present all results in this paper in terms of the more general \emph{Doeblin curve} $\mathrm{F}_\Phi^\mathsf{UA}$ for conceptual simplicity.}

Formally, following \cite{CalmonGeneralizationError}, we begin by considering a general (possibly non-convex) stochastic optimization problem,   \begin{gather}
    \min_{w \in \mathcal{W}} G_\mu(w), \\
    \text{where}\quad  G_\mu(w) \triangleq \Ewrt{Z \sim \mu}{g(w, Z)} = \int_{\mathcal{Z}} g(w, \wildcard) \diff \mu ,
\end{gather}
where $w \in \mathcal{W} \subseteq \mathbb{R}^d$ are the model parameters constrained to the (potentially unbounded) feasible set $\mathcal{W}$, $\mu$ is the underlying data distribution over the samples $Z$ belonging to the sample space $\mathcal{Z}$, and $g\colon \mathcal{W} \times \mathcal{Z} \rightarrow \mathbb{R}_+$ is the loss function. In practice, the true data distribution $\mu$ is usually unknown, and only a dataset $\mathcal{S} \triangleq ({Z}_1, \tdots, {Z}_n)$ consisting of $n$ independent and identically distributed samples ${Z}_i \sim \mu$ is available. Thus, we instead consider the empirical risk minimization (ERM) problem
\begin{equation}
    \min_{w \in \mathcal{W}} G_{\mathcal{S}}(w) , \quad \text{where} \quad G_{\mathcal{S}}(w) \triangleq \frac{1}{n} \sum_{i=1}^n g({w}, {Z}_i) ,
\end{equation}
with the aim of finding a solution which generalizes well to the original problem of minimizing $G_\mu$. To this end, we consider the following noisy iterative algorithm. The algorithm is initialized with an arbitrary parameter $W_0 \in \mathcal{W}$. Before the optimization process, $T$ disjoint mini-batch index sets $\mathcal{B}_1, \tdots, \mathcal{B}_T$ are chosen deterministically and fixed, where $\mathcal{B}_t \subseteq [n]$ contains the indices of the samples comprising the mini-batch at iteration $t$. Next, the algorithm performs $T$ iterations by updating the parameters $W_t$ according to the rule
\begin{equation}
    W_t \triangleq \proj_\mathcal{W} \biggprn{W_{t-1} - \frac{\eta_t}{\abs{\mathcal{B}_t}} \sum_{i \in \mathcal{B}_t} \nabla g(W_{t-1}, Z_i) + m_t N_t} \label{eq:noisy-iterative-algorithm-iteration}
\end{equation}
for each iteration $t \in [T]$, where the projection operator $\proj_\mathcal{W}$ ensures that the parameters remain within the feasible set $\mathcal{W}$,\footnote{Technically, our results hold with $\proj_\mathcal{W}$ being any function whose range is $\mathcal{W}$, not necessarily the orthogonal projection operator or even a function which is idempotent. Note that some choices for the projection operator might require the feasible set to be closed, convex, or both.} $\eta_t$ is the learning rate, $\nabla g$ represents the partial gradient of $g$ with respect to its first argument (i.e., the model parameters), $N_t \sim \mathsf{Normal}(0, I)$ represents independent additive standard Gaussian noise (where $I$ is the identity matrix of appropriate dimension),\footnote{We focus on Gaussian noise for conceptual simplicity, although our results can be easily extended to other noise distributions such as Laplace.} and $m_t$ controls the noise magnitude. We remark that several algorithms can be expressed in this form, such as Stochastic Gradient Langevin Dynamics (SGLD) and Differentially-Private Stochastic Gradient Descent (DP-SGD).

We assume that the loss function is bounded above by a constant $A > 0$, i.e.,   \begin{equation}
    \forall w \in \mathcal{W}, \, \forall z \in \mathcal{Z}, \quad g(w, z) \leq A  , \label{eq:A-def}
\end{equation}
and that the gradient of the loss is bounded in magnitude, i.e.,
\begin{equation}
    \forall w \in \mathcal{W}, \, \forall z \in \mathcal{Z}, \quad \norm{\nabla g(w, z)}_2 \leq L  . \label{eq:generalization-error-gradient-bound}
\end{equation}
Our objective is to upper-bound the expected generalization error at the final model parameters $W_T$ after $T$ iterations, i.e.,   \begin{math}
    \abs{\mathbb{E}[G_\mu(W_T) - G_{\mathcal{S}}(W_T)]}
\end{math},
where the expectation is taken with respect to the randomness in the dataset $\mathcal{S}$ and the noise $N_1, \tdots, N_T$ throughout the optimization process.

Let $\mathrm{F}_\Phi^\mathsf{UA}$ denote the uniform average power-constrained Doeblin curve of the $\mathsf{Normal}(0, I)$ additive noise channel under the power function $M(z) = z^2$. Our main result is formally stated as follows.

\begin{theorem}[Expected Generalization Error] \label{theorem:expected-generalization-error}
    The expected generalization error of the noisy iterative algorithm described by \cref{eq:noisy-iterative-algorithm-iteration} satisfies  
    \begin{multline}
        \abs{\E{G_\mu(W_T) - G_{\mathcal{S}}(W_T)}}\\ \leq A \sum_{t=1}^T \frac{\abs{\mathcal{B}_t}}{n} \, \invbreve{\mathrm{F}}_\Phi^\mathsf{UA} \mprn{\tcdots \invbreve{\mathrm{F}}_\Phi^\mathsf{UA} \mprn{\frac{\eta_t \sigma_{t-1}}{m_t \abs{\mathcal{B}_t}}; p_{t+1}} \tcdots; p_T} , \label{eq:expected-generalization-error}
    \end{multline}
    where $A$ upper bounds the loss function as assumed in \cref{eq:A-def} and for each $t \in \ibrc{0, \tdots, T - 1}$ we define
    \begin{equation}
        \sigma_t \triangleq \mathop{\mathbb{E}}_{\substack{(W_t, Z) \\ \sim P_{W_t} \otimes \mu}} \biggbkt{\Bignorm{\nabla g(W_t, Z) - \mathop{\mathbb{E}}_{\substack{(W_t, Z) \\ \sim P_{W_t} \otimes \mu}} \mbkt{\nabla g(W_t, Z)}}_2} ,
    \end{equation}
    for each $t \in \ibrc{2, \tdots, T}$ we define 
    \begin{equation}
        p_t \triangleq \frac{1}{m_t^2} \Biggprn{2 \max_{i \in \bigcup\limits_{s=1}^{t-1} \mathcal{B}_s} \sup_{z \in \mathcal{Z}} \E{\norm{W_{t-1}}_2^2 \given Z_i = z} + 2 \eta_t^2 L^2}
    \end{equation}
    where $L$ upper bounds the gradient of the loss as assumed in \cref{eq:generalization-error-gradient-bound}, and the composition of Doeblin curves in \cref{eq:expected-generalization-error} reduces to the identity function (i.e., no Doeblin curves are applied) for the $t = T$ term of the sum.
\end{theorem}

\cref{theorem:expected-generalization-error} is proved in \cref{subsection:proofs-for-generalization-error}. The argument utilizes the following lemma %
to express the expected generalization error in terms of the TV-information between the model parameter $W_T$ and each data sample $Z_i$.

\begin{lemma}[Generalization Error and TV-Information {\cite[Lemma 1]{CalmonGeneralizationError}}] \label{lemma:generalization-error-and-t-information}
    Define the TV-information between two random variables $X$ and $Y$ as 
    \begin{equation}
        I_{\mathsf{TV}}(X; Y) \triangleq \norm{P_{X,Y} - P_X \otimes P_Y}_\mathsf{TV} , \label{eq:t-information}
    \end{equation}
    where $P_{X,Y}$ is the joint distribution of $(X, Y)$, and $P_X$ and $P_Y$ are the marginal distributions of $X$ and $Y$, respectively. Then, the expected generalization error satisfies
    \begin{align}
        \abs{\E{G_\mu(W_T) - G_{\mathcal{S}}(W_T)}} \leq \frac{A}{n} \sum_{i=1}^n I_{\mathsf{TV}}(W_T; Z_i) .
    \end{align}
\end{lemma}

Our principal contribution is to bound the TV-information between the output $W_T$ of the learning algorithm and the data samples $Z_i$ by utilizing properties of the Gaussian noise channel's Doeblin curve, as formalized in the following lemma which holds for possibly non-compact $\mathcal{W}$. %

\begin{lemma}[Recursive Bound on TV-Information] \label{lemma:recursive-bound-on-t-information}
    The TV-information between the final model parameter $W_T$ and a data sample $Z_i$ satisfies   \begin{multline}
        I_{\mathsf{TV}}(W_T; Z_i)\\ \leq \invbreve{\mathrm{F}}_\Phi^\mathsf{UA} \Bigprn{\tcdots \invbreve{\mathrm{F}}_\Phi^\mathsf{UA} \Bigprn{\invbreve{\mathrm{F}}_\Phi^\mathsf{UA} \Bigprn{I_{\mathsf{TV}}(W_t; Z_i); p_{t+1}}; p_{t+2}} \tcdots; p_T}  , \label{eq:recursive-bound-on-t-information}
    \end{multline}
    where $t$ is the iteration on which $Z_i$ is used in the update rule (i.e., $t \in \mathcal{B}_i$), for each $s \in \ibrc{t + 1, \tdots, T}$ we define
    \begin{align}
        p_s \triangleq \frac{1}{m_s^2} \Biggprn{2 \max_{i \in \bigcup\limits_{r=1}^{s-1} \mathcal{B}_r} \sup_{z \in \mathcal{Z}} \E{\norm{W_{s-1}}_2^2 \given Z_i = z} + 2 \eta_s^2 L^2}  , \label{eq:generalization-error-recursive-bound-power}
    \end{align}
    and the composition of Doeblin curves in \cref{eq:recursive-bound-on-t-information} reduces to the identity function if $t = T$ (in which case \cref{eq:recursive-bound-on-t-information} trivially holds with equality).
\end{lemma}

\cref{lemma:recursive-bound-on-t-information} is proved in \cref{subsection:proofs-for-generalization-error}. We remark that bounds on the second moment $\mathbb{E}[\inorm{W_t}_2^2]$ of the iterates $W_t$ have been derived in the literature, e.g., under additive Gaussian noise \cite{RaginskyEnergyBound}. These results can often be directly translated into bounds on $p_t$. In principle, tighter bounds can be obtained by considering the second moment $\mathbb{E}[\inorm{W_t - \tilde{w}}_2^2]$ of the iterates with respect to a reference point $\tilde{w} \in \mathcal{W}$ chosen as the center of the feasible set or, ideally, the optimal solution $w^*$. However, since $w^*$ is unknown in practice, we use the worst-case bounds on the second moment of the iterates instead. 

\subsection{Reliable Computation Using Noisy \texorpdfstring{$q$}{q}-ary Gates} \label{subsection:reliable-computation-using-noisy-q-ary-gates}

In this subsection, we use Doeblin curves to establish information-theoretic bounds on the expressive power of circuits of noisy gates for the computation of $q$-ary functions. This builds on results for Boolean functions in \cite{PolyanskiyWu2016}. We consider $n$-input circuits which compute a single output value using $b$-input gates, where the result of each gate is perturbed by independent noise. Historically, the study of such architectures began with von Neumann's pioneering work on fault-tolerant Boolean circuits \cite{vonNeumann1956,HajekWeller1991,EvansSchulman2003}, and has been largely limited to binary operations. We expand this line of study to multivalued logic systems, as motivated by recent advancements in classical and quantum information processing \cite{TanHo2024}.

Formally, we model a noisy circuit over a $q$-ary alphabet $\mathcal{Q} = \ibrc{\xi_1, \tdots, \xi_q} \subset \mathbb{R}^d$ as a directed acyclic graph with $n$ \emph{input vertices} $X_1, \tdots, X_n \in \mathcal{Q}$ indexed by $i \in [n]$. The inputs are processed by a collection of $m$ \emph{gate vertices} indexed by $j \in [m]$, each of which takes $b_j \leq b$ $\mathbb{R}^d$-valued arguments and computes a $\mathcal{Q}$-valued result $Z_j \in \mathcal{Q}$, which is then corrupted by independent noise from some arbitrary fixed kernel $\Phi$ to produce $Y_j \in \mathbb{R}^d$. Formally, letting $\mathcal{N}_j$ and $\mathcal{M}_j$ denote the \emph{incoming edge sets} of input vertices and gate vertices feeding into gate $j$, respectively, we have   \begin{equation}
    Z_j \triangleq \Gamma_j \mprn{\brc{X_i}_{i \in \mathcal{N}_j}, \brc{Y_{j'}}_{j' \in \mathcal{M}_j}}, \quad
    Y_j \sim \Phi(\lwildcard \mid Z_j)
\end{equation}
for some deterministic gate function $\Gamma_j\colon (\mathbb{R}^d)^{b_j} \rightarrow \mathcal{Q}$, where $b_j = |\mathcal{N}_j| + |\mathcal{M}_j| \leq b$. Without loss of generality, we assume the gates are indexed in topological order, i.e., $\mathcal{M}_j \subseteq [j - 1]$ for each $j \in [m]$. At the end, the circuit produces a single \emph{output vertex} $Y_m$ with no outgoing edges. We assume there exists a directed path from each $X_i$ to $Y_m$, so none of the circuit inputs are discarded. We note that this model can be seen as a $q$-ary generalization of the setup considered in \cite[Section 5.3]{PolyanskiyWu2016}.

One goal of such a circuit is to compute a function $G\colon \mathcal{Q}^n \rightarrow \mathcal{Q}$ with error probability at most $\mathrm{P}_\mathsf{error} > 0$. Namely, there must exist some fixed decoding function $\hat{G}\colon \mathbb{R}^d \rightarrow \mathcal{Q}$ such that for any input vector $(x_1, \tdots, x_n) \in \mathcal{Q}^n$,   \begin{math}
    \iP{\hat{G}(Y_m) = G(x_1, \tdots, x_n)} \geq 1 - \mathrm{P}_\mathsf{error}
\end{math},
where the probability is taken with respect to all the noise in the circuit. In the binary case ($q = |\mathcal{Q}| = 2$), by Le Cam's relation, the error probability of the optimal decoder can be controlled by the TV distance between the marginal distributions of $Y_m$ induced by setting some ``initial'' input vertex to $0$ or $1$, respectively, while fixing all other input values \cite[Eq. 146]{PolyanskiyWu2016}. So, several works study the question of how information measures contract over fault-tolerant circuits, e.g., \cite{PolyanskiyWu2016} recursively upper-bounds TV distance while \cite{evans1999signal} bounds mutual information. Propelled by such analyses, and in light of the interpretation of complementary Doeblin coefficients as a multi-way generalization of TV distance \cite[Theorem 2]{MakurSingh2023b}, we analyze the complementary Doeblin coefficient of the $q$ induced marginal distributions of $Y_m$ as a first step towards understanding information propagation in noisy $q$-ary circuits. (Intuitively, perceiving a circuit as a Markov chain, the analysis of TV distance in \cite{PolyanskiyWu2016} is related to the coupling time of two copies of the chain with two different values at the ``initial'' vertex \cite{LevinPeresWilmer2009}. On the other hand, our analysis of Doeblin coefficients is related to the coupling time of $q$ copies of the chain initiated at all $q$ possible values.) Our main contribution is the following theorem, which relates the Doeblin coefficient to the noise mechanism's Doeblin curve.

\begin{theorem}[Upper Bound on Circuit Output Divergence] \label{theorem:upper-bound-on-circuit-output-divergence}
    Fix $\epsilon > 0$. For all $n$-input circuits of noisy $b$-input gates with sufficiently large $n$, there exists some input vertex, which we take to be $X_1$ without loss of generality, such that for any fixed values $x_2, \tdots, x_n \in \mathcal{Q}$ of the remaining inputs,   \begin{multline}
        \rho \Bigprn{\bigbkt{P_{Y_m}^{(1)}, \tdots, P_{Y_m}^{(q)}}} \\\leq \max \brc{t \in [0, 1]: \invbreve{\mathrm{F}}_\Phi^\mathsf{UA} \mprn{\min \mbrc{1, bt}; p} = t} + \epsilon ,
    \end{multline}
    where $P_{Y_m}^{(\ell)}$ is the marginal distribution of $Y_m$ induced by the circuit when setting $X_1 = \ell$ and $X_i = x_i$ for all $i > 1$, $\mathrm{F}_\Phi^\mathsf{UA}$ is the uniform average power-constrained Doeblin curve of the noise kernel $\Phi$ for some norm $\inorm{\wildcard}$ on $\mathbb{R}^d$ and power function $M\colon \mathbb{R}_+ \rightarrow \mathbb{R}_+$, and   \begin{math}
        p \triangleq \max_{\xi \in \mathcal{Q}} M \mprn{\norm{\xi}}
    \end{math}.
\end{theorem}

\cref{theorem:upper-bound-on-circuit-output-divergence} represents an information-theoretic limit on the quality of any output decoder, since decoding performance improves as output divergence increases. We defer the proof to \cref{subsection:proofs-for-reliable-computation}. The proof constructs a coupling of the induced marginal distributions at each gate output $Y_j$, using the maximal coupling characterization of Doeblin coefficients \cref{eq:maximal-coupling} to refine the couplings from earlier gates. This is formalized in the following lemma. For notational convenience, we denote collections of superscripted variables as $w^{(1:q)} = (w^{(1)}, \tdots, w^{(q)})$.

\begin{lemma}[Stepwise Coupling Construction] \label{lemma:stepwise-coupling-construction}
    Let $\mathcal{U}$ be a finite subset of a separable Banach space. Let $(\mathcal{V}, \mathcal{F}_\mathcal{V})$ be a Polish space and let $\Phi\colon \mathcal{U} \times \mathcal{F}_\mathcal{V} \rightarrow [0, 1]$ be a Markov kernel from $\mathcal{U}$ to $\mathcal{V}$. Let $P_U^{(1)}, \tdots, P_U^{(q)}$ be a collection of probability measures on $\mathcal{U}$, and let $P_V^{(1)}, \tdots, P_V^{(q)}$ be the respective probability measures on $\mathcal{V}$ induced by pushing $P_U^{(\ell)}$ through $\Phi$ for each $\ell \in [q]$, i.e.,
    \begin{equation}
        \forall A \in \mathcal{F}_\mathcal{V}, \quad P_V^{(\ell)}(A) \triangleq \int_\mathcal{U} \Phi(A \mid u) \diff P_U^{(\ell)}(u) . \label{eq:stepwise-coupling-construction-p-v}
    \end{equation}
    Then, given any coupling $\pi_U$ of $P_U^{(1)}, \tdots, P_U^{(q)}$, there exists a coupling $\pi_V$ of $P_V^{(1)}, \tdots, P_V^{(q)}$ satisfying   \begin{multline}
        \Pwrt{V^{(1:q)} \sim \pi_V}{\lnot \mprn{V^{(1)} = \tcdots = V^{(q)}}}\\ \leq \invbreve{\mathrm{F}}_\Phi^\mathsf{UA} \mprn{\Pwrt{U^{(1:q)} \sim \pi_U}{\lnot \mprn{U^{(1)} = \tcdots = U^{(q)}}}; p}  , \label{eq:stepwise-coupling-construction}
    \end{multline}
    where we define   \begin{math}
        p \triangleq \max_{u \in \mathcal{U}} M(\norm{u})
    \end{math}.
\end{lemma}

\cref{lemma:stepwise-coupling-construction} is proved in \cref{subsection:proofs-for-reliable-computation} by constructing $\pi_V$ as the weighted average (with respect to $U^{(1:q)} \sim \pi_U$) of the maximal coupling \cref{eq:maximal-coupling} of $\Phi(\lwildcard \mid U^{(1)}), \tdots, \Phi(\lwildcard \mid U^{(q)})$. Iterating this construction for each gate leads to repeated application of the Doeblin curve $\mathrm{F}_\Phi^\mathsf{UA}$, which converges to the greatest fixed point of a transformed version of $\mathrm{F}_\Phi^\mathsf{UA}$.

\subsection{Relation to Differential Privacy and Online Algorithms} \label{subsection:relation-to-differential-privacy-and-privacy-amplification}

In this subsection, we utilize our variational characterization of Doeblin coefficients to motivate a new definition of group local differential privacy (LDP). Recall that a mechanism $K\colon \mathcal{X} \times \mathcal{F}_\mathcal{Y} \rightarrow [0, 1]$, which is nothing but a Markov kernel, is $(\epsilon, \delta)$-LDP if for any inputs $x, x^\prime \in \mathcal{X}$ both $K(\lwildcard \mid x)$ and $K(\lwildcard \mid x^\prime)$ exhibit similar distributions \cite{evfimievski2003, kasiviswanathan2011, asoodeh2021local,AsoodehLiaoCalmon2021,ErlingssonPihur2014,TruexLiuChow2020},\footnote{This differs slightly from the standard definition of (central) differential privacy (DP), which considers input \emph{datasets of multiple users} differing by only one user \cite[Definition 2.4]{DworkRoth2014}.} i.e.,
\begin{equation}
    \forall x, x^\prime \in \mathcal{X}; \, \forall A \in \mathcal{F}_\mathcal{Y}, \quad K(A \mid x) - e^\epsilon K(A \mid x^\prime) \leq \delta . \label{eq:differential-privacy}
\end{equation}
Subtracting both sides from $1$, this is equivalent to
\begin{align}
    \forall x, x^\prime \in \mathcal{X}; \, \forall A \in \mathcal{F}_\mathcal{Y}, \quad K(A^\complement \mid x) + e^\epsilon K(A \mid x^\prime) \geq 1 - \delta . \label{eq:differential-privacy-rearranged}
\end{align}
This rearrangement has an interesting interpretation: Consider a binary hypothesis testing scenario that aims to distinguish whether the output $Y$ of a mechanism is derived from the distribution $K(\lwildcard \mid x)$ versus $K(\lwildcard \mid x^\prime)$ for any fixed $x$ and $x^\prime$. Let $A \subseteq \mathcal{Y}$ be any possible choice of the rejection region for the hypothesis $Y \sim K(\lwildcard \mid x^\prime)$. The reformulated version in \cref{eq:differential-privacy-rearranged} essentially highlights the impossibility of achieving a very low weighted sum of type-\Romannum{1} and type-\Romannum{2} errors from the data $Y$ derived from a differentially private mechanism $K$ (see \cite[Theorem 2.1]{kairouz2015composition} for more details).

In addition, observe that the left-hand side of \cref{eq:differential-privacy} is bounded by an $\epsilon$-weighted TV distance (called the $\mathsf{E}_{\gamma}$-divergence) between $K(\lwildcard \mid x)$ and $K(\lwildcard \mid x')$. Indeed, this connection has been noted by, e.g., \cite{asoodeh2021local,PrivacyAmplification}. Considering that $\mathsf{E}_{\gamma}$-divergence, a weighted variant of TV distance (or Dobrushin coefficient), exactly characterizes local differential privacy, it is natural to ask whether a similarly weighted variant of the Doeblin coefficient can admit an analogous interpretation.

Motivated by this observation and our variational characterization of Doeblin coefficients from \cref{theorem:variational-characterization-of-doeblin-coefficient}, we extend the standard definition of LDP to a ``group'' setting.\footnote{Accordingly, this is different from the usual definition of group DP based on distance between the input datasets \cite[Theorem 2.2]{DworkRoth2014}.} Consider any group size $n \geq 2$ and any $\boldsymbol{\epsilon} = (\epsilon_1, \tdots, \epsilon_n) \in \mathbb{R}_+^n$ where we assume without loss of generality that $\epsilon_1 = 0$ and $\epsilon_i \leq \epsilon_j$ for all $i < j$ . A mechanism $K$ is said to be \emph{$(\boldsymbol{\epsilon}, \delta, n)$-LDP} if for any $x_1, \tdots, x_n \in \mathcal{X}$ and any $\mathcal{F}_\mathcal{Y}$-measurable $n$-partition $A_1, \tdots, A_n$ of $\mathcal{Y}$,
\begin{equation}
    \sum_{i=1}^n e^{\epsilon_i} K(A_i \mid x_i) \geq 1 - \delta . \label{eq:group-differential-privacy}
\end{equation}
This definition may be interpreted in the context of an $n$-ary hypothesis testing problem. Given a set of $n$ hypotheses $H_i\colon Y \sim K(\lwildcard \mid x_i)$ indexed by $i \in [n]$, where $Y$ is the observed variable, the generalized definition in \cref{eq:group-differential-privacy} states that the problem of identifying a single false hypothesis has a large weighted sum error for any possible choices of rejection regions. Any reasonable test for this hypothesis testing problem would partition the space $\mathcal{Y}$ into $n$ rejection regions $A_1, \tdots, A_n$, and the test rejects hypothesis $H_i$ if $Y$ takes on a value in $A_i$. Therefore, $K(A_i \mid x_i)$ represents the conditional probability of error given $H_i$ in this scenario. For $n = 2$, the notion of $(\boldsymbol{\epsilon}, \delta, 2)$-LDP reduces to the standard definition in \cref{eq:differential-privacy}. Furthermore, $K$ being $(\boldsymbol{\epsilon}, \delta, n)$-LDP for $n \geq 2$ is a stronger condition, as it implies that $K$ is also $([\epsilon_1,\tdots,\epsilon_m], \delta, m)$-LDP for all $2 \leq m \leq n$. Intuitively, while an $(\epsilon, \delta)$-LDP algorithm makes it difficult for an attacker to distinguish between a pair of inputs, an $(\boldsymbol{\epsilon}, \delta, n)$-LDP algorithm protects against arbitrary $n$-way comparisons between inputs, which is a stronger notion.

Motivated by this formulation, we define weighted analogues of the Doeblin and complementary Doeblin coefficients for a Markov kernel $K\colon \mathcal{X} \times \mathcal{F}_\mathcal{Y} \rightarrow [0, 1]$:
\begin{align}
    \tau_{\boldsymbol{\epsilon}}(K, n) &\triangleq \inf_{x_1, \dots, x_n \in \mathcal{X}} \inf_{\substack{\text{$n$-partition of $\mathcal{Y}$} \\ A_1, \dots, A_n}} \sum_{i=1}^n e^{\epsilon_i} K(A_i \mid x_i),\, \label{eq:tau-epsilon} \\
    \rho_{\boldsymbol{\epsilon}}(K, n) &\triangleq 1 - \tau_{\boldsymbol{\epsilon}}(K, n)  . \label{eq:rho-epsilon}
\end{align}
We remark that $\tau_{\boldsymbol{\epsilon}}(K, n) \in [0, 1]$ (and so $\rho_{\boldsymbol{\epsilon}}(K, n) \in [0, 1]$), because for any $x_1, \tdots, x_n \in \mathcal{X}$, we have   \begin{align}
    \tau_{\boldsymbol{\epsilon}}(K, n) &\stackclap{(a)}{\leq} e^{\epsilon_1} K(\mathcal{Y} \mid x_1) + \sum_{i=2}^n e^{\epsilon_i} K(\emptyset \mid x_i) \\&\stackclap{(b)}{=} K(\mathcal{Y} \mid x_1) \stackclap{(c)}{=} 1,
\end{align}
where (a) holds by upper-bounding the infima in \cref{eq:tau-epsilon} with a specific instance, (b) holds because $\epsilon_1 = 0$, and (c) holds because $K$ is a Markov kernel. Moreover, akin to how \cite{PrivacyAmplification,asoodeh2021local} reformulated local differential privacy in terms of the $\mathsf{E}_\gamma$ contraction coefficient,\footnote{We note that the $\mathsf{E}_\gamma$-divergence is a weighted version of TV distance, which is equivalent (when appropriately scaled) to Bayes statistical information or DeGroot distance \cite{DeGroot1962}.} $\rho_{\boldsymbol{\epsilon}}(K, n)$ captures $(\boldsymbol{\epsilon},\delta,n)$-LDP exactly: $\rho_{\boldsymbol{\epsilon}}(K, n) \leq \delta$ if and only if $K$ is $(\boldsymbol{\epsilon},\delta,n)$-LDP (see proof of \cref{theorem:contraction-of-rho-epsilon} in \cref{subsection:proofs-for-differential-privacy}).

In addition, we note a connection between $\rho_{\boldsymbol{\epsilon}}$ and the concept of min-DeGroot distance from \cite{MakurSingh2023b}, which is construed as a generalization of Bayes statistical information. In this context, a hidden random variable $X \in \mathcal{X}$ (where $\mathcal{X} = \ibrc{x_1, \tdots, x_n}$ has cardinality $|\mathcal{X}| = n$) follows a prior distribution $\boldsymbol{\lambda} \in \mathbb{R}^n$, and a random variable $Y \in \mathcal{Y}$ is observed according to an observation model denoted by $K$. Let $\hat{X} \in \mathcal{X}$ be any (possibly randomized) estimator of $X$ based on $Y$, such that $X \rightarrow Y \rightarrow \hat{X}$ forms a Markov chain. \cite{MakurSingh2023b} defines the min-DeGroot distance $\tau_{\min}(\boldsymbol{\lambda}, K)$ as the reduction in Bayes risk when observing the data $Y$ compared to not observing $Y$. Specifically, 
\begin{align}
    \tau_{\min}(\boldsymbol{\lambda}, K) \triangleq \min_{i \in [n]} \lambda_i - \int_{\mathcal{Y}} \prn{\min_{i \in [n]} \lambda_i \frac{\diff K}{\diff \mu}(\lwildcard \mid x_i)} \diff \mu ,
\end{align}
where $\mu$ is the common dominating measure for $K(\lwildcard \mid x_1),\allowbreak \tdots, K(\lwildcard \mid x_n)$ from \cref{eq:common-dominating-measure}. By setting $\lambda_i = e^{\epsilon_i} / \sum_{i=1}^n e^{\epsilon_i}$ and applying \cref{proposition:integral-characterization-of-infimum}, the min-DeGroot distance can be interpreted as a rescaled version of $\rho_{\boldsymbol{\epsilon}}(K, n)$, namely,
\begin{align}
    \rho_{\boldsymbol{\epsilon}}(K, n) = \prn{\sum_{i=1}^n e^{\epsilon_i}} \tau_{\min}(\boldsymbol{\lambda}, K) . \label{eq:rho-epsilon-min-degroot-connection}
\end{align}
Therefore, the definition of $(\boldsymbol{\epsilon}, \delta, n)$-LDP is essentially a rescaled version of the generalized notion of Bayes statistical information obtained from the set $\{x_1, \tdots, x_n\}$. Our proposed definition of an $(\boldsymbol{\epsilon}, \delta, n)$-LDP mechanism underscores that, given the processed data of the entire group, obtaining any substantial ``information'' (quantified in the Bayes statistical information sense) about any individual member is ``hard.'' 

In the following theorem, we prove contraction properties of $\rho_{\boldsymbol{\epsilon}}(K,n)$. By \cref{eq:rho-epsilon-min-degroot-connection}, these results translate to contraction properties of $\tau_{\min}(\boldsymbol{\lambda}, K)$.

\begin{theorem}[Contraction of $\rho_{\boldsymbol{\epsilon}}(K,n)$] \label{theorem:contraction-of-rho-epsilon}
    Let $K\colon \mathcal{X} \times \mathcal{F}_\mathcal{Y} \rightarrow [0, 1]$ be an absolutely continuous Markov kernel. For any group size $n \geq 2$ and any $\boldsymbol{\epsilon} = (\epsilon_1, \tdots, \epsilon_n) \in \mathbb{R}_+^n$ such that $\epsilon_1 = 0$ and $\epsilon_i \leq \epsilon_j$ for all $i < j$, we have the following properties:
    \begin{enumerate}
        \item \emph{(Submultiplicativity)} For any Markov kernel $W\colon \mathcal{U} \times \mathcal{F}_\mathcal{X} \rightarrow [0, 1]$,
        \begin{align}
            \rho_{\boldsymbol{\epsilon}}(WK, n) \leq \rho_{\boldsymbol{\epsilon}}(W, n) \, \rho_{\boldsymbol{\epsilon}}(K, n)  . \label{eq:contraction-of-rho-epsilon-submultiplicativity}
        \end{align}
        \item \emph{(Contraction behavior)} $K$ is $(\boldsymbol{\epsilon}, \delta, n)$-LDP iff for any Markov kernel $W\colon \mathcal{U} \times \mathcal{F}_\mathcal{X} \rightarrow [0, 1]$,
        \begin{align}
            \rho_{\boldsymbol{\epsilon}}(WK, n) \leq \delta \rho_{\boldsymbol{\epsilon}}(W, n)  .
        \end{align}
        \item \emph{(Doeblin curves)} For any Markov kernel $W\colon \mathcal{U} \times \mathcal{F}_\mathcal{X} \rightarrow [0, 1]$ satisfying $\itnorm{W}_\mathsf{UA} \leq p$,
        \begin{align}
            \rho_{\boldsymbol{\epsilon}}(WK, n) \leq \mathrm{F}_K^\mathsf{UA}(\rho_{\boldsymbol{\epsilon}}(W, n); p)  . \label{eq:differential-privacy-doeblin-curves}
        \end{align}
    \end{enumerate}
\end{theorem}

\cref{theorem:contraction-of-rho-epsilon} is proved in \cref{subsection:proofs-for-differential-privacy}. Parts 1 and 2 imply that $\rho_{\boldsymbol{\epsilon}}$ exhibits contraction-coefficient-like properties akin to $\rho$, with Part 1 acting as a meta-SDPI for $\rho_{\boldsymbol{\epsilon}}$ (cf. \cite[Proposition 3]{MakurZheng2020}) and Part 2 characterizing $\delta$ as the contraction coefficient of this meta-SDPI. In particular, the \emph{data processing inequality} for $\rho_{\boldsymbol{\epsilon}}$, i.e.,
\begin{equation}
    \rho_{\boldsymbol{\epsilon}}(WK, n) \leq \rho_{\boldsymbol{\epsilon}}(W, n) , \label{eq:data-processing-inequality-rho-epsilon}
\end{equation}
is an immediate consequence due to $\rho_{\boldsymbol{\epsilon}}(K, n) \in [0, 1]$.
In addition, considering the connection between $\rho_{\boldsymbol{\epsilon}}$ and the min-DeGroot distance highlighted above, Part 1 can be restated as an SDPI for min-DeGroot distances, i.e., for any Markov kernel $W\colon \mathcal{U} \times \mathcal{F}_\mathcal{X} \rightarrow [0, 1]$ with finite $\mathcal{U}$,   \begin{math}
    \tau_{\min}(\boldsymbol{\lambda}, WK) \leq \rho_{\boldsymbol{\epsilon}}(K, \abs{\mathcal{U}}) \, \tau_{\min}(\boldsymbol{\lambda}, W)
\end{math}.\par
Equation \cref{eq:differential-privacy-doeblin-curves} and the techniques developed above provide tools for tighter privacy guarantees for various mechanisms. Using these results, we now derive differential privacy guarantees for online iterative algorithms using Doeblin curves. While prior works such as \cite{PrivacyAmplification} rely on compactness assumptions, our focus is on extending such analyses to potentially unbounded spaces where it can be difficult to capture privacy using classical contraction-coefficient-based techniques. To this end, we consider and build on the online learning framework presented in \cite[Section IV-B]{PrivacyAmplification}, where the learner sequentially minimizes a sequence of convex objectives $\ibrc{g_t}_{t=1}^n$ over a parameter space $\mathcal{W} \subseteq \mathbb{R}^d$. The protocol proceeds as follows. The learner initializes with a random parameter $W_0 \in \mathcal{W}$, and at each iteration $t \in [T]$, the cost function $g_t$ is revealed. The learner then performs the update
\begin{equation}
    W_t \triangleq \proj_\mathcal{W} \mprn{W_{t-1} - \eta_t \nabla g_t(W_{t-1}) + m_t N_t} , \label{eq:dp-update-rule}
\end{equation}
where the projection operator $\proj_\mathcal{W}$ ensures that the parameters remain within the feasible set $\mathcal{W}$,\footnote{As before, our results hold with $\proj_\mathcal{W}$ being any function whose range is $\mathcal{W}$.} $\eta_t > 0$ is the learning rate, and $m_t$ scales the standard Gaussian noise $N_t \sim \mathsf{Normal}(0, I)$ (where $I$ is the identity matrix of appropriate dimension). By taking $g_t(w) \triangleq \ell(w, Z_t)$ where $\ell\colon \mathcal{W} \times \mathcal{Z} \rightarrow \mathbb{R}_+$ is a loss function, this framework subsumes one-pass ERM, where the learner approximates the solution of   \begin{math}
    \min_{w \in \mathcal{W}} \frac{1}{T} \sum_{t=1}^T \ell(w, Z_t)
\end{math}
with data points $Z_1, \tdots, Z_T \in \mathcal{Z}$ revealed sequentially. The goal is to ensure that the final iterate $W_T$ satisfies a privacy guarantee. We assume the gradient of each objective $g_t$ is uniformly bounded, i.e.,
\begin{equation}
    \forall t \in [T], \, \forall w \in \mathcal{W}, \quad \norm{\nabla g_t(w)}_2 \leq L . \label{eq:dp-amplification-gradient-bound}
\end{equation}

Existing LDP guarantees for online algorithms \cite{PrivacyAmplification} rely on the contraction coefficient of $\mathsf{E}_\gamma$-divergence, which quantifies the decay of information due to additive noise. However, for $\mathcal{W}$ with infinite diameter, the contraction coefficient for $\mathsf{E}_\gamma$-divergence degenerates to $1$. Below, we utilize the contraction bound on $\rho_{\boldsymbol{\epsilon}}$ in terms of Doeblin curves $\mathrm{F}_\Phi^\mathsf{UA}$ from \cref{theorem:contraction-of-rho-epsilon} to quantify the contraction induced by the noise $N_t$.

\begin{theorem}[Differential Privacy for Unconstrained Online Learning] \label{theorem:differential-privacy-for-unconstrained-online-learning}
    Let $W_0^{(i)} \sim P_{W_0}^{(i)}$ for $i \in [n]$ denote $n$ different initializations of the online learning process. Let $W_T^{(i)} \sim P_{W_T}^{(i)}$ be the respective output parameters after $T$ iterations of the update rule \cref{eq:dp-update-rule}. Let $\mathrm{F}_\Phi^\mathsf{UA}$ be the uniform average power-constrained Doeblin curve (with power function $M(z) = z^2$) of the $\mathsf{Normal}(0, I)$ additive noise channel. Then,   \begin{multline}
        \rho_{\boldsymbol{\epsilon}} \mprn{\bkt{P_{W_T}^{(1)}, \tdots, P_{W_T}^{(n)}}, n} \\\leq \mathrm{F}_\Phi^\mathsf{UA} \biggprn{\tcdots \mathrm{F}_\Phi^\mathsf{UA} \biggprn{\rho_{\boldsymbol{\epsilon}} \mprn{\bkt{P_{W_0}^{(1)}, \tdots, P_{W_0}^{(n)}}, n}; p_1} \tcdots; p_T}  , \label{eq:differential-privacy-for-unconstrained-online-learning}
    \end{multline}
    where for each $t \in [T]$ we define
    \begin{equation}
        p_t \triangleq \frac{1}{m_t^2} \prn{2 \max_{i \in [n]} \E{\norm{W_{t-1}^{(i)}}_2^2} + 2 \eta_t^2 L^2} . \label{eq:dp-amplification-power}
    \end{equation}
\end{theorem}

\cref{theorem:differential-privacy-for-unconstrained-online-learning} is proved in \cref{subsection:proofs-for-differential-privacy} using \cref{theorem:contraction-of-rho-epsilon}, Part 3. %
We note that the above result can also be used to derive privacy amplification bounds for online learning algorithms as in \cite{PrivacyAmplification} for more general $\mathcal{W}$. In addition, \cref{theorem:differential-privacy-for-unconstrained-online-learning} can also be perceived as a variation of \cref{lemma:recursive-bound-on-t-information} that bounds an $n$-way divergence instead of TV-information. Finally, it is worth mentioning that one could also replace the nested applications of Doeblin curves in \cref{theorem:differential-privacy-for-unconstrained-online-learning} or \cref{lemma:recursive-bound-on-t-information} with a product of power-constrained Doeblin \emph{coefficients} to obtain a looser bound that remains meaningful for unbounded feasible sets if desired.

\section{Proofs of Main Results on Doeblin Curves} \label{section:proofs-of-main-results-on-doeblin-curves}

In this section, we prove the main results presented in \cref{section:main-results-on-doeblin-curves}, pertaining to fundamental properties of Doeblin coefficients and curves.

\subsection{Proofs of Doeblin Characterizations} \label{subsection:proofs-of-doeblin-characterizations}

In this subsection, we first prove \cref{proposition:maximal-coupling-characterization-of-doeblin-coefficients,proposition:integral-characterization-of-infimum}. Next, we introduce essential preliminaries (such as lattice infimum) required for generalizing the argument to Polish spaces. Finally, we present a detailed step-by-step proof of \cref{theorem:variational-characterization-of-doeblin-coefficient}.

\begin{proof}[Proof of \cref{proposition:maximal-coupling-characterization-of-doeblin-coefficients}]
    We have   \begin{equation}
        \tau(K) \stackclap{(a)}{=} \bigwedge_{x \in \mathcal{X}} K(\mathcal{Y} \mid x) \stackclap{(b)}{=} \!\!\sup_{\mathbb{P}: Y_x \sim K(\cdot \mid x)} \!\! \P{\forall x, x' \in \mathcal{X}, \, Y_x = Y_{x'}} , \label{eq:doeblin-coefficient-max-coupling-proof}
    \end{equation}
    where (a) holds by the greatest common component characterization of Doeblin coefficient \cref{eq:doeblin-coefficient-gcc} and (b) holds by \cite[Theorem 2]{MakurSingh2023b}, \cite[Chapter 3, Theorem 7.3, p. 107]{mainbook}.
\end{proof}

We remark that for any kernel $K$, the supremum in \cref{eq:doeblin-coefficient-max-coupling-proof} is achieved by the \emph{maximal coupling} given by (cf. \cite[Eq. 34]{MakurSingh2023b}, \cite[Chapter 3, Theorem 7.3, p. 107]{mainbook})   %
\begin{equation}
    \forall x \in \mathcal{X}, \quad Y_x \triangleq \begin{cases}
        Y^* , & \text{if $I = 1$} , \\
        \tilde{Y}_x , & \text{if $I = 0$} ,
    \end{cases} \label{eq:maximal-coupling}
\end{equation}
where $\alpha = \ibigwedge_{x \in \mathcal{X}} K(\mathcal{Y} \mid x)$ and the random variables $I$, $Y^*$, and $\ibrc{\tilde{Y}_x}_{x \in \mathcal{X}}$ are sampled independently from the probability measures
\begin{gather}
    I \sim \mathsf{Bernoulli}(\alpha), \quad
    Y^* \sim \frac{\bigwedge_{x \in \mathcal{X}} K(\lwildcard \mid x)}{\alpha}, \\
    \forall x \in \mathcal{X}, \quad \tilde{Y}_x \sim \frac{K(\lwildcard \mid x) - \bigwedge_{x' \in \mathcal{X}} K(\lwildcard \mid x')}{1 - \alpha}.
\end{gather}
(Note that $Y^*$ and $\tilde{Y}_x$ are unused in the case that their respective probability measures are undefined, i.e., $\alpha = 0$ or $\alpha = 1$.)

\begin{proof}[Proof of \cref{proposition:integral-characterization-of-infimum}]
    Fix an absolutely continuous Markov kernel $K\colon \mathcal{X} \times \mathcal{F}_\mathcal{Y} \rightarrow [0, 1]$ with common dominating measure $\mu\colon \mathcal{F}_\mathcal{Y} \rightarrow \mathbb{R}_+$. Fix $x_1, \tdots, x_n \in \mathcal{X}$ and $\gamma_1, \tdots, \gamma_n \geq 0$.
    
    First, we will show that
    \begin{equation}
        \inf_{\substack{\text{$n$-partition of $\mathcal{Y}$} \\ A_1, \dots, A_n}} \sum_{i=1}^n \gamma_i K(A_i \mid x_i) \leq \int_{\mathcal{Y}} \prn{\min_{i \in [n]} \gamma_i \frac{\diff K}{\diff \mu} (\lwildcard \mid x_i)} \diff \mu .
    \end{equation}
    Let $A^*_1, \tdots, A^*_n$ be the specific partition of $\mathcal{Y}$ given by
    \begin{equation}
        A^*_i \triangleq \brc{y \in \mathcal{Y}: \min \mprn{\arg \min_{j \in [n]} \gamma_j \frac{\diff K}{\diff\mu}(y \mid x_j)} = i}
    \end{equation}
    for each $i \in [n]$, where the $\min$ operator is added to break ties. We have   \begin{align}
        &\phantom{{}={}}\inf_{\substack{\text{$n$-partition of $\mathcal{Y}$} \\ A_1, \dots, A_n}} \sum_{i=1}^n \gamma_i K(A_i \mid x_i) \\&\stackclap{(a)}{\leq} \sum_{i=1}^n \gamma_i K(A^*_i \mid x_i) \\&\stackclap{(b)}{=} \sum_{i=1}^n \int_{A^*_i} \prn{\gamma_i \frac{\diff K}{\diff \mu}(\lwildcard \mid x_i)} \diff \mu \\&\stackclap{(c)}{=} \sum_{i=1}^n \int_{A^*_i} \prn{\min_{j \in [n]} \gamma_j \frac{\diff K}{\diff\mu}(\lwildcard \mid x_j)} \diff  \mu \\&\stackclap{(d)}{=} \int_{\mathcal{Y}} \prn{\min_{i \in [n]} \gamma_i \frac{\diff K}{\diff  \mu} (\lwildcard \mid x_i)} \diff  \mu
    \end{align}
    as desired, where (a) holds by upper-bounding the infimum with a particular instance, (b) holds by definition of Radon-Nikodym derivative, (c) holds by definition of $A^*_i$, and (d) holds because $A^*_1, \tdots, A^*_n$ is a partition of $\mathcal{Y}$.
    
    Next, we will show that
    \begin{equation}
        \inf_{\substack{\text{$n$-partition of $\mathcal{Y}$} \\ A_1, \dots, A_n}} \sum_{i=1}^n \gamma_i K(A_i \mid x_i) \geq \int_{\mathcal{Y}} \prn{\min_{i \in [n]} \gamma_i \frac{\diff K}{\diff  \mu} (\lwildcard \mid x_i)} \diff  \mu  . \label{eq:infimum-to-integral-lower-bound}
    \end{equation}
    For any partition $A_1, \tdots, A_n$ of $\mathcal{Y}$, we have   \begin{align}
        \sum_{i=1}^n \gamma_i K(A_i \mid x_i) &\stackclap{(a)}{\geq} \sum_{i=1}^n \bigwedge_{j \in [n]} \gamma_j K(A_i \mid x_j) \\ &\stackclap{(b)}{=} \bigwedge_{j \in [n]} \gamma_j K(\mathcal{Y} \mid x_j) \\
        &\stackclap{(c)}{=} \sup \brc{\nu(\mathcal{Y}): \forall j \in [n], \, \gamma_j K(\lwildcard \mid x_j) \geq \nu} \\
        &\stackclap{(d)}{\geq} \int_\mathcal{Y} \prn{\min_{i \in [n]} \gamma_i \frac{\diff K}{\diff \mu}(\lwildcard \mid x_i)} \diff \mu  ,
    \end{align}
    where (a) holds because the greatest common component of a collection of measures is a lower bound on each individual measure, (b) holds because $A_1, \tdots, A_n$ is a partition of $\mathcal{Y}$, (c) holds by definition of greatest common component \cref{eq:greatest-common-component}, and (d) holds by lower-bounding the supremum with the specific measure
    \begin{equation}
        \forall A \in \mathcal{F}_\mathcal{Y}, \quad \nu^*(A) \triangleq \int_A \prn{\min_{i \in [n]} \gamma_i \frac{\diff K}{\diff \mu}(\lwildcard \mid x_i)} \diff \mu ,
    \end{equation}
    which satisfies $\gamma_j K(\lwildcard \mid x_j) \geq \nu^*$ for each $j \in [n]$. Lastly, since $A_1, \tdots, A_n$ was arbitrary, taking the infimum over all $n$-partitions of $\mathcal{Y}$ proves \cref{eq:infimum-to-integral-lower-bound} as desired.
\end{proof}

Next, we present some preliminaries before carrying out the proof in the general setting of Polish spaces.

\begin{definition}[Lattice Infimum {\cite[Section 2, p. 253]{HajlaszMaly2002}}] \label{definition:lattice-infimum}
    Let $\mathcal{X}$ be a Polish domain, and let $\mathcal{F}$ be the set of all bounded measurable functions $f\colon \mathcal{X} \rightarrow \mathbb{R}$. Given a measure $\mu\colon \mathcal{X} \rightarrow \mathbb{R}_+$, the lattice infimum of an arbitrary subset of functions $\mathcal{G} \subseteq \mathcal{F}$, denoted $\ilatinf_{g \in \mathcal{G}} g$, is the function in $\mathcal{F}$ such that:
    \begin{enumerate}
        \item For each $h \in \mathcal{G}$,
        \begin{equation}
            \latinf_{g \in \mathcal{G}} g \leq h \quad \text{($\mu$-a.e.)}. \label{eq:lattice-infimum-condition-one}
        \end{equation}
        \item For each $f \in \mathcal{F}$, if
        \begin{equation}
            \forall h \in \mathcal{G}, \quad \Bigprn{f \leq h \quad \text{($\mu$-a.e.)}}  , \label{eq:lattice-infimum-condition-two-antecedent}
        \end{equation}
        then
        \begin{equation}
            f \leq \latinf_{g \in \mathcal{G}} g \quad \text{($\mu$-a.e.)}. \label{eq:lattice-infimum-condition-two-consequent}
        \end{equation}
    \end{enumerate}
\end{definition}

For any $\mathcal{G}$, the lattice infimum is $\mu$-a.e. unique. To see this, consider any two functions $f_1, f_2 \in \mathcal{F}$ satisfying \cref{definition:lattice-infimum}. Since both functions satisfy the first condition \cref{eq:lattice-infimum-condition-one}, both functions satisfy the antecedent in the second condition \cref{eq:lattice-infimum-condition-two-antecedent}, and so by the consequent in the second condition \cref{eq:lattice-infimum-condition-two-consequent}, we have $f_1 \leq f_2$ ($\mu$-a.e.) and $f_2 \leq f_1$ ($\mu$-a.e.) Moreover, if $\mathcal{G}$ is countable, the lattice infimum is simply the pointwise infimum of functions in $\mathcal{G}$, i.e.,   \begin{math}
    \ilatinf_{g \in \mathcal{G}} g(x) = \inf \brc{g(x): g \in \mathcal{G}}
\end{math}.
However, when $\mathcal{G}$ is uncountable, the lattice infimum and pointwise infimum are different in general. We remark that the notion of lattice infimum serves primarily to avoid measurability issues, and the proof of \cref{theorem:variational-characterization-of-doeblin-coefficient} can be carried out without such notions for ``well-behaved'' kernels as shown in \cref{appendix:variational-characterization-under-equicontinuity} for completeness.

The following lemma shows that the density of the greatest common component of a collection of measures is the lattice infimum of the individual measures' densities.

\begin{lemma}[Density of Greatest Common Component] \label{lemma:density-of-greatest-common-component}
    Let $\ibrc{\pi_i}_{i \in \mathcal{I}}$ be a collection of measures on a Polish space $(\mathcal{Y}, \mathcal{F}_\mathcal{Y})$, where $\mathcal{I}$ is an arbitrary index set. Assume each $\pi_i$ is dominated by a common $\sigma$-finite measure $\mu$ (i.e., $\pi_i \ll \mu$ for all $i \in \mathcal{I}$). Then, we have the following properties.
    \begin{enumerate}
        \item The greatest common component $\ibigwedge_{i \in \mathcal{I}} \pi_i$ is dominated by $\mu$, i.e., $\ibigwedge_{i \in \mathcal{I}} \pi_i \ll \mu$.
        \item The lattice infimum $\ilatinf_{i \in \mathcal{I}} \frac{\idiff \pi_i}{\idiff \mu}$ exists.
        \item The density of the greatest common component is   \begin{math}
            \frac{\idiff}{\idiff \mu} (\ibigwedge_{i \in \mathcal{I}} \pi_i) = \ilatinf_{i \in \mathcal{I}} \frac{\idiff \pi_i}{\idiff \mu}
        \end{math} ($\mu$-a.e.).
    \end{enumerate}
\end{lemma}

\begin{proof}
    Fix measures $\ibrc{\pi_i}_{i \in \mathcal{I}}$ on $(\mathcal{Y}, \mathcal{F}_\mathcal{Y})$ satisfying $\pi_i \ll \mu$ for each $i \in \mathcal{I}$.
    
    \textbf{Part 1:} For any set $A \in \mathcal{F}_\mathcal{Y}$ with $\mu(A) = 0$, we have
    \begin{equation}
        \bigwedge_{i \in \mathcal{I}} \pi_i(A) \stackclap{(a)}{\leq} \pi_{i^*}(A) \stackclap{(b)}{=} 0
    \end{equation}
    as desired, where (a) holds for any $i^* \in \mathcal{I}$ because the greatest common component of a collection of measures is a lower bound on each measure, and (b) holds because $\pi_{i^*} \ll \mu$.

    \textbf{Part 2:} We refer readers to the analogous argument in \cite[Lemma 2.6]{HajlaszMaly2002}.

    \textbf{Part 3:} For notational convenience, define a measure $\nu^*\colon \mathcal{F}_\mathcal{Y} \rightarrow \mathbb{R}_+$ as
    \begin{equation}
        \forall A \in \mathcal{F}_\mathcal{Y}, \quad \nu^*(A) \triangleq \int_A \prn{\latinf_{i \in \mathcal{I}} \frac{\diff \pi_i}{\diff \mu}} \diff \mu . \label{eq:density-of-greatest-common-component-measure}
    \end{equation}
    First, observe that for any $i \in \mathcal{I}$ and $A \in \mathcal{F}_\mathcal{Y}$, we have
    \begin{align}
        \nu^*(A) \stackclap{(a)}{\leq} \int_A \frac{\diff \pi_i}{\diff \mu} \diff \mu \stackclap{(b)}{=} \pi_i(A) ,
    \end{align}
    where (a) holds by \cref{definition:lattice-infimum}, Part 1 and (b) holds by definition of Radon-Nikodym derivative. Since $\nu^* \leq \pi_i$ for each $i \in \mathcal{I}$, it follows that for each $A \in \mathcal{F}_\mathcal{Y}$,
    \begin{equation}
        \nu^*(A) \leq \sup \brc{\nu(A): \forall i \in \mathcal{I}, \, \nu \leq \pi_i} . \label{eq:density-of-greatest-common-component-upper-bound}
    \end{equation}

    Next, observe that any measure $\nu$ satisfying $\nu \leq \pi_i$ for each $i \in \mathcal{I}$ is itself dominated by $\mu$ (i.e., $\nu \ll \mu$), because
    \begin{equation}
        \nu(A) \leq \pi_i(A) \stackclap{(a)}{=} 0
    \end{equation}
    for any set $A \in \mathcal{F}_\mathcal{Y}$ with $\mu(A) = 0$, where (a) holds because $\pi_i \ll \mu$. Hence, by the Radon-Nikodym theorem, $\frac{\idiff \nu}{\idiff \mu}$ is well-defined. Moreover, we have
    \begin{equation}
        \frac{\diff \nu}{\diff \mu} \leq \frac{\diff \pi_i}{\diff \mu} \quad \text{($\mu$-a.e.)} \label{eq:density-of-greatest-common-component-step-one}
    \end{equation}
    for each $i \in \mathcal{I}$. To see this, suppose the contrary: Assume that for some $i^* \in \mathcal{I}$, we have $\frac{\diff \nu}{\diff \mu} > \frac{\diff \pi_{i^*}}{\diff \mu}$ on some set $A^* \in \mathcal{F}_\mathcal{Y}$ with $\mu(A^*) > 0$. Then,
    \begin{equation}
        \nu(A^*) = \int_{A^*} \frac{\diff \nu}{\diff \mu} \diff \mu > \int_{A^*} \frac{\diff \pi_{i^*}}{\diff \mu} \diff \mu = \pi_{i^*}(A^*) ,
    \end{equation}
    which contradicts the fact that $\nu \leq \pi_{i^*}$. Following from \cref{eq:density-of-greatest-common-component-step-one}, we have
    \begin{equation}
        \frac{\diff \nu}{\diff \mu} \leq \latinf_{i \in \mathcal{I}} \frac{\diff \pi_i}{\diff \mu} \quad \text{($\mu$-a.e.)}
    \end{equation}
    by \cref{definition:lattice-infimum}, Part 2, and therefore $\nu \leq \nu^*$ by \cref{eq:density-of-greatest-common-component-measure}. Since $\nu$ was arbitrary, it follows that for each $A \in \mathcal{F}_\mathcal{Y}$,
    \begin{align}
        \nu^*(A) \geq \sup \brc{\nu(A): \forall i \in \mathcal{I}, \, \nu \leq \pi_i} . \label{eq:density-of-greatest-common-component-lower-bound}
    \end{align}

    Proceeding onwards, for each $A \in \mathcal{F}_\mathcal{Y}$, we have   \begin{equation}
        \int_A\! \prn{\latinf_{i \in \mathcal{I}} \frac{\diff \pi_i}{\diff \mu}\!}\! \diff \mu \stackclap{(a)}{=} \sup \brc{\nu(A): \forall i \in \mathcal{I}, \, \nu \leq \pi_i} \stackclap{(b)}{=} \!\bigwedge_{i \in \mathcal{I}} \pi_i(A) ,
    \end{equation}
    where (a) holds by combining \cref{eq:density-of-greatest-common-component-measure,eq:density-of-greatest-common-component-lower-bound,eq:density-of-greatest-common-component-upper-bound} and (b) holds by \cref{eq:greatest-common-component}. Hence, by definition of Radon-Nikodym derivative, $\frac{\idiff}{\idiff \mu} (\ibigwedge_{i \in \mathcal{I}} \pi_i) = \ilatinf_{i \in \mathcal{I}} \frac{\idiff \pi_i}{\idiff \mu}$ as desired.
\end{proof}

Finally, we prove \cref{theorem:variational-characterization-of-doeblin-coefficient}.

\begin{proof}[Proof of \cref{theorem:variational-characterization-of-doeblin-coefficient}]
    Fix an absolutely continuous Markov kernel $K\colon \mathcal{X} \times \mathcal{F}_\mathcal{Y} \rightarrow [0, 1]$ with common dominating measure $\mu\colon \mathcal{F}_\mathcal{Y} \rightarrow \mathbb{R}_+$. We have   \begin{align}
        \tau(K) &\stackclap{(a)}{=} \bigwedge_{x \in \mathcal{X}} K(\mathcal{Y} \mid x) \\&\stackclap{(b)}{=} \int_{\mathcal{Y}} \prn{\frac{\diff}{\diff \mu} \bigwedge_{x \in \mathcal{X}} K(\lwildcard \mid x)} \diff \mu\\  &\stackclap{(c)}{=} \int_{\mathcal{Y}} \prn{\latinf_{x \in \mathcal{X}} \frac{\diff K}{\diff \mu}(\lwildcard \mid x)} \diff \mu  , \label{eq:variational-characterization-step-one}
    \end{align}
    where (a) holds by the greatest common component characterization of Doeblin coefficient \cref{eq:doeblin-coefficient-gcc}, (b) holds by definition of Radon-Nikodym derivative, and (c) holds by \cref{lemma:density-of-greatest-common-component}, Part 3.
    
    Next, we compute the lattice infimum $\ilatinf_{x \in \mathcal{X}} \frac{\idiff K}{\idiff \mu}(\lwildcard \mid x)$. For notational convenience, let
    \begin{align}
        \circled{1} \triangleq \inf_{n \in \mathbb{N}} \inf_{x_1, \dots, x_n \in \mathcal{X}} \int_{\mathcal{Y}} \prn{\min_{i \in [n]} \frac{\diff K}{\diff \mu}(\lwildcard \mid x_i)} \diff \mu . \label{eq:variational-characterization-step-two}
    \end{align}
    For each $k \in \mathbb{N}$, let $s_k \triangleq \ibrc{\hat{x}_{k,1}, \tdots, \hat{x}_{k,n_k}} \subseteq \mathcal{X}$ be a finite sequence such that
    \begin{align}
        \int_{\mathcal{Y}} \prn{\min_{i \in [n_k]} \frac{\diff K}{\diff \mu}(\lwildcard \mid \hat{x}_{k,i})} \diff \mu \leq \circled{1} + \frac{1}{k} ,
    \end{align}
    where the existence of such a sequence is guaranteed by the definition of $\circled{1}$ as an infimum in \cref{eq:variational-characterization-step-two}. Construct an infinite sequence $\ibrc{\tilde{x}_1, \tilde{x}_2, \tdots} \subseteq \mathcal{X}$ by concatenating the finite sequences $s_k$ in ascending order of $k$. Define a sequence of functions $g_n\colon \mathcal{Y} \rightarrow \mathbb{R}$ as
    \begin{equation}
        \forall n \in \mathbb{N}, \quad g_n(y) \triangleq \min_{i \in [n]} \frac{\diff K}{\diff \mu}(y \mid \tilde{x}_i) . \label{eq:variational-characterization-gn}
    \end{equation}
    By construction, it holds that
    \begin{equation}
        \lim_{n \rightarrow \infty} \int_{\mathcal{Y}} g_n \, d \mu = \circled{1} . \label{eq:variational-characterization-gn-integral}
    \end{equation}
    Moreover, each $g_n$ is non-negative by the non-negativity of the Radon-Nikodym derivative, and $\ibrc{g_n}_{n \in \mathbb{N}}$ is a non-increasing sequence since each successive $g_n$ is the minimum over a larger set of $i$. Hence, $\ibrc{g_n}_{n \in \mathbb{N}}$ converges $\mu$-a.e. to a measurable function $g\colon \mathcal{Y} \rightarrow \mathbb{R}$, i.e.,
    \begin{equation}
        \lim_{n \rightarrow \infty} g_n = g \quad \text{($\mu$-a.e.)}. \label{eq:variational-characterization-g}
    \end{equation}
    The limiting function $g$ satisfies
    \begin{equation}
        \int_{\mathcal{Y}} g \diff \mu \stackclap{(a)}{=} \lim_{n \rightarrow \infty} \int_{\mathcal{Y}} g_n \diff \mu \stackclap{(b)}{=} \circled{1}  , \label{eq:variational-characterization-g-integral}
    \end{equation}
    where (a) holds by the dominated convergence theorem because $g_n \leq g_1$ for all $n \in \mathbb{N}$ and   \begin{math}
        \int_{\mathcal{Y}} g_1 \diff \mu = K(\mathcal{Y} \mid \tilde{x}_1) = 1 < \infty
    \end{math},
    and (b) holds by \cref{eq:variational-characterization-gn-integral}.
    
    Observe that for every $x \in \mathcal{X}$, we have
    \begin{equation}
        g \leq \frac{\diff K}{\diff \mu}(\lwildcard \mid x) \quad \text{($\mu$-a.e.).}
    \end{equation}
    To see this, suppose the contrary: assume that $g > \frac{\idiff K}{\idiff \mu}(\lwildcard \mid x^*)$ for some $x^* \in \mathcal{X}$, on some set $A^* \in \mathcal{F}_\mathcal{Y}$ with positive measure $\mu(A^*) > 0$. Then, for any arbitrary $\epsilon > 0$, it follows that   \begin{align}
        \circled{1} &\stackclap{(a)}{=} \int_{A^*} g \diff \mu + \int_{A^{*\complement}} g \diff \mu \\&\stackclap{(b)}{>} \left.\underbracket{\int_{A^*} \frac{\diff K}{\diff \mu}(\lwildcard \mid x^*) \diff \mu + \int_{A^{*\complement}} g \diff \mu}\right._{\mathrlap{\circled{2}}}\\&\stackclap{(c)}{=} \int_{A^*} \frac{\diff K}{\diff \mu}(\lwildcard \mid x^*) \diff \mu + \int_{A^{*\complement}} \!\prn{\lim_{n \rightarrow \infty} \min_{i \in [n]} \frac{\diff K}{\diff \mu}(\lwildcard \mid \tilde{x}_i)}\! \diff \mu \\
        &\stackclap{(d)}{\geq} \int_{\mathcal{Y}} \prn{\lim_{n \rightarrow \infty} \min_{i \in [n]} \frac{\diff K}{\diff \mu}(\lwildcard \mid x^*_i)} \diff \mu \\&\stackclap{(e)}{=} \lim_{n \rightarrow \infty} \int_{\mathcal{Y}} \prn{\min_{i \in [n]} \frac{\diff K}{\diff \mu}(\lwildcard \mid x^*_i)} \diff \mu \\&\stackclap{(f)}{\geq} \int_{\mathcal{Y}} \prn{\min_{i \in [n^*]} \frac{\diff K}{\diff \mu}(\lwildcard \mid x^*_i)} \diff \mu - \epsilon \\&\stackclap{(g)}{\geq} \circled{1} - \epsilon
    \end{align}
    and we obtain the contradiction $\circled{1} > \circled{2} \geq \circled{1}$, where (a) holds by \cref{eq:variational-characterization-g-integral} and because $\mathcal{Y} = A^* \cup A^{*\complement}$, (b) holds by the supposition, (c) holds by \cref{eq:variational-characterization-gn,eq:variational-characterization-g}, (d) holds for the sequence $x^*_i$ given by $x^*_1 = x^*$ and $x^*_i = \tilde{x}_{i-1}$ for all $i \geq 2$, (e) holds by the dominated convergence theorem, (f) holds for some $n^*$ (possibly depending on $\epsilon$) by definition of limit, and (g) holds by lower-bounding the value for the specific sequence $\ibrc{x^*_1, \tdots, x^*_{n^*}}$ with the infimum over all finite sequences.

    Similarly, observe that any function $h\colon \mathcal{Y} \rightarrow \mathbb{R}$ where $h \leq \frac{\idiff K}{\idiff \mu}(\lwildcard \mid x)$ $\mu$-a.e. for every $x \in \mathcal{X}$ satisfies $h \leq g$ $\mu$-a.e. To see this, suppose the contrary: assume that $h \leq \frac{\idiff K}{\idiff \mu}(\lwildcard \mid x)$ $\mu$-a.e. for every $x \in \mathcal{X}$, and $h > g$ on some set $A^* \in \mathcal{F}_\mathcal{Y}$ with positive measure $\mu(A^*) > 0$. Then, for any arbitrary $\epsilon > 0$, it follows that   \begin{align}
        \circled{1} &\stackclap{(a)}{=} \int_{A^*} g \diff \mu + \int_{A^{*\complement}} g \diff \mu \\&\stackclap{(b)}{<} \left. \underbracket{\int_{A^*} h \diff \mu + \int_{A^{*\complement}} g \diff \mu} \right._{\mathrlap{\circled{3}}} \\ &\stackclap{(c)}{\leq} \int_{A^*} \prn{\min_{i \in [n^*]} \frac{\diff K}{\diff \mu}(\lwildcard \mid \tilde{x}_i)} \diff \mu + \int_{A^{*\complement}} g \diff \mu \\
        &\stackclap{(d)}{\leq} \lim_{n \rightarrow \infty} \int_{A^*} \prn{\min_{i \in [n]} \frac{\diff K}{\diff \mu}(\lwildcard \mid \tilde{x}_i)} \diff \mu + \epsilon + \int_{A^{*\complement}} g \diff \mu \\&\stackclap{(e)}{=} \int_{A^*} \prn{\lim_{n \rightarrow \infty} \min_{i \in [n]} \frac{\diff K}{\diff \mu}(\lwildcard \mid \tilde{x}_i)} \diff \mu + \epsilon + \int_{A^{*\complement}} g \diff \mu \\&\stackclap{(f)}{=} \int_{A^*} g \diff \mu + \epsilon + \int_{A^{*\complement}} g \diff \mu \\&\stackclap{(g)}{=} \circled{1} + \epsilon
    \end{align}
    and we obtain the contradiction $\circled{1} < \circled{3} \leq \circled{1}$, where (a) holds by \cref{eq:variational-characterization-g-integral} and because $\mathcal{Y} = A^* \cup A^{*\complement}$, (b) and (c) hold by the supposition, (c) holds for \emph{any} finite $n^* \in \mathbb{N}$ because the finite union of measure-zero sets has measure zero, there always exists some $n^*$ (possibly depending on $\epsilon$) to make (d) hold by definition of limit, (e) holds by the dominated convergence theorem, (f) holds by \cref{eq:variational-characterization-gn,eq:variational-characterization-g}, and (g) holds by \cref{eq:variational-characterization-g-integral} and because $\mathcal{Y} = A^* \cup A^{*\complement}$.

    Hence, by definition of lattice infimum (\cref{definition:lattice-infimum}), $g = \ilatinf_{x \in \mathcal{X}} \frac{\diff K}{\diff \mu}(\lwildcard \mid x)$. Proceeding from \cref{eq:variational-characterization-step-one}, we have   \begin{align}
        \tau(K) &= \int_{\mathcal{Y}} g \diff \mu \\&\stackrel{(a)}{=} \circled{1} \\&\stackrel{(b)}{=} \inf_{n \in \mathbb{N}} \inf_{x_1, \dots, x_n \in \mathcal{X}} \inf_{\substack{\text{$n$-partition of $\mathcal{Y}$} \\ A_1, \dots, A_n}} \sum_{i=1}^n K(A_i \mid x_i)
    \end{align}
    as desired, where (a) holds by \cref{eq:variational-characterization-g-integral} and (b) holds by \cref{eq:variational-characterization-step-two,proposition:integral-characterization-of-infimum}.
\end{proof}

\subsection{Proofs of Doeblin Curve Properties} \label{subsection:proofs-of-doeblin-curve-properties}

In this subsection, we first prove \cref{proposition:properties-of-doeblin-curves}. We begin by establishing two technical lemmata used in the proofs of \cref{proposition:properties-of-doeblin-curves} and subsequent results.

\begin{lemma}[Doeblin Coefficients Under Affine Transformation] \label{lemma:doeblin-coefficients-under-affine-transformation}
    Let $W\colon \mathcal{U} \times \mathcal{F}_\mathcal{X} \rightarrow [0, 1]$ and $K\colon \mathcal{X} \times \mathcal{F}_\mathcal{Y} \rightarrow [0, 1]$ be Markov kernels. For any probability measures $\pi, \mu\colon \mathcal{F}_\mathcal{X} \rightarrow [0, 1]$ and scalars $\lambda, \alpha \geq 0$ such that $W - \alpha \pi \geq 0$ and $1 - \lambda \bar{\alpha} \geq 0$ (where we define $\bar{\alpha} \triangleq 1 - \alpha$ for convenience), consider the Markov kernel $\hat{W}\colon \mathcal{U} \times \mathcal{F}_\mathcal{X} \rightarrow [0, 1]$ given by   \begin{math}
        \hat{W} = \lambda \prn{W - \alpha \pi} + \prn{1 - \lambda \bar{\alpha}} \mu
    \end{math}.
    Then, the complementary Doeblin coefficients of $\hat{W}$ and $\hat{W} K$ are   $\rho(\hat{W}) = \lambda \, \rho(W)$ and $\rho(\hat{W} K) = \lambda \, \rho(WK)$.
\end{lemma}

\begin{proof}
    Fix Markov kernels $W$ and $K$, probability measures $\pi$ and $\mu$, and scalars $\lambda$ and $\alpha$ satisfying the conditions in \cref{lemma:doeblin-coefficients-under-affine-transformation}. The complementary Doeblin coefficient of $\hat{W}$ is   \begin{align}
        &\phantom{{}={}}\rho(\hat{W})
        \\&\stackclap{(a)}{=} 1 - \bigwedge_{u \in \mathcal{U}} \hat{W}(\mathcal{X} \mid u) \\&\stackclap{(b)}{=} 1 - \prn{\lambda \prn{\bigwedge_{u \in \mathcal{U}} W(\mathcal{X} \mid u) - \alpha \pi(\mathcal{X})} + \prn{1 - \lambda \bar{\alpha}} \mu(\mathcal{X})} \\
        &\stackclap{(c)}{=} 1 - \bigprn{\lambda \prn{\tau(W) - \alpha \pi(\mathcal{X})} + \prn{1 - \lambda \bar{\alpha}} \mu(\mathcal{X})} \\&\stackclap{(d)}{=} 1 - \bigprn{\lambda \prn{\tau(W) - \alpha} + \prn{1 - \lambda \bar{\alpha}}} \\&= \lambda \, \rho(W)
    \end{align}
    as desired, where (a) and (c) hold by the greatest common component characterization of Doeblin coefficient \cref{eq:doeblin-coefficient-gcc}, (b) holds by \cref{lemma:greatest-common-component-under-affine-transformation}, and (d) holds because $\pi$ and $\mu$ are probability measures. The composition of $\hat{W}$ with $K$ is   \begin{multline}
        \forall u \in \mathcal{U}, \\\hat{W} K(\lwildcard \mid u) = \lambda \bigprn{WK(\lwildcard \mid u) - \alpha \pi K(\wildcard)} + \prn{1 - \lambda \bar{\alpha}} \mu K(\wildcard) , \label{eq:doeblin-coefficients-under-affine-transformation-wk}
    \end{multline}
    so the complementary Doeblin coefficient of $\hat{W} K$ is   \begin{align}
        \rho(\hat{W} K)
        &\stackclap{(a)}{=} 1 - \bigwedge_{u \in \mathcal{U}} \hat{W} K(\mathcal{Y} \mid u) \\&\stackclap{(b)}{=} 1 - \Biggl( \lambda \prn{\bigwedge_{u \in \mathcal{U}} WK(\mathcal{Y} \mid u) - \alpha \pi K(\mathcal{Y})} \\&\phantom{\stackclap{(b)}{=} 1 - \Biggl(}+ \prn{1 - \lambda \bar{\alpha}} \mu K(\mathcal{Y}) \Biggr) \\
        &\stackclap{(c)}{=} 1 - \bigprn{\lambda \prn{\tau(WK) - \alpha \pi K(\mathcal{Y})} + \prn{1 - \lambda \bar{\alpha}} \mu K(\mathcal{Y})} \\&\stackclap{(d)}{=} 1 - \bigprn{\lambda \prn{\tau(WK) - \alpha} + \prn{1 - \lambda \bar{\alpha}}} \\&= \lambda \, \rho(WK)
    \end{align}
    as desired, where (a) and (c) hold by the greatest common component characterization of Doeblin coefficient \cref{eq:doeblin-coefficient-gcc}, (b) holds by \cref{lemma:greatest-common-component-under-affine-transformation,eq:doeblin-coefficients-under-affine-transformation-wk}, and (d) holds because $\pi K$ and $\mu K$ are probability measures.
\end{proof}

\begin{lemma}[Properties of Identity Kernel] \label{lemma:properties-of-identity-kernel}
    Let $(\mathcal{U}, \mathcal{F}_\mathcal{U})$ and $(\mathcal{X}, \mathcal{F}_\mathcal{X})$ be Polish spaces with $\mathcal{U} \subseteq \mathcal{X}$. Let $K\colon \mathcal{X} \times \mathcal{F}_\mathcal{Y} \rightarrow [0, 1]$ be a Markov kernel. The identity kernel from $\mathcal{U}$ to $\mathcal{X}$, i.e., $\Xi\colon \mathcal{U} \times \mathcal{F}_\mathcal{X} \rightarrow [0, 1]$ where   \begin{math}
        \Xi(A \mid u) = \updelta_u(A)
    \end{math} for all $ u \in \mathcal{U}$ and $A \in \mathcal{F}_\mathcal{X}$,
    satisfies $\rho(\Xi) = 1$ and $\rho(\Xi K) = \rho_\mathcal{U}(K)$.
\end{lemma}

\begin{proof}
    By inspection, the greatest common component of $\Xi$ is   \begin{math}
        \ibigwedge_{u \in \mathcal{U}} \Xi(\lwildcard \mid u) = 0
    \end{math},
    and so the complementary Doeblin coefficient of $\Xi$, by \cref{eq:doeblin-coefficient-gcc}, is   \begin{math}
        \rho(\Xi) = 1 - \ibigwedge_{u \in \mathcal{U}} \Xi(\mathcal{X} \mid u) = 1
    \end{math}
    as desired. By inspection, the composition of $\Xi$ with $K$ is   \begin{math}
        \Xi K(A \mid u) = K(A \mid u)
    \end{math} for all $u \in \mathcal{U}$ and $A \in \mathcal{F}_\mathcal{Y}$,
    and so the complementary Doeblin coefficient of $\Xi K$ is   \begin{math}
        \rho(\Xi K) = 1 - \ibigwedge_{u \in \mathcal{U}} K(\mathcal{Y} \mid u) = \rho_\mathcal{U}(K)
    \end{math}
    as desired.
\end{proof}

Now, we are ready to prove \cref{proposition:properties-of-doeblin-curves}.

\begin{proof}[Proof of \cref{proposition:properties-of-doeblin-curves}] ~\newline
    \indent \textbf{Part 1:} Fix Markov kernels $K_1$ and $K_2$ and constraint sets $\mathcal{G}_1$ and $\mathcal{G}_2$. We have   \begin{align}
        \mathrm{F}_{K_1 K_2}(t; \mathcal{G}) &\stackclap{(a)}{=} \sup \brc{\rho(WK_1 K_2)\!: \rho(W) \leq t, \, W \in \mathcal{G}} \\&\stackclap{(b)}{\leq} \sup \brc{\rho(VK_2): \rho(V) \leq \mathrm{F}_{K_1}(t; \mathcal{G}_1), V \in \mathcal{G}_2} \\&\stackclap{(c)}{=} \mathrm{F}_{K_2}(\mathrm{F}_{K_1}(t; \mathcal{G}_1); \mathcal{G}_2)
    \end{align}
    as desired, where (a) and (c) hold by definition of Doeblin curve \cref{eq:doeblin-curve}, and (b) holds because any kernel $W \in \mathcal{G}$ with $\rho(W) \leq t$ satisfies   \begin{math}
        \rho(WK_1) \leq \sup \brc{\rho(WK_1): \rho(W) \leq t, \, W \in \mathcal{G}_1} = \mathrm{F}_{K_1}(t; \mathcal{G}_1)
    \end{math}
    and $WK_1 \in \mathcal{G}_2$.

    \textbf{Part 2:} Fix a Markov kernel $K$ and a convex constraint set $\mathcal{G}$ containing the identity kernel $\Xi(\lwildcard \mid u) = \updelta_u(\wildcard)$ and a constant kernel $\bar{W}(\lwildcard \mid u) = \pi(\wildcard)$. For $t = 0$, we trivially have   \begin{math}
        \mathrm{F}_K(t; \mathcal{G}) = 0 = \rho(K) \, t
    \end{math},
    so fix $t \in (0, 1]$ for the remainder of this part. Let $\hat{W}\colon \mathcal{X} \times \mathcal{F}_\mathcal{X} \rightarrow [0, 1]$ be the Markov kernel given by   \begin{math}
        \hat{W}(\lwildcard \mid u) \triangleq t \, \Xi(\lwildcard \mid u) + \prn{1 - t} \bar{W}(\lwildcard \mid u) = t \, \Xi(\lwildcard \mid u) + \iprn{1 - t} \pi(\wildcard)
    \end{math} for all $u \in \mathcal{X}$.
    The complementary Doeblin coefficients of $\hat{W}$ and $\hat{W} K$ are
    \begin{align}
        \rho(\hat{W}) &\stackclap{(a)}{=} t \, \rho(\Xi) \stackclap{(b)}{=} t , \\
        \rho(\hat{W} K) &\stackclap{(a)}{=} t \, \rho(\Xi K) \stackclap{(b)}{=} t \, \rho_\mathcal{X}(K) = t \, \rho(K)  ,
    \end{align}
    where in each line (a) holds by \cref{lemma:doeblin-coefficients-under-affine-transformation} and (b) holds by \cref{lemma:properties-of-identity-kernel}. By convexity of $\mathcal{G}$, we have $\hat{W} \in \mathcal{G}$. Hence,   \begin{align}
        \mathrm{F}_K(t; \mathcal{G}) &\stackclap{(a)}{=} \sup \brc{\rho(WK): \rho(W) \leq t, \, W \in \mathcal{G}} \\&\geq \rho(\hat{W} K) = t \, \rho(K) ,
    \end{align}
    where (a) holds by definition of Doeblin curve \cref{eq:doeblin-curve}. Combining this lower bound with \cref{eq:doeblin-curve-submultiplicativity-bound} completes the proof.

    \textbf{Part 3:} Fix a Markov kernel $K$, a convex constraint set $\mathcal{G}$ containing a constant kernel $\bar{W}(\lwildcard \mid u) = \pi(\wildcard)$, and scalars $\lambda \in [0, 1]$ and $t \in [0, 1]$. Fix an arbitrary $\epsilon > 0$. By definition of Doeblin curve as a supremum \cref{eq:doeblin-curve}, there exists a Markov kernel $W^* \in \mathcal{G}$ such that $\rho(W^*) \leq t$ and $\rho(W^* K) \geq \mathrm{F}_K(t; \mathcal{G}) - \epsilon$. Define another Markov kernel $\hat{W}$ as   \begin{math}
        \hat{W}(\lwildcard \mid u) \triangleq \lambda W^*(\lwildcard \mid u) + \prn{1 - \lambda} \bar{W}(\lwildcard \mid u) = \lambda W^*(\lwildcard \mid u) + \prn{1 - \lambda} \pi(\wildcard)
    \end{math} for all $u \in \mathcal{U}$.
    By \cref{lemma:doeblin-coefficients-under-affine-transformation}, we have   $\rho(\hat{W}) = \lambda \, \rho(W^*) \leq \lambda t$ and $\rho(\hat{W} K) = \lambda \, \rho(W^* K) \geq \lambda \prn{\mathrm{F}_K(t; \mathcal{G}) - \epsilon}$.
    By convexity of $\mathcal{G}$, we have $\hat{W} \in \mathcal{G}$. Hence,   \begin{align}
        \mathrm{F}_K(\lambda t; \mathcal{G}) &\stackclap{(a)}{=} \sup \brc{\rho(WK): \rho(W) \leq \lambda t, \, W \in \mathcal{G}} \\&\geq \rho(\hat{W} K) \\&\geq \lambda \prn{\mathrm{F}_K(t; \mathcal{G}) - \epsilon} , \label{eq:properties-of-doeblin-curves-super-homogeneity}
    \end{align}
    where (a) holds by definition of Doeblin curve \cref{eq:doeblin-curve}. Since $\epsilon > 0$ was arbitrary, we have $\mathrm{F}_K(\lambda t; \mathcal{G}) \geq \lambda \mathrm{F}_K(t; \mathcal{G})$ as desired.
    
    Finally, to show the equivalent characterization that $t \mapsto \mathrm{F}_K(t; \mathcal{G}) / t$ is non-increasing, fix $s, t \in (0, 1]$ such that $s < t$ and observe that
    \begin{equation}
        \frac{\mathrm{F}_K(s; \mathcal{G})}{s} \stackclap{(a)}{\geq} \prn{\frac{s}{t}} \frac{\mathrm{F}_K(t; \mathcal{G})}{s} = \frac{\mathrm{F}_K(t; \mathcal{G})}{t} ,
    \end{equation}
    where (a) holds by applying \cref{eq:properties-of-doeblin-curves-super-homogeneity} with $\lambda \triangleq s/t < 1$.
    
    \textbf{Part 4:} By definition of Lipschitz continuity, we want to show that   \begin{math}
        \abs{\mathrm{F}_K(s; \mathcal{G}) - \mathrm{F}_K(t; \mathcal{G})} \leq \abs{s - t}
    \end{math} for all $s, t \in [0, 1]$.
    Fix $s, t \in [0, 1]$ such that $s < t$. If $s = 0$, the result trivially follows from \cref{eq:doeblin-curve-trivial-bound}. Otherwise, we have   \begin{align}
        \abs{\mathrm{F}_K(s; \mathcal{G}) - \mathrm{F}_K(t; \mathcal{G})} &\stackclap{(a)}{=} \frac{t}{s} \prn{\frac{s}{t}} \mathrm{F}_K(t; \mathcal{G}) - \mathrm{F}_K(s; \mathcal{G}) \\&\stackclap{(b)}{\leq} \frac{t}{s} \, \mathrm{F}_K(s; \mathcal{G}) - \mathrm{F}_K(s; \mathcal{G}) \\&= (t - s) \, \frac{\mathrm{F}_K(s; \mathcal{G})}{s} \\&\stackclap{(c)}{\leq} t - s \\&\stackclap{(d)}{=} \abs{s - t}
    \end{align}
    as desired, where (a) holds because $\mathrm{F}_K(s; \mathcal{G}) \leq \mathrm{F}_K(t; \mathcal{G})$ by monotonicity of Doeblin curves, (b) holds by applying Part 3 with $\lambda \triangleq s/t < 1$, (c) holds because $\mathrm{F}_K(s) \leq s$ by \cref{eq:doeblin-curve-trivial-bound}, and (d) holds because $s < t$.
\end{proof}

Next, we prove \cref{proposition:joint-range-of-2-by-n-discrete-channels}.

\begin{proof}[Proof of \cref{proposition:joint-range-of-2-by-n-discrete-channels}]
    Fix row stochastic matrices $\mathbf{K} \in \mathbb{R}_\mathsf{sto}^{2 \times n}$ and $\mathbf{W} \in \mathbb{R}_\mathsf{sto}^{m \times 2}$. By definition of joint range, it suffices to show that $\rho(\mathbf{W} \mathbf{K}) = \rho(\mathbf{W}) \rho(\mathbf{K})$. We have   \begin{align}
        &\rho(\mathbf{W} \mathbf{K}) \stackclap{(a)}{=} 1 - \sum_{j=1}^n \min_{i \in [m]} \bkt{\mathbf{W} \mathbf{K}}_{i,j} \\&\stackclap{(b)}{=} 1 - \sum_{j=1}^n \min_{i \in [m]} \mbrc{\bkt{\mathbf{W}}_{i,1} \bkt{\mathbf{K}}_{1,j} + \prn{1 - \bkt{\mathbf{W}}_{i,1}} \bkt{\mathbf{K}}_{2,j}} \\
        &= 1 - \sum_{j=1}^n \min_{i \in [m]} \mbrc{\bkt{\mathbf{W}}_{i,1} \prn{\bkt{\mathbf{K}}_{1,j} - \bkt{\mathbf{K}}_{2,j}}} - \sum_{j=1}^n \bkt{\mathbf{K}}_{2,j} \\&\stackclap{(c)}{=} \sum_{j=1}^n \max_{i \in [m]} \mbrc{\bkt{\mathbf{W}}_{i,1} \prn{\bkt{\mathbf{K}}_{2,j} - \bkt{\mathbf{K}}_{1,j}}}  , \label{eq:joint-range-step-one}
    \end{align}
    where (a) holds by definition of complementary Doeblin coefficient \cref{eq:doeblin-coefficient-matrix}, (b) holds because $\mathbf{W}$ is row stochastic, and (c) holds because $\mathbf{K}$ is row stochastic. Next, let $\mathcal{S} = \ibrc{j \in [n]: [\mathbf{K}]_{2,j} \geq [\mathbf{K}]_{1,j}}$ akin to \cite[Proposition 4.2]{LevinPeresWilmer2009}. Proceeding from \cref{eq:joint-range-step-one}, we have   \begin{align}
        \rho(\mathbf{W} \mathbf{K}) =& \sum_{j \in \mathcal{S}} \prn{\bkt{\mathbf{K}}_{2,j} - \bkt{\mathbf{K}}_{1,j}} \max_{i \in [m]} \bkt{\mathbf{W}}_{i,1} \\&+ \sum_{j \in \mathcal{S}^\complement} \prn{\bkt{\mathbf{K}}_{2,j} - \bkt{\mathbf{K}}_{1,j}} \min_{i \in [m]} \bkt{\mathbf{W}}_{i,1} . \label{eq:joint-range-step-two}
    \end{align}
    Since $\mathbf{K}$ is row stochastic, it holds that
    \begin{equation}
        \sum_{j=1}^n \bkt{\mathbf{K}}_{1,j} \stackclap{(a)}{=} \sum_{j=1}^n \bkt{\mathbf{K}}_{2,j} = 1 .
    \end{equation}
    Rearranging (a), we obtain
    \begin{equation}
        \sum_{j \in \mathcal{S}} \prn{\bkt{\mathbf{K}}_{2,j} - \bkt{\mathbf{K}}_{1,j}} = \sum_{j \in \mathcal{S}^\complement} \prn{\bkt{\mathbf{K}}_{1,j} - \bkt{\mathbf{K}}_{2,j}} . \label{eq:joint-range-rearrangement}
    \end{equation}
    Combining \cref{eq:joint-range-step-two,eq:joint-range-rearrangement},   \begin{multline}
        \rho(\mathbf{W} \mathbf{K}) \\= \underbracket{\Biggprn{\smash{\sum_{j \in \mathcal{S}} \prn{\bkt{\mathbf{K}}_{2,j} - \bkt{\mathbf{K}}_{1,j}}}\!\!}}_{\circled{1}}\underbracket{\Biggprn{\max_{i \in [m]} \bkt{\mathbf{W}}_{i,1} - \min_{i \in [m]} \bkt{\mathbf{W}}_{i,1}}}_{\circled{2}} \!\!.
    \end{multline}
    Next, we evaluate $\circled{1}$. We have   \begin{align}
        \circled{1} &\stackclap{(a)}{=} \frac{1}{2} \Biggprn{\sum_{j \in \mathcal{S}} \prn{\bkt{\mathbf{K}}_{2,j} - \bkt{\mathbf{K}}_{1,j}} + \sum_{j \in \mathcal{S}^\complement} \prn{\bkt{\mathbf{K}}_{1,j} - \bkt{\mathbf{K}}_{2,j}}} \\&\stackclap{(b)}{=} \frac{1}{2} \sum_{j=1}^n \mprn{\max_{\ell \in [2]} \bkt{\mathbf{K}}_{\ell,j} - \min_{\ell \in [2]} \bkt{\mathbf{K}}_{\ell,j}} \\
        &= \frac{1}{2} \sum_{j=1}^n \Biggprn{\sum_{\ell=1}^2 \bkt{\mathbf{K}}_{\ell,j} - 2 \min_{\ell \in [2]} \bkt{\mathbf{K}}_{\ell,j}} \\&\stackclap{(c)}{=} 1 - \sum_{j=1}^n \min_{\ell \in [2]} \bkt{\mathbf{K}}_{\ell,j} \stackclap{(d)}{=} \rho(\mathbf{K}) ,
    \end{align}
    where (a) holds by \cref{eq:joint-range-rearrangement}, (b) holds by definition of $\mathcal{S}$, (c) holds because $\mathbf{K}$ is row stochastic, and (d) holds by definition of complementary Doeblin coefficient \cref{eq:doeblin-coefficient-matrix}. Next, we evaluate $\circled{2}$. We have   \begin{align}
        \circled{2} &\stackclap{(a)}{=} \prn{1 - \min_{i \in [m]} \bkt{\mathbf{W}}_{i,2}} - \min_{i \in [m]} \bkt{\mathbf{W}}_{i,1} \\&= 1 - \sum_{\ell=1}^2 \min_{i \in [m]} \bkt{\mathbf{W}}_{i,\ell} \stackclap{(b)}{=} \rho(\mathbf{W}) ,
    \end{align}
    where (a) holds because $\mathbf{W}$ is row stochastic and (b) holds by definition of complementary Doeblin coefficient \cref{eq:doeblin-coefficient-matrix}. Hence, $\rho(\mathbf{W} \mathbf{K}) = \rho(\mathbf{W}) \rho(\mathbf{K})$ as desired.
\end{proof}

Next, we prove \cref{proposition:homogeneity-of-power-constrained-doeblin-curves}. We begin by establishing an additional technical lemma used in the proof of \cref{proposition:homogeneity-of-power-constrained-doeblin-curves}.

\begin{lemma}[Power Levels Under Affine Transformation] \label{lemma:power-levels-under-affine-transformation}
    Let $W\colon \mathcal{U} \times \mathcal{F}_\mathcal{X} \rightarrow [0, 1]$ be a Markov kernel. Let $\mathbf{0} \in \mathcal{X}$ denote the element in the Banach space $\mathcal{X}$ such that $\inorm{\mathbf{0}} = 0$. For any probability measure $\pi\colon \mathcal{F}_\mathcal{X} \rightarrow [0, 1]$ and scalars $\lambda, \alpha \geq 0$ such that $W - \alpha \pi \geq 0$ and $1 - \lambda \bar{\alpha} \geq 0$ (where we define $\bar{\alpha} \triangleq 1 - \alpha$ for convenience), the Markov kernel
    \begin{equation}
        \hat{W} = \lambda \prn{W - \alpha \pi} + \prn{1 - \lambda \bar{\alpha}} \updelta_\mathbf{0}
    \end{equation}
    has power levels   \begin{align}
        \itnorm{\hat{W}}_\mathsf{UA} \leq \lambda \tnorm{W}_\mathsf{UA}  , \label{eq:power-levels-under-affine-transformation-ua} \\
        \itnorm{\hat{W}}_\mathsf{AE} \leq \lambda \tnorm{W}_\mathsf{AE}  , \label{eq:power-levels-under-affine-transformation-ae}
    \end{align}
    where \cref{eq:power-levels-under-affine-transformation-ae} holds if average extremal power is well-defined for $W$. Furthermore, if $\alpha = 0$, the bounds in \cref{eq:power-levels-under-affine-transformation-ua,eq:power-levels-under-affine-transformation-ae} hold with equality.
\end{lemma}

\begin{proof}
    Fix a Markov kernel $W$, a probability measure $\pi$, and scalars $\lambda$ and $\alpha$ satisfying the conditions in \cref{lemma:power-levels-under-affine-transformation}.
    
    Firstly, the uniform average power of $\hat{W}$ is   \begin{align}
        \itnorm{\hat{W}}_\mathsf{UA} &\stackclap{(a)}{=} \lambda \sup_{u \in \mathcal{U}} \int_\mathcal{X} M(\norm{x}) \,\Bigprn{W(\diff x \mid u) - \alpha \pi(\diff x)} \\&\phantom{{}={}}+ \prn{1 - \lambda \bar{\alpha}} \int_\mathcal{X} M(\norm{x}) \, \updelta_\mathbf{0}(\diff x) \\
        &\stackclap{(b)}{=} \lambda \sup_{u \in \mathcal{U}} \int_\mathcal{X} M(\norm{x}) \,\Bigprn{W(\diff x \mid u) - \alpha \pi(\diff x)} \\&\leq \lambda \sup_{u \in \mathcal{U}} \int_\mathcal{X} M(\norm{x}) \, W(\diff x \mid u) \stackclap{(c)}{=} \lambda \tnorm{W}_\mathsf{UA}\label{eq:power-levels-under-affine-transformation-bound-ua} 
    \end{align}
    as desired, where (a) holds by definition of uniform average power \cref{eq:uniform-average-power} and linearity of Lebesgue integration with respect to the measure, (b) holds because $M(\inorm{\mathbf{0}}) = M(0) = 0$, and (c) holds by definition of uniform average power \cref{eq:uniform-average-power}.

    Secondly, to compute the average extremal power of $W$, we consider the maximal coupling of random variables $\ibrc{X_u}_{u \in \mathcal{U}}$ with $X_u \sim W(\lwildcard \mid u)$ for each $u \in \mathcal{U}$, defined in \cref{eq:maximal-coupling} as
    \begin{equation}
        \forall u \in \mathcal{U}, \quad X_u \triangleq \begin{cases}
            X^*  , & \text{if $I = 1$}  , \\
            \tilde{X}_u  , & \text{if $I = 0$}  ,
        \end{cases} \label{eq:power-levels-under-affine-transformation-maximal-coupling-w}
    \end{equation}
    where the random variables $I$, $X^*$, and $\ibrc{\tilde{X}_u}_{u \in \mathcal{U}}$ are sampled independently from the probability measures
    \begin{align}
        I &\sim \mathsf{Bernoulli}(\tau(W))  , \label{eq:power-levels-under-affine-transformation-maximal-coupling-w-i} \\
        X^* &\sim \frac{\bigwedge_{u \in \mathcal{U}} W(\lwildcard \mid u)}{\tau(W)}  , \label{eq:power-levels-under-affine-transformation-maximal-coupling-w-x-star} \\
        \forall u \in \mathcal{U}, \quad \tilde{X}_u &\sim \frac{W(\lwildcard \mid u) - \bigwedge_{u' \in \mathcal{U}} W(\lwildcard \mid u')}{\rho(W)}  . \label{eq:power-levels-under-affine-transformation-maximal-coupling-w-x-tilde}
    \end{align}
    For notational convenience, let $\nu$ be the product measure
    \begin{equation}
        \nu \triangleq \bigotimes_{u \in \mathcal{U}} \frac{W(\lwildcard \mid u) - \bigwedge_{u' \in \mathcal{U}} W(\lwildcard \mid u')}{\rho(W)} .
    \end{equation}
    The average extremal power of $W$ is   \begin{align}
        \tnorm{W}_\mathsf{AE}\! &\stackclap{(a)}{=} \P{I = 1} \, \E{\sup_{u \in \mathcal{U}} M(\norm{X_u}) \given I = 1}\\
        &\phantom{{}={}}+ \P{I = 0} \, \E{\sup_{u \in \mathcal{U}} M(\norm{X_u}) \given I = 0} \\& \stackclap{(b)}{=} \tau(W) \, \E{M(\norm{X^*})} + \rho(W) \, \E{\sup_{u \in \mathcal{U}} M \bigprn{\inorm{\tilde{X}_u}}} \\
        &\stackclap{(c)}{=} \tau(W) \!\!\int_\mathcal{X} M(\norm{x^*}) \, \frac{\bigwedge_{u \in \mathcal{U}} W(\diff x^* \mid u)}{\tau(W)} \\
        &\phantom{{}={}}+ \rho(W)\!\! \int_{\mathcal{X}^\mathcal{U}} \!\prn{\sup_{u \in \mathcal{U}} M(\norm{\tilde{x}_u})\!} \diff \nu \mprn{\brc{\tilde{x}_u}_{u \in \mathcal{U}}} \\&= \int_\mathcal{X} M(\norm{x^*}) \bigwedge_{u \in \mathcal{U}} W(\diff x^* \mid u)  \\
        &\phantom{{}={}}+ \rho(W)\!\! \int_{\mathcal{X}^\mathcal{U}} \!\prn{\sup_{u \in \mathcal{U}} M(\norm{\tilde{x}_u})\!} \diff \nu \mprn{\brc{\tilde{x}_u}_{u \in \mathcal{U}}},\label{eq:power-levels-under-affine-transformation-w-ae}
    \end{align}
    where all expectations are taken with respect to the maximal coupling \cref{eq:power-levels-under-affine-transformation-maximal-coupling-w}, (a) holds by definition of average extremal power \cref{eq:average-extremal-power} and the law of total expectation, (b) holds by \cref{eq:power-levels-under-affine-transformation-maximal-coupling-w,eq:power-levels-under-affine-transformation-maximal-coupling-w-i}, and (c) holds by \cref{eq:power-levels-under-affine-transformation-maximal-coupling-w-x-star,eq:power-levels-under-affine-transformation-maximal-coupling-w-x-tilde}.
    
    Similarly, to compute the average extremal power of $\hat{W}$, we consider the maximal coupling of random variables $\ibrc{X_u}_{u \in \mathcal{U}}$ with $X_u \sim \hat{W}(\lwildcard \mid u)$ for each $u \in \mathcal{U}$, defined in \cref{eq:maximal-coupling} as
    \begin{align}
        \forall u \in \mathcal{U}, \quad X_u \triangleq \begin{cases}
            X^*  , & \text{if $I = 1$}  , \\
            \tilde{X}_u  , & \text{if $I = 0$}  ,
        \end{cases} \label{eq:power-levels-under-affine-transformation-maximal-coupling-w-hat}
    \end{align}
    where the random variables $I$, $X^*$, and $\ibrc{\tilde{X}_u}_{u \in \mathcal{U}}$ are sampled independently from the probability measures
    \begin{align}
        I &\sim \mathsf{Bernoulli} \mprn{\tau(\hat{W})}  , \label{eq:power-levels-under-affine-transformation-maximal-coupling-w-hat-i} \\
        X^* &\sim \frac{\bigwedge_{u \in \mathcal{U}} \hat{W}(\lwildcard \mid u)}{\tau(\hat{W})}  , \label{eq:power-levels-under-affine-transformation-maximal-coupling-w-hat-x-star} \\
        \forall u \in \mathcal{U}, \quad \tilde{X}_u &\sim \frac{\hat{W}(\lwildcard \mid u) - \bigwedge_{u' \in \mathcal{U}} \hat{W}(\lwildcard \mid u')}{\rho(\hat{W})}  . \label{eq:power-levels-under-affine-transformation-maximal-coupling-w-hat-x-tilde}
    \end{align}
    By \cref{lemma:doeblin-coefficients-under-affine-transformation,lemma:greatest-common-component-under-affine-transformation} respectively, we have
    \begin{align}
        \rho(\hat{W}) &= \lambda \, \rho(W) , \label{eq:power-levels-under-affine-transformation-rho-w-hat} \\
        \bigwedge_{u \in \mathcal{U}} \hat{W}(\lwildcard \mid u) &= \lambda \Bigprn{\bigwedge_{u \in \mathcal{U}}\! W(\lwildcard \mid u) - \alpha \pi(\wildcard)} + \prn{1 - \lambda \bar{\alpha}} \updelta_\mathbf{0}(\wildcard) , \label{eq:power-levels-under-affine-transformation-gcc-w-hat}
    \end{align}
    and so,   \begin{align}
        &\phantom{{}={}}\bigotimes_{u \in \mathcal{U}} \frac{\hat{W}(\lwildcard \mid u) - \bigwedge_{u' \in \mathcal{U}} \hat{W}(\lwildcard \mid u')}{\rho(\hat{W})} \\&= \bigotimes_{u \in \mathcal{U}} \frac{\lambda \prn{W(\lwildcard \mid u) - \bigwedge_{u' \in \mathcal{U}} W(\lwildcard \mid u')}}{\lambda \, \rho(W)} = \nu .
    \end{align}
    Thus, the average extremal power of $\hat{W}$ is   \begin{align}
        \itnorm{\hat{W}}_\mathsf{AE}
        &\stackclap{(a)}{=} \P{I = 1} \, \E{\sup_{u \in \mathcal{U}} M(\norm{X_u}) \given I = 1} \\
        &\phantom{{}={}}+ \P{I = 0} \, \E{\sup_{u \in \mathcal{U}} M(\norm{X_u}) \given I = 0} \\
        &\stackclap{(b)}{=} \tau(\hat{W}) \, \E{M(\norm{X^*})} + \rho(\hat{W}) \, \E{\sup_{u \in \mathcal{U}} M \mprn{\inorm{\tilde{X}_u}}} \\
        &\stackclap{(c)}{=} \underbracket{\tau(\hat{W}) \int_\mathcal{X} M(\norm{x^*}) \, \frac{\bigwedge_{u \in \mathcal{U}} \hat{W}(\diff x^* \mid u)}{\tau(\hat{W})}} _{\circled{1}}\\&\phantom{{}={}}+
        \underbracket{\rho(\hat{W}) \int_{\mathcal{X}^\mathcal{U}} \prn{\sup_{u \in \mathcal{U}} M(\norm{\tilde{x}_u})} \diff \nu \mprn{\brc{\tilde{x}_u}_{u \in \mathcal{U}}}}_{\circled{2}}, \label{eq:power-levels-under-affine-transformation-w-hat-ae}
    \end{align}
    where all expectations are taken with respect to the maximal coupling \cref{eq:power-levels-under-affine-transformation-maximal-coupling-w-hat}, (a) holds by definition of average extremal power \cref{eq:average-extremal-power} and the law of total expectation, (b) holds by \cref{eq:power-levels-under-affine-transformation-maximal-coupling-w-hat,eq:power-levels-under-affine-transformation-maximal-coupling-w-hat-i}, and (c) holds by \cref{eq:power-levels-under-affine-transformation-maximal-coupling-w-hat-x-star,eq:power-levels-under-affine-transformation-maximal-coupling-w-hat-x-tilde}. Next, we evaluate $\circled{1}$. We have   \begin{align}
        \circled{1} &\stackclap{(a)}{=} \lambda \int_\mathcal{X} M(\norm{x^*}) \,\Bigprn{\bigwedge_{u \in \mathcal{U}} W(\diff x^* \mid u) - \alpha \pi(\diff x^*)}\\
        &\phantom{{}={}}+ \prn{1 - \lambda \bar{\alpha}} \int_\mathcal{X} M(\norm{x^*}) \, \updelta_\mathbf{0}(\diff x^*) \\
        &\stackclap{(b)}{=} \lambda \int_\mathcal{X} M(\norm{x^*}) \,\Bigprn{\bigwedge_{u \in \mathcal{U}} W(\diff x^* \mid u) - \alpha \pi(\diff x^*)} \\
        &\leq \lambda \int_\mathcal{X} M(\norm{x^*}) \bigwedge_{u \in \mathcal{U}} W(\diff x^* \mid u) , \label{eq:power-levels-under-affine-transformation-bound-ae}
    \end{align}
    where (a) holds by \cref{eq:power-levels-under-affine-transformation-gcc-w-hat} and linearity of Lebesgue integration with respect to the measure, and (b) holds because $M(\inorm{\mathbf{0}}) = M(0) = 0$. Next, we evaluate $\circled{2}$. By \cref{eq:power-levels-under-affine-transformation-rho-w-hat}, we have
    \begin{align}
        \circled{2} = \lambda \, \rho(W) \int_{\mathcal{X}^\mathcal{U}} \prn{\sup_{u \in \mathcal{U}} M(\norm{\tilde{x}_u})} \diff \nu \mprn{\brc{\tilde{x}_u}_{u \in \mathcal{U}}} . \label{eq:power-levels-under-affine-transformation-term-two}
    \end{align}
    Combining \cref{eq:power-levels-under-affine-transformation-w-ae,eq:power-levels-under-affine-transformation-w-hat-ae,eq:power-levels-under-affine-transformation-bound-ae,eq:power-levels-under-affine-transformation-term-two} yields $\tnorm{\smash{\hat{W}}}_\mathsf{AE} \leq \lambda \tnorm{W}_\mathsf{AE}$ as desired.
    
    Lastly, if $\alpha = 0$, the bounds in \cref{eq:power-levels-under-affine-transformation-bound-ua,eq:power-levels-under-affine-transformation-bound-ae} hold with equality, as desired.
\end{proof}

Now, we are ready to prove \cref{proposition:homogeneity-of-power-constrained-doeblin-curves}.

\begin{proof}[Proof of \cref{proposition:homogeneity-of-power-constrained-doeblin-curves}]
    Fix a Markov kernel $K\colon \mathcal{X} \times \mathcal{F}_\mathcal{Y} \rightarrow [0, 1]$. Fix scalars $t \in [0, 1]$, $p \in [0, \infty)$, and $\lambda > 0$ such that $\lambda t \leq 1$.\footnote{The proposition trivially holds for $\lambda = 0$, so we assume $\lambda > 0$ throughout this proof.} Throughout this proof, let $\mathrm{F}_K$ and $\itnorm{\wildcard}$ denote either $\mathrm{F}_K^\mathsf{UA}$ and $\itnorm{\wildcard}_\mathsf{UA}$, or $\mathrm{F}_K^\mathsf{AE}$ and $\itnorm{\wildcard}_\mathsf{AE}$, in general. We will consider two cases.
    
    \textit{Case 1:} $\lambda \leq 1$. We follow the argument from \cite[Proposition 2]{PolyanskiyWu2016}. First, we will show that
    \begin{equation}
        \mathrm{F}_K(\lambda t; \lambda p) \geq \lambda \mathrm{F}_K(t; p) . \label{eq:homogeneity-step-one}
    \end{equation}
    Fix any arbitrary $\epsilon > 0$. By definition of Doeblin curve, there exists a Markov kernel $W^*\colon \mathcal{U} \times \mathcal{F}_\mathcal{X} \rightarrow [0, 1]$ such that   $\rho(W^*) \leq t$, $\itnorm{W^*} \leq p$, and $\rho(W^* K) \geq \mathrm{F}_K(t; p) - \epsilon$.
    Since $\mathcal{X}$ is a Banach space, there exists an element $\mathbf{0} \in \mathcal{X}$ such that $\inorm{\mathbf{0}} = 0$; we denote this element in bold font to distinguish it from the scalar $0 \in \mathbb{R}$. Define another Markov kernel $\hat{W}\colon \mathcal{U} \times \mathcal{F}_\mathcal{X} \rightarrow [0, 1]$ as   $\hat{W}(\lwildcard \mid u) = \lambda W^*(\lwildcard \mid u) + \prn{1 - \lambda} \updelta_\mathbf{0}(\wildcard)$ for all $u \in \mathcal{U}$.
    By \cref{lemma:doeblin-coefficients-under-affine-transformation,lemma:power-levels-under-affine-transformation}, we have   \begin{align}
        \rho(\hat{W}) &= \lambda \, \rho(W^*) \leq \lambda t  , \label{eq:homogeneity-step-one-w-hat-rho} \\
        \rho(\hat{W} K) &= \lambda \, \rho(W^* K) \geq \lambda \prn{\mathrm{F}_K(t; p) - \epsilon} , \label{eq:homogeneity-step-one-w-hat-k-rho} \\
        \itnorm{\hat{W}} &= \lambda \tnorm{W^*} \leq \lambda p  . \label{eq:homogeneity-step-one-w-hat-power}
    \end{align}
    Hence,   \begin{align}
        \mathrm{F}_K(\lambda t; \lambda p) &\stackclap{(a)}{=} \sup \brc{\rho(WK): \rho(W) \leq \lambda t, \, \tnorm{W} \leq \lambda p} \\&\stackclap{(b)}{\geq} \rho(\hat{W} K) \stackclap{(c)}{\geq} \lambda \prn{\mathrm{F}_K(t; p) - \epsilon} ,
    \end{align}
    where (a) holds by definition of Doeblin curve, (b) holds by \cref{eq:homogeneity-step-one-w-hat-rho,eq:homogeneity-step-one-w-hat-power}, and (c) holds by \cref{eq:homogeneity-step-one-w-hat-k-rho}. Since $\epsilon > 0$ was arbitrary, this shows \cref{eq:homogeneity-step-one} as desired.
    
    Next, we will show that
    \begin{equation}
        \lambda \mathrm{F}_K(t; p) \geq \mathrm{F}_K(\lambda t; \lambda p) . \label{eq:homogeneity-step-two}
    \end{equation}
    Fix any arbitrary $\epsilon > 0$. By definition of Doeblin curve, there exists a Markov kernel $W^*\colon \mathcal{U} \times \mathcal{F}_\mathcal{X} \rightarrow [0, 1]$ such that   $\rho(W^*) \leq \lambda t$, $\tnorm{W^*} \leq \lambda p$, and $\rho(W^* K) \geq \mathrm{F}_K(\lambda t; \lambda p) - \epsilon$.
    Let $\alpha$ be any scalar such that $1 - \lambda \leq \alpha \leq \tau(W^*)$.\footnote{Such an $\alpha$ exists because $t \leq 1$ and $\lambda > 0$, and so $1 - \lambda \leq 1 - \lambda t \leq 1 - \rho(W^*) = \tau(W^*)$.} For notational convenience, define $\bar{\alpha} \triangleq 1 - \alpha$. By definition of Doeblin coefficient, there exists a probability measure $\pi^*$ such that $W^* \geq \alpha \pi^*$.\footnote{This is true even for $\alpha = \tau(W^*)$ because the supremum in \cref{eq:doeblin-coefficient} is always achieved, as per the discussion following \cref{eq:doeblin-coefficient}.} Define another Markov kernel $\hat{W}\colon \mathcal{U} \times \mathcal{F}_\mathcal{X} \rightarrow [0, 1]$ as   \begin{math}
        \hat{W}(\lwildcard \mid u) = \frac{1}{\lambda}\iprn{{W^*(\lwildcard \mid u)} - \alpha \pi^*(\wildcard)} + \iprn{1 - \frac{\bar{\alpha}}{\lambda}} \updelta_\mathbf{0}(\wildcard)
    \end{math} for all $u \in \mathcal{U}$.
    This is a valid Markov kernel because $W^* \geq \alpha \pi^*$, non-negative linear combinations of measures are measures,\footnote{We have $1 - \bar{\alpha} / \lambda \geq 0$ because $1 - \lambda \leq \alpha$.} and for any $u \in \mathcal{U}$,   \begin{math}
        \hat{W}(\mathcal{X} \mid u) = \frac{1}{\lambda} \bigprn{W^*(\mathcal{X} \mid u) - \alpha \pi^*(\mathcal{X})} + \bigprn{1 - \frac{\bar{\alpha}}{\lambda}} \updelta_\mathbf{0}(\mathcal{X}) = \frac{1}{\lambda} \iprn{1 - \alpha} + \iprn{1 - \frac{\bar{\alpha}}{\lambda}} = 1
    \end{math}.
    By \cref{lemma:doeblin-coefficients-under-affine-transformation,lemma:power-levels-under-affine-transformation}, we have   \begin{align}
        \rho(\hat{W}) &= \frac{\rho(W^*)}{\lambda} \leq t  , \label{eq:homogeneity-step-two-w-hat-rho} \\
        \rho(\hat{W} K) &= \frac{\rho(W^* K)}{\lambda} \geq \frac{\mathrm{F}_K(\lambda t; \lambda p) - \epsilon}{\lambda}  , \label{eq:homogeneity-step-two-w-hat-k-rho} \\
        \itnorm{\hat{W}} &= \frac{\tnorm{W^*}}{\lambda} \leq p  . \label{eq:homogeneity-step-two-w-hat-power}
    \end{align}
    Hence,   \begin{align}
        \lambda \mathrm{F}_K(t; p) &\stackclap{(a)}{=} \lambda \sup \brc{\rho(WK): \rho(W) \leq t, \, \tnorm{W} \leq p} \\&\stackclap{(b)}{\geq} \lambda \, \rho(\hat{W} K) \stackclap{(c)}{\geq} \mathrm{F}_K(\lambda t; \lambda p) - \epsilon  ,
    \end{align}
    where (a) holds by definition of Doeblin curve, (b) holds by \cref{eq:homogeneity-step-two-w-hat-rho,eq:homogeneity-step-two-w-hat-power}, and (c) holds by \cref{eq:homogeneity-step-two-w-hat-k-rho}. Since $\epsilon > 0$ was arbitrary, this shows \cref{eq:homogeneity-step-two} as desired.

    Finally, combining \cref{eq:homogeneity-step-one,eq:homogeneity-step-two}, we have $\mathrm{F}_K(\lambda t; \lambda p) = \lambda \mathrm{F}_K(t; p)$ as desired.
    
    \textit{Case 2:} $\lambda > 1$. We have   \begin{align}
        \mathrm{F}_K(\lambda t; \lambda p) &= \lambda \prn{\frac{1}{\lambda}} \mathrm{F}_K(\lambda t; \lambda p) \stackclap{(a)}{=} \lambda \mathrm{F}_K \mprn{\!\mprn{\frac{1}{\lambda}} \lambda t; \prn{\frac{1}{\lambda}} \lambda p\!} \\&= \lambda \mathrm{F}_K(t; p)
    \end{align}
    as desired, where (a) holds by applying Case 1 with $1/\lambda < 1$.
\end{proof}

Next, we prove \cref{proposition:scale-invariant-doeblin-curves-of-additive-noise-channels}.

\begin{proof}[Proof of \cref{proposition:scale-invariant-doeblin-curves-of-additive-noise-channels}]
First, we define scalar multiplication of sets as $\lambda A \triangleq \ibrc{\lambda a: a \in A}$ for any $\lambda\in\mathbb{R}$ and set $A$. Note that for all $\lambda>0$, $A \in \mathcal{F}_{\mathbb{R}^d} \iff \lambda A \in \mathcal{F}_{\mathbb{R}^d}$. Next, note that for all $A \in \mathcal{F}_\mathcal{Y}$ and $\mathbf{x} \in \mathcal{X}=\mathbb{R}^d$,
\begin{align}
K_{\sigma} (\sigma A\mid\sigma \mathbf{x})&=\int_{\sigma A} \frac{1}{\sigma^d} f\mprn{\frac{\mathbf{y}-\sigma \mathbf{x}}{\sigma}}\diff\mathbf{y}
\\&= \int_{A} f\mprn{\mathbf{y}'-\mathbf{x}}\diff\mathbf{y}'
=  K_1(A\mid\mathbf{x}). \label{eq:additive-noise-scale-invariance-proof-scaled-input-output}
\end{align}
Given any $W\colon \mathcal{U}\times \mathcal{F}_\mathcal{X}\to[0,1]$, define $W_\sigma$ such that $W_\sigma(A\mid u) \triangleq W(A/\sigma\mid u)$ (i.e., $W_\sigma(\lwildcard\mid u)$ is the pushforward measure of $W(\lwildcard\mid u)$ by $\mathbf{g}(\mathbf{x})=\sigma \mathbf{x}$). Then,
\begin{align}
    \itnorm{W_\sigma}_\mathsf{UA} &= \sup_{u\in\mathcal{U}} \int_\mathcal{X} \norm{\mathbf{x}}^2 W_\sigma(\diff\mathbf{x}\mid u) \\&= \sup_{u\in\mathcal{U}} \int_\mathcal{X} \norm{\mathbf{x}}^2 W(\diff\mathbf{x} / \sigma \mid u) \\&= \sup_{u\in\mathcal{U}} \int_\mathcal{X} \norm{\sigma \mathbf{x}'}^2 W(\diff\mathbf{x}'\mid u) \\&= \sigma^2 \itnorm{W}_\mathsf{UA}.
\end{align}
In addition,
\begin{align}
    \rho(W_\sigma) &= 1-\sup\{\alpha \in \mathbb{R}: \exists \pi \in \mathscr{P}, \forall u \in \mathcal{U}, \forall A \in \mathcal{F}_\mathcal{X},\\&\hspace{12em} W_\sigma(A\mid u) \geq \alpha \pi(A)\} \\
    &=1-\sup\{\alpha \in \mathbb{R}: \exists \pi \in \mathscr{P}, \forall u \in \mathcal{U}, \forall A' \in \mathcal{F}_\mathcal{X}, \\&\hspace{12em} W(A'\mid u) \geq \alpha \pi(\sigma A')\} \\
    &=1-\sup\{\alpha \in \mathbb{R}: \exists \pi' \in \mathscr{P}, \forall u \in \mathcal{U}, \forall A' \in \mathcal{F}_\mathcal{X}, \\&\hspace{12em} W(A'\mid u) \geq \alpha \pi'(A')\} \\&= \rho(W),
\end{align}
and similarly $\rho(K_{\sigma} W_\sigma)=\rho(K_ W)$ since $(K_{\sigma} W_\sigma) (\sigma A\mid u)=(K_1 W)(A\mid u)$ due to \cref{eq:additive-noise-scale-invariance-proof-scaled-input-output} and the definition of $W_\sigma$. Thus, $\mathrm{F}_{K_{\sigma}}^{\mathsf{UA}}(t;\sigma^2 p) = \mathrm{F}_{K_1}^{\mathsf{UA}}(t;p)$ for all $p \geq 0$, completing the proof.
\end{proof}

\subsection{Proofs of Bounds on Power-Constrained Doeblin Curves} \label{subsection:proofs-of-bounds-on-power-constrained-doeblin-curves}

First, we prove \cref{theorem:bounds-on-average-extremal-doeblin-curve}.

\begin{proof}[Proof of \cref{theorem:bounds-on-average-extremal-doeblin-curve}]
    Fix an absolutely continuous Markov kernel $K\colon \mathcal{X} \times \mathcal{F}_\mathcal{Y} \rightarrow [0, 1]$ with common dominating measure $\mu\colon \mathcal{F}_\mathcal{Y} \rightarrow \mathbb{R}_+$.

    \textbf{Upper bound:} Fix $t \in (0, 1]$ and $p \in [0, \infty)$. By definition of Doeblin curve, we want to show that for any Markov kernel $W\colon \mathcal{U} \times \mathcal{F}_\mathcal{X} \rightarrow [0, 1]$ satisfying $\rho(W) \leq t$ and $\itnorm{W}_\mathsf{AE} \leq p$, we have
    \begin{align}
        \rho(WK) \leq t \, \invbreve{\Theta} \mprn{2 \gamma \, M^{-1} \mprn{\frac{p}{t}}}  . \label{eq:bounds-step-zero}
    \end{align}
    Fix a Polish space $(\mathcal{U}, \mathcal{F}_\mathcal{U})$ and such a Markov kernel $W$. We have   %
    \begin{align}
        \rho(WK) &\stackclap{(a)}{=} 1 - \inf_{n \in \mathbb{N}} \inf_{u_1, \dots, u_n \in \mathcal{U}} \inf_{\substack{\text{$n$-partition of $\mathcal{Y}$} \\ A_1, \dots, A_n}} \sum_{i=1}^n WK(A_i \mid u_i) \\&\stackclap{(b)}{=} 1 - \inf_{n \in \mathbb{N}} \inf_{u_1, \dots, u_n \in \mathcal{U}}\\&\hspace{2em}\underbracket{\inf_{\substack{\text{$n$-partition of $\mathcal{Y}$} \\ A_1, \dots, A_n}} \sum_{i=1}^n \int_{\mathcal{X}} \!\!K(A_i \mid x) \, W(\diff x \mid u_i)}_{\mathrlap{\circled{1}}}, \label{eq:bounds-step-one}
    \end{align}
    where (a) holds by \cref{theorem:variational-characterization-of-doeblin-coefficient} and (b) holds by definition of composition of two kernels \cref{eq:kernel-composition}. Next, we lower-bound $\circled{1}$ for any $n \in \mathbb{N}$ and $u_1, \tdots, u_n \in \mathcal{U}$. Let $\mathbb{P}$ be the maximal coupling of random variables $\ibrc{X_u}_{u \in \mathcal{U}}$ with $X_u \sim W(\lwildcard \mid u)$ for each $u \in \mathcal{U}$, as defined in \cref{eq:maximal-coupling}. We have   \begin{align}
        \circled{1} &\stackclap{(a)}{=} \inf_{\substack{\text{$n$-partition of $\mathcal{Y}$} \\ A_1, \dots, A_n}} \sum_{i=1}^n \int_{\mathcal{X}} K(A_i \mid x) \, \mathbb{P}_{X_{u_i}}(\diff x) \\&\stackclap{(b)}{=} \inf_{\substack{\text{$n$-partition of $\mathcal{Y}$} \\ A_1, \dots, A_n}} \sum_{i=1}^n \E{K(A_i \mid X_{u_i})} \\
        &\stackclap{(c)}{=} \inf_{\substack{\text{$n$-partition of $\mathcal{Y}$} \\ A_1, \dots, A_n}} \E{\sum_{i=1}^n K(A_i \mid X_{u_i})} \\&\geq \E{\inf_{\substack{\text{$n$-partition of $\mathcal{Y}$} \\ A_1, \dots, A_n}} \sum_{i=1}^n K(A_i \mid X_{u_i})} \\&\stackclap{(d)}{=} \E{\int_{\mathcal{Y}} \prn{\min_{i \in [n]} \frac{\diff K}{\diff \mu}(\lwildcard \mid X_{u_i})} \diff \mu} , \label{eq:bounds-step-two}
    \end{align}
    where (a) holds by definition of a coupling, the expectations from (b) onwards are taken with respect to $\ibrc{X_u}_{u \in \mathcal{U}} \sim \mathbb{P}$, (c) holds by linearity of expectation, and (d) holds by \cref{proposition:integral-characterization-of-infimum}. Combining \cref{eq:bounds-step-one,eq:bounds-step-two}, we have   %
    \begin{equation}
        \rho(WK) \leq 1 - \underbracket{\inf_{n \in \mathbb{N}} \inf_{u_1, \dots, u_n \in \mathcal{U}} \!\!\E{\int_{\mathcal{Y}} \prn{\min_{i \in [n]} \frac{\diff K}{\diff \mu}(\lwildcard \mid X_{u_i})} \diff \mu}}_{\circled{2}}. \label{eq:bounds-step-three}
    \end{equation}
    
    Next, we lower-bound $\circled{2}$. For each $k \in \mathbb{N}$, let $s_k \triangleq \ibrc{\hat{u}_{k,1}, \tdots, \hat{u}_{k,n_k}} \subseteq \mathcal{U}$ be a finite sequence such that
    \begin{align}
        \E{\int_{\mathcal{Y}} \prn{\min_{i \in [n_k]} \frac{\diff K}{\diff \mu}(\lwildcard \mid X_{\hat{u}_{k,i}})} \diff \mu} \leq \circled{2} + \frac{1}{k} ,
    \end{align}
    where the existence of such a sequence is guaranteed by the definition of $\circled{2}$ as an infimum in \cref{eq:bounds-step-three}. Construct an infinite sequence $\ibrc{\tilde{u}_1, \tilde{u}_2, \tdots} \subseteq \mathcal{U}$ by concatenating the finite sequences $s_k$ in ascending order of $k$. Define a sequence of functions $g_n\colon \mathcal{X}^\mathcal{U} \rightarrow \mathbb{R}$ (with respect to the cylinder $\sigma$-algebra on $\mathcal{X}^\mathcal{U}$) as
    \begin{align}
        \forall n \in \mathbb{N}, \quad g_n \mprn{\brc{x_u}_{u \in \mathcal{U}}} \triangleq \int_{\mathcal{Y}} \prn{\min_{i \in [n]} \frac{\diff K}{\diff \mu}(\lwildcard \mid x_{\tilde{u}_i})} \diff \mu .
    \end{align}
    By construction, it holds that
    \begin{equation}
        \lim_{n \rightarrow \infty} \E{g_n \mprn{\brc{X_u}_{u \in \mathcal{U}}}} = \circled{2} . \label{eq:bounds-step-four}
    \end{equation}
    Each $g_n$ is non-negative, by the non-negativity of the Radon-Nikodym derivative. Hence, we may define $h\colon \mathcal{X}^\mathcal{U} \rightarrow \mathbb{R}_+$ as the infimum $h \triangleq \inf_{n \in \mathbb{N}} g_n$, i.e.,
    \begin{equation}
        h \mprn{\brc{x_u}_{u \in \mathcal{U}}} \triangleq \inf_{n \in \mathbb{N}} \int_{\mathcal{Y}} \prn{\min_{i \in [n]} \frac{\diff K}{\diff \mu}(\lwildcard \mid x_{\tilde{u}_i})} \diff \mu . \label{eq:bounds-h}
    \end{equation}
    By construction, $g_n \geq h$ for all $n \in \mathbb{N}$. Proceeding from \cref{eq:bounds-step-four}, we thus have
    \begin{equation}
        \circled{2} \geq \E{h \mprn{\brc{X_u}_{u \in \mathcal{U}}}}  . \label{eq:bounds-step-five}
    \end{equation}
    Define the event $\mathcal{E} = \ibrc{\forall u, v \in \mathcal{U}, \, X_u = X_v}$, which is measurable since $\mathcal{U}$ is Polish. Combining \cref{eq:bounds-step-three,eq:bounds-step-five}, we have   \begin{align}
        \rho(WK) &\leq 1 - \E{h \mprn{\brc{X_u}_{u \in \mathcal{U}}}} \\
        &\stackclap{(a)}{=} 1 - \P{\mathcal{E}} \, \E{h \mprn{\brc{X_u}_{u \in \mathcal{U}}} \given \mathcal{E}} \\
        &\phantom{{}={}}- \P{\mathcal{E}^\complement} \, \E{h \mprn{\brc{X_u}_{u \in \mathcal{U}}} \given \mathcal{E}^\complement} \\
        &\stackclap{(b)}{=} 1 - \P{\mathcal{E}} - \P{\mathcal{E}^\complement} \, \E{h \mprn{\brc{X_u}_{u \in \mathcal{U}}} \given \mathcal{E}^\complement} \\
        &\stackclap{(c)}{=} \P{\mathcal{E}^\complement} \underbracket{\E{1 - h \mprn{\brc{X_u}_{u \in \mathcal{U}}} \given \mathcal{E}^\complement}}_{\circled{3}} , \label{eq:bounds-step-six}
    \end{align}
    where (a) holds by the law of total expectation; (b) holds because, under the event $\mathcal{E}$, there exists some $x^* \in \mathcal{X}$ such that $X_u = x^*$ for all $u \in \mathcal{U}$ and so
    \begin{equation}
        h \mprn{\brc{X_u}_{u \in \mathcal{U}}} = \int_{\mathcal{Y}} \prn{\frac{\diff K}{\diff \mu}(\lwildcard \mid x^*)} \diff \mu = K(\mathcal{Y} \mid x^*) = 1
    \end{equation}
    given $\mathcal{E}$; and (c) holds by the complement rule of probability and linearity of expectation. If $\mathbb{P}(\mathcal{E}^\complement) = 0$, we have $\rho(WK) = 0$ and the desired upper bound \cref{eq:bounds-step-zero} is trivially satisfied. Hence, assume $\mathbb{P}(\mathcal{E}^\complement) > 0$ for the remainder of this argument.

    Next, we upper-bound $\circled{3}$. Observe that for any bounded (multi)set $\mathcal{S} = \ibrc{x_u}_{u \in \mathcal{U}} \in \mathcal{X}^\mathcal{U}$ (i.e., $\tnorm{\mathcal{S}}_\infty < \infty$) and any $\epsilon > 0$,   \begin{align}
        1 - h(\mathcal{S})\! &\stackclap{(a)}{=} 1 - \inf_{n \in \mathbb{N}} \int_{\mathcal{Y}} \prn{\min_{i \in [n]} \frac{\diff K}{\diff \mu}(\lwildcard \mid x_{\tilde{u}_i})} \diff \mu \\
        &\stackclap{(b)}{\leq} 1 - \inf_{n \in \mathbb{N}} \inf_{u_1, \dots, u_n \in \mathcal{U}} \int_{\mathcal{Y}} \prn{\min_{i \in [n]} \frac{\diff K}{\diff \mu}(\lwildcard \mid x_{u_i})} \diff \mu \\
        &\stackclap{(c)}{=} 1 - \inf_{n \in \mathbb{N}} \inf_{u_1, \dots, u_n \in \mathcal{U}} \inf_{\substack{\text{$n$-partition of $\mathcal{Y}$} \\ A_1, \dots, A_n}} \sum_{i=1}^n K(A_i \mid x_{u_i})\\
        &\stackclap{(d)}{=} \rho_{\mathcal{S}}(K) \\&\stackclap{(e)}{\leq} \rho_{\mathcal{B}(a^*, \rad(\mathcal{S}) + \epsilon)}(K) \\
        &\stackclap{(f)}{\leq} \sup_{a \in \mathcal{X}} \rho_{\mathcal{B}(a, \rad(\mathcal{S}) + \epsilon)}(K) \\&\stackclap{(g)}{\leq} \sup_{a \in \mathcal{X}} \rho_{\mathcal{B}(a, \gamma \tnorm{\mathcal{S}}_\infty\! + \epsilon)}(K)\\&\stackclap{(h)}{=} \Theta(\gamma \tnorm{\mathcal{S}}_\infty + \epsilon) \\&\stackclap{(i)}{\leq} \invbreve{\Theta}(\gamma \tnorm{\mathcal{S}}_\infty + \epsilon) ,
    \end{align}
    where (a) holds by definition of $h$ \cref{eq:bounds-h}; (b) holds by lower-bounding the integrand for the specific sequence $\ibrc{\tilde{u}_1, \tdots, \tilde{u}_n}$ with the infimum over all length-$n$ sequences; (c) holds by \cref{proposition:integral-characterization-of-infimum}; (d) holds by \cref{theorem:variational-characterization-of-doeblin-coefficient} and the definition of complementary Doeblin coefficient; (e) holds for \emph{some} center $a^* \in \mathcal{X}$ which satisfies $\sup_{x \in \mathcal{S}} \norm{x - a^*} \leq \rad(\mathcal{S}) + \epsilon$, because $\mathcal{S}' \mapsto \rho_{\mathcal{S}'}(K)$ is monotonically non-decreasing in its input set $\mathcal{S}'$; (f) holds by upper-bounding the value for the specific center $a^*$ with the supremum over \emph{all} centers $a \in \mathcal{X}$; (g) holds because $\rad(\mathcal{S}) \leq \gamma \itnorm{S}_\infty$ by rearranging \cref{eq:jung-constant};   %
    (h) holds by definition of $\Theta$; and (i) holds because $\invbreve{\Theta}$ is the upper concave envelope (and hence an upper bound) of $\Theta$. Therefore, if $\mathcal{S}$ has strictly positive diameter $\itnorm{\mathcal{S}}_\infty > 0$,   \begin{align}
        1 - h(\mathcal{S}) &\stackclap{(a)}{\leq} \invbreve{\Theta}(\gamma \tnorm{\mathcal{S}}_\infty) \\&\stackclap{(b)}{=} \invbreve{\Theta} \mprn{\gamma \sup_{u, v \in \mathcal{U}} \norm{x_u - x_v}} \\
        &\stackclap{(c)}{\leq} \invbreve{\Theta} \mprn{2 \gamma \sup_{u \in \mathcal{U}} \norm{x_u}} \\&\stackclap{(d)}{=} \invbreve{\Theta} \mprn{2 \gamma \, M^{-1} \circ M \mprn{\sup_{u \in \mathcal{U}} \norm{x_u}}} \\
        &\stackclap{(e)}{=} \invbreve{\Theta} \mprn{2 \gamma \, M^{-1} \mprn{\sup_{u \in \mathcal{U}} M(\norm{x_u})}}  , \label{eq:bounds-step-seven}
    \end{align}
    where (a) holds because $\invbreve{\Theta}$ is concave and thus continuous on its interior which contains $\gamma \tnorm{\mathcal{S}}_\infty > 0$, and because $\epsilon > 0$ was arbitrary; (b) holds by definition of diameter \cref{eq:diameter}; (c) holds by the triangle inequality, and because $\invbreve{\Theta}$ is non-decreasing by \cref{lemma:properties-of-upper-concave-envelope}, Part 2; (d) holds because $M$ is strictly increasing and hence invertible; and (e) holds because $M$ is increasing. For notational convenience, define a function $G\colon \mathbb{R}_+ \rightarrow [0, 1]$ as   \begin{math}
        G(r) \triangleq \invbreve{\Theta} \iprn{2 \gamma \, M^{-1}(r)}
    \end{math}.
    Since $\invbreve{\Theta}$ is non-decreasing and concave, and since $M^{-1}$ is increasing and concave (being the inverse of an increasing convex function), it follows that $G$ is non-decreasing and concave. Combining \cref{eq:bounds-step-six,eq:bounds-step-seven}, we have, akin to the argument of \cite[Theorem 4]{PolyanskiyWu2016},\footnote{It suffices to define $G$ over $\mathbb{R}_+$ instead of $\mathbb{R}_+ \cup \ibrc{\infty}$, because $\tnorm{W}_\mathsf{AE} < \infty$ and so $\sup_{u \in \mathcal{U}} M(\inorm{X_u}) < \infty$ almost surely. Hence, the expectation $\mathbb{E}[G(\sup_{u \in \mathcal{U}} M(\inorm{X_u}))]$ is well-defined even if $G$ is undefined at $\infty$.}   \begin{align}
        \circled{3} &\leq \E{G \mprn{\sup_{u \in \mathcal{U}} M(\norm{X_u})} \given \mathcal{E}^\complement} \\&\stackclap{(a)}{\leq} G \mprn{\E{\sup_{u \in \mathcal{U}} M(\norm{X_u}) \given \mathcal{E}^\complement}} \\&\stackclap{(b)}{\leq} G \mprn{\frac{1}{\P{\mathcal{E}^\complement}} \, \E{\sup_{u \in \mathcal{U}} M(\norm{X_u})}}\\&\stackclap{(c)}{=} G \mprn{\frac{\tnorm{W}_\mathsf{AE}}{\P{\mathcal{E}^\complement}}} , \label{eq:bounds-step-eight}
    \end{align}
    where (a) holds by Jensen's inequality and the concavity of $G$, (b) holds by the law of total expectation because   \begin{align}
        &\phantom{{}={}}\E{\sup_{u \in \mathcal{U}} M(\norm{X_u})} \\&= \P{\mathcal{E}}  \E{\sup_{u \in \mathcal{U}} M(\norm{X_u}) \!\given\! \mathcal{E}} \!+ \P{\mathcal{E}^\complement}  \E{\sup_{u \in \mathcal{U}} M(\norm{X_u}) \!\given\! \mathcal{E}^\complement} \\&\stackclap{(d)}{\geq} \P{\mathcal{E}^\complement}  \E{\sup_{u \in \mathcal{U}} M(\norm{X_u}) \!\given\! \mathcal{E}^\complement} ,
    \end{align}
    the division in (b) is well-defined because we assumed $\mathbb{P}(\mathcal{E}^\complement) > 0$, (c) holds by definition of average extremal power \cref{eq:average-extremal-power} because $\mathbb{P}$ is the maximal coupling \cref{eq:maximal-coupling}, and (d) holds because $M$ is non-negative. Combining \cref{eq:bounds-step-six,eq:bounds-step-eight}, we thus have
    \begin{equation}
        \rho(WK) \leq \P{\mathcal{E}^\complement} \, G \mprn{\frac{\tnorm{W}_\mathsf{AE}}{\P{\mathcal{E}^\complement}}} . \label{eq:bounds-step-ten}
    \end{equation}
    
    Observe that the \emph{perspective function} $H: (0, 1] \mathrel{\times} \mathbb{R}_+ \rightarrow [0, 1]$ given by
    \begin{equation}
        H(s; a) \triangleq s \, G \mprn{\frac{a}{s}}
    \end{equation}
    is non-decreasing in $s$ for any $a$, because for any $s, t \in (0, 1]$, $s < t$ and $a > 0$,\footnote{The $a = 0$ case is obvious by inspection, since $H(s; 0) = s \, G(0)$ and $G(0)$ is non-negative by the range of $G$.} we have   \begin{equation}
        H(s; a) \stackclap{(a)}{=} a \, \frac{G(a/s) - G(0)}{a/s - 0}\stackclap{(b)}{\leq} a \, \frac{G(a/t) - G(0)}{a/t - 0} \stackclap{(c)}{=} H(t; a) ,
    \end{equation}
    where (a) and (c) hold because $G(0) = \invbreve{\Theta}(0) = \Theta(0) = 0$, where the second equality follows from \cref{lemma:properties-of-upper-concave-envelope}, Part 1; and (b) holds because $G$ is concave and so the \emph{difference quotient}   %
      $(G(r) - G(q))/(r - q)$
    is monotonically non-decreasing in $r > q$ for any fixed $q$ by the first inequality in \cite[Exercise 3.1.b]{BoydVandenberghe2009}. (Alternatively, if $G$ is differentiable, we may use the first derivative test: for any $s \in (0, 1]$ and $a \in \mathbb{R}_+$, we have   \begin{align}
        \frac{\partial H}{\partial s} &= G \mprn{\frac{a}{s}} - \prn{\frac{a}{s}} G' \mprn{\frac{a}{s}} \stackclap{(a)}{\geq} G \mprn{\frac{a}{s}} - \prn{\frac{a}{s}} G'(\xi) \\&\stackclap{(b)}{=} G \mprn{\frac{a}{s}} - \prn{G \mprn{\frac{a}{s}} - G(0)} \stackclap{(c)}{=} 0 ,
    \end{align}
    where (a) holds for \emph{any} $\xi \in [0, a/s]$ by concavity of $G$; (b) holds for \emph{some} $\xi \in [0, a/s]$ by the mean value theorem; and (c) holds because $G(0) = \invbreve{\Theta}(0) = \Theta(0) = 0$, where the second equality follows from \cref{lemma:properties-of-upper-concave-envelope}, Part 1.) Next, observe that
    \begin{equation}
        \P{\mathcal{E}^\complement} \stackclap{(a)}{=} \rho(W) \leq t , \label{eq:bounds-step-eleven}
    \end{equation}
    where (a) holds by \cref{proposition:maximal-coupling-characterization-of-doeblin-coefficients} because $\mathbb{P}$ is the maximal coupling \cref{eq:maximal-coupling}. Proceeding from \cref{eq:bounds-step-ten}, we have   \begin{align}
        \rho(WK) &\stackclap{(a)}{=} H \mprn{\P{\mathcal{E}^\complement}; \tnorm{W}_\mathsf{AE}} \stackclap{(b)}{\leq} H \mprn{t; \tnorm{W}_\mathsf{AE}} \\&\stackclap{(c)}{=} t \, \invbreve{\Theta} \mprn{2 \gamma \, M^{-1} \mprn{\frac{\tnorm{W}_\mathsf{AE}}{t}}\!} \label{eq:upper-bound-shared-argument} \stackclap{(d)}{\leq} t \, \invbreve{\Theta} \mprn{2 \gamma \, M^{-1} \mprn{\frac{p}{t}}\!}
    \end{align}
    as desired, where (a) holds by definition of $H$, (b) holds by \cref{eq:bounds-step-eleven} and because $H$ is non-decreasing in $s$, (c) holds by definition of $H$ and $G$, and (d) holds by the power constraint $\tnorm{W}_\mathsf{AE} \leq p$ and because $\invbreve{\Theta}$ and $M^{-1}$ are non-decreasing.

    \textbf{Lower bound:} Fix $t \in (0, 1]$ and $p \in [0, \infty)$. By definition of Doeblin curve, we want to show that there exists a Markov kernel $W^*\colon \mathcal{U} \times \mathcal{F}_\mathcal{X} \rightarrow [0, 1]$ satisfying $\rho(W^*) \leq t$ and $\itnorm{W^*}_\mathsf{AE} \leq p$ such that   \begin{math}
        \rho(W^* K) \geq t \, \theta \iprn{\mathbf{0}, M^{-1} \iprn{\ifrac{p}{t}}}
    \end{math}.
    Consider $\mathcal{U} \triangleq \mathcal{B}(\mathbf{0}, M^{-1}(p/t)) \subset \mathcal{X}$, equipped with the Borel $\sigma$-algebra induced on $\mathcal{U}$ by $\mathcal{X}$. Let $W\colon \mathcal{U} \times \mathcal{F}_\mathcal{X} \rightarrow [0, 1]$ be the identity kernel   \begin{math}
        {W(\lwildcard \mid u) \triangleq \updelta_u(\wildcard)}
    \end{math} for all $u \in \mathcal{U}$,
    and choose $W^*\colon \mathcal{U} \times \mathcal{F}_\mathcal{X} \rightarrow [0, 1]$ to be the Markov kernel   \begin{math}
        W^*(\lwildcard \mid u) \triangleq t \, W(\lwildcard \mid u) + \iprn{1 - t} \updelta_\mathbf{0}(\wildcard)
    \end{math} for all $u \in \mathcal{U}$.

    By inspection, the greatest common component of $W$ is   \begin{math}
        \ibigwedge_{u \in \mathcal{U}} W(\lwildcard \mid u) = 0
    \end{math}.
    Thus, the complementary Doeblin coefficient of $W$, by \cref{eq:doeblin-coefficient-gcc}, is   \begin{math}
        \rho(W) = 1 - \ibigwedge_{u \in \mathcal{U}} W(\mathcal{X} \mid u) = 1
    \end{math},
    and the maximal coupling of random variables $\ibrc{X_u}_{u \in \mathcal{U}}$ with $X_u \sim W(\lwildcard \mid u)$ for each $u \in \mathcal{U}$, as defined in \cref{eq:maximal-coupling}, reduces to
    \begin{equation}
        \forall u \in \mathcal{U}, \quad X_u \sim \updelta_u(\wildcard) , \label{eq:average-extremal-lower-bound-maximal-coupling}
    \end{equation}
    where $\ibrc{X_u}_{u \in \mathcal{U}}$ are independent. Hence, the average extremal power of $W$ is   \begin{align}
        \tnorm{W}_\mathsf{AE} &\stackclap{(a)}{=} \E{\sup_{u \in \mathcal{U}} M(\norm{X_u})} \stackclap{(b)}{=} \sup_{u \in \mathcal{U}} M(\norm{u})\\&\stackclap{(c)}{=} M \circ M^{-1} \mprn{\frac{p}{t}}= \frac{p}{t} ,
    \end{align}
    where (a) holds by definition of average extremal power \cref{eq:average-extremal-power}, (b) holds by \cref{eq:average-extremal-lower-bound-maximal-coupling}, and (c) holds because $\mathcal{U} = \mathcal{B}(\mathbf{0}, M^{-1}(p/t))$ and $M$ is increasing. By inspection, the composition of $W$ with $K$ is   \begin{math}
        WK(\lwildcard \mid u) = K(\lwildcard \mid u)
    \end{math} for all $u \in \mathcal{U}$,
    and so the complementary Doeblin coefficient of $WK$ is
    \begin{align}
        \rho(WK) = \rho_\mathcal{U}(K) \stackclap{(a)}{=} \theta \mprn{\mathbf{0}, M^{-1} \mprn{\frac{p}{t}}} ,
    \end{align}
    where (a) holds by definition of $\theta$ and because $\mathcal{U} = \mathcal{B}(\mathbf{0}, M^{-1}(p/t))$. Thus, by \cref{lemma:doeblin-coefficients-under-affine-transformation,lemma:power-levels-under-affine-transformation}, we have   $\rho(W^*) = t \, \rho(W) = t$,
        $\rho(W^* K) = t \, \rho(WK) = t \, \theta \iprn{\mathbf{0}, M^{-1} \iprn{\ifrac{p}{t}}}$, and 
        $\tnorm{W^*}_\mathsf{AE} = t \tnorm{W}_\mathsf{AE} = p$
    as desired.
\end{proof}

Next, we prove \cref{corollary:examples-of-average-extremal-doeblin-curves}.

\begin{proof}[Proof of \cref{corollary:examples-of-average-extremal-doeblin-curves}]
    Throughout this proof, let $K$ denote one of the kernels $K_1$, $K_2$, or $K_3$ in general. For each kernel, consider the density function $g_K(z) = g_K(y - x) = \frac{\idiff K}{\idiff y}(y \mid x)$ with respect to the Lebesgue measure on $\mathbb{R}$:   
    \begin{align}
        g_{K_1} \mprn{z; \sigma^2} &= \frac{1}{\sigma \sqrt{2 \pi}} \exp \mprn{-\frac{z^2}{2 \sigma^2}},\\ 
        g_{K_2}(z; b) &= \frac{1}{2b} \exp \mprn{-\frac{\abs{z}}{b}},\\
        g_{K_3}(z; \beta) &= \frac{\sqrt{\beta}}{\pi} \mprn{\frac{1}{1 + \beta z^2}},
    \end{align}
    where $g_K$ depends only on the difference $z = y - x$ because $K$ is a convolution kernel. For any $r > 0$, we have   \begin{align}
        \Theta(r) &\stackclap{(a)}{=} \theta(0, r) \stackclap{(b)}{=} \rho_{[-r, r]}(K) \stackclap{(c)}{=} 1 -\!\!\!\!\bigwedge_{x \in [-r, r]} \!\!\!\!K(\mathbb{R} \mid x) \\&\stackclap{(d)}{=} 1 - \int_{\mathbb{R}} \, \biggprn{\latinf_{x \in [-r, r]}\!\! \frac{\diff K}{\diff y}(y \mid x)} \diff y\\&= 1 - \int_{\mathbb{R}} \prn{\min_{x \in [-r, r]} g_K(y - x)} \diff y , \label{eq:examples-of-average-extremal-doeblin-curves-step-one}
    \end{align}
    where (a) holds because $K$ is a convolution kernel, (b) holds by definition of $\theta$, (c) holds by the greatest common component characterization of Doeblin coefficient \cref{eq:doeblin-coefficient-gcc}, and (d) holds by \cref{lemma:density-of-greatest-common-component}. By inspection, $g_K$ is increasing on $(-\infty, 0]$ and decreasing on $[0, \infty)$. Therefore, the minimum of $g_K$ over any interval $[a, b] \subset \mathbb{R}$ is achieved at one of the endpoints, i.e.,
    \begin{align}
        \min_{z \in [a, b]} g_K(z) = \min \brc{g_K(a), g_K(b)} . \label{eq:examples-of-average-extremal-doeblin-curves-endpoints}
    \end{align}
    Proceeding from \cref{eq:examples-of-average-extremal-doeblin-curves-step-one}, we thus have (cf. \cite[Eq. 40]{PolyanskiyWu2016})   \begin{align}
        \Theta(r) &\stackclap{(a)}{=} 1 - \int_{\mathbb{R}} \min \brc{g_K(y - r), g_K(y + r)} \diff y \\
        &\stackclap{(b)}{=} 1 - \int_{-\infty}^0\!\!\!\!\!\min \brc{g_K(y - r), g_K(-|y + r|)} \diff y \\
        &\phantom{{}={}}- \int_0^\infty\!\!\! \min \brc{g_K(|y - r|), g_K(y + r)} \diff y \\
        &\stackclap{(c)}{=} 1 - \int_{-\infty}^0\!\!\!\!\! g_K(y - r) \diff y - \int_0^\infty\!\!\! g_K(y + r) \diff y \\&\stackclap{(d)}{=} 1 - 2 f_K(-r) , \label{eq:examples-of-average-extremal-doeblin-curves-step-two}
    \end{align}
    where (a) holds by \cref{eq:examples-of-average-extremal-doeblin-curves-endpoints}; (b) holds because $g_K$ is symmetric about $z = 0$ by inspection; (c) holds because $y - r \leq -|y + r| \leq 0$ for $y \leq 0$ and $g_K$ is increasing on $z \leq 0$, and $0 \leq |y - r| \leq y + r$ for $y \geq 0$ and $g_K$ is decreasing on $z \geq 0$; $f_K$ in (d) denotes the CDF   \begin{math}
        f_K(y) = \int_{-\infty}^y g_K(z) \idiff z
    \end{math};
    and (d) holds because $g_K$ is symmetric about $z = 0$. Also, $\Theta$ is concave because
    \begin{align}
        \Theta''(r) = -2 g_K'(-r) \stackclap{(a)}{<} 0 \label{eq:examples-of-average-extremal-doeblin-curves-concave}
    \end{align}
    for all $r > 0$, where (a) holds because $g_K$ is increasing on $z < 0$. Hence, by \cref{corollary:average-extremal-doeblin-curves-of-convolution-kernels,eq:examples-of-average-extremal-doeblin-curves-step-two},
    \begin{align}
        \mathrm{F}_K^\mathsf{AE}(t; p) = t \prn{1 - 2 f_K \mprn{-\sqrt{\frac{p}{t}}}\!} . \label{eq:examples-of-average-extremal-doeblin-curves-step-three}
    \end{align}
    It remains to compute the closed-form CDFs for the Gaussian, Laplace, and $q$-Gaussian kernels. For the Gaussian kernel $K_1$,   \begin{equation}
        f_{K_1}(y) = \int_{-\infty}^y \frac{1}{\sigma \sqrt{2 \pi}} \exp \mprn{-\frac{z^2}{2 \sigma^2}}\diff z = \Phi \mprn{\frac{y}{\sigma}}  . \label{eq:examples-of-average-extremal-doeblin-curves-cdf-one}
    \end{equation}
    For the Laplace kernel $K_2$, for all $y \leq 0$,   \begin{equation}
        f_{K_2}(y) = \int_{-\infty}^y \frac{1}{2b} \exp \mprn{-\frac{\abs{z}}{b}} \diff z = \frac{1}{2} \exp \mprn{\frac{y}{b}}  . \label{eq:examples-of-average-extremal-doeblin-curves-cdf-two}
    \end{equation}
    Lastly, for the $q$-Gaussian kernel $K_3$,   \begin{equation}
        f_{K_3}(y) = \int_{-\infty}^y \!\frac{\sqrt{\beta}}{\pi} \mprn{\frac{1}{1 + \beta z^2}} \diff z = \frac{1}{\pi} \arctan \mprn{\sqrt{\beta} y} + \frac{1}{2}  . \label{eq:examples-of-average-extremal-doeblin-curves-cdf-three}
    \end{equation}
    Combining \cref{eq:examples-of-average-extremal-doeblin-curves-step-three,eq:examples-of-average-extremal-doeblin-curves-cdf-one,eq:examples-of-average-extremal-doeblin-curves-cdf-two,eq:examples-of-average-extremal-doeblin-curves-cdf-three} produces the desired Doeblin curves.
    Note that the supremum of \cref{eq:doeblin-curve} is attained by the kernel $W\colon \{0,1\}\times \mathcal{F}_\mathcal{X}\to[0,1]$ with $W(\lwildcard\mid 0)=t \updelta_{-(p/t)^{1/2}}(\wildcard)+(1-t) \updelta_0(\wildcard)$ and $W(\wildcard\mid 1)= t \updelta_{(p/t)^{1/2}}(\wildcard)+(1-t) \updelta_0(\wildcard)$. To show that the joint range $\mathfrak{F}(K; \mathcal{G}_p)$ encompasses the entire area under the Doeblin curve $\mathrm{F}_K^\mathsf{AE}(t; p)$, due to \cref{lemma:doeblin-coefficients-under-affine-transformation}, it suffices to show that the joint range contains the line segment connecting $(1, 0)$ and $(1, \mathrm{F}_K^\mathsf{AE}(1; p))$. This follows from observing that for all $y \in [0, \mathrm{F}_K^\mathsf{AE}(1; p)]$, there must exist some $q \in [0, p]$ such that $\mathrm{F}_K^\mathsf{AE}(1; q)=y$ since $\mathrm{F}_K^\mathsf{AE}(1; p)$ is continuous with respect to $p$ and $\mathrm{F}_K^\mathsf{AE}(1; 0)=0$. This completes the proof.
\end{proof}

Next, we prove \cref{proposition:upper-bound-on-uniform-average-doeblin-curve-finite}.

\begin{proof}[Proof of \cref{proposition:upper-bound-on-uniform-average-doeblin-curve-finite}]
    Fix an absolutely continuous Markov kernel $K\colon \mathcal{X} \times \mathcal{F}_\mathcal{Y} \rightarrow [0, 1]$. By definition of Doeblin curve, we want to show that for any Markov kernel $W \in \mathcal{G}$ (i.e., satisfying the assumptions of \cref{proposition:upper-bound-on-uniform-average-doeblin-curve-finite}) such that $\rho(W) \leq t$, we have
    \begin{equation}
        \rho(WK) \leq t \, \invbreve{\Theta} \mprn{2 \gamma \, M^{-1} \mprn{\frac{p + \sigma \sqrt{2 \log_e \abs{\mathcal{U}}}}{t}}\!} . \label{eq:upper-bound-uniform-average-finite}
    \end{equation}
    By following the proof of the upper bound in \cref{theorem:bounds-on-average-extremal-doeblin-curve} up to \cref{eq:upper-bound-shared-argument}, step (c)\ignorespaces, we have
    \begin{equation}
        \rho(WK) \leq t \, \invbreve{\Theta} \mprn{2 \gamma \, M^{-1} \mprn{\frac{\tnorm{W}_\mathsf{AE}}{t}}\!}  . \label{eq:upper-bound-uniform-average-finite-step-one}
    \end{equation}
    Next, we upper-bound $\tnorm{W}_\mathsf{AE}$. For notational convenience, let $Z_u = M(\inorm{X_u})$ for each $u \in \mathcal{U}$. We have   \begin{align}
        \tnorm{W}_\mathsf{AE} &\stackclap{(a)}{=} \E{\max_{u \in \mathcal{U}} Z_u} \\&\leq \E{\max_{u \in \mathcal{U}} \E{Z_u} + \max_{u \in \mathcal{U}} \brc{Z_u - \E{Z_u}}} \\&= \max_{u \in \mathcal{U}} \E{Z_u} + \E{\max_{u \in \mathcal{U}} \brc{Z_u - \E{Z_u}}} \\&\stackclap{(b)}{\leq} p + \underbracket{\E{\max_{u \in \mathcal{U}} \brc{Z_u - \E{Z_u}}}}_{\circled{1}}, \label{eq:upper-bound-uniform-average-finite-step-two}
    \end{align}
    where (a) holds by definition of average extremal power \cref{eq:average-extremal-power} and we use $\max$ instead of $\sup$ because $\mathcal{U}$ is finite, and (b) holds because any kernel $W \in \mathcal{G}$ satisfies $\tnorm{W}_\mathsf{UA} \leq p$. Next, we upper-bound $\circled{1}$. For any $\lambda > 0$, we have (cf. \cite[Eq. 2]{ChuRaginsky2025})   %
      \begin{align}
        \circled{1} &= \frac{1}{\lambda} \, \E{\log_e \exp \mprn{\max_{u \in \mathcal{U}} \lambda \prn{Z_u - \E{Z_u}}}} \\&\stackclap{(a)}{=} \frac{1}{\lambda} \, \E{\log_e \max_{u \in \mathcal{U}} e^{\lambda \prn{Z_u - \E{Z_u}}}} \\&\stackclap{(b)}{\leq} \frac{1}{\lambda} \log_e \E{\max_{u \in \mathcal{U}} e^{\lambda \prn{Z_u - \E{Z_u}}}} \\
        &\stackclap{(c)}{\leq} \frac{1}{\lambda} \log_e \E{\sum_{u \in \mathcal{U}} e^{\lambda \prn{Z_u - \E{Z_u}}}} \\&\stackclap{(d)}{=} \frac{1}{\lambda} \log_e \sum_{u \in \mathcal{U}} \E{e^{\lambda \prn{Z_u - \E{Z_u}}}} \\&\stackclap{(e)}{\leq} \frac{1}{\lambda} \log_e \sum_{u \in \mathcal{U}} \exp \mprn{\frac{\sigma^2 \lambda^2}{2}} \\&= \frac{\log_e \abs{\mathcal{U}}}{\lambda} + \frac{\sigma^2 \lambda}{2}  ,
    \end{align}
    where (a) holds because $\exp$ is increasing, (b) holds by Jensen's inequality, (c) holds by upper-bounding the maximum of a collection of positive terms with their sum, (d) holds by linearity of expectation, and (e) holds by the sub-Gaussianity of $W$. Choosing the optimal value $\lambda = \sqrt{2 \log_e |\mathcal{U}|} / \sigma$ to balance the summands, we have
    \begin{equation}
        \circled{1} \leq \sigma \sqrt{2 \log_e \abs{\mathcal{U}}} . \label{eq:upper-bound-uniform-average-finite-step-three}
    \end{equation}
    Since $\invbreve{\Theta}$ and $M^{-1}$ are non-decreasing (as established in the proof of \cref{theorem:bounds-on-average-extremal-doeblin-curve}), combining \cref{eq:upper-bound-uniform-average-finite-step-one,eq:upper-bound-uniform-average-finite-step-two,eq:upper-bound-uniform-average-finite-step-three} proves \cref{eq:upper-bound-uniform-average-finite} as desired.
\end{proof}

Next, we prove \cref{proposition:upper-bound-on-uniform-average-doeblin-curve-totally-bounded}.

\begin{proof}[Proof of \cref{proposition:upper-bound-on-uniform-average-doeblin-curve-totally-bounded}]
    Fix an absolutely continuous Markov kernel $K\colon \mathcal{X} \times \mathcal{F}_\mathcal{Y} \rightarrow [0, 1]$. By definition of Doeblin curve, we want to show that for any Markov kernel $W \in \mathcal{G}$ (i.e., satisfying the assumptions of \cref{proposition:upper-bound-on-uniform-average-doeblin-curve-totally-bounded}) such that $\rho(W) \leq t$, we have   \begin{equation}
        \rho(WK) \leq t \, \invbreve{\Theta} \Biggprn{2 \gamma \, M^{-1}\! \Biggprn{\frac{p}{t} + \frac{32 \sigma}{t} \hspace{-1em} \int\limits_0^{\tnorm{\mathcal{U}}_\infty} \hspace{-1em} \sqrt{\log_e N(\epsilon, \mathcal{U}, d_\mathcal{U})} \diff \epsilon} \!\! } . \label{eq:upper-bound-uniform-average-totally-bounded}
    \end{equation}
    By following the proof of the upper bound in \cref{theorem:bounds-on-average-extremal-doeblin-curve} up to \cref{eq:upper-bound-shared-argument}, step (c), we have
    \begin{align}
        \rho(WK) \leq t \, \invbreve{\Theta} \mprn{2 \gamma \, M^{-1} \mprn{\frac{\tnorm{W}_\mathsf{AE}}{t}}\!}  . \label{eq:upper-bound-uniform-average-totally-bounded-step-one}
    \end{align}
    Next, we upper-bound $\itnorm{W}_\mathsf{AE}$. For notational convenience, let $Z_u = M(\norm{X_u})$ for each $u \in \mathcal{U}$. Observe that   \begin{align}
        \sup_{u \in \mathcal{U}} Z_u &\leq \sup_{u \in \mathcal{U}} \E{Z_u} + \sup_{u \in \mathcal{U}} \brc{Z_u - \E{Z_u}} \\&\stackclap{(a)}{=} \tnorm{W}_\mathsf{UA} + \sup_{u \in \mathcal{U}} \brc{Z_u - \E{Z_u}} \\&\stackclap{(b)}{\leq} p + \sup_{u \in \mathcal{U}} \brc{Z_u - \E{Z_u}} , \label{eq:upper-bound-uniform-average-totally-bounded-step-two}
    \end{align}
    where (a) holds by definition of uniform average power \cref{eq:uniform-average-power} and (b) holds because any kernel $W \in \mathcal{G}$ satisfies $\tnorm{W}_\mathsf{UA} \leq p$. Hence,   \begin{align}
        \tnorm{W}_\mathsf{AE} &\stackclap{(a)}{=} \E{\sup_{u \in \mathcal{U}} Z_u} \\&\stackclap{(b)}{\leq} p + \E{\sup_{u \in \mathcal{U}} \brc{Z_u - \E{Z_u}}} \\&\stackclap{(c)}{=} p + \E{\sup_{u \in \mathcal{U}} \brc{Z_u - \E{Z_u}} - \prn{Z_{v^*} - \E{Z_{v^*}}}} \\
        &\leq p + \E{\sup_{u, v \in \mathcal{U}} \brc{\prn{Z_u - \E{Z_u}} - \prn{Z_v - \E{Z_v}}}} \\&\stackclap{(d)}{\leq} p + 32 \sigma \int_0^{\tnorm{\mathcal{U}}_\infty}\hspace{-1em} \sqrt{\log_e N(\epsilon, \mathcal{U}, d_\mathcal{U})} \diff \epsilon , \label{eq:upper-bound-uniform-average-totally-bounded-step-three}
    \end{align}
    where all the expectations are taken with respect to the maximal coupling $\ibrc{X_u}_{u \in \mathcal{U}} \sim \mathbb{P}$ as defined in \cref{eq:maximal-coupling}, (a) holds by definition of average extremal power \cref{eq:average-extremal-power}, (b) holds by \cref{eq:upper-bound-uniform-average-totally-bounded-step-two}, (c) holds for \emph{any fixed} $v^* \in \mathcal{U}$ (i.e., $v^*$ is chosen independently of $\ibrc{Z_u}_{u \in \mathcal{U}}$), and (d) holds by \emph{Dudley's entropy integral bound} \cite[Theorem 5.22]{Wainwright2019}. Since $\invbreve{\Theta}$ and $M^{-1}$ are non-decreasing (as established in the proof of \cref{theorem:bounds-on-average-extremal-doeblin-curve}), combining \cref{eq:upper-bound-uniform-average-totally-bounded-step-one,eq:upper-bound-uniform-average-totally-bounded-step-three} proves \cref{eq:upper-bound-uniform-average-totally-bounded} as desired.
\end{proof}

Lastly, we prove \cref{proposition:lower-bound-on-uniform-average-doeblin-curve}.

\begin{proof}[Proof of \cref{proposition:lower-bound-on-uniform-average-doeblin-curve}]
    Fix an absolutely continuous Markov kernel $K\colon \mathcal{X} \times \mathcal{F}_\mathcal{Y} \rightarrow [0, 1]$. Fix $t \in (0, 1]$. By following the proof of the lower bound in \cref{theorem:bounds-on-average-extremal-doeblin-curve}, we obtain that the Markov kernel $W^*\colon \mathcal{U} \times \mathcal{F}_\mathcal{X} \rightarrow [0, 1]$ from $\mathcal{U} \triangleq \mathcal{B}(\mathbf{0}, M^{-1}(p/t))$ to $\mathcal{X}$ given by   \begin{math}
        W^*(A \mid u) \triangleq \prn{1 - t} \updelta_\mathbf{0}(A) + t \, \updelta_u(A)
    \end{math} for all $u \in \mathcal{U}$ and $A \in \mathcal{F}_\mathcal{X}$
    satisfies $\rho(W^*) = t$, $\itnorm{W^*}_\mathsf{AE} = p$, and $\rho(W^* K) = t \, \theta(\mathbf{0}, M^{-1}(p/t))$. By the discussion immediately following \cref{definition:power-level}, we have $\itnorm{W^*}_\mathsf{UA} \leq \itnorm{W^*}_\mathsf{AE} = p$. Hence,   \begin{math}
        \mathrm{F}_K^\mathsf{UA}(t; p) = \sup \ibrc{\rho(WK): \rho(W) \leq t, \, \tnorm{W}_\mathsf{UA} \leq p} \geq \rho(W^* K) = t \, \theta \iprn{\mathbf{0}, M^{-1} \iprn{\ifrac{p}{t}}}
    \end{math}
    as desired.
\end{proof}

\section{Proofs of Main Results on Applications} \label{section:proofs-of-main-results-on-applications}

In this section, we prove the main results presented in \cref{section:main-results-on-applications}, pertaining to applications of Doeblin curves.

\subsection{Proofs for Generalization Error} \label{subsection:proofs-for-generalization-error}

First, we prove \cref{lemma:recursive-bound-on-t-information}.

\begin{proof}[Proof of \cref{lemma:recursive-bound-on-t-information}]
    For each $t \in [T]$, define the following random variables representing intermediate computations within the update rule:   \begin{gather}
        U_t \triangleq W_{t-1} - \frac{\eta_t}{\abs{\mathcal{B}_t}} \sum_{j \in \mathcal{B}_t} \nabla g(W_{t-1}, Z_j)  , \\
        V_t \triangleq U_t + m_t N_t  , \quad
        W_t \triangleq \proj_\mathcal{W}(V_t)  .
    \end{gather}
    Conditioned on the event that $i \in \mathcal{B}_t$ (i.e., $Z_i$ is used in the $t$th iteration), the following Markov chain holds, because $Z_i$ will not be used in any future iteration (cf. \cite{CalmonGeneralizationError}):   \begin{multline}
        Z_i \rightarrow \underline{U_t \rightarrow V_t \rightarrow W_t} \rightarrow \underline{U_{t+1} \rightarrow V_{t+1} \rightarrow W_{t+1}} \rightarrow \tcdots\\ \rightarrow \underline{U_{T-1}\rightarrow V_{T-1}\rightarrow W_{T-1}} \rightarrow \underline{U_T \rightarrow V_T \rightarrow W_T}  .
    \end{multline}
    It follows that   \begin{align} &\phantom{{}={}}
        I_{\mathsf{TV}}(W_T; Z_i) \stackclap{(a)}{\leq} I_{\mathsf{TV}}(V_T; Z_i) \\&\stackclap{(b)}{=} \norm{P_{V_T, Z_i} - P_{V_T} \otimes P_{Z_i}}_\mathsf{TV} \\&\stackclap{(c)}{=} \Ewrt{z \sim P_{Z_i}}{\norm{P_{V_T|Z_i = z} - P_{V_T}}_\mathsf{TV}} \\
        &\stackclap{(d)}{=} \mathop{\mathbb{E}}_{z \sim P_{Z_i}}\big[\big\lVert P_{U_T|Z_i = z} * \mathsf{Normal} \mprn{0, m_T^2 I} \\
        &\phantom{=\mathop{\mathbb{E}}_{z \sim P_{Z_i}}\big[\big\lVert} - P_{U_T} * \mathsf{Normal} \mprn{0, m_T^2 I}\big\rVert_\mathsf{TV}\big] , \label{eq:recursive-bound-step-one}
    \end{align}
    where (a) holds by the data processing inequality for $f$-information \cite[Theorem 7.16]{PolyanskiyWu2025}, (b) holds by definition of TV-information \cref{eq:t-information}, (c) holds because $P_{V_T, Z_i} = P_{Z_i} P_{V_T|Z_i}$ by the chain rule of probability, and (d) holds by definition of $V_T$. Next, observe that for any $z \in \mathcal{Z}$, the power of the distribution $P_{U_T|Z_i = z}$ is bounded as   \begin{align}
        &\phantom{{}={}}\E{\norm{U_T}_2^2 \given Z_i = z} \\&\stackclap{(a)}{=} \mathbb{E} \biggbkt{\Bignorm{W_{T-1} - \frac{\eta_T}{\abs{\mathcal{B}_T}} \sum_{j \in \mathcal{B}_T} \nabla g(W_{T-1}, Z_j)}_2^2 \:\bigg|\: Z_i = z} \\
        &\stackclap{(b)}{\leq} 2 \E{\norm{W_{T-1}}_2^2 \given Z_i = z} \\&\phantom{{}={}}+ 2 \mathbb{E} \biggbkt{\Bignorm{\frac{\eta_T}{\abs{\mathcal{B}_T}} \sum_{j \in \mathcal{B}_T} \nabla g(W_{T-1}, Z_j)}_2^2 \:\bigg|\: Z_i = z} \\
        &\stackclap{(c)}{\leq} 2 \E{\norm{W_{T-1}}_2^2 \given Z_i = z}\\&\phantom{{}={}} + 2 \mathbb{E} \biggbkt{\Bigprn{\frac{\eta_T}{\abs{\mathcal{B}_T}} \sum_{j \in \mathcal{B}_T} \norm{\nabla g(W_{T-1}, Z_j)}_2}^2 \:\bigg|\: Z_i = z} \\
        &\stackclap{(d)}{\leq} 2 \E{\norm{W_{T-1}}_2^2 \given Z_i = z} + 2 \eta_T^2 L^2\stackclap{(e)}{\leq} m_T^2 p_T , \label{eq:recursive-bound-step-two}
    \end{align}
    where (a) holds by definition of $U_T$, (b) holds by the identity
    \begin{equation}
        \norm{x + y}_2^2 \leq \norm{x + y}_2^2 + \norm{x - y}_2^2 = 2 \norm{x}_2^2 + 2 \norm{y}_2^2
    \end{equation}
    and linearity of expectation, (c) holds by the triangle inequality, (d) holds by the bound on the loss gradient \cref{eq:generalization-error-gradient-bound}, and (e) holds by definition of $p_T$ \cref{eq:generalization-error-recursive-bound-power}. Similarly, the power of the distribution $P_{U_T}$ is bounded as   \begin{align}
        \E{\norm{U_T}_2^2} &\stackclap{(a)}{=} \Ewrt{z \sim P_{Z_i}}{\E{\norm{U_T}_2^2 \given Z_i = z}} \\&\stackclap{(b)}{\leq} \Ewrt{z \sim P_{Z_i}}{m_T^2 p_T} = m_T^2 p_T  ,
    \end{align}
    where (a) holds by the tower rule of expectation and (b) holds by \cref{eq:recursive-bound-step-two}. Continuing from \cref{eq:recursive-bound-step-one}, we have   \begin{align}
        I_{\mathsf{TV}}(W_T; Z_i)&\stackclap{(a)}{\leq} \Ewrt{z \sim P_{Z_i}}{\mathrm{F}_\Phi^\mathsf{UA} \mprn{\norm{P_{U_T|Z_i = z} - P_{U_T}}_\mathsf{TV}; p_T}} \\&\stackclap{(b)}{\leq} \Ewrt{z \sim P_{Z_i}}{\invbreve{\mathrm{F}}_\Phi^\mathsf{UA} \mprn{\norm{P_{U_T|Z_i = z} - P_{U_T}}_\mathsf{TV}; p_T}} \\
        &\stackclap{(c)}{\leq} \invbreve{\mathrm{F}}_\Phi^\mathsf{UA} \biggprn{\Ewrt{z \sim P_{Z_i}}{\norm{P_{U_T|Z_i = z} - P_{U_T}}_\mathsf{TV}}; p_T} \\&\stackclap{(d)}{=} \invbreve{\mathrm{F}}_\Phi^\mathsf{UA} \Bigprn{\norm{P_{U_T, Z_i} - P_{U_T} \otimes P_{Z_i}}_\mathsf{TV}; p_T} \\
        &\stackclap{(e)}{=} \invbreve{\mathrm{F}}_\Phi^\mathsf{UA} \Bigprn{I_{\mathsf{TV}}(U_T; Z_i); p_T} \\&\stackclap{(f)}{\leq} \invbreve{\mathrm{F}}_\Phi^\mathsf{UA} \Bigprn{I_{\mathsf{TV}}(W_{T-1}; Z_i); p_T} , \label{eq:recursive-bound-step-three}
    \end{align}
    where (a) holds by \cref{proposition:scale-invariant-doeblin-curves-of-additive-noise-channels}; (b) holds because $\invbreve{\mathrm{F}}_\Phi^\mathsf{UA}$ is the upper concave envelope, and hence an upper bound, of $\mathrm{F}_\Phi^\mathsf{UA}$; (c) holds by Jensen's inequality; (d) holds by the chain rule of probability; (e) holds by definition of TV-information \cref{eq:t-information}; and (f) holds by the data processing inequality, and because $\invbreve{\mathrm{F}}_\Phi^\mathsf{UA}$ is non-decreasing by \cref{lemma:properties-of-upper-concave-envelope}, Part 2.
    
    Finally, by recursively applying the arguments in \cref{eq:recursive-bound-step-one,eq:recursive-bound-step-three} above, we obtain   \begin{multline}
        I_{\mathsf{TV}}(W_T; Z_i)\\\leq \invbreve{\mathrm{F}}_\Phi^\mathsf{UA} \Bigprn{\tcdots \invbreve{\mathrm{F}}_\Phi^\mathsf{UA} \Bigprn{\invbreve{\mathrm{F}}_\Phi^\mathsf{UA} \Bigprn{I_{\mathsf{TV}}(W_t; Z_i); p_{t+1}}; p_{t+2}} \tcdots; p_T}
    \end{multline}
    as desired, where the direction of the bound holds at each recursive step because each layer of $\invbreve{\mathrm{F}}_\Phi^\mathsf{UA}$ is non-decreasing by \cref{lemma:properties-of-upper-concave-envelope}, Part 2.
\end{proof}

Now, we are ready to prove \cref{theorem:expected-generalization-error}.

\begin{proof}[Proof of \cref{theorem:expected-generalization-error}]
    Using the high-level argument of \cite{CalmonGeneralizationError}, we first bound the expected generalization error in terms of the TV-information between model parameters and data samples. We have   \begin{align}
        &\phantom{{}={}}\abs{\E{G_\mu(W_T) - G_{\mathcal{S}}(W_T)}} \\&\stackclap{(a)}{\leq} \frac{A}{n} \sum_{i=1}^n I_{\mathsf{TV}}(W_T; Z_i)\\&\stackclap{(b)}{=} \frac{A}{n} \sum_{t=1}^T \sum_{i \in \mathcal{B}_t} I_{\mathsf{TV}}(W_T; Z_i) + \frac{A}{n} \sum_{i \notin \cup_{t=1}^T \mathcal{B}_t} I_{\mathsf{TV}}(W_T; Z_i) \\
        &\stackclap{(c)}{=} \frac{A}{n} \sum_{t=1}^T \sum_{i \in \mathcal{B}_t} I_{\mathsf{TV}}(W_T; Z_i)\\&\stackclap{(d)}{\leq} \frac{A}{n} \sum_{t=1}^T \sum_{i \in \mathcal{B}_t} \invbreve{\mathrm{F}}_\Phi^\mathsf{UA} \Bigprn{\cdots \invbreve{\mathrm{F}}_\Phi^\mathsf{UA} \Bigprn{\underbracket{I_{\mathsf{TV}}(W_t; Z_i)}_{\circled{1}}; p_{t+1}} \cdots; p_T}  , \label{eq:generalization-error-step-one}
    \end{align}
    where (a) holds by \cref{lemma:generalization-error-and-t-information}; (b) holds by grouping the data indices $i \in [n]$ by the iteration $t \in [T]$ in which they are used (if any), since the mini-batches $\mathcal{B}_1, \tdots, \mathcal{B}_T$ are disjoint; (c) holds because data samples which are unused during training are statistically independent of the final parameters and thus have zero TV-information; and (d) holds by \cref{lemma:recursive-bound-on-t-information}.

    Next, we upper-bound $\circled{1}$. Define the following random variables representing intermediate steps in the update rule:   \begin{gather}
        U_t \triangleq W_{t-1} - \frac{\eta_t}{\abs{\mathcal{B}_t}} \sum_{j \in \mathcal{B}_t} \nabla g(W_{t-1}, Z_j)  , \\
        V_t \triangleq U_t + m_t N_t  , \quad
        W_t \triangleq \proj_\mathcal{W}(V_t)  ,
    \end{gather}
    which form the Markov chain   \begin{math}
        {Z_i \rightarrow U_t \rightarrow V_t \rightarrow W_t}
    \end{math}.
    Following the argument in \cite[Lemma 9]{CalmonGeneralizationError}, we have   \begin{align}
        \circled{1} &\stackclap{(a)}{\leq} I_{\mathsf{TV}}(V_t; Z_i)\\&\stackclap{(b)}{=} \norm{P_{V_t, Z_i} - P_{V_t} \otimes P_{Z_i}}_\mathsf{TV} \\&\stackclap{(c)}{=} \!\!\Ewrt{z \sim \mu}{\norm{P_{V_t|Z_i = z} - P_{V_t}}_\mathsf{TV}} \\&\stackclap{(d)}{=}\!\! \mathop{\mathbb{E}}_{z \sim \mu} \bigbkt{\underbracket{\norm{P_{U_t + m_t N_t|Z_i = z} - P_{U_t + m_t N_t}}_\mathsf{TV}}_{\circled{2}}} , \label{eq:generalization-error-step-two}
    \end{align}
    where (a) holds by the data processing inequality for $f$-information \cite[Theorem 7.16]{PolyanskiyWu2025}, (b) holds by definition of TV-information \cref{eq:t-information}, (c) holds because $P_{V_t, Z_i} = P_{Z_i} P_{V_t|Z_i}$ by the chain rule of probability and $P_{Z_i} = \mu$ since the data samples are drawn i.i.d., and (d) holds by definition of $V_t$.

    Next, we upper-bound $\circled{2}$ for any fixed $z \in \mathcal{Z}$. For any two distributions $P$ and $Q$, define the \emph{optimal transport cost}    \begin{equation}
        \mathsf{W}(P, Q; m_t) \triangleq  \frac{1}{2 m_t} \inf_{\substack{P_{X,Y}: \\ P_X = P, \, P_Y = Q}} \Ewrt{(X,Y) \sim P_{X,Y}}{\norm{X - Y}_2} , \label{eq:generalization-error-optimal-transport-cost}
    \end{equation}
    where the infimum is taken over all couplings $P_{X,Y}$ of the random variables $X$ and $Y$ with respective marginals $P_X = P$ and $P_Y = Q$. By \cite[Lemma 6]{CalmonGeneralizationError}, it holds that
    \begin{equation}
        \circled{2} \leq \mathsf{W} \mprn{P_{U_t|Z_i = z}, P_{U_t}; m_t}  . \label{eq:generalization-error-step-three}
    \end{equation}
    Define two random variables   \begin{align}
        U^* &\triangleq W_{t-1} - \frac{\eta_t}{\abs{\mathcal{B}_t}} \Biggprn{\sum_{\substack{j \in \mathcal{B}_t: \\ j \neq i}} \nabla g(W_{t-1}, Z_j) + \nabla g(W_{t-1}, z)}  , \\
        U^\dagger &\triangleq W_{t-1} - \frac{\eta_t}{\abs{\mathcal{B}_t}} \sum_{j \in \mathcal{B}_t} \nabla g(W_{t-1}, Z_j) . \label{eq:generalization-error-u-star-dagger}
    \end{align}
    These random variables have marginals $U^* \sim P_{U_t|Z_i = z}$ and $U^\dagger \sim P_{U_t}$, by definition of $U_t$. Hence, continuing from \cref{eq:generalization-error-step-three},   \begin{align}
        \circled{2} &\stackclap{(a)}{\leq} \frac{1}{2 m_t} \, \E{\norm{U^* - U^\dagger}_2} \\&\stackclap{(b)}{=} \frac{\eta_t}{2 m_t \abs{\mathcal{B}_t}} \, \mathbb{E} \Bigbkt{\norm{\nabla g(W_{t-1}, Z_i) - \nabla g(W_{t-1}, z)}_2} , \label{eq:generalization-error-step-four}
    \end{align}
    where (a) holds by upper-bounding the infimum in \cref{eq:generalization-error-optimal-transport-cost} with the specific instance \cref{eq:generalization-error-u-star-dagger}, and (b) holds by definition of $U^*$ and $U^\dagger$.
    
    Next, observe that $W_{t-1}$ and $Z_i$ are statistically independent, because $i \in \mathcal{B}_t$ by \cref{eq:generalization-error-step-one} and thus $Z_i$ is not used in any iteration except $t$. Combining \cref{eq:generalization-error-step-two,eq:generalization-error-step-four}, we have   \begin{align}
        &\phantom{{}={}}\circled{1} \\&\leq \!\!\mathop{\mathbb{E}}_{z \sim \mu}\! \biggl[ \frac{\eta_t}{2 m_t \abs{\mathcal{B}_t}}\!\! \mathop{\mathbb{E}}_{\substack{(W_{t-1}, Z_i) \\ \sim P_{W_{t-1}} \otimes \mu}} \!\!\Bigbkt{\norm{\nabla g(W_{t-1}, Z_i) -\! \nabla g(W_{t-1}, z)}_2} \!\biggr] \\
        &\stackclap{(a)}{=} \frac{\eta_t}{2 m_t \abs{\mathcal{B}_t}} \mathop{\mathbb{E}}_{\substack{(W_{t-1}, Z_i, z) \\ \sim P_{W_{t-1}} \otimes \mu \otimes \mu}} \!\Bigbkt{\norm{\nabla g(W_{t-1}, Z_i) - \nabla g(W_{t-1}, z)}_2} \\
        &\stackclap{(b)}{\leq} \frac{\eta_t}{2 m_t \abs{\mathcal{B}_t}}\! \biggl(\! \mathop{\mathbb{E}}_{\substack{(W_{t-1}, Z_i) \\ \sim P_{W_{t-1}} \otimes \mu}} \!\!\biggbkt{\!\Bignorm{\nabla g(W_{t-1}, Z_i) \!- \hspace{-1.9em} \mathop{\mathbb{E}}_{\substack{(W_{t-1}, Z) \\ \sim P_{W_{t-1}} \otimes \mu}} \hspace{-1.7em}\mbkt{\nabla g(W_{t-1}, Z)}}_2\!} \\
        &\phantom{\leq \frac{\eta_t}{2 m_t \abs{\mathcal{B}_t}}}+ \kern-1em \mathop{\mathbb{E}}_{\substack{(W_{t-1}, z) \\ \sim P_{W_{t-1}} \otimes \mu}} \!\!\biggbkt{\!\Bignorm{\nabla g(W_{t-1}, z) \!- \hspace{-1.9em} \mathop{\mathbb{E}}_{\substack{(W_{t-1}, Z) \\ \sim P_{W_{t-1}} \otimes \mu}}\hspace{-1.7em} \mbkt{\nabla g(W_{t-1}, Z)}}_2\!} \!\biggr) \\
        &\stackclap{(c)}{\leq} \frac{\eta_t \sigma_{t-1}}{m_t \abs{\mathcal{B}_t}}  , \label{eq:generalization-error-step-five}
    \end{align}
    where (a) holds by Tonelli's theorem and linearity of expectation, (b) holds by the triangle inequality and linearity of expectation, and (c) holds by definition of $\sigma_{t-1}$. Combining \cref{eq:generalization-error-step-one,eq:generalization-error-step-five}, we have   \begin{align}
        &\phantom{{}={}}\abs{\E{G_\mu(W_T) - G_{\mathcal{S}}(W_T)}}
        \\&\leq \frac{A}{n} \sum_{t=1}^T \sum_{i \in \mathcal{B}_t} \invbreve{\mathrm{F}}_\Phi^\mathsf{UA} \Bigprn{\tcdots \invbreve{\mathrm{F}}_\Phi^\mathsf{UA} \Bigprn{\frac{\eta_t \sigma_{t-1}}{m_t \abs{\mathcal{B}_t}}; p_{t+1}} \tcdots; p_T\!} \\&= A \sum_{t=1}^T \frac{\abs{\mathcal{B}_t}}{n} \, \invbreve{\mathrm{F}}_\Phi^\mathsf{UA} \Bigprn{\tcdots \invbreve{\mathrm{F}}_\Phi^\mathsf{UA} \Bigprn{\frac{\eta_t \sigma_{t-1}}{m_t \abs{\mathcal{B}_t}}; p_{t+1}} \tcdots; p_T\!}
    \end{align}
    as desired.
\end{proof}

\subsection{Proofs for Reliable Computation} \label{subsection:proofs-for-reliable-computation}

First, we prove \cref{lemma:stepwise-coupling-construction}.

\begin{proof}[Proof of \cref{lemma:stepwise-coupling-construction}]
    Let $\Psi\colon \mathcal{U}^q \times \mathcal{F}_\mathcal{V}^{\otimes q} \rightarrow [0, 1]$ denote the Markov kernel from $\mathcal{U}^q$ to $\mathcal{V}^q$ such that for all $u^{(1:q)} \in \mathcal{U}^q$, $\Psi(\lwildcard \mid u^{(1:q)})$ is the maximal coupling of $\Phi(\lwildcard \mid u^{(1)}), \tdots, \Phi(\lwildcard \mid u^{(q)})$ defined in \cref{eq:maximal-coupling}. Given $\pi_U$, construct the coupling $\pi_V\colon \mathcal{F}_\mathcal{V}^{\otimes q} \rightarrow [0, 1]$ as
    \begin{align}
        \forall B \in \mathcal{F}_\mathcal{V}^{\otimes q}, \quad \pi_V(B) &\triangleq \!\int_{\mathcal{U}^q} \hspace{-0.7em}\Psi \bigprn{B \:\big\vert\: u^{(1:q)}} \diff \pi_U \bigprn{u^{(1:q)}} \label{eq:stepwise-coupling-construction-pi-v-integral} \\
        &= \Ewrt{u^{(1:q)} \sim \pi_U}{\Psi \bigprn{B \:\big\vert\: u^{(1:q)}}}  . \label{eq:stepwise-coupling-construction-pi-v-expectation}
    \end{align}
    This defines a valid coupling of $P_V^{(1)}, \tdots, P_V^{(q)}$, because for any $\ell \in [q]$ and any $A \in \mathcal{F}_\mathcal{V}$,   \begin{align}
        &\phantom{{}={}}\pi_V \mprn{\mathcal{V}^{\ell-1} \times A \times \mathcal{V}^{q-\ell}}\\&\stackclap{(a)}{=} \int_{\mathcal{U}^q} \Psi \mprn{\mathcal{V}^{\ell-1} \times A \times \mathcal{V}^{q-\ell} \given u^{(1:q)}} \diff \pi_U \mprn{u^{(1:q)}} \\
        &\stackclap{(b)}{=} \int_{\mathcal{U}^q} \Phi \mprn{A \given u^{(\ell)}} \diff \pi_U \mprn{u^{(1:q)}}\\&\stackclap{(c)}{=} \int_\mathcal{U} \Phi \mprn{A \given u^{(\ell)}} \diff P_U^{(\ell)} \mprn{u^{(\ell)}}\\&\stackclap{(d)}{=} P_V^{(\ell)}(A)  ,
    \end{align}
    where (a) holds by the definition of $\pi_V$ \cref{eq:stepwise-coupling-construction-pi-v-integral}, (b) holds because $\Psi(\lwildcard \mid u^{(1:q)})$ is a coupling of $\Phi(\lwildcard \mid u^{(1)}), \tdots, \Phi(\lwildcard \mid u^{(q)})$, (c) holds because $\pi_U$ is a coupling of $P_U^{(1)}, \tdots, P_U^{(q)}$, and (d) holds by \cref{eq:stepwise-coupling-construction-p-v}. Also, $\pi_V$ satisfies   \begin{align}
        &\phantom{{}={}} \Pwrt{V^{(1:q)} \sim \pi_V}{\lnot \mprn{V^{(1)} = \tcdots = V^{(q)}}} \\
        &\stackclap{(a)}{=} \Ewrt{U^{(1:q)} \sim \pi_U}{\Pwrt{V^{(1:q)} \sim \Psi(\cdot \mid U^{(1:q)})}{\lnot \mprn{V^{(1)} = \tcdots = V^{(q)}}}} \\&\stackclap{(b)}{=} \Ewrt{U^{(1:q)} \sim \pi_U}{\rho \Bigprn{\bigbkt{\Phi(\cdot \mid U^{(1)}), \tdots, \Phi(\cdot \mid U^{(q)})}}} \\
        &= \Ewrt{U^{(1:q)} \sim \pi_U}{\rho \Bigprn{\bigbkt{\updelta_{U^{(1)}}, \tdots, \updelta_{U^{(q)}}} \Phi}} \\&\stackclap{(c)}{\leq} \Ewrt{U^{(1:q)} \sim \pi_U}{\mathrm{F}_\Phi^\mathsf{UA} \mprn{\rho \Bigprn{\bigbkt{\updelta_{U^{(1)}}, \tdots, \updelta_{U^{(q)}}}}; p}} \\
        &= \Ewrt{U^{(1:q)} \sim \pi_U}{\mathrm{F}_\Phi^\mathsf{UA} \mprn{\Iv{\lnot \mprn{U^{(1)} = \tcdots = U^{(q)}}}; p}}\\&\stackclap{(d)}{\leq} \Ewrt{U^{(1:q)} \sim \pi_U}{\invbreve{\mathrm{F}}_\Phi^\mathsf{UA} \mprn{\Iv{\lnot \mprn{U^{(1)} = \tcdots = U^{(q)}}}; p}} \\
        &\stackclap{(e)}{\leq} \invbreve{\mathrm{F}}_\Phi^\mathsf{UA} \mprn{\Ewrt{U^{(1:q)} \sim \pi_U}{\Iv{\lnot \mprn{U^{(1)} = \tcdots = U^{(q)}}}}; p}\\&\stackclap{(f)}{=} \invbreve{\mathrm{F}}_\Phi^\mathsf{UA} \mprn{\Pwrt{U^{(1:q)} \sim \pi_U}{\lnot \mprn{U^{(1)} = \tcdots = U^{(q)}}}; p}
    \end{align}
    as desired, where (a) holds by \cref{eq:stepwise-coupling-construction-pi-v-expectation}, (b) holds by the maximal coupling characterization of Doeblin coefficients (\cref{proposition:maximal-coupling-characterization-of-doeblin-coefficients}), (c) holds by definition of Doeblin curve \cref{eq:doeblin-curve} because the kernel of Dirac measures has uniform average power   \begin{multline}
        \tnorm{\bigbkt{\updelta_{U^{(1)}}, \tdots, \updelta_{U^{(q)}}}}_\mathsf{UA} = \max_{\ell \in [q]} \int_\mathcal{U} M(\norm{u}) \diff \updelta_{U^{(\ell)}}(u)
        \\= \max_{\ell \in [q]} M \mprn{\inorm{U^{(\ell)}}}
        \leq \max_{u \in \mathcal{U}} M(\norm{u})
        = p ,
    \end{multline}
    (d) holds because $\invbreve{\mathrm{F}}_\Phi^\mathsf{UA}$ is the upper concave envelope (and hence an upper bound) of $\mathrm{F}_\Phi^\mathsf{UA}$, (e) holds by Jensen's inequality, and (f) holds because the probability of an event is the expectation of its indicator.
\end{proof}

Now, we are ready to prove \cref{theorem:upper-bound-on-circuit-output-divergence}.

\begin{proof}[Proof of \cref{theorem:upper-bound-on-circuit-output-divergence}]
    Define a function $f\colon [0, 1] \rightarrow [0, 1]$ as $f(t) \triangleq \invbreve{\mathrm{F}}_\Phi^\mathsf{UA} \mprn{\min \mbrc{1, bt}; p}$.   %
    As a prelude to our main argument, we establish the following useful result concerning fixed point convergence of $f$. Observe that $[0, 1]$ is a compact and totally ordered set so that any decreasing sequence $\ibrc{t_s}_{s \in \mathbb{N}} \subset [0, 1]$ has $\inf_{s \in \mathbb{N}} t_s \in [0, 1]$.
    Furthermore, $\invbreve{\mathrm{F}}_\Phi^\mathsf{UA}(\lwildcard \, ; p)$ is non-decreasing and continuous on $[0, 1]$ by \cref{lemma:properties-of-upper-concave-envelope}, Parts 2 and 3. Hence, $f$ is non-decreasing and continuous on $[0, 1]$ (being the composition of two non-decreasing and continuous functions), and so $f(\inf_{s \in \mathbb{N}} t_s) = \inf_{s \in \mathbb{N}} f(t_s)$.   %
    Also, $1 \geq f(1)$ by the range of $f$. Hence, applying Kleene's fixed point theorem \cite{Baranga1991} on $([0, 1], \geq)$,   \begin{equation}
        \lim_{s \rightarrow \infty} f^{(s)}(1) \stackclap{(a)}{=} \inf_{s \in \mathbb{N}}{ f^{(s)}(1)} = \max \brc{t \in [0, 1]: f(t) = t} ,
    \end{equation}
    where $f^{(s)}$ denotes the $s$-fold composition of $f$, and (a) holds because $f^{(s)}(1)$ is monotonically non-increasing in $s$.   %

    With this groundwork established, we commence our main argument. Fix $\epsilon > 0$, and fix $n$ sufficiently large such that for all $s \geq \lceil \log_b(n) \rceil$,
    \begin{equation}
        f^{(s)}(1) \leq \max \brc{t \in [0, 1]: f(t) = t} + \epsilon  . \label{eq:upper-bound-on-circuit-output-divergence-n}
    \end{equation}
    Consider any $n$-input circuit of noisy gates where each gate has at most $b$ inputs. Define the \emph{height} of an input vertex $X_i$ as the length of the shortest directed path from $X_i$ to the output vertex $Y_m$, where the existence of such a path is guaranteed for each $X_i$ by the formal model defined in \cref{subsection:reliable-computation-using-noisy-q-ary-gates}. Since the in-degree of each gate vertex is at most $b$, there must exist at least one input vertex whose height is $\lceil \log_b(n) \rceil$ or greater.\footnote{To see this, consider all circuits where the in-degree of each gate is at most $b$ and the height of each input vertex is at most $h_{\max} \triangleq \lceil \log_b(n) - 1 \rceil$, and let $C$ be such a circuit with the greatest number of input vertices. We may assume that no vertex in $C$ has multiple outgoing edges, since we may equivalently consider the tree generated by running breadth-first traversal starting from $Y_m$ and following edges in reverse direction, which preserves the heights of all input vertices. Furthermore, each gate vertex in $C$ must have exactly $b$ incoming edges; otherwise, additional input vertices may be added to the gate and hence to $C$. Lastly, each vertex in $C$ with height less than $h_{\max}$ must be a gate, since if such a vertex was an input, replacing it with a gate taking $b$ new input vertices would add $b - 1$ input vertices to $C$ overall. Hence, $C$ is a perfect $b$-ary anti-arborescence where all input vertices have height $h_{\max}$, and so $C$ has $b^{h_{\max}} < b^{\log_b(n)} = n$ inputs.} Without loss of generality, let $X_1$ be such an input vertex.
    
    Fix any values $x_2, \tdots, x_n \in \mathcal{Q}$ for the remaining inputs. We define some notation used throughout the remainder of our proof. To refer to input and gate vertices in a unified manner, let $W_k \triangleq X_{-k}$ for all $k \in \ibrc{-1, \tdots, -n}$, let $W_k \triangleq Y_k$ for all $k \in [m]$, and let $\mathcal{O}_j \triangleq \ibrc{-i: i \in \mathcal{N}_j} \cup \mathcal{M}_j$ for all $j \in [m]$. Given a subset of indices $\mathcal{S} \subseteq \ibrc{-1, \tdots, -n} \cup [m]$, let $w_\mathcal{S} \triangleq \ibrc{w_k}_{k \in \mathcal{S}}$ refer to the corresponding collection of subscripted variables. For any $W$, let $P_W^{(\ell)}$ be the marginal distribution of $W$ induced by the circuit when setting $X_1 = \ell$ and $X_i = x_i$ for all $i > 1$. Lastly, let $a \land b \triangleq \min \ibrc{a, b}$ for any scalars $a, b \in \mathbb{R}$.

    Recall from \cref{subsection:reliable-computation-using-noisy-q-ary-gates} that we index the gate vertices in topological order, such that $\mathcal{M}_j \subseteq [j - 1]$ for each $j \in [m]$. We follow the strategy in \cite[Section 5.3]{PolyanskiyWu2016}. Consider the following algorithm for constructing a coupling of the marginal distributions at each vertex in the circuit:
    \begin{enumerate}
        \item For each $i \in [n]$, let $\pi_{W_{-i}}$ be the maximal coupling of $P_{X_i}^{(1)}, \tdots, P_{X_i}^{(q)}$.
        \item For each $j \in [m]$:
        \begin{enumerate}
            \item Let $\pi_{Z_j}$ be the pushforward measure of $\bigotimes_{k \in \mathcal{O}_j} \pi_{W_k}$ through the function $\Gamma_j^{\otimes q}\colon (\mathbb{R}^d)^{b_j \times q} \rightarrow \mathcal{Q}^q$ which accepts $q$ copies of input variables and independently applies $\Gamma_j$ on each copy to produce $q$ outputs, i.e.,   \begin{math}
                Z_j^{(1:q)} = \Gamma_j^{\otimes q} \iprn{W_{\mathcal{O}_j}^{(1:q)}}
            \end{math},
            where for each $\ell \in [q]$,
            \begin{equation}
                Z_j^{(\ell)} = \Gamma_j \bigprn{W_{\mathcal{O}_j}^{(\ell)}} . \label{eq:upper-bound-on-circuit-output-divergence-gamma}
            \end{equation}
            Namely, $\pi_{Z_j}$ is the coupling of $P_{Z_j}^{(1)}, \tdots, P_{Z_j}^{(q)}$ induced by independently sampling $W_k^{(1:q)} \sim \pi_{W_k}$ for each $k \in \mathcal{O}_j$ and then computing \cref{eq:upper-bound-on-circuit-output-divergence-gamma} for each $\ell \in [q]$.
            \item Construct a coupling $\pi_{W_j}$ of $P_{Y_j}^{(1)}, \tdots, P_{Y_j}^{(q)}$ by applying \cref{lemma:stepwise-coupling-construction} with input space $\mathcal{U} \triangleq \mathcal{Q}$, output space $\mathcal{V} \triangleq \mathbb{R}^d$, input coupling $\pi_U \triangleq \pi_{Z_j}$, output coupling $\pi_V \triangleq \pi_{W_j}$, and $\Phi$ being the circuit noise mechanism.
        \end{enumerate}
    \end{enumerate}
    
    We analyze the coupling $\pi_{W_m}$ constructed by the algorithm above to upper-bound $\rho \bigprn{[P_{Y_m}^{(1)}, \tdots, P_{Y_m}^{(q)}]}$ using the maximal coupling characterization of Doeblin coefficients (\cref{proposition:maximal-coupling-characterization-of-doeblin-coefficients}). For notational convenience, for each $k \in \ibrc{-1, \tdots, -n} \cup [m]$, let
    \begin{equation}
        t_{W_k} \triangleq \Pwrt{W_k^{(1:q)} \sim \pi_{W_k}}{\lnot \mprn{W_k^{(1)} = \tcdots = W_k^{(q)}}}  . \label{eq:upper-bound-on-circuit-output-divergence-t-wk}
    \end{equation}
    The value of $t_{W_{-i}}$ for each input $i \in [n]$ is
    \begin{align}
        t_{W_{-1}} &= 1  , \label{eq:upper-bound-on-circuit-output-divergence-t-x1} \\
        \forall i > 1, \quad t_{W_{-i}} &= 0  , \label{eq:upper-bound-on-circuit-output-divergence-t-xi}
    \end{align}
    where \cref{eq:upper-bound-on-circuit-output-divergence-t-x1} holds because $\pi_{W_{-1}}$ is a coupling of distinct Dirac measures $P_{X_1}^{(\ell)} = \updelta_{\xi_\ell}$ for $\ell \in [q]$, and \cref{eq:upper-bound-on-circuit-output-divergence-t-xi} holds because $\pi_{W_{-i}}$ for $i > 1$ is a coupling of identical Dirac measures $P_{X_i}^{(1)} = \tcdots = P_{X_i}^{(q)} = \updelta_{x_i}$. For each gate $j \in [m]$, $t_{W_j}$ satisfies the recurrence (cf. \cite[Eq. 149]{PolyanskiyWu2016})   \begin{align}
        &\phantom{{}={}}t_{W_j} \\&\stackclap{(a)}{\leq} \invbreve{\mathrm{F}}_\Phi^\mathsf{UA} \biggprn{\Pwrt{Z_j^{(1:q)} \sim \pi_{Z_j}}{\lnot \mprn{Z_j^{(1)} = \tcdots = Z_j^{(q)}}}; p} \\&\stackclap{(b)}{\leq} \invbreve{\mathrm{F}}_\Phi^\mathsf{UA} \biggprn{\Pwrt{W_{\mathcal{O}_j}^{(1:q)} \sim \bigotimes\limits_{k \in \mathcal{O}_j} \pi_{W_k}}{\lnot \mprn{W_{\mathcal{O}_j}^{(1)} = \tcdots = W_{\mathcal{O}_j}^{(q)}}}; p} \\
        &\stackclap{(c)}{\leq} \invbreve{\mathrm{F}}_\Phi^\mathsf{UA} \biggprn{1 \land \sum_{k \in \mathcal{O}_j} \Pwrt{W_k^{(1:q)} \sim \pi_{W_k}}{\lnot \mprn{W_k^{(1)} = \tcdots = W_k^{(q)}}}; p \!} \\
        &\stackclap{(d)}{=} \invbreve{\mathrm{F}}_\Phi^\mathsf{UA} \biggprn{1 \land \sum_{k \in \mathcal{O}_j} t_{W_k}; p} \\&\stackclap{(e)}{\leq} \invbreve{\mathrm{F}}_\Phi^\mathsf{UA} \mprn{1 \land b \max_{k \in \mathcal{O}_j} t_{W_k}; p} \\&\stackclap{(f)}{=} f \mprn{\max_{k \in \mathcal{O}_j} t_{W_k}}  , \label{eq:upper-bound-on-circuit-output-divergence-recurrence}
    \end{align}
    where (a) holds by \cref{eq:stepwise-coupling-construction}; (b) holds by \cref{eq:upper-bound-on-circuit-output-divergence-gamma} because the gate operation is a deterministic function of its inputs, and because $\invbreve{\mathrm{F}}_\Phi^\mathsf{UA}$ is non-decreasing by \cref{lemma:properties-of-upper-concave-envelope}, Part 2; (c) holds by the union bound; (d) holds by \cref{eq:upper-bound-on-circuit-output-divergence-t-wk}; (e) holds because $|\mathcal{O}_j| \leq b$ since gates have at most $b$ inputs; and (f) holds by definition of $f$.
    
    For convenience, let $\mathcal{D} \subseteq [m]$ denote the set of gates which are descendants of $X_1$, i.e., there exists a directed path from $X_1$ to any gate $j \in \mathcal{D}$. Next, we will construct a directed path from $X_1$ to $Y_m$ by searching backwards starting at $Y_m$, following the algorithm below (cf. \cite[Section 5.3]{PolyanskiyWu2016}):
    \begin{enumerate}
        \item Let $j \in [m]$ denote the current gate vertex. Start at $j \triangleq m$.
        \item If $X_1$ is an incoming neighbor of gate $j$ (i.e., $1 \in \mathcal{N}_j$), move to $X_1$ and terminate.
        \item Otherwise, find and move to any incoming gate neighbor $j' \in \mathcal{M}_j$ such that   \begin{equation}
            j' \in \Bigprn{\arg \max_{j' \in \mathcal{M}_j} t_{W_{j'}}} \cap \mathcal{D}  . \label{eq:upper-bound-on-circuit-output-divergence-next-vertex}
        \end{equation}
        Set $j \triangleq j'$ and return to Step 2.
    \end{enumerate}
    This algorithm is guaranteed to terminate at $X_1$, because the circuit has no directed cycles and the invariant $j \in \mathcal{D}$ holds throughout execution. (We start at the output gate $m \in \mathcal{D}$ in Step 1, and move to some gate $j' \in \mathcal{D}$ on each iteration of Step 3.) To establish well-definedness, it remains to show that the set in \cref{eq:upper-bound-on-circuit-output-divergence-next-vertex} is always non-empty if the algorithm did not terminate in Step 2 (i.e., $1 \notin \mathcal{N}_j$). Since $j \in \mathcal{D}$ and $1 \notin \mathcal{N}_j$, we must have $\mathcal{M}_j \neq \emptyset$. By \cref{eq:upper-bound-on-circuit-output-divergence-recurrence,eq:upper-bound-on-circuit-output-divergence-t-xi} and the fact that $f(0) = \invbreve{\mathrm{F}}_\Phi^\mathsf{UA}(0; p) = \mathrm{F}_\Phi^\mathsf{UA}(0; p) = 0$ (where the second equality holds by \cref{lemma:properties-of-upper-concave-envelope}, Part 1), it follows that any gate $j' \notin \mathcal{D}$ has $t_{W_{j'}} = 0$. Hence, $(\arg \max_{j' \in \mathcal{M}_j} t_{W_{j'}}) \not\subseteq \mathcal{D}^\complement$, and so the set in \cref{eq:upper-bound-on-circuit-output-divergence-next-vertex} is non-empty, as desired.

    Let $j_1 < \tcdots < j_s = m$ be the sequence of gates from $X_1$ to $Y_m$ constructed above, where $s$ is the path length. We have $\max_{k \in \mathcal{O}_{j_1}} t_{W_k} = 1$, because $1 \in \mathcal{N}_{j_1}$ and $t_{W_{-1}} = 1$ by \cref{eq:upper-bound-on-circuit-output-divergence-t-x1}. For any $r \in \ibrc{2, \tdots, s}$, we have
    \begin{equation}
        \max_{k \in \mathcal{O}_{j_r}} t_{W_k} \stackclap{(a)}{=} \max_{j' \in \mathcal{M}_{j_r}} t_{W_{j'}} \stackclap{(b)}{=} t_{W_{j_{r-1}}} ,
    \end{equation}
    where (a) holds because $1 \notin \mathcal{N}_{j_r}$ and so $t_{W_{-i}} = 0$ for all $i \in \mathcal{N}_{j_r}$ by \cref{eq:upper-bound-on-circuit-output-divergence-t-xi}, and (b) holds by Step 3 of the algorithm \cref{eq:upper-bound-on-circuit-output-divergence-next-vertex}. Therefore,   \begin{align}
        \rho \bigprn{\bigbkt{P_{Y_m}^{(1)}, \tdots, P_{Y_m}^{(q)}}} &\stackclap{(a)}{=} 1 - \!\!\!\!\!\sup_{\mathbb{P}: Y_m^{(\ell)} \sim P_{Y_m}^{(\ell)}}\!\!\!\!\! \mathbb{P}\Bigprn{Y_m^{(1)} \!=\! \tcdots \!=\! Y_m^{(q)}} \\&\stackclap{(b)}{\leq} 1 -\!\!\!\!\!\!\! \mathop{\mathbb{P}}_{Y_m^{(1:q)} \sim \pi_{W_m}}\!\!\!\!\!\!\Bigprn{Y_m^{(1)} \!=\! \tcdots \!=\! Y_m^{(q)}} \\&\stackclap{(c)}{=} t_{W_m} \\&\stackclap{(d)}{\leq} f \bigprn{t_{W_{j_{s-1}}}}\\&\leq \tcdots\stackclap{(e)}{\leq} f^{(s)} \mprn{1} ,
    \end{align}
    where (a) holds by \cref{proposition:maximal-coupling-characterization-of-doeblin-coefficients}, where the supremum is taken over all couplings of $P_{Y_m}^{(1)}, \tdots, P_{Y_m}^{(q)}$; (b) holds by lower-bounding the supremum with a specific coupling $\pi_{W_m}$; (c) holds by definition of $t_{W_k}$ \cref{eq:upper-bound-on-circuit-output-divergence-t-wk}; and (d) and (e) hold by repeatedly applying the recurrence \cref{eq:upper-bound-on-circuit-output-divergence-recurrence} along with the above characterizations of $\max_{k \in \mathcal{O}_{j_r}}$. Since the height of $X_1$ is $\lceil \log_b(n) \rceil$ or greater, $s \geq \lceil \log_b(n) \rceil$. By \cref{eq:upper-bound-on-circuit-output-divergence-n}, we have   \begin{math}
        \rho \bigprn{[P_{Y_m}^{(1)}, \tdots, P_{Y_m}^{(q)}]} \leq \max \ibrc{t \in [0, 1]: f(t) = t} + \epsilon
    \end{math}
    as desired.
\end{proof}

We remark on an alternative way to guarantee that starting from $Y_m$ and repeatedly moving to the gate maximizing $t_{W_{j'}}$ leads to a path terminating at $X_1$. For each $j \notin \mathcal{D}$, let $N_j$ represent the randomness in the noise mechanism which produces $Y_j$ from $Z_j$ \cite[Lemma 4.22]{Kallenberg2021}. Then,   \begin{align}
    &\phantom{{}={}}\rho \bigprn{\bigbkt{P_{Y_m}^{(1)}, \dots, P_{Y_m}^{(q)}}} \\&\stackclap{(a)}{=} \rho \Biggprn{\Ewrt{\varphi_{\mathcal{D}^\complement} \stackrel{\mathsf{iid}}{\sim} \Phi}{\bigbkt{P_{Y_m|N_{\mathcal{D}^\complement} = \varphi_{\mathcal{D}^\complement}}^{(1)}, \tdots, P_{Y_m|N_{\mathcal{D}^\complement} = \varphi_{\mathcal{D}^\complement}}^{(q)}}}} \\&\stackclap{(b)}{\leq} \Ewrt{\varphi_{\mathcal{D}^\complement} \stackrel{\mathsf{iid}}{\sim} \Phi}{\rho \Bigprn{\bigbkt{P_{Y_m|N_{\mathcal{D}^\complement} = \varphi_{\mathcal{D}^\complement}}^{(1)}, \tdots, P_{Y_m|N_{\mathcal{D}^\complement} = \varphi_{\mathcal{D}^\complement}}^{(q)}}}}
\end{align}
where (a) holds by the law of total probability, and (b) holds by convexity of complementary Doeblin coefficients \cite[Theorem 1, Part 3]{MakurSingh2023b} and Jensen's inequality. The marginal distributions in (b) correspond to analyzing the circuit where all gate vertices which are not descendants of $X_1$ are deterministic (after conditioning on $N_{\mathcal{D}^\complement}$), and so may be removed from the circuit before proceeding with the coupling constructions. Since all directed paths in the circuit now start with $X_1$, the search algorithm is guaranteed to build a path back to $X_1$.

\subsection{Proofs for Differential Privacy} \label{subsection:proofs-for-differential-privacy}

In this section, we prove \cref{theorem:contraction-of-rho-epsilon,theorem:differential-privacy-for-unconstrained-online-learning}.

\begin{proof}[Proof of \cref{theorem:contraction-of-rho-epsilon}]
    Fix an absolutely continuous Markov kernel $K\colon \mathcal{X} \times \mathcal{F}_\mathcal{Y} \rightarrow [0, 1]$ with common dominating measure $\mu\colon \mathcal{F}_\mathcal{Y} \rightarrow \mathbb{R}_+$. Fix $n \geq 2$ and $\boldsymbol{\epsilon} = (\epsilon_1, \tdots, \epsilon_n) \in \mathbb{R}_+^n$ with $\epsilon_1 = 0$ and $\epsilon_i \leq \epsilon_j$ for all $i < j$.
    
    \textbf{Part 1:} Fix a Markov kernel $W\colon \mathcal{U} \times \mathcal{F}_\mathcal{X} \rightarrow [0, 1]$ from a Polish space $(\mathcal{U}, \mathcal{F}_\mathcal{U})$. We have   \begin{align}
        &\phantom{{}={}}\rho_{\boldsymbol{\epsilon}}(WK, n)\\&\stackclap{(a)}{=} 1 - \inf_{u_1, \dots, u_n \in \mathcal{U}} \inf_{\substack{\text{$n$-partition of $\mathcal{Y}$} \\ A_1, \tdots, A_n}} \sum_{i=1}^n e^{\epsilon_i} WK(A_i \mid u_i) \\&\stackclap{(b)}{=} 1 - \inf_{u_1, \tdots, u_n \in \mathcal{U}} \underbracket{\int_{\mathcal{Y}} \prn{\min_{i \in [n]} e^{\epsilon_i} \frac{\diff WK}{\diff \mu}(\lwildcard \mid u_i)} \diff \mu}_{\circled{1}} , \label{eq:differential-privacy-part-one-step-one}
    \end{align}
    where (a) holds by definition of $\rho_{\boldsymbol{\epsilon}}$ \cref{eq:rho-epsilon} and (b) holds by \cref{proposition:integral-characterization-of-infimum}.
    
    Next, we lower-bound $\circled{1}$. Let $\nu\colon\mathcal{F}_\mathcal{X} \rightarrow \mathbb{R}_+$ given by $\nu(\wildcard) \triangleq \ibigwedge_{i \in [n]} e^{\epsilon_i} W(\lwildcard \mid u_i)$ be the weighted greatest common component of $W$ restricted to $\ibrc{u_1, \tdots, u_n}$. Let $\tau_0 \triangleq \nu(\mathcal{X})$ (for convenience, we elide the dependence of $\nu$ and $\tau_0$ on $u_1, \tdots, u_n$ when denoting them). Notice that
    \begin{equation}
        \tau_0 \stackclap{(a)}{\leq} \min_{i \in [n]} e^{\epsilon_i} W(\mathcal{X} \mid u_i) \stackclap{(b)}{=} \min_{i \in [n]} e^{\epsilon_i} \stackclap{(c)}{=} 1  , \label{eq:differential-privacy-part-one-tau-zero-bound}
    \end{equation}
    where (a) holds by definition of greatest common component \cref{eq:greatest-common-component}, (b) holds because $W$ is a Markov kernel, and (c) holds because $\epsilon_1 = 0$ and $\epsilon_i \geq \epsilon_1$ for all $i > 1$. We consider two cases based on the value of $\tau_0$.
    
    \textit{Case 1:} $\tau_0 = 1$. Proceeding from \cref{eq:differential-privacy-part-one-step-one}, we have   \begin{equation}
        \circled{1} \stackclap{(a)}{=} \!\int_{\mathcal{Y}} \Biggprn{\frac{\diff}{\diff \mu} \!\bigwedge_{i \in [n]}\! e^{\epsilon_i} WK(\lwildcard \mid u_i)\!}\! \diff \mu\stackclap{(b)}{=} \!\bigwedge_{i \in [n]}\! e^{\epsilon_i} WK(\mathcal{Y} \mid u_i)  , \label{eq:differential-privacy-part-one-case-one}
    \end{equation}
    where (a) holds by \cref{lemma:density-of-greatest-common-component} and because the lattice infimum of a finite family is the minimum, and (b) holds by the Radon-Nikodym theorem. For notational convenience, let $w_i\colon \mathcal{F}_\mathcal{X} \rightarrow [0, 1]$ be the probability measure $w_i(\wildcard) \triangleq W(\lwildcard \mid u_i)$. Observe that for any $i \in [n]$,   \begin{equation}
        e^{\epsilon_i} WK(\lwildcard \mid u_i) = e^{\epsilon_i} w_i K(\wildcard) = \prn{e^{\epsilon_i} w_i - \nu} K(\wildcard) + \nu K(\wildcard) ,
    \end{equation}
    and so by \cref{lemma:greatest-common-component-under-affine-transformation},
    \begin{equation}
        \bigwedge_{i \in [n]} e^{\epsilon_i} WK(\lwildcard \mid u_i) = \bigwedge_{i \in [n]} \prn{e^{\epsilon_i} w_i - \nu} K(\wildcard) + \nu K(\wildcard) . \label{eq:differential-privacy-part-one-step-two}
    \end{equation}
    Furthermore, it holds that $\nu = w_1$. To see this, first observe that $\nu \leq w_1$, because
    \begin{align}
        \forall A \in \mathcal{F}_\mathcal{X}, \quad \nu(A) \stackclap{(a)}{\leq} e^{\epsilon_1} w_1(A) \stackclap{(b)}{=} w_1(A) ,
    \end{align}
    where (a) holds by definition of greatest common component and (b) holds because $\epsilon_1 = 0$. Next, suppose for the sake of contradiction that $\nu(A) \neq w_1(A)$ for some $A \in \mathcal{F}_\mathcal{X}$. Then,   \begin{align}
        1 &\stackclap{(a)}{=} \nu(\mathcal{X})\stackclap{(b)}{=} \nu(A) + \nu(A^\complement) \stackclap{(c)}{<} w_1(A) + \nu(A^\complement) \\&\stackclap{(d)}{\leq} w_1(A) + w_1(A^\complement) \stackclap{(e)}{=} w_1(\mathcal{X}) \stackclap{(f)}{=} 1
    \end{align}
    and we obtain the contradiction $1 < 1$, where (a) holds by the assumption that $\tau_0 = 1$, (b) holds by additivity of measures, (c) holds by combining the supposition $\nu(A) \neq w_1(A)$ with the fact that $\nu \leq w_1$, (d) holds because $\nu \leq w_1$, (e) holds by additivity of measures, and (f) holds because $w_1$ is a probability measure. Finally, combining \cref{eq:differential-privacy-part-one-step-two} with the fact that $\nu = w_1$, we have
    \begin{align}
        \bigwedge_{i \in [n]} e^{\epsilon_i} WK(\lwildcard \mid u_i) = \nu K(\wildcard) , \label{eq:differential-privacy-part-one-gcc}
    \end{align}
    and combining \cref{eq:differential-privacy-part-one-case-one,eq:differential-privacy-part-one-gcc}, we have   \begin{equation}
        \circled{1} = \nu K(\mathcal{Y}) \stackclap{(a)}{=} \prn{1 - \tau_0} \tau_{\boldsymbol{\epsilon}}(K, n) + \nu K(\mathcal{Y})  , \label{eq:differential-privacy-part-one-step-three}
    \end{equation}
    where (a) holds by the assumption that $\tau_0 = 1$.
    
    \textit{Case 2:} $\tau_0 < 1$. Define a Markov kernel $V\colon \ibrc{u_1, \tdots, u_n} \times \mathcal{F}_\mathcal{X} \rightarrow [0, 1]$ given by
    \begin{align}
        \forall i \in [n], \quad V(\lwildcard \mid u_i) \triangleq \frac{e^{\epsilon_i} W(\lwildcard \mid u_i) - \nu(\wildcard)}{e^{\epsilon_i} - \tau_0}  . \label{eq:dp-part-one-v}
    \end{align}
    This is a valid Markov kernel because for any ${i \in [n]}$,   \begin{math}
        V(\mathcal{X} \mid u_i) = \ifracab{e^{\epsilon_i} W(\mathcal{X} \mid u_i) - \nu(\mathcal{X})}{e^{\epsilon_i} - \tau_0} = (e^{\epsilon_i} - \tau_0)/\allowbreak(e^{\epsilon_i} - \tau_0) = 1
    \end{math}.
    By rearranging \cref{eq:dp-part-one-v}, $W$ may be written in terms of $V$ and $\nu$ as   \begin{math}
        e^{\epsilon_i} W(\lwildcard \mid u_i) = \iprn{e^{\epsilon_i} - \tau_0} V(\lwildcard \mid u_i) + \nu(\wildcard)
    \end{math},
    and therefore the composition $WK$ may be written as
    \begin{equation}
        e^{\epsilon_i} WK(\lwildcard \mid u_i) = \iprn{e^{\epsilon_i} - \tau_0} VK(\lwildcard \mid u_i) + \nu K(\wildcard)  . \label{eq:differential-privacy-wk}
    \end{equation}
    Proceeding from \cref{eq:differential-privacy-part-one-step-one}, we have   \begin{align}
        \circled{1} &\stackclap{(a)}{=} \int_{\mathcal{Y}} \min_{i \in [n]} \brc{\prn{e^{\epsilon_i} - \tau_0} \frac{\diff VK}{\diff \mu}(\lwildcard \mid u_i) + \frac{\diff \nu K}{\diff \mu}} \diff \mu \\&\stackclap{(b)}{=}\underbracket{\int_{\mathcal{Y}} \prn{\min_{i \in [n]} \prn{e^{\epsilon_i} - \tau_0} \frac{\diff VK}{\diff\mu}(\lwildcard \mid u_i)} \diff \mu}_{\circled{2}} + \int_{\mathcal{Y}} \frac{\diff \nu K}{\diff \mu} \diff \mu , \label{eq:differential-privacy-part-one-step-four}
    \end{align}
    where (a) holds by differentiating both sides of \cref{eq:differential-privacy-wk} and (b) holds by linearity of integration. Next, we lower-bound $\circled{2}$. We have   \begin{align}
        \circled{2} &\stackclap{(a)}{=} \inf_{\substack{\text{$n$-partition of $\mathcal{Y}$} \\ A_1, \dots, A_n}} \sum_{i=1}^n \prn{e^{\epsilon_i} - \tau_0} VK(A_i \mid u_i)\\&\stackclap{(b)}{=} \inf_{\substack{\text{$n$-partition of $\mathcal{Y}$} \\ A_1, \dots, A_n}} \sum_{i=1}^n \prn{e^{\epsilon_i} - \tau_0} \int_{\mathcal{X}} K(A_i \mid x)  V(\diff x \mid u_i) \\
        &\stackclap{(c)}{\geq} \inf_{\substack{\text{$n$-partition of $\mathcal{Y}$} \\ A_1, \dots, A_n}} \sum_{i=1}^n \prn{e^{\epsilon_i} - \tau_0} \inf_{x_i \in \mathcal{X}} K(A_i \mid x_i)\\&\stackclap{(d)}{=} \inf_{x_1, \dots, x_n \in \mathcal{X}} \inf_{\substack{\text{$n$-partition of $\mathcal{Y}$} \\ A_1, \dots, A_n}} \sum_{i=1}^n \prn{e^{\epsilon_i} - \tau_0} K(A_i \mid x_i) \\
        &\stackclap{(e)}{=} \prn{1 - \tau_0} \inf_{x_1, \dots, x_n \in \mathcal{X}} \inf_{\substack{\text{$n$-partition of $\mathcal{Y}$} \\ A_1, \dots, A_n}} \sum_{i=1}^n \frac{e^{\epsilon_i} - \tau_0}{1 - \tau_0} \, K(A_i \mid x_i)\\&\stackclap{(f)}{\geq} \prn{1 - \tau_0} \inf_{x_1, \dots, x_n \in \mathcal{X}} \inf_{\substack{\text{$n$-partition of $\mathcal{Y}$} \\ A_1, \dots, A_n}} \sum_{i=1}^n e^{\epsilon_i} K(A_i \mid x_i) \\
        &\stackclap{(g)}{=} \prn{1 - \tau_0} \tau_{\boldsymbol{\epsilon}}(K, n)  , \label{eq:differential-privacy-part-one-step-five}
    \end{align}
    where (a) holds by \cref{proposition:integral-characterization-of-infimum}, (b) holds by definition of kernel composition \cref{eq:kernel-composition}, (c) holds by replacing a weighted average with an infimum, (d) holds by interchanging the order of infimums, (e) holds because $1 - \tau_0 > 0$, (f) holds because $\epsilon_i \geq 0$ and so   \begin{math}
        \ifracab{e^{\epsilon_i} - \tau_0}{1 - \tau_0} \geq \ifracab{e^{\epsilon_i} - e^{\epsilon_i} \tau_0}{1 - \tau_0} = e^{\epsilon_i}
    \end{math}
    for each $i \in [n]$, and (g) holds by definition of $\tau_{\boldsymbol{\epsilon}}$ \cref{eq:tau-epsilon}. Combining \cref{eq:differential-privacy-part-one-step-four,eq:differential-privacy-part-one-step-five}, we have   \begin{align}
        \circled{1} &\geq \prn{1 - \tau_0} \tau_{\boldsymbol{\epsilon}}(K, n) + \int_{\mathcal{Y}} \frac{\diff \nu K}{\diff \mu} \diff \mu\\&\stackclap{(a)}{=} \prn{1 - \tau_0} \tau_{\boldsymbol{\epsilon}}(K, n) + \nu K(\mathcal{Y})  , \label{eq:differential-privacy-part-one-step-six}
    \end{align}
    where (a) holds by the Radon-Nikodym theorem.

    \textit{Proceeding from both cases:} Now, by \cref{eq:differential-privacy-part-one-step-three,eq:differential-privacy-part-one-step-six}, we have the same lower bound for $\circled{1}$ in both cases. For notational convenience, let $\hat{\nu}\colon \mathcal{F}_\mathcal{X} \rightarrow \mathbb{R}_+$ be the measure given by
    \begin{equation}
        \hat{\nu}(\cdot) \triangleq \begin{cases}
            \frac{\nu(\cdot)}{\tau_0}  , & \text{if $\tau_0 > 0$}  , \\
            0  , & \text{if $\tau_0 = 0$}  .
        \end{cases}
    \end{equation}
    Proceeding onwards,   \begin{equation}
        \circled{1} \geq \prn{1 - \tau_0} \tau_{\boldsymbol{\epsilon}}(K, n) + \tau_0 \, \hat{\nu} K(\mathcal{Y}) \stackclap{(a)}{=} \prn{1 - \tau_0} \tau_{\boldsymbol{\epsilon}}(K, n) + \tau_0  , \label{eq:differential-privacy-part-one-step-seven}
    \end{equation}
    where (a) holds in the case $\tau_0 > 0$ because $\hat{\nu}$ is a probability measure on $\mathcal{F}_\mathcal{X}$ in this case and so $\hat{\nu} K$ is a probability measure on $\mathcal{F}_\mathcal{Y}$. Combining \cref{eq:differential-privacy-part-one-step-one,eq:differential-privacy-part-one-step-seven}, we have   \begin{align}
        \rho_{\boldsymbol{\epsilon}}(WK, n) &\leq 1 - \inf_{u_1, \dots, u_n \in \mathcal{U}} \brc{\prn{1 - \tau_0} \tau_{\boldsymbol{\epsilon}}(K, n) + \tau_0} \\&= \biggprn{1 - \underbracket{\inf_{u_1, \dots, u_n \in \mathcal{U}} \tau_0}_{\circled{3}}} \prn{1 - \tau_{\boldsymbol{\epsilon}}(K, n)} . \label{eq:differential-privacy-part-one-step-eight}
    \end{align}
    Next, we evaluate $\circled{3}$. Given any $u_1, \tdots, u_n \in \mathcal{U}$, let $\xi\colon \mathcal{F}_\mathcal{X} \rightarrow \mathbb{R}_+$ be the measure $\xi(\wildcard) \triangleq \sum_{i=1}^n W(\lwildcard \mid u_i)$, where we elide the dependence on $u_1, \tdots, u_n$ in the notation $\xi$ for convenience. We have   \begin{align}
        \circled{3} &\stackclap{(a)}{=} \inf_{u_1, \dots, u_n \in \mathcal{U}} \bigwedge_{i \in [n]} e^{\epsilon_i} W(\mathcal{X} \mid u_i) \\&\stackclap{(b)}{=} \inf_{u_1, \dots, u_n \in \mathcal{U}} \int_{\mathcal{X}} \Biggprn{\frac{\diff}{\diff \xi} \bigwedge_{i \in [n]} e^{\epsilon_i} W(\lwildcard \mid u_i)} \diff \xi \\
        &\stackclap{(c)}{=} \inf_{u_1, \dots, u_n \in \mathcal{U}} \int_{\mathcal{X}} \prn{\min_{i \in [n]} e^{\epsilon_i} \frac{\diff W}{\diff \xi}(\lwildcard \mid u_i)} \diff \xi \\&\stackclap{(d)}{=} \inf_{u_1, \dots, u_n \in \mathcal{U}} \inf_{\substack{\text{$n$-partition of $\mathcal{X}$} \\ A_1, \dots, A_n}} \sum_{i=1}^n e^{\epsilon_i} W(A_i \mid u_i) \\&\stackclap{(e)}{=} \tau_{\boldsymbol{\epsilon}}(W, n) , \label{eq:differential-privacy-part-one-step-nine}
    \end{align}
    where (a) holds by definition of $\tau_0$, (b) holds by the Radon-Nikodym theorem because $\ibigwedge_{i \in [n]} e^{\epsilon_i} W(\lwildcard \mid u_i) \ll \xi$, (c) holds by \cref{lemma:density-of-greatest-common-component} and because the lattice infimum of a finite family is the minimum, (d) holds by \cref{proposition:integral-characterization-of-infimum}, and (e) holds by definition of $\tau_{\boldsymbol{\epsilon}}$ \cref{eq:tau-epsilon}. Combining \cref{eq:differential-privacy-part-one-step-eight,eq:differential-privacy-part-one-step-nine}, we have   \begin{align}
        \rho_{\boldsymbol{\epsilon}}(WK, n) &\leq \prn{1 - \tau_{\boldsymbol{\epsilon}}(W, n)} \prn{1 - \tau_{\boldsymbol{\epsilon}}(K, n)}\\&\stackclap{(a)}{=} \rho_{\boldsymbol{\epsilon}}(W, n) \, \rho_{\boldsymbol{\epsilon}}(K, n)
    \end{align}
    as desired, where (a) holds by definition of $\rho_{\boldsymbol{\epsilon}}$ \cref{eq:rho-epsilon}.

    \textbf{Part 2:} First, observe that   \begin{align}
        &\phantom{{}\iff{}}\text{$K$ is $(\boldsymbol{\epsilon}, \delta, n)$-LDP}
        \\&\stackclap{(a)}{\iff}\!\!\! \inf_{x_1, \dots, x_n \in \mathcal{X}} \inf_{\substack{\text{$n$-partition of $\mathcal{Y}$} \\ A_1, \dots, A_n}} \sum_{i=1}^n e^{\epsilon_i} K(A_i \mid x_i) \geq 1 - \delta \\&\stackclap{(b)}{\iff}\! \tau_{\boldsymbol{\epsilon}}(K, n) \geq 1 - \delta \!\stackclap{(c)}{\iff}\! \rho_{\boldsymbol{\epsilon}}(K, n) \leq \delta  , \label{eq:differential-privacy-part-two-step-one}
    \end{align}
    where (a) holds by definition of $(\boldsymbol{\epsilon}, \delta, n)$-LDP \cref{eq:group-differential-privacy}, (b) holds by definition of $\tau_{\boldsymbol{\epsilon}}$ \cref{eq:tau-epsilon}, and (c) holds by subtracting both sides of the inequality from $1$. We note that \cref{eq:differential-privacy-part-two-step-one} generalizes \cite[Theorem 1]{asoodeh2021local}.
    
    With \cref{eq:differential-privacy-part-two-step-one} in mind, observe also that
    \begin{align}
        \rho_{\boldsymbol{\epsilon}}(K, n) \leq \delta \implies \forall W, \quad \rho_{\boldsymbol{\epsilon}}(WK, n) \leq \delta \rho_{\boldsymbol{\epsilon}}(W, n)  , \label{eq:differential-privacy-part-two-step-two}
    \end{align}
    because for any Markov kernel $W\colon \mathcal{U} \times \mathcal{F}_\mathcal{X} \rightarrow [0, 1]$ from any Polish space $(\mathcal{U}, \mathcal{F}_\mathcal{U})$, we have   \begin{equation}
        \rho_{\boldsymbol{\epsilon}}(WK, n) \stackclap{(a)}{\leq} \rho_{\boldsymbol{\epsilon}}(W, n) \, \rho_{\boldsymbol{\epsilon}}(K, n) \stackclap{(b)}{\leq} \delta \rho_{\boldsymbol{\epsilon}}(W, n)  ,
    \end{equation}
    where (a) follows by Part 1 and (b) follows by the antecedent in \cref{eq:differential-privacy-part-two-step-two}. Next, we will show that the converse
    \begin{align}
        \bigprn{\forall W, \quad \rho_{\boldsymbol{\epsilon}}(WK, n) \leq \delta \rho_{\boldsymbol{\epsilon}}(W, n)} \implies \rho_{\boldsymbol{\epsilon}}(K, n) \leq \delta \label{eq:differential-privacy-part-two-step-three}
    \end{align}
    also holds. Note that if $\tau_{\boldsymbol{\epsilon}}(K, n) = 1$, the consequent in \cref{eq:differential-privacy-part-two-step-three} is trivially satisfied, because   \begin{math}
        \rho_{\boldsymbol{\epsilon}}(K, n) = 1 - \tau_{\boldsymbol{\epsilon}}(K, n) = 0 \leq \delta
    \end{math}.
    Hence, we assume $\tau_{\boldsymbol{\epsilon}}(K, n) < 1$ for the remainder of this argument. By definition of $\tau_{\boldsymbol{\epsilon}}$ as an infimum in \cref{eq:tau-epsilon}, for any arbitrary $0 < \gamma < 1 - \tau_{\boldsymbol{\epsilon}}(K, n)$, there exists $x^*_1, \tdots, x^*_n \in \mathcal{X}$ (possibly depending on $\gamma$) such that
    \begin{equation}
        \inf_{\substack{\text{$n$-partition of $\mathcal{Y}$} \\ A_1, \dots, A_n}} \sum_{i=1}^n e^{\epsilon_i} K(A_i \mid x^*_i) \leq \tau_{\boldsymbol{\epsilon}}(K, n) + \gamma  . \label{eq:differential-privacy-part-two-step-four}
    \end{equation}
    Furthermore, there must be at least two distinct values within $(x^*_1, \tdots, x^*_n)$, because if we suppose the contrary that $x^*_1 = \tcdots = x^*_n = x^*$ for some $x^* \in \mathcal{X}$, we obtain the contradiction   \begin{multline}
        \hspace{-1em} \inf_{\substack{\text{$n$-partition of $\mathcal{Y}$} \\ A_1, \dots, A_n}} \sum_{i=1}^n e^{\epsilon_i} K(A_i \mid x^*_i) =\hspace{-1em} \inf_{\substack{\text{$n$-partition of $\mathcal{Y}$} \\ A_1, \dots, A_n}} \sum_{i=1}^n e^{\epsilon_i} K(A_i \mid x^*) \\\stackclap{(a)}{=} e^{\epsilon_1} K(\mathcal{Y} \mid x^*) \stackclap{(b)}{=} 1 > \tau_{\boldsymbol{\epsilon}}(K, n) + \gamma  ,
    \end{multline}
    where (a) holds because $\epsilon_i \geq \epsilon_1$ for all $i > 1$, and (b) holds because $\epsilon_1 = 0$ and $K$ is a Markov kernel. Let $W^*\colon [n] \times \mathcal{F}_\mathcal{X} \rightarrow [0, 1]$ be a Markov kernel such that $W^*(\lwildcard \mid i) \triangleq \updelta_{x^*_i}(\wildcard)$ for all $i \in [n]$. Then,   \begin{align}
        &\phantom{{}={}}\rho_{\boldsymbol{\epsilon}}(K, n)
        \\&\stackclap{(a)}{=} 1 - \tau_{\boldsymbol{\epsilon}}(K, n) \\&\stackclap{(b)}{\leq} 1 - \inf_{\substack{\text{$n$-partition of $\mathcal{Y}$} \\ A_1, \dots, A_n}} \sum_{i=1}^n e^{\epsilon_i} K(A_i \mid x^*_i) + \gamma \\&\stackclap{(c)}{=} 1 - \inf_{\substack{\text{$n$-partition of $\mathcal{Y}$} \\ A_1, \dots, A_n}} \sum_{i=1}^n e^{\epsilon_i} W^* K(A_i \mid i) + \gamma \\
        &\stackclap{(d)}{\leq} 1 - \inf_{u_1, \dots, u_n \in [n]} \inf_{\substack{\text{$n$-partition of $\mathcal{Y}$} \\ A_1, \dots, A_n}} \sum_{i=1}^n e^{\epsilon_i} W^* K(A_i \mid u_i) + \gamma \\&\stackclap{(e)}{=} \rho_{\boldsymbol{\epsilon}}(W^* K, n) + \gamma \stackclap{(f)}{\leq} \delta \underbracket{\rho_{\boldsymbol{\epsilon}}(W^*, n)}_{\circled{1}} \!\mathrel{+} \gamma  , \label{eq:differential-privacy-part-two-step-five}
    \end{align}
    where (a) holds by definition of $\rho_{\boldsymbol{\epsilon}}$ \cref{eq:rho-epsilon}, (b) holds by \cref{eq:differential-privacy-part-two-step-four}, (c) holds by definition of $W^*$ and because each $\ibrc{x^*_i}$ is measurable (since $\mathcal{X}$ is Polish), (d) holds by lower-bounding the value of a particular instance with an infimum, (e) holds by definition of $\rho_{\boldsymbol{\epsilon}}$ \cref{eq:rho-epsilon}, and (f) holds by the antecedent in \cref{eq:differential-privacy-part-two-step-three}. Next, we evaluate $\circled{1}$. Let $s, t \in [n]$ be indices such that $x^*_s \neq x^*_t$. We have   \begin{align}
        \circled{1} &\stackclap{(a)}{=} 1 - \inf_{u_1, \dots, u_n \in \mathcal{U}} \inf_{\substack{\text{$n$-partition of $\mathcal{X}$} \\ A_1, \dots, A_n}} \sum_{i=1}^n e^{\epsilon_i} W^*(A_i \mid u_i)\\&\stackclap{(b)}{=} 1 - \Bigprn{e^{\epsilon_1} W^* \mprn{\mathcal{X} - \brc{x^*_s} \mid s} + e^{\epsilon_2} W^* \mprn{\brc{x^*_s} \mid t}} \\
        &\stackclap{(c)}{=} 1 - \Bigprn{e^{\epsilon_1} \updelta_{x^*_s} \mprn{\mathcal{X} - \brc{x^*_s}} + e^{\epsilon_2} \updelta_{x^*_t} \mprn{\brc{x^*_s}}}\stackclap{(d)}{=} 1 , \label{eq:differential-privacy-part-two-step-six}
    \end{align}
    where (a) holds by definition of $\rho_{\boldsymbol{\epsilon}}$ \cref{eq:rho-epsilon}; the events $\mathcal{X} - \ibrc{x^*_s}$ and $\ibrc{x^*_s}$ in (b) are measurable because $\mathcal{X}$ is Polish, and so $\ibrc{x}$ is measurable for any $x \in \mathcal{X}$; (b) follows by replacing the infima with a specific instance, and equality holds because $\sum_{i=1}^n e^{\epsilon_i} W^*(A_i \mid u_i) \geq 0$ for all $u_1, \tdots, u_n$ and $A_1, \tdots, A_n$; (c) holds by definition of $W^*$; and (d) holds by definition of Dirac measure. Since $\gamma > 0$ was arbitrary, combining \cref{eq:differential-privacy-part-two-step-five,eq:differential-privacy-part-two-step-six} shows that $\rho_{\boldsymbol{\epsilon}}(K, n) \leq \delta$, thus proving \cref{eq:differential-privacy-part-two-step-three}.

    Finally, combining \cref{eq:differential-privacy-part-two-step-one,eq:differential-privacy-part-two-step-two,eq:differential-privacy-part-two-step-three}, we have
    \begin{equation}
        \text{$K$ is $(\boldsymbol{\epsilon}, \delta, n)$-LDP} \iff \forall W, \quad \rho_{\boldsymbol{\epsilon}}(WK, n) \leq \delta \rho_{\boldsymbol{\epsilon}}(W, n)
    \end{equation}
    as desired.

    \textbf{Part 3:} Fix a Markov kernel $W\colon \mathcal{U} \times \mathcal{F}_\mathcal{X} \rightarrow [0, 1]$. Following the argument from \cref{eq:differential-privacy-part-one-step-one}, we have   \begin{equation}
        \rho_{\boldsymbol{\epsilon}}(WK, n) \stackclap{(a)}{=} 1 - \inf_{u_1, \dots, u_n \in \mathcal{U}} \underbracket{\int_{\mathcal{Y}} \prn{\min_{i \in [n]} e^{\epsilon_i} \frac{\diff WK}{\diff \mu}(\lwildcard \mid u_i)} \diff \mu}_{\circled{1}}\! . \label{eq:dp-part-three-step-one}
    \end{equation}
    Note that if $\inf_{u_1, \dots, u_n \in \mathcal{U}} \circled{1} = 1$, the result \cref{eq:differential-privacy-doeblin-curves} trivially follows because
    \begin{equation}
        \rho_{\boldsymbol{\epsilon}}(WK, n) = 0 \stackclap{(a)}{\leq} \mathrm{F}_K^\mathsf{UA}(\rho_{\boldsymbol{\epsilon}}(W, n); p)  ,
    \end{equation}
    where (a) holds by the range of $\mathrm{F}_K^\mathsf{UA}$. Hence, assume $\inf_{u_1, \dots, u_n \in \mathcal{U}} \circled{1} < 1$ for the remainder of this argument. Next, we lower-bound $\circled{1}$. Let $\nu\colon \mathcal{F}_\mathcal{X} \rightarrow \mathbb{R}_+$ given by $\nu(\wildcard) \triangleq \ibigwedge_{i \in [n]} e^{\epsilon_i} W(\lwildcard \mid u_i)$ be the weighted greatest common component of $W$ restricted to $\ibrc{u_1, \tdots, u_n}$. Let $\tau_0 \triangleq \nu(\mathcal{X})$ ~(as before, we elide the dependence of $\nu$ and $\tau_0$ on $u_1, \tdots, u_n$ when denoting them).  Following the argument from \cref{eq:differential-privacy-part-one-tau-zero-bound}, we have $\tau_0 \leq 1$. We consider two cases based on the value of $\tau_0$.

    \textit{Case 1:} $\tau_0 = 1$. We have
    \begin{equation}
        \circled{1} \stackclap{(a)}{=} \nu K(\mathcal{Y}) \stackclap{(b)}{=} 1 , \label{eq:dp-part-three-step-two}
    \end{equation}
    where (a) holds by following the arguments from \cref{eq:differential-privacy-part-one-case-one,eq:differential-privacy-part-one-gcc}, and (b) holds because $\nu$ is a probability measure on $\mathcal{F}_\mathcal{X}$ (since $\nu(\mathcal{X}) = \tau_0 = 1$) and so $\nu K$ is a probability measure on $\mathcal{F}_\mathcal{Y}$.

    \textit{Case 2:} $\tau_0 < 1$. Define a Markov kernel $V\colon \ibrc{u_1, \tdots, u_n} \times \mathcal{F}_\mathcal{X} \rightarrow [0, 1]$ given by
    \begin{align}
        \forall i \in [n], \quad V(\lwildcard \mid u_i) \triangleq \frac{e^{\epsilon_i} W(\lwildcard \mid u_i) - \nu(\wildcard)}{e^{\epsilon_i} - \tau_0} . \label{eq:dp-part-three-v}
    \end{align}
    Following the argument from \cref{eq:differential-privacy-part-one-step-four}, we have   \begin{equation}
        \circled{1} = \underbracket{\int_{\mathcal{Y}} \prn{\min_{i \in [n]} \prn{e^{\epsilon_i} - \tau_0} \frac{\diff VK}{\diff \mu}(\lwildcard \mid u_i)} \diff \mu}_{\circled{2}}\! \mathrel{+} \nu K(\mathcal{Y})  . \label{eq:dp-part-three-step-three}
    \end{equation}
    Next, we lower-bound $\circled{2}$. We have   \begin{align}
        \circled{2} &\stackclap{(a)}{\geq} \prn{1 - \tau_0} \int_{\mathcal{Y}} \prn{\min_{i \in [n]} \frac{\diff VK}{\diff \mu}(\lwildcard \mid u_i)} \diff \mu \\&\stackclap{(b)}{=} \prn{1 - \tau_0} \hspace{-0.7em}\inf_{\substack{\text{$n$-partition of $\mathcal{Y}$} \\ A_1, \dots, A_n}} \sum_{i=1}^n VK(A_i \mid u_i) \\&\stackclap{(c)}{\geq} \prn{1 - \tau_0} \tau(VK)  , \label{eq:dp-part-three-step-four}
    \end{align}
    where (a) holds because $\epsilon_i \geq 0$ for all $i \in [n]$, (b) holds by \cref{proposition:integral-characterization-of-infimum}, and (c) holds by \cref{theorem:variational-characterization-of-doeblin-coefficient}. For notational convenience, let $\hat{\nu}\colon \mathcal{F}_\mathcal{X} \rightarrow \mathbb{R}_+$ be the measure given by
    \begin{equation}
        \hat{\nu}(\cdot) \triangleq \begin{cases}
            \frac{\nu(\cdot)}{\tau_0}  , & \text{if $\tau_0 > 0$}  , \\
            0  , & \text{if $\tau_0 = 0$}  .
        \end{cases}
    \end{equation}
    Combining \cref{eq:dp-part-three-step-three,eq:dp-part-three-step-four},   \begin{equation}
        \circled{1} \geq \prn{1 - \tau_0} \tau(VK) + \tau_0 \, \hat{\nu} K(\mathcal{Y}) \stackclap{(a)}{=} \prn{1 - \tau_0} \tau(VK) + \tau_0  , \label{eq:dp-part-three-step-five}
    \end{equation}
    where (a) holds in the case $\tau_0 > 0$ because $\hat{\nu}$ is a probability measure on $\mathcal{F}_\mathcal{X}$ in this case and so $\hat{\nu} K$ is a probability measure on $\mathcal{F}_\mathcal{Y}$.

    \textit{Proceeding from both cases:} Now, by \cref{eq:dp-part-three-step-two,eq:dp-part-three-step-five}, we have lower bounds on $\circled{1}$ for any value of $\tau_0$. Proceeding from \cref{eq:dp-part-three-step-one}, we have   \begin{align}
        \rho_{\boldsymbol{\epsilon}}(WK, n) &\stackclap{(a)}{=} 1 - \!\!\!\inf_{\substack{u_1, \dots, u_n \in \mathcal{U}: \\ \tau_0 < 1}} \circled{1} \\&\stackclap{(b)}{\leq} 1 -\!\!\! \inf_{\substack{u_1, \dots, u_n \in \mathcal{U}: \\ \tau_0 < 1}} \brc{\prn{1 - \tau_0} \tau(VK) + \tau_0} \\&\stackclap{(c)}{\leq} 1 - \inf_{u_1, \dots, u_n \in \mathcal{U}} \brc{\prn{1 - \tau_0} \tau(VK) + \tau_0} \\
        &= \!\!\!\sup_{u_1, \dots, u_n \in \mathcal{U}}\! \brc{\prn{1 - \tau_0} \rho(VK)} \\&\stackclap{(d)}{\leq} \!\!\!\sup_{u_1, \dots, u_n \in \mathcal{U}} \!\brc{\prn{1 - \tau_0} \mathrm{F}_K^\mathsf{UA} \mprn{\rho(V); \tnorm{V}_{\mathsf{UA}}}} \\&\stackclap{(e)}{\leq} \!\!\!\sup_{u_1, \dots, u_n \in \mathcal{U}}\! \brc{\prn{1 - \tau_0} \mathrm{F}_K^\mathsf{UA} \mprn{1; \tnorm{V}_{\mathsf{UA}}}} , \label{eq:dp-part-three-step-six}
    \end{align}
    where (a) holds because the infimum $\inf_{u_1, \tdots, u_n \in \mathcal{U}} \circled{1} < 1$ is not achieved by any $(u_1, \tdots, u_n)$ for which $\tau_0 = 1$, by \cref{eq:dp-part-three-step-two}; (b) holds by \cref{eq:dp-part-three-step-five}; (c) holds because an infimum does not increase when taken over a larger feasible set; (d) holds by definition of Doeblin curve \cref{eq:doeblin-curve}; and (e) holds because Doeblin curves are non-decreasing in their first argument. The uniform average power of $V$ may be bounded as   \begin{align}
        \tnorm{V}_{\mathsf{UA}} &\stackclap{(a)}{=} \max_{i \in [n]} \int_{\mathcal{X}} M(\norm{x}) \, V(\diff x \mid u_i) \\&\stackclap{(b)}{=} \max_{i \in [n]} \int_{\mathcal{X}} M(\norm{x}) \, \frac{e^{\epsilon_i} W(\diff x \mid u_i) - \nu(\diff x)}{e^{\epsilon_i} - \tau_0} \\
        &\leq \max_{i \in [n]} \int_{\mathcal{X}} M(\norm{x}) \, \frac{e^{\epsilon_i}}{e^{\epsilon_i} - \tau_0} \, W(\diff x \mid u_i) \\&\stackclap{(c)}{\leq} p \max_{i \in [n]} \frac{e^{\epsilon_i}}{e^{\epsilon_i} - \tau_0} \\&\stackclap{(d)}{=} \frac{p}{1 - \tau_0}  , \label{eq:dp-part-three-step-seven}
    \end{align}
    where (a) holds by definition of uniform average power \cref{eq:uniform-average-power}; (b) holds by definition of $V$ \cref{eq:dp-part-three-v}; (c) holds because $W$ satisfies $\tnorm{W}_{\mathsf{UA}} \leq p$; and (d) holds because $s \mapsto s / (s - \tau_0)$ is decreasing for $s > \tau_0$, and $e^{\epsilon_i} \geq 1 > \tau_0$ for each $i \in [n]$. Combining \cref{eq:dp-part-three-step-six,eq:dp-part-three-step-seven},   \begin{align}
        \rho_{\boldsymbol{\epsilon}}(WK, n) &\leq \sup_{u_1, \dots, u_n \in \mathcal{U}} \brc{\prn{1 - \tau_0} \mathrm{F}_K^\mathsf{UA} \mprn{1; \frac{p}{1 - \tau_0}}} \\&\stackclap{(a)}{=} \sup_{u_1, \dots, u_n \in \mathcal{U}} \mathrm{F}_K^\mathsf{UA}(1 - \tau_0; p) \\
        &\stackclap{(b)}{=} \mathrm{F}_K^\mathsf{UA} \mprn{\sup_{u_1, \dots, u_n \in \mathcal{U}} \brc{1 - \tau_0}; p} \\&\stackclap{(c)}{=} \mathrm{F}_K^\mathsf{UA}(1 - \tau_{\boldsymbol{\epsilon}}(W, n); p) \\&\stackclap{(d)}{=} \mathrm{F}_K^\mathsf{UA}(\rho_{\boldsymbol{\epsilon}}(W, n); p)
    \end{align}
    as desired, where (a) holds by \cref{proposition:homogeneity-of-power-constrained-doeblin-curves}, (b) holds because Doeblin curves are non-decreasing in their first argument, (c) holds by following the argument in \cref{eq:differential-privacy-part-one-step-nine}, and (d) holds by definition of $\rho_{\boldsymbol{\epsilon}}$ \cref{eq:rho-epsilon}.
\end{proof}

Next, we prove \cref{theorem:differential-privacy-for-unconstrained-online-learning}.

\begin{proof}[Proof of \cref{theorem:differential-privacy-for-unconstrained-online-learning}]
    For each $t \in [T]$ and each initialization $i \in [n]$, define the following random variables representing intermediate computations within the update rule:   \begin{gather}
        U_t^{(i)} \triangleq W_{t-1}^{(i)} - \eta_t \nabla g_t \bigprn{W_{t-1}^{(i)}} , \\
        V_t^{(i)} \triangleq U_t^{(i)} + m_t N_t , \quad
        W_t^{(i)} \triangleq \proj_\mathcal{W} \bigprn{V_t^{(i)}}.\end{gather}
    Clearly, the random variables form the Markov chain (similar to the proof of \cref{lemma:recursive-bound-on-t-information} and \cite{CalmonGeneralizationError,PrivacyAmplification})   \begin{equation}
        W_0^{(i)} \rightarrow \underline{U_1^{(i)} \rightarrow V_1^{(i)} \rightarrow W_1^{(i)}}
        \rightarrow \tcdots
        \rightarrow \underline{U_T^{(i)} \rightarrow V_T^{(i)} \rightarrow W_T^{(i)}}
    \end{equation}
    for each $i \in [n]$. We have   \begin{align}
        &\phantom{{}={}}\rho_{\boldsymbol{\epsilon}} \bigprn{\bigbkt{P_{W_T}^{(1)}, \tdots, P_{W_T}^{(n)}}\!, n}\stackclap{(a)}{\leq} \rho_{\boldsymbol{\epsilon}} \bigprn{\bigbkt{P_{V_T}^{(1)}, \tdots, P_{V_T}^{(n)}}\!, n} \\&\stackclap{(b)}{=} \rho_{\boldsymbol{\epsilon}} \bigl( \bigl[ P_{U_T}^{(1)} * \mathsf{Normal} (0, m_T^2 I), \tdots, \\&\hspace{7em}P_{U_T}^{(n)} * \mathsf{Normal} (0, m_T^2 I) \bigr], n \bigr)  , \label{eq:dp-amplification-step-one}
    \end{align}
    where (a) holds by the data processing inequality for $\rho_{\boldsymbol{\epsilon}}$ \cref{eq:data-processing-inequality-rho-epsilon}, and (b) holds by definition of $V_T$. The power of the input distributions is bounded as   \begin{align}
        &\phantom{{}={}}\tnorm{\bigbkt{P_{U_T}^{(1)}, \tdots, P_{U_T}^{(n)}}}_{\mathsf{UA}}
        \\&\stackclap{(a)}{=} \max_{i \in [n]} \mathbb{E} \bkt{\inorm{U_T^{(i)}}_2^2}\\&\stackclap{(b)}{=} \max_{i \in [n]} \mathbb{E} \bkt{\norm{W_{T-1}^{(i)} - \eta_T \nabla g_T \bigprn{W_{T-1}^{(i)}}}_2^2} \\
        &\stackclap{(c)}{\leq} \max_{i \in [n]} \brc{2 \E{\inorm{W_{T-1}^{(i)}}_2^2} + 2 \E{\norm{\eta_T \nabla g_T \bigprn{W_{T-1}^{(i)}}}_2^2}}\\&\stackclap{(d)}{\leq} 2 \max_{i \in [n]} \E{\inorm{W_{T-1}^{(i)}}_2^2} + 2 \eta_T^2 L^2\\&\stackclap{(e)}{=} m_T^2 p_T  ,
    \end{align}
    where (a) holds by definition of uniform average power \cref{eq:uniform-average-power}, (b) holds by definition of $U_T$, (c) holds by the identity   \begin{math}
        \inorm{x + y}_2^2 \leq \inorm{x + y}_2^2 + \inorm{x - y}_2^2 = 2 \inorm{x}_2^2 + 2 \inorm{y}_2^2
    \end{math}
    and linearity of expectation, (d) holds by the bound on the objective gradients \cref{eq:dp-amplification-gradient-bound}, and (e) holds by definition of $p_T$ \cref{eq:dp-amplification-power}. Continuing from \cref{eq:dp-amplification-step-one}, we have   \begin{align}
        &\rho_{\boldsymbol{\epsilon}} \bigprn{\bigbkt{P_{W_T}^{(1)}, \tdots, P_{W_T}^{(n)}}, n}\stackclap{(a)}{\leq} \mathrm{F}_\Phi^\mathsf{UA} \bigprn{\rho_{\boldsymbol{\epsilon}} \bigprn{\bigbkt{P_{U_T}^{(1)}, \tdots, P_{U_T}^{(n)}}, n}; p_T}\\&\quad\stackclap{(b)}{\leq} \mathrm{F}_\Phi^\mathsf{UA} \bigprn{\rho_{\boldsymbol{\epsilon}} \bigprn{\bigbkt{P_{W_{T-1}}^{(1)}, \tdots, P_{W_{T-1}}^{(n)}}, n}; p_T} , \label{eq:dp-amplification-step-two}
    \end{align}
    where (a) holds by \cref{theorem:contraction-of-rho-epsilon,proposition:scale-invariant-doeblin-curves-of-additive-noise-channels}, and (b) holds by the data processing inequality for $\rho_{\boldsymbol{\epsilon}}$ \cref{eq:data-processing-inequality-rho-epsilon}. Finally, by recursively applying the arguments in \cref{eq:dp-amplification-step-one,eq:dp-amplification-step-two} above, we obtain   \begin{multline}
        \rho_{\boldsymbol{\epsilon}} \bigprn{\bigbkt{P_{W_T}^{(1)}, \tdots, P_{W_T}^{(n)}}, n}\\\leq \mathrm{F}_\Phi^\mathsf{UA} \bigprn{\tcdots \mathrm{F}_\Phi^\mathsf{UA} \bigprn{\rho_{\boldsymbol{\epsilon}} \bigprn{\bigbkt{P_{W_0}^{(1)}, \tdots, P_{W_0}^{(n)}}, n}; p_1} \tcdots; p_T}
    \end{multline}
    as desired.
\end{proof}

\section{Conclusion}

In closing, we review our main contributions and propose some directions for follow-up work. In this paper, we formulated the notion of a Doeblin curve to quantify information contraction of Markov kernels on collections of arbitrarily many input distributions with specific divergence and power levels, building upon existing literature on Doeblin coefficients and nonlinear information contraction. After introducing a new variational characterization of Doeblin coefficients, we established several properties of Doeblin curves and derived bounds on Doeblin curves under canonical power constraints and regularity conditions. With this more nuanced measure of information contraction in place, we presented three theoretical applications in noisy iterative optimization, reliable computation, and differential privacy, leveraging Doeblin curves to generalize results in these areas to multi-way or unbounded settings where Doeblin coefficients fail to capture information contraction.

We suggest two potential extensions of our present work. Firstly, our current understanding of Doeblin curves may be further enriched by establishing precise conditions for concavity and by studying examples of closed-form uniform average Doeblin curves. Secondly, future research may explore and broaden our proposed definition of group local differential privacy \cref{eq:group-differential-privacy}, since the standard definition of differential privacy is sometimes considered very stringent for many applications \cite{feldman2021,levy2021learning}.

\appendices
\crefalias{section}{appendix}   %

\section{Technical Lemmata} \label{appendix:technical-lemmata}

In this appendix, we state and prove two miscellaneous results used throughout our paper.

\begin{lemma}[Greatest Common Component Under Affine Transformation] \label{lemma:greatest-common-component-under-affine-transformation}
    For any kernel $W\colon \mathcal{U} \times \mathcal{F}_\mathcal{X} \rightarrow \mathbb{R}_+$ and signed measure $\pi\colon \mathcal{F}_\mathcal{X} \rightarrow \mathbb{R}$, the greatest common component of the kernel $\hat{W}\colon \mathcal{U} \times \mathcal{F}_\mathcal{X} \rightarrow \mathbb{R}_+$ given by
    \begin{equation}
        \forall u \in \mathcal{U}, \quad \hat{W}(\lwildcard \mid u) \triangleq \alpha W(\lwildcard \mid u) + \pi(\wildcard) , \label{eq:greatest-common-component-under-affine-transformation-w-hat}
    \end{equation}
    where $\alpha \geq 0$ and $\pi$ are such that $\hat{W} \geq 0$, is
    \begin{equation}
        \bigwedge_{u \in \mathcal{U}} \hat{W}(\lwildcard \mid u) = \alpha \bigwedge_{u \in \mathcal{U}} W(\lwildcard \mid u) + \pi(\wildcard) . \label{eq:greatest-common-component-under-affine-transformation-gcc}
    \end{equation}
\end{lemma}

\begin{proof}
    The lemma trivially holds for $\alpha = 0$, so assume $\alpha > 0$ henceforth. For notational convenience, let $\mu^*$ denote the measure given by the right-hand side of \cref{eq:greatest-common-component-under-affine-transformation-gcc}. By definition of greatest common component as a supremum in \cref{eq:greatest-common-component}, we seek to verify two properties of $\mu^*$.
    
    Firstly, $\mu^*$ satisfies the condition of the supremum $\hat{W} \geq \mu^*$, because for any $u \in \mathcal{U}$,   \begin{equation}
        \hat{W}(\lwildcard \mid u) \!\stackclap{(a)}{=}\! \alpha W(\lwildcard \mid u) + \pi(\wildcard) \!\stackclap{(b)}{\geq}\! \alpha \!\bigwedge_{u \in \mathcal{U}}\! W(\lwildcard \mid u) + \pi(\wildcard) \!\stackclap{(c)}{=}\! \mu^*(\wildcard) ,
    \end{equation}
    where (a) holds by definition of $\hat{W}$ \cref{eq:greatest-common-component-under-affine-transformation-w-hat}, (b) holds because the greatest common component of a kernel is a lower bound on the kernel, and (c) holds by definition of $\mu^*$ as the right-hand side of \cref{eq:greatest-common-component-under-affine-transformation-gcc}.
    
    Secondly, we will prove that for any measure $\mu$ satisfying $\hat{W} \geq \mu$, we have $\mu^* \geq \mu$. Fix a measure $\mu$ such that $\hat{W} \geq \mu$. Observe that
    \begin{equation}
        \forall u \in \mathcal{U}, \quad W(\lwildcard \mid u) \stackrel{(a)}{=} \frac{\hat{W}(\lwildcard \mid u) - \pi(\wildcard)}{\alpha} \geq \frac{\mu(\wildcard) - \pi(\wildcard)}{\alpha} , \label{eq:greatest-common-component-under-affine-transformation-step-one}
    \end{equation}
    where (a) holds by rearranging \cref{eq:greatest-common-component-under-affine-transformation-w-hat}. Consider a Hahn-Jordan decomposition of $\mu - \pi$, i.e.,   \begin{equation}
        \forall A \in \mathcal{F}_\mathcal{X}, \quad \begin{aligned}
            \prn{\mu - \pi}^+(A) &\triangleq \prn{\mu - \pi}(A \cap P) , \\
        \prn{\mu - \pi}^-(A) &\triangleq -\prn{\mu - \pi}(A \cap N) ,
        \end{aligned} \label{eq:greatest-common-component-under-affine-transformation-positive-part} 
    \end{equation}
    where $P \subseteq \mathcal{X}$ and $N = P^\complement$ are the supports (modulo null sets) of the positive and negative parts of $\mu - \pi$, respectively. Then,   \begin{align}
        W(A \mid u) &\stackclap{(a)}{=} W(A \cap P \mid u) + W(A \cap N \mid u) \\&\stackclap{(b)}{\geq} \frac{\prn{\mu - \pi}(A \cap P)}{\alpha} + 0 \\&\stackclap{(c)}{=} \frac{\prn{\mu - \pi}^+(A)}{\alpha} \label{eq:greatest-common-component-under-affine-transformation-step-two}
    \end{align} for all $u\in\mathcal{U}$ and $A \in \mathcal{F}_\mathcal{X}$,
    where (a) holds by additivity of measures because $N = P^\complement$, (b) holds because $W \geq (\mu - \pi) / \alpha$ (by \cref{eq:greatest-common-component-under-affine-transformation-step-one}) and $W \geq 0$, and (c) holds by \cref{eq:greatest-common-component-under-affine-transformation-positive-part}. Now, suppose for the sake of contradiction that there exists $A \in \mathcal{F}_\mathcal{X}$ such that $\mu^*(A) < \mu(A)$. Then,   \begin{multline}
        \bigwedge_{u \in \mathcal{U}} W(A \mid u) \stackclap{(a)}{=} \sup \brc{\nu(A): W \geq \nu} \stackclap{(b)}{\geq} \frac{\prn{\mu - \pi}^+(A)}{\alpha} \\\geq \frac{\mu(A) - \pi(A)}{\alpha} > \frac{\mu^*(A) - \pi(A)}{\alpha} \stackclap{(c)}{=} \bigwedge_{u \in \mathcal{U}} W(A \mid u)
    \end{multline}
    and we obtain the contradiction ${\ibigwedge_{u \in \mathcal{U}} W(A \mid u)} > \ibigwedge_{u \in \mathcal{U}} W(A \mid u)$, where (a) holds by definition of greatest common component \cref{eq:greatest-common-component}, (b) holds because $W \geq (\mu - \pi)^+ / \alpha$ by \cref{eq:greatest-common-component-under-affine-transformation-step-two} and so we may lower-bound the supremum by this specific instance, and (c) holds because $\mu^*$ is the right-hand side of \cref{eq:greatest-common-component-under-affine-transformation-gcc}.
\end{proof}

\begin{lemma}[Properties of Upper Concave Envelope] \label{lemma:properties-of-upper-concave-envelope}
    Let $f\colon \mathcal{I} \rightarrow [0, 1]$ be a function defined on a (possibly infinite) interval $\mathcal{I} \subseteq \mathbb{R}_+$. The upper concave envelope $\invbreve{f}\colon \mathcal{I} \rightarrow [0, 1]$ of $f$ satisfies the following properties:   %
    \begin{enumerate}
        \item For all boundary points $t$ of $\mathcal{I}$, we have $\invbreve{f}(t) = f(t)$.
        \item If $f$ is non-decreasing, then $\invbreve{f}$ is non-decreasing.
        \item If $f$ is non-decreasing, $\mathcal{I} = [0, 1]$, and $f(t) \leq t$ for all $t \in \mathcal{I}$, then $\invbreve{f}$ is continuous on $\mathcal{I}$.
    \end{enumerate}
\end{lemma}

\begin{proof} ~\newline
    \indent \textbf{Part 1:} Consider a boundary point $a \in \mathcal{I}$. Suppose for the sake of contradiction that $\invbreve{f}(a) \neq f(a)$. Since $\invbreve{f} \geq f$ everywhere on $\mathcal{I}$ by definition of upper concave envelope, we must have $\invbreve{f}(a) > f(a)$. Consider the function $g\colon \mathcal{I} \rightarrow [0, 1]$ given by
    \begin{equation}
        \forall t \in \mathcal{I}, \quad g(t) \triangleq \begin{cases}
            f(a)  , & \text{if $t = a$}  , \\
            \invbreve{f}(t)  , & \text{if $t \neq a$}  .
        \end{cases} \label{eq:properties-of-upper-concave-envelope-g}
    \end{equation}
    Clearly, $f \leq g \leq \invbreve{f}$ everywhere on $\mathcal{I}$. Moreover, $g$ is concave, because for all distinct $s, t \in \mathcal{I}$ and all $\theta \in (0, 1)$,   \begin{align}
        g(\theta s + \prn{1 - \theta} t) &\stackclap{(a)}{=} \invbreve{f}(\theta s + \prn{1 - \theta} t) \stackclap{(b)}{\geq} \theta \, \invbreve{f}(s) + \prn{1 - \theta} \invbreve{f}(t) \\&\stackclap{(c)}{\geq} \theta \, g(s) + \prn{1 - \theta} g(t)  ,
    \end{align}
    where (a) holds by the second case in \cref{eq:properties-of-upper-concave-envelope-g} because $\theta s + (1 - \theta) t$ is an interior point of $\mathcal{I}$, (b) holds because $\invbreve{f}$ is concave, and (c) holds because $g \leq \invbreve{f}$ everywhere, as mentioned above. This contradicts the fact that $\invbreve{f}$ is the pointwise infimum of all concave upper bounds of $f$.

    \textbf{Part 2:} Let $\bar{f}: \mathcal{I} \rightarrow \mathbb{R}$ be the closed upper concave envelope (or ``closed convex hull'') of $f$ \cite[Chapter B, Proposition 2.5.2, Definition 2.5.3]{HiriartUrrutyLemarechal2001}, i.e.,
    \begin{multline}
        \forall t \in \mathcal{I}, \enspace \bar{f}(t) \triangleq \inf \{\alpha t + \beta: \alpha \geq 0, \, \beta \in \mathbb{R} \text{ such that } \\ \forall s \in \mathcal{I}, \, \alpha s + \beta \geq f(s)\} , \label{eq:properties-of-upper-concave-envelope-bar-f}
    \end{multline}
    where it suffices to consider only $\alpha \geq 0$ because $f$ is non-decreasing. Fix $s, t \in \mathcal{I}$ such that $s < t$. By definition of $\bar{f}$ as an infimum \cref{eq:properties-of-upper-concave-envelope-bar-f}, there exist sequences of coefficients $\ibrc{\alpha_n}_{n \in \mathbb{N}} \subset \mathbb{R}_+$ and $\ibrc{\beta_n}_{n \in \mathbb{N}} \subset \mathbb{R}$, satisfying $\alpha_n s + \beta_n \geq f(s)$ for all $s \in \mathcal{I}$ and $n \in \mathbb{N}$, such that $\bar{f}(t) = \lim_{n \rightarrow \infty} \ibrc{\alpha_n t + \beta_n}$. It follows that   \begin{align}
        \bar{f}(s) &\stackclap{(a)}{\leq} \inf_{n \in \mathbb{N}} \brc{\alpha_n s + \beta_n} \leq \lim_{n \rightarrow \infty} \brc{\alpha_n s + \beta_n} \\&\stackclap{(b)}{\leq} \lim_{n \rightarrow \infty} \brc{\alpha_n t + \beta_n} = \bar{f}(t)  ,
    \end{align}
    where (a) holds by taking the infimum in \cref{eq:properties-of-upper-concave-envelope-bar-f} over a smaller set, and (b) holds because $s < t$ and $\alpha_n \geq 0$ for all $n \in \mathbb{N}$. This establishes that $\bar{f}$ is non-decreasing on $\mathcal{I}$. Since $\bar{f}$ is the closure of $\invbreve{f}$ \cite[Chapter B, Proposition 2.5.2]{HiriartUrrutyLemarechal2001}, $\bar{f}$ and $\invbreve{f}$ agree on the interior of $\mathcal{I}$ \cite[Chapter B, Proposition 1.2.6]{HiriartUrrutyLemarechal2001} and $\invbreve{f} \leq \bar{f}$ on any boundary points of $\mathcal{I}$. Hence, $\invbreve{f}$ is non-decreasing on $\mathcal{I} - \ibrc{\sup \mathcal{I}}$. If $\mathcal{I}$ is finite and includes its right boundary point $b \in \mathcal{I}$, then $\invbreve{f}(b) \leq \bar{f}(b) \leq f(b)$, where the second inequality holds by upper-bounding the infimum in \cref{eq:properties-of-upper-concave-envelope-bar-f} with the specific majorant $\alpha = 0$ and $\beta = f(b)$. By Part 1, $\invbreve{f}(b) = f(b)$, and so $\invbreve{f}(b) = \bar{f}(b)$. Thus, $\invbreve{f}$ is non-decreasing on all of $\mathcal{I}$, as desired.

    \textbf{Part 3:} By concavity, $\invbreve{f}$ is continuous on the interior of $\mathcal{I}$. To establish the continuity of $\invbreve{f}$ at $0$, observe that
    \begin{equation}
        \forall t \in \mathcal{I}, \enspace \invbreve{f}(t) \stackclap{(a)}{\geq} \invbreve{f}(0) \stackclap{(b)}{=} f(0) \stackclap{(c)}{=} 0  ,
    \end{equation}
    where (a) holds by Part 2, (b) holds by Part 1, and (c) holds by the assumption $f(t) \leq t$ and the range of $f$. Furthermore,
    \begin{equation}
        \forall t \in \mathcal{I}, \enspace \invbreve{f}(t) \leq \bar{f}(t) \stackclap{(a)}{\leq} t  ,
    \end{equation}
    where (a) holds by upper-bounding the infimum in \cref{eq:properties-of-upper-concave-envelope-bar-f} with the specific majorant $\alpha = 1$ and $\beta = 0$ due to the assumption $f(t) \leq t$. Hence, we have $\lim_{t \rightarrow 0^+} \invbreve{f}(t) = 0 = \invbreve{f}(0)$ as desired. To establish the continuity of $\invbreve{f}$ at $1$, observe that $\lim_{t \rightarrow 1^-} \invbreve{f}(t) \leq \invbreve{f}(1)$ because $\invbreve{f}$ is non-decreasing by Part 2, and $\lim_{t \rightarrow 1^-} \invbreve{f}(t) \geq \lim_{t \rightarrow 1^-} \ibrc{t \, \invbreve{f}(1) + \prn{1 - t} \invbreve{f}(0)} = \invbreve{f}(1)$ because $\invbreve{f}$ is concave. Hence, $\lim_{t \rightarrow 1^-} \invbreve{f}(t) = \invbreve{f}(1)$ as desired.
\end{proof}

\section{Variational Characterization Under Equicontinuity} \label{appendix:variational-characterization-under-equicontinuity}

In this appendix, we provide an alternative statement and proof of \cref{theorem:variational-characterization-of-doeblin-coefficient} without the use of lattice infima, under the assumption of equicontinuity of the kernel $K$.

\begin{theorem}[Variational Characterization of Doeblin Coefficient] \label{theorem:variational-characterization-of-doeblin-coefficient-equicontinuity}
    Let $(\mathcal{X}, d_\mathcal{X})$ and $(\mathcal{Y}, \mathcal{F}_\mathcal{Y})$ be Polish spaces, where $\mathcal{X}$ is endowed with the metric $d_\mathcal{X}\colon \mathcal{X} \times \mathcal{X} \rightarrow \mathbb{R}_+$. Let $K\colon \mathcal{X} \times \mathcal{F}_\mathcal{Y} \rightarrow [0, 1]$ be an absolutely continuous Markov kernel with respect to the $\sigma$-finite measure $\mu\colon \mathcal{F}_\mathcal{Y} \rightarrow \mathbb{R}_+$. Assume the family of functions $\ibrc{x \mapsto \frac{\idiff K}{\idiff \mu}(y \mid x): y \in \mathcal{Y}}$ is equicontinuous \cite[Definition 7.22]{Rudin1976}, i.e.,   \begin{multline}
        \forall \epsilon > 0, \, \exists \delta > 0, \, \forall x, x' \in \mathcal{X}, \, \forall y \in \mathcal{Y}, \\ d_\mathcal{X}(x, x') < \delta \implies \abs{\frac{\diff K}{\diff \mu}(y \mid x) - \frac{\diff K}{\diff \mu}(y \mid x')} < \epsilon . \label{eq:equicontinuous-definition}
    \end{multline}
    Then, the Doeblin coefficient of $K$ admits the characterization
    \begin{align}
        \tau(K) = \inf_{n \in \mathbb{N}} \inf_{x_1, \dots, x_n \in \mathcal{X}} \inf_{\substack{\text{$n$-partition of $\mathcal{Y}$} \\ A_1, \dots, A_n}} \sum_{i=1}^n K(A_i \mid x_i) . \label{eq:variational-characterization-of-doeblin-coefficient-equicontinuity}
    \end{align}
\end{theorem}

\begin{proof}
    Fix a Markov kernel $K\colon \mathcal{X} \times \mathcal{F}_\mathcal{Y} \rightarrow [0, 1]$ satisfying the assumptions in \cref{theorem:variational-characterization-of-doeblin-coefficient-equicontinuity}. For notational convenience, denote the right-hand side of \cref{eq:variational-characterization-of-doeblin-coefficient-equicontinuity} as
    \begin{equation}
        \circled{1} \triangleq \inf_{n \in \mathbb{N}} \inf_{x_1, \dots, x_n \in \mathcal{X}} \inf_{\substack{\text{$n$-partition of $\mathcal{Y}$} \\ A_1, \dots, A_n}} \sum_{i=1}^n K(A_i \mid x_i)  .
    \end{equation}

    First, we will show that $\tau(K) \leq \circled{1}$. By definition of Doeblin coefficient as a supremum \cref{eq:doeblin-coefficient} which is achieved as per the discussion following \cref{eq:doeblin-coefficient}, there exists some probability measure $\pi^*\colon \mathcal{F}_\mathcal{Y} \rightarrow [0, 1]$ such that $K \geq \tau(K) \, \pi^*$. For any $n \in \mathbb{N}$, any $x_1, \tdots, x_n \in \mathcal{X}$, and any partition $A_1, \tdots, A_n$ of $\mathcal{Y}$, we have   \begin{equation}
        \tau(K) \stackclap{(a)}{=} \tau(K) \, \pi^*(\mathcal{Y}) \stackclap{(b)}{=} \tau(K) \sum_{i=1}^n \pi^*(A_i) \stackclap{(c)}{\leq} \sum_{i=1}^n K(A_i \mid x_i)  ,
    \end{equation}
    where (a) holds because $\pi^*$ is a probability measure, (b) holds because $A_1, \tdots, A_n$ is a partition of $\mathcal{Y}$, and (c) holds because $K \geq \tau(K) \, \pi^*$. Since $n$, $x_1, \tdots, x_n$, and $A_1, \tdots, A_n$ were arbitrary, it follows that $\tau(K) \leq \circled{1}$ as desired.

    Next, we will show that $\circled{1} \leq \tau(K)$. Since $\mathcal{X}$ is Polish and thus separable, there exists a countable dense net $\ibrc{x^*_1, x^*_2, \tdots} \subseteq \mathcal{X}$, i.e.,
    \begin{equation}
        \forall x \in \mathcal{X}, \, \forall \delta > 0, \, \exists i \in \mathbb{N}, \quad d_\mathcal{X}(x, x^*_i) \leq \delta . \label{eq:variational-characterization-equicontinuity-step-one}
    \end{equation}
    For notational convenience, define a sequence of functions $g_n\colon \mathcal{Y} \rightarrow \mathbb{R}_+$ as
    \begin{equation}
        \forall n \in \mathbb{N}, \quad g_n(y) \triangleq \min_{i \in [n]} \frac{\diff K}{\diff \mu}(y \mid x^*_i) .
    \end{equation}
    We have   \begin{align}
        \circled{1} &\stackclap{(a)}{=} \inf_{n \in \mathbb{N}} \inf_{x_1, \dots, x_n \in \mathcal{X}} \int_{\mathcal{Y}} \prn{\min_{i \in [n]} \frac{\diff K}{\diff \mu}(\lwildcard \mid x_i)} \diff \mu \\&\stackclap{(b)}{\leq} \inf_{n \in \mathbb{N}} \int_{\mathcal{Y}} g_n \diff \mu \stackclap{(c)}{=} \lim_{n \rightarrow \infty} \int_{\mathcal{Y}} g_n \diff \mu , \label{eq:variational-characterization-equicontinuity-step-two}
    \end{align}
    where (a) holds by \cref{proposition:integral-characterization-of-infimum}, (b) holds by upper-bounding the inner infimum with the specific instance $\ibrc{x^*_1, \tdots, x^*_n}$, and (c) holds because $\ibrc{g_n}_{n \in \mathbb{N}}$ is a non-increasing sequence, since each successive $g_n$ is the minimum over a larger set of $i$. Next, define a function $g\colon \mathcal{Y} \rightarrow \mathbb{R}_+$ as
    \begin{equation}
        g(y) \triangleq \lim_{n \rightarrow \infty} g_n(y) = \inf_{i \in \mathbb{N}} \frac{\diff K}{\diff \mu}(y \mid x^*_i) . \label{eq:variational-characterization-equicontinuity-step-three}
    \end{equation}
    Observe that $g$ is measurable, because it is the countable infimum of measurable functions.   %
    Furthermore, observe that $g_n \leq g_1$ for all $n \in \mathbb{N}$, $g \leq g_1$, and   \begin{math}
        \int_\mathcal{Y} g_1 \diff \mu = K(\mathcal{Y} \mid x^*_1) = 1 < \infty
    \end{math}.
    Hence, proceeding from \cref{eq:variational-characterization-equicontinuity-step-two} and applying the dominated convergence theorem, we have   \begin{equation}
        \circled{1} \leq \int_{\mathcal{Y}} \prn{\lim_{n \rightarrow \infty} g_n} \diff \mu = \underbracket{\int_{\mathcal{Y}} g \diff \mu}_{=\alpha^*} , \label{eq:variational-characterization-equicontinuity-step-four}
    \end{equation}
    where we define $\alpha^* \in [0, 1]$ above for convenience hereafter. If $\alpha^* = 0$, we trivially have $\circled{1} = 0 \leq \tau(K)$ as desired. Hence, assume $\alpha^* > 0$ for the remainder of this argument. Define a probability measure $\pi^*\colon \mathcal{F}_\mathcal{Y} \rightarrow [0, 1]$ as   $\pi^*(A) \triangleq (\ifrac{1}{\alpha^*}) \int_A g \diff \mu$ for all $A \in \mathcal{F}_\mathcal{Y}$.
    Fix an arbitrary $\epsilon > 0$. Consider $\delta > 0$ (possibly depending on $\epsilon$) such that \cref{eq:equicontinuous-definition} holds. Observe that for any $x \in \mathcal{X}$ and $y \in \mathcal{Y}$, we have
    \begin{equation}
        g(y) \stackclap{(a)}{\leq} \frac{\diff K}{\diff \mu}(y \mid x^*_i) \stackclap{(b)}{\leq} \frac{\diff K}{\diff \mu}(y \mid x) + \epsilon , \label{eq:variational-characterization-equicontinuity-step-five}
    \end{equation}
    where (a) holds for \emph{all} $i \in \mathbb{N}$ by definition of $g$ \cref{eq:variational-characterization-equicontinuity-step-three}, and (b) holds for \emph{some} $i \in \mathbb{N}$ (possibly depending on $x$) such that $d_\mathcal{X}(x, x^*_i) < \delta$, by equicontinuity \cref{eq:equicontinuous-definition}; we remark that the existence of such an $i$ is guaranteed by \cref{eq:variational-characterization-equicontinuity-step-one}. Hence, for any $x \in \mathcal{X}$ and $A \in \mathcal{F}_\mathcal{Y}$,
    \begin{equation}
        \alpha^* \pi^*(A) = \int_A g \diff \mu \stackclap{(a)}{\leq} \int_A \prn{\frac{\diff K}{\diff \mu}(\lwildcard \mid x)} \diff \mu = K(A \mid x) , \label{eq:variational-characterization-equicontinuity-step-six}
    \end{equation}
    where (a) holds by \cref{eq:variational-characterization-equicontinuity-step-five} because $\epsilon > 0$ was arbitrary. Proceeding from \cref{eq:variational-characterization-equicontinuity-step-four}, we have   \begin{equation}
        \circled{1} \leq \alpha^* \stackclap{(a)}{\leq} \sup \brc{\alpha \in \mathbb{R}: \exists \pi \in \mathscr{P}, \, K \geq \alpha \pi} \stackclap{(b)}{=} \tau(K)
    \end{equation}
    as desired, where (a) holds because $K \geq \alpha^* \pi^*$ by \cref{eq:variational-characterization-equicontinuity-step-six}, and (b) holds by definition of Doeblin coefficient \cref{eq:doeblin-coefficient}.
\end{proof}

We remark that the equicontinuity assumption on $K$ allows us to circumvent measurability issues in the proof by defining $g$ as the pointwise infimum over a countable dense subset of $\mathcal{X}$, which is guaranteed to be measurable while approximating the infimum over all of general uncountable $\mathcal{X}$ to any $\epsilon > 0$ accuracy. Otherwise, $g$ would have to be defined as the lattice or essential infimum, as was done in the proof of \cref{theorem:variational-characterization-of-doeblin-coefficient} in \cref{subsection:proofs-of-doeblin-characterizations}.

\section{Markov Kernel Examples} \label{appendix:markov-kernel-examples}

In this appendix, we provide examples of non-trivial Markov kernels satisfying the preconditions for various results in our paper. Throughout this appendix, let $\Phi\colon \mathbb{R} \rightarrow (0, 1)$ denote the standard Gaussian CDF   \begin{math}
    \Phi(x) \triangleq \int_{-\infty}^x (1 / \sqrt{2 \pi}) \exp(-t^2 / 2) \diff t
\end{math}.

First, we present a Markov kernel $W\colon \mathcal{U} \times \mathcal{F}_\mathcal{X} \rightarrow [0, 1]$ with a countably infinite source space $\mathcal{U}$, unbounded support on the target space $\mathcal{X}$, and finite average extremal power.

\begin{proposition}[Average Extremal Power Example] \label{proposition:average-extremal-power-example-countable}
    Let $\mathcal{U}$ be countable (without loss of generality, $\mathcal{U} \triangleq \mathbb{N}$) and $\mathcal{X} \triangleq \mathbb{R}$. Let $W\colon \mathcal{U} \times \mathcal{F}_\mathcal{X} \rightarrow [0, 1]$ be a Markov kernel such that for each $i \in \mathcal{U}$, $X_i \sim W(\lwildcard \mid i)$ is sub-Gaussian with mean $0$ and variance factor $\sigma_i^2 \leq c / \log_e(i + 1)$, i.e.,
    \begin{equation}
        \forall \lambda \in \mathbb{R}, \quad \E{e^{\lambda \prn{X_i - \E{X_i}}}} \leq \exp \mprn{\frac{\sigma_i^2 \lambda^2}{2}}  ,
    \end{equation}
    where $c > 0$ is a fixed constant. Then, under the norm $\inorm{x} \triangleq \iabs{x}$ and power function $M(z) \triangleq z^2$, the average extremal power of $W$ satisfies   \begin{math}
        \tnorm{W}_\mathsf{AE} \leq 4c (1 + \pi^2 / 6)
    \end{math}.
\end{proposition}

\begin{proof}
    Let $\mathbb{P}$ be any coupling of random variables $\ibrc{X_i}_{i \in \mathbb{N}}$ with $X_i \sim W(\lwildcard \mid i)$ for each $i \in \mathbb{N}$. In the following, all expectations are taken with respect to $\ibrc{X_i}_{i \in \mathbb{N}} \sim \mathbb{P}$. We have   \begin{align}
        \tnorm{W}_\mathsf{AE} &\stackclap{(a)}{=} \E{\sup_{i \in \mathbb{N}} M(\norm{X_i})} \\&\stackclap{(b)}{=} \E{\sup_{i \in \mathbb{N}} \abs{X_i}^2} \\&= \E{\lim_{n \rightarrow \infty} \max_{i \in [n]} \abs{X_i}^2} \\&\stackclap{(c)}{=} \lim_{n \rightarrow \infty} \underbracket{\E{\max_{i \in [n]} \abs{X_i}^2}}_{\circled{1}} , \label{eq:ap-power-constraint-step-one}
    \end{align}
    where (a) holds by definition of average extremal power \cref{eq:average-extremal-power}, (b) holds by definition of $\inorm{\wildcard}$ and $M$, and (c) holds by the monotone convergence theorem because the sequence of functions $\ibrc{g_n}_{n \in \mathbb{N}}$ given by   \begin{math}
        g_n \iprn{\ibrc{x_i}_{i \in \mathbb{N}}} \triangleq \max_{i \in [n]} \iabs{x_i}^2
    \end{math}
    is monotonically non-decreasing in $n$. Next, we upper-bound $\circled{1}$ by adapting the result in \cite[Lemma 2.3]{VanHandel2017}.   %
    For any $t_0 \geq 0$, we have   \begin{align}
        \circled{1}&\stackclap{(a)}{\leq} \int_0^\infty \P{\max_{i \in [n]} \abs{X_i}^2 \geq t} \diff t\\&\stackclap{(b)}{=} \int_0^{t_0} \P{\max_{i \in [n]} \abs{X_i}^2 \geq t} \diff t + \!\int_{t_0}^\infty \P{\max_{i \in [n]} \abs{X_i}^2 \geq t} \diff t \\
        &\stackclap{(c)}{\leq} t_0 + \int_{t_0}^\infty \P{\max_{i \in [n]} \abs{X_i}^2 \geq t} \diff t \\&\stackclap{(d)}{\leq} t_0 + \int_{t_0}^\infty \sum_{i=1}^n \P{\abs{X_i}^2 \geq t} \diff t \\&= t_0 + \sum_{i=1}^n \int_{t_0}^\infty \P{\abs{X_i} \geq t^{\frac{1}{2}}} \diff t \\
        &\stackclap{(e)}{\leq} t_0 + \sum_{i=1}^n \int_{t_0}^\infty 2 \exp \mprn{-\frac{t}{2 \sigma_i^2}} \diff t \\&\stackclap{(f)}{=} t_0 + 4 \sum_{i=1}^n \sigma_i^2 \exp \mprn{-\frac{t_0}{2 \sigma_i^2}} , \label{eq:ap-power-constraint-step-two}
    \end{align}
    where (a) holds by the layer cake representation \cite[Lemma 1.2.1]{Vershynin2018}, (b) holds for any $t_0 \in [0, \infty]$ by splitting the region of integration into two intervals,   %
    (c) holds because probability values are bounded by $1$, (d) holds by the union bound, (e) holds by \cite[Eq. 2.9]{Wainwright2019} because $X_i$ is sub-Gaussian with mean $0$ and variance factor $\sigma_i^2$, and (f) holds by evaluating the integral. Choosing   \begin{math}
        t_0 \triangleq 4 \max_{i \in [n]} \ibrc{\sigma_i^2 \log_e(i + 1)}
    \end{math}
    and continuing from \cref{eq:ap-power-constraint-step-two}, we have   \begin{align}
        &\phantom{{}={}}\circled{1} \\&\stackclap{(a)}{\leq} 4 \max_{i \in [n]} \brc{\sigma_i^2 \log_e(i + 1)} \!+\! 4 \sum_{i=1}^n \sigma_i^2 \exp \mprn{\!\!-\frac{4 \sigma_i^2 \log_e(i + 1)}{2 \sigma_i^2}\!} \\
        &= 4 \max_{i \in [n]} \brc{\sigma_i^2 \log_e(i + 1)} \!+\! 4 \sum_{i=1}^n \frac{\sigma_i^2}{(i + 1)^2} \\&\stackclap{(b)}{\leq} 4c + 4 \sum_{i=1}^n \frac{c}{(i + 1)^2 \log_e(i + 1)} , \label{eq:ap-power-constraint-step-three}
    \end{align}
    where (a) holds by the choice of $t_0$ and (b) holds because $\sigma_i^2 \leq c / \log_e(i + 1)$ for each $i \in \mathbb{N}$. Combining \cref{eq:ap-power-constraint-step-one,eq:ap-power-constraint-step-three}, we have   \begin{align}
        \tnorm{W}_\mathsf{AE} &\leq 4c \biggprn{1 + \sum_{i=1}^\infty \frac{1}{(i + 1)^2 \log_e(i + 1)}} \\&\leq 4c \biggprn{1 + \sum_{i=1}^\infty \frac{1}{i^2}} = 4c \biggprn{1 + \frac{\pi^2}{6}}
    \end{align}
    as desired.
\end{proof}

Next, we present a Markov kernel $W\colon \mathcal{U} \times \mathcal{F}_\mathcal{X} \rightarrow [0, 1]$ with an uncountable source space $\mathcal{U}$, unbounded support on the target space $\mathcal{X}$, and finite average extremal power.

\begin{proposition}[Average Extremal Power Example] \label{proposition:average-extremal-power-example-uncountable}
    Let $\mathcal{U} \triangleq (0, \infty)$ and $\mathcal{X} \triangleq \mathbb{R}$. Let $W\colon \mathcal{U} \times \mathcal{F}_\mathcal{X} \rightarrow [0, 1]$ be the Markov kernel such that for each $u \in \mathcal{U}$, $X_u \sim W(\lwildcard \mid u)$ is given by
    \begin{equation}
        X_u \sim \begin{cases}
            \mathsf{Normal} \mprn{0, 1}  , & \text{if $J_u = 1$}  , \\
            2 \, \mathsf{Beta}(u, u) - 1  , & \text{if $J_u = 0$}  ,
        \end{cases}
    \end{equation}
    where $J_u \sim \mathsf{Bernoulli}(1/2)$ is independent of the Beta and Gaussian random variables. Namely, the cumulative distribution function $f_{X_u}\colon \mathbb{R} \rightarrow (0, 1)$ is   \begin{multline}
        f_{X_u}(x) = \frac{1}{2} \, \Phi(x)\\+ \begin{cases}
            0  , & \text{if $x < -1$}  , \\
            \frac{\int_{-1}^x \prn{\frac{t + 1}{2}}^{u-1} \prn{1 - \frac{t + 1}{2}}^{u-1} \diff t}{4 \mathrm{B}(u, u)}   , & \text{if $-1 \leq x \leq 1$}  , \\
            \frac{1}{2}  , & \text{if $x > 1$}  ,
        \end{cases}
    \end{multline}
    where $\mathrm{B}(\alpha, \beta)$ denotes the Beta function   \begin{math}
        \mathrm{B}(\alpha, \beta) \triangleq \int_0^1 t^{\alpha-1}\allowbreak \iprn{1 - t}^{\beta-1} \diff t
    \end{math},
    and the probability density function $p_{X_u}$ is   \begin{equation}
        p_{X_u}(x) = \begin{cases}
            \frac{\prn{\frac{x + 1}{2}}^{u-1} \prn{1 - \frac{x + 1}{2}}^{u-1}}{4 \mathrm{B}(u, u)}  + \frac{\exp \mprn{-\frac{x^2}{2}}}{2 \sqrt{2 \pi}}   , & \text{if $\abs{x} < 1$} , \\
            \frac{\exp \mprn{-\frac{x^2}{2}}}{2 \sqrt{2 \pi}}   , & \text{if $\abs{x} > 1$}  .
        \end{cases} \label{eq:average-extremal-power-example-uncountable-pdf}
    \end{equation}
    Then, under the norm $\inorm{x} \triangleq |x|$ and power function $M(z) \triangleq z^2$, the average extremal power of $W$ satisfies   \begin{math}
        \itnorm{W}_\mathsf{AE} \leq 1
    \end{math}.
\end{proposition}

\begin{proof}
    First, we compute the pointwise infimum of the probability densities $\ibrc{p_{X_u}}_{u \in \mathcal{U}}$. We have   \begin{multline}
        \inf_{u \in \mathcal{U}} p_{X_u}(x) \stackclap{(a)}{=} \frac{1}{2 \sqrt{2 \pi}} \exp \mprn{-\frac{x^2}{2}} \\+ \Iv{\abs{x} < 1} \underbracket{\inf_{u \in \mathcal{U}} \frac{1}{4 \mathrm{B}(u, u)} \prn{\frac{x + 1}{2}}^{u-1} \prn{\frac{1 - x}{2}}^{u-1}}_{\circled{1}}, \label{eq:average-extremal-power-example-uncountable-step-one}
    \end{multline}
    where (a) holds by \cref{eq:average-extremal-power-example-uncountable-pdf}. For any $x \in (-1, 1)$,   \begin{align}
        \circled{1} &\stackclap{(a)}{\leq} \lim_{u \rightarrow 0} \frac{1}{4 \mathrm{B}(u, u)} \prn{\frac{x + 1}{2}}^{u-1} \prn{\frac{1 - x}{2}}^{u-1} \\&= \frac{1}{4} \lim_{u \rightarrow 0} \brc{\frac{1}{\mathrm{B}(u, u)}} \lim_{u \rightarrow 0} \biggbrc{\prn{\frac{x + 1}{2}}^{u-1} \prn{\frac{1 - x}{2}}^{u-1}} \\
        &\stackclap{(b)}{=} \frac{1}{4} \prn{\frac{4}{1 - x^2}} \lim_{u \rightarrow 0} \frac{1}{\mathrm{B}(u, u)} \\&\stackclap{(c)}{=} \frac{1}{1 - x^2} \lim_{u \rightarrow 0} \frac{\Gamma(2u)}{\Gamma(u)^2} \\&\stackclap{(d)}{=} \frac{1}{1 - x^2} \lim_{u \rightarrow 0} \frac{2^{2u-1} \Gamma \mprn{u + \frac{1}{2}}}{\sqrt{\pi} \, \Gamma(u)} \\&= \frac{\lim_{u \rightarrow 0} \frac{1}{\Gamma(u)}}{2 \prn{1 - x^2}}  = 0  , \label{eq:average-extremal-power-example-uncountable-step-two}
    \end{align}
    where (a) holds because $u = 0$ is a limit point of $\mathcal{U}$, (b) holds because $x \in (-1, 1)$ and so $((x + 1) / 2)^{-1}$ and $((1 - x) / 2)^{-1}$ are well-defined, (c) holds by the identity   \begin{math}
        \mathrm{B}(\alpha, \beta) = \ifrac{\Gamma(\alpha) \, \Gamma(\beta)}{\Gamma(\alpha + \beta)}
    \end{math},
    where $\Gamma$ denotes the Gamma function   \begin{math}
        \Gamma(\alpha) \triangleq \int_0^\infty t^{\alpha - 1} e^{-t} \diff t
    \end{math},
    and (d) holds by the \emph{Legendre duplication formula}   \begin{math}
        \Gamma(2 \alpha) = \ifrac{2^{2 \alpha - 1} \Gamma(\alpha) \, \Gamma \mprn{\alpha + \ifrac{1}{2}}}{\sqrt{\pi}}
    \end{math}.
    Also, clearly $\circled{1} \geq 0$ by inspection. Hence, combining \cref{eq:average-extremal-power-example-uncountable-step-one,eq:average-extremal-power-example-uncountable-step-two} yields
    \begin{equation}
        \inf_{u \in \mathcal{U}} p_{X_u}(x) = \frac{1}{2 \sqrt{2 \pi}} \exp \mprn{-\frac{x^2}{2}}  .
    \end{equation}
    
    Following from the above, the greatest common component of $W$ is   \begin{equation}
        \bigwedge_{u \in \mathcal{U}} W(A \mid u) \stackclap{(a)}{=} \int_A \prn{\inf_{u \in \mathcal{U}} p_{X_u}(x)} \diff x = \int_A \frac{\exp \mprn{-\frac{x^2}{2}}}{2 \sqrt{2 \pi}} \diff x ,
    \end{equation}
    where (a) holds by \cref{lemma:density-of-greatest-common-component} and because the lattice infimum reduces to the pointwise infimum when the latter is measurable. Hence, the maximal coupling of random variables $\ibrc{X_u}_{u \in \mathcal{U}}$ with $X_u \sim W(\lwildcard \mid u)$ for each $u \in \mathcal{U}$, defined in \cref{eq:maximal-coupling}, is
    \begin{align}
        \forall u \in \mathcal{U}, \quad X_u \triangleq \begin{cases}
            X^*  , & \text{if $I = 1$}  , \\
            \tilde{X}_u  , & \text{if $I = 0$}  ,
        \end{cases} \label{eq:average-extremal-power-example-uncountable-maximal-coupling}
    \end{align}
    where the random variables $I$, $X^*$, and $\ibrc{\tilde{X}_u}_{u \in \mathcal{U}}$ are sampled independently from the distributions
    \begin{align}
        I &\sim \mathsf{Bernoulli} \mprn{\frac{1}{2}}  , \label{eq:average-extremal-power-example-uncountable-maximal-coupling-i} \\
        X^* &\sim \mathsf{Normal}(0, 1)  , \label{eq:average-extremal-power-example-uncountable-maximal-coupling-x-star} \\
        \forall u \in \mathcal{U}, \quad \tilde{X}_u &\sim 2 \, \mathsf{Beta}(u, u) - 1  . \label{eq:average-extremal-power-example-uncountable-maximal-coupling-x-tilde}
    \end{align}
    Therefore, the average extremal power of $W$ is   \begin{align}
        \tnorm{W}_\mathsf{AE} &\stackclap{(a)}{=} \P{I = 1} \, \E{\sup_{u \in \mathcal{U}} M(\norm{X_u}) \given I = 1} \\&\phantom{{}={}} + \P{I = 0} \, \E{\sup_{u \in \mathcal{U}} M(\norm{X_u}) \given I = 0} \\&\stackclap{(b)}{=} \underbracket{\frac{1}{2} \, \E{\prn{X^*}^2}}_{= \frac{1}{2} \text{ \cref{eq:average-extremal-power-example-uncountable-maximal-coupling-x-star}}} + \underbracket{\frac{1}{2} \, \E{\sup_{u \in \mathcal{U}} \tilde{X}_u^2}}_{\leq \frac{1}{2} \text{ \cref{eq:average-extremal-power-example-uncountable-maximal-coupling-x-tilde}}} \leq 1
    \end{align}
    as desired, where all expectations are taken with respect to the maximal coupling \cref{eq:average-extremal-power-example-uncountable-maximal-coupling}, (a) holds by definition of average extremal power \cref{eq:average-extremal-power} and the law of total expectation, and (b) holds by definitions of $\inorm{\wildcard}$, $M$, and $X_u$.
\end{proof}

Next, we present a Markov kernel $W\colon \mathcal{U} \times \mathcal{F}_\mathcal{X} \rightarrow [0, 1]$ with unbounded support on the target space $\mathcal{X}$ which satisfies the preconditions for \cref{proposition:upper-bound-on-uniform-average-doeblin-curve-finite}.

\begin{proposition}[Upper Bound Example] \label{proposition:upper-bound-example-finite}
    Let $n \in \mathbb{N}$ and $\varsigma > 0$ be fixed constants. Let $\mathcal{U} \triangleq [n]$ and let $\mathcal{X} \triangleq \mathbb{R}_+$ be equipped with the norm $\inorm{\wildcard} \triangleq |\wildcard|$. Let $W\colon \mathcal{U} \times \mathcal{F}_\mathcal{X} \rightarrow [0, 1]$ be the Markov kernel such that for each $i \in \mathcal{U}$, $X_i \sim W(\lwildcard \mid i)$ is given by
    \begin{equation}
        X_i \triangleq \begin{cases}
            X_{i \mid 1} \sim \mu_i + \mathsf{HalfNormal}(\varsigma)  , & \text{if $J_i = 1$}  , \\
            X_{i \mid 0} \sim \mathsf{Uniform} \mprn{i, \mu_i}  , & \text{if $J_i = 0$}  ,
        \end{cases}
    \end{equation}
    where $J_i \sim \mathsf{Bernoulli}(1/2)$, $X_{i \mid 1}$, and $X_{i \mid 0}$ are independent, and $\mu_i \triangleq i + 2 \varsigma \sqrt{2/\pi}$. Namely, the cumulative distribution function $f_{X_i}\colon \mathbb{R} \rightarrow [0, 1)$ is
    \begin{equation}
        f_{X_i}(x) = \begin{cases}
            0  , & \text{if $x < i$}  , \\
            \frac{x - i}{4 \varsigma} \sqrt{\frac{\pi}{2}}  , & \text{if $i \leq x < \mu_i$}  , \\
            \Phi \mprn{\frac{x - \mu_i}{\varsigma}}  , & \text{if $x \geq \mu_i$}  ,
        \end{cases}
    \end{equation}
    and the probability density function $p_{X_i}$ is
    \begin{align}
        p_{X_i}(x) = \begin{cases}
            0  , & \text{if $x < i$}  , \\
            \frac{1}{4 \varsigma} \sqrt{\frac{\pi}{2}}  , & \text{if $i \leq x < \mu_i$}  , \\
            \frac{1}{\varsigma \sqrt{2 \pi}} \exp \mprn{-\frac{\prn{x - \mu_i}^2}{2 \varsigma^2}}  , & \text{if $x > \mu_i$}  .
        \end{cases}
    \end{align}
    Then, under the definition of $\mathcal{G}$ in \cref{proposition:upper-bound-on-uniform-average-doeblin-curve-finite} with power function $M(z) \triangleq z$, $p \triangleq \mu_n$, and $\sigma \triangleq \sqrt{2} \prn{4 \varsigma / \pi}$, we have $W \in \mathcal{G}$.
\end{proposition}

\begin{proof}
    Throughout this proof, for notational convenience, let   \begin{math}
        Z_i \triangleq M(\inorm{X_i}) \stackclap{(a)}{=} X_i
    \end{math}
    for all $i \in \mathcal{U}$, where (a) holds because $X_i \geq i > 0$.
    
    \textbf{Uniform average power constraint:} Observe that for each $i \in \mathcal{U}$, we have   \begin{align}
        \E{Z_i} = \E{X_i} &\stackclap{(a)}{=} \P{J_i = 0}  \E{X_{i \mid 0}} + \P{J_i = 1}  \E{X_{i \mid 1}} \\&\stackclap{(b)}{=} \frac{1}{2}  \E{X_{i \mid 0}} + \frac{1}{2}  \E{X_{i \mid 1}} \\&\stackclap{(c)}{=} \frac{1}{2} \mprn{\frac{i + \mu_i}{2}} + \frac{1}{2} \mprn{\mu_i + \varsigma \sqrt{2 / \pi}} = \mu_i  , \label{eq:upper-bound-example-finite-step-one}
    \end{align}
    where (a) holds by the law of total expectation, (b) holds because $J_i \sim \mathsf{Bernoulli}(1/2)$, and (c) holds because $X_{i \mid 0} \sim \mathsf{Uniform}(i, \mu_i)$ and $X_{i \mid 1} \sim \mu_i + \mathsf{HalfNormal}(\varsigma)$. Therefore,
    \begin{equation}
        \tnorm{W}_\mathsf{UA} \stackclap{(a)}{=} \sup_{i \in \mathcal{U}} \E{Z_i} \stackclap{(b)}{=} \mu_n = p
    \end{equation}
    as desired, where (a) holds by definition of uniform average power \cref{eq:uniform-average-power} and (b) holds from \cref{eq:upper-bound-example-finite-step-one} because $\mathcal{U} = [n]$.

    \textbf{Sub-Gaussianity:} Fix $i \in \mathcal{U}$. For any $t \geq 0$, we have   \begin{align}
        &\phantom{{}={}}\P{\abs{Z_i - \E{Z_i}} \geq t}
        \\&= \P{\abs{X_i - \E{X_i}} \geq t}\\&\stackclap{(a)}{=} \P{J_i = 0} \, \P{\abs{X_{i \mid 0} - \E{X_i}} \geq t} \\&\phantom{{}={}}+ \P{J_i = 1} \, \P{\abs{X_{i \mid 1} - \E{X_i}} \geq t} \\
        &\stackclap{(b)}{=} \frac{1}{2} \, \P{\abs{X_{i \mid 0} - \E{X_i}} \geq t} + \frac{1}{2} \, \P{\abs{X_{i \mid 1} - \E{X_i}} \geq t} \\&\stackclap{(c)}{=} \frac{1}{2} \, \P{\abs{X_{i \mid 0} - \mu_i} \geq t} + \frac{1}{2} \, \P{\abs{X_{i \mid 1} - \mu_i} \geq t} \\
        &\stackclap{(d)}{=} \underbracket{\frac{1}{2} \, \P{X_{i \mid 0} \leq \mu_i - t}}_{\circled{1}} + \underbracket{\frac{1}{2} \, \P{X_{i \mid 1} \geq \mu_i + t}}_{\circled{2}}  , \label{eq:upper-bound-example-finite-step-two}
    \end{align}
    where (a) holds by the law of total probability, (b) holds because $J_i \sim \mathsf{Bernoulli}(1/2)$, (c) holds by \cref{eq:upper-bound-example-finite-step-one}, and (d) holds because $X_{i \mid 0} \leq \mu_i$ and $X_{i \mid 1} \geq \mu_i$. Next, we upper-bound $\circled{1}$ and $\circled{2}$. Let $N_i \sim \mathsf{Normal}(0, (4 \varsigma / \pi)^2)$ be an independent Gaussian random variable. Then,   \begin{align}
        \circled{1} &\stackclap{(a)}{=} \frac{1}{2} \max \brc{0, 1 - \int_{-t}^0 \frac{1}{2 \varsigma \sqrt{2 / \pi}} \diff x} \\&= \max \brc{0, \frac{1}{2} - \int_{-t}^0 \frac{1}{4 \varsigma \sqrt{2 / \pi}} \diff x} \\
        &\stackclap{(b)}{\leq} \max \brc{0, \frac{1}{2} - \int_{-t}^0 \frac{1}{(4 \varsigma / \pi) \sqrt{2 \pi}} \exp \mprn{-\frac{x^2}{2 (4 \varsigma / \pi)^2}} \diff x} \\&\stackclap{(c)}{=} \P{N_i \leq -t} \, \label{eq:upper-bound-example-finite-step-three}
    \end{align}
    where (a) holds by the PDF of a uniform random variable supported on an interval of length $\mu_i - i = 2 \varsigma \sqrt{2/\pi}$, (b) holds by lower-bounding the integrand pointwise, and (c) holds by the PDF of $N_i \sim \mathsf{Normal}(0, (4 \varsigma / \pi)^2)$. Similarly, we also have   \begin{align}
        \circled{2} &\stackclap{(a)}{=} \frac{1}{2} \int_t^\infty \frac{\sqrt{2}}{\varsigma \sqrt{\pi}} \exp \mprn{-\frac{x^2}{2 \varsigma^2}} \diff x = 1 - \Phi \mprn{\frac{t}{\varsigma}} \\&\stackclap{(b)}{\leq} 1 - \Phi \mprn{\frac{t}{4 \varsigma / \pi}} = \P{N_i \geq t}  , \label{eq:upper-bound-example-finite-step-four}
    \end{align}
    where (a) holds by the PDF of a $\mathsf{HalfNormal}(\varsigma)$ random variable and (b) holds because CDFs are non-decreasing. Combining \cref{eq:upper-bound-example-finite-step-two,eq:upper-bound-example-finite-step-three,eq:upper-bound-example-finite-step-four}, we thus have   \begin{math}
        \iP{\iabs{Z_i - \E{Z_i}} \geq t} \leq \iP{\iabs{N_i} \geq t}
    \end{math},
    and so by \cite[Theorem 2.6]{Wainwright2019} we have sub-Gaussianity with variance factor $\sigma = \sqrt{2} \iprn{4 \varsigma / \pi}$ as desired, i.e.,   %
      $\iE{\exp(\lambda \iprn{Z_i - \iE{Z_i}})} \leq \exp \iprn{\ifrac{16 \varsigma^2 \lambda^2}{\pi^2}}$ for all $\lambda \in \mathbb{R}$.
\end{proof}

Finally, we present a Markov kernel $W\colon \mathcal{U} \times \mathcal{F}_\mathcal{X} \rightarrow [0, 1]$ with an uncountable source space $\mathcal{U}$ and unbounded support on the target space $\mathcal{X}$ which satisfies the preconditions for \cref{proposition:upper-bound-on-uniform-average-doeblin-curve-totally-bounded}.

\begin{proposition}[Upper Bound Example] \label{proposition:upper-bound-example-totally-bounded}
    Let $a > 0$ and $\varsigma > 0$ be fixed constants. Let $\mathcal{U} \triangleq [-a, a]$ be equipped with the metric $d_\mathcal{U}(u, v) \triangleq |u - v|$. Let $\mathcal{X} \triangleq \mathbb{R}$ be equipped with the norm $\inorm{\wildcard} \triangleq |\wildcard|$. Let $W\colon \mathcal{U} \times \mathcal{F}_\mathcal{X} \rightarrow [0, 1]$ be the Markov kernel such that for each $u \in \mathcal{U}$, $X_u \sim W(\lwildcard \mid u)$ is given by
    \begin{equation}
        X_u \triangleq \begin{cases}
            X_{u \mid 1} \sim \mathsf{Normal} \mprn{0, \varsigma^2}  , & \text{if $J_u = 1$}  , \\
            X_{u \mid 0} \sim \updelta_u  , & \text{if $J_u = 0$}  ,
        \end{cases}
    \end{equation}
    where $J_u \sim \mathsf{Bernoulli}(1/2)$, $X_{u \mid 1}$, and $X_{u \mid 0}$ are independent. Namely, the probability measure $W(\lwildcard \mid u)\colon \mathcal{F}_\mathcal{X} \rightarrow [0, 1]$ is   \begin{equation}
        \forall A \in \mathcal{F}_\mathcal{X}, \quad W(A \mid u) = \int_A \frac{\exp \mprn{-\frac{x^2}{2 \varsigma^2}}}{2 \varsigma \sqrt{2 \pi}}  \diff x + \frac{1}{2} \, \updelta_u(A)  ,
    \end{equation}
    and the cumulative distribution function $f_{X_u}\colon \mathbb{R} \rightarrow (0, 1)$ is
    \begin{equation}
        f_{X_u}(x) = \begin{cases}
            \frac{1}{2} \Phi \mprn{\frac{x}{\varsigma}}  , & \text{if $x < u$}  , \\
            \frac{1}{2} \Phi \mprn{\frac{x}{\varsigma}} + \frac{1}{2}  , & \text{if $x \geq u$}  .
        \end{cases}
    \end{equation}
    Then, under the definition of $\mathcal{G}$ in \cref{proposition:upper-bound-on-uniform-average-doeblin-curve-totally-bounded} with power function $M(z) \triangleq z^2$, $p \triangleq (a^2 + \varsigma^2)/2$, and $\sigma \triangleq a$, we have $W \in \mathcal{G}$.
\end{proposition}

This example resembles an erasure channel wherein erasures are replaced by independent Gaussian variables, akin to the construction of symmetric channels using erasure channels on finite alphabets (cf. \cite{Evansetal2000,MakurMosselPolyanskiy2022}).

\begin{proof}
    Throughout this proof, for notational convenience, let $Z_u \triangleq M(\inorm{X_u}) = X_u^2$ for all $u \in \mathcal{U}$.
    
    \textbf{Uniform average power constraint:} Observe that for each $u \in \mathcal{U}$, we have   \begin{align}
        \E{Z_u} &= \E{X_u^2} \\&\stackclap{(a)}{=} \P{J_u = 0} \, \E{X_{u \mid 0}^2}\! + \P{J_u = 1} \, \E{X_{u \mid 1}^2} \\&\stackclap{(b)}{=} \frac{1}{2} \, \E{X_{u \mid 0}^2} + \frac{1}{2} \, \E{X_{u \mid 1}^2} \\&\stackclap{(c)}{=} \frac{u^2 + \varsigma^2}{2}  , \label{eq:upper-bound-example-totally-bounded-step-one}
    \end{align}
    where (a) holds by the law of total expectation, (b) holds because $J_u \sim \mathsf{Bernoulli}(1/2)$, and (c) holds because $X_{u \mid 0} \sim \updelta_u$ and $X_{u \mid 1} \sim \mathsf{Normal}(0, \varsigma^2)$. Therefore,   \begin{equation}
        \tnorm{W}_\mathsf{UA} \stackclap{(a)}{=} \sup_{u \in \mathcal{U}} \E{Z_u} \stackclap{(b)}{=} (a^2 + \varsigma^2) / 2 = p
    \end{equation}
    as desired, where (a) holds by definition of uniform average power \cref{eq:uniform-average-power} and (b) holds from \cref{eq:upper-bound-example-totally-bounded-step-one} because $\mathcal{U} = [-a, a]$.

    \textbf{Sub-Gaussian increments:} Let $\ibrc{X_u}_{u \in \mathcal{U}}$ be distributed according to the maximal coupling defined in \cref{eq:maximal-coupling}. Fix $u, v \in \mathcal{U}$. Since   \begin{equation}
        \mathbb{E}\bigbkt{{e^{\lambda \prn{\prn{Z_u - Z_v} - \E{Z_u - Z_v}}}}} = \underbracket{\mathbb{E}\bigbkt{e^{\lambda \prn{Z_u - Z_v}}}}_{\circled{1}}\underbracket{e^{\lambda \prn{\E{Z_v} - \E{Z_u}}}}_{\circled{2}} , \label{eq:upper-bound-example-totally-bounded-step-two}
    \end{equation}
    we separately consider $\circled{1}$ and $\circled{2}$. To evaluate $\circled{1}$, notice that the greatest common component of $W$ is   \begin{align}
        \bigwedge_{u \in \mathcal{U}} W(A \mid u) &\stackclap{(a)}{=} \int_A \frac{1}{2 \varsigma \sqrt{2 \pi}} \exp \mprn{-\frac{x^2}{2 \varsigma^2}} \diff x + \frac{1}{2} \bigwedge_{u \in \mathcal{U}} \updelta_u(A) \\&= \int_A \frac{1}{2 \varsigma \sqrt{2 \pi}} \exp \mprn{-\frac{x^2}{2 \varsigma^2}} \diff x ,
    \end{align}
    where (a) holds by \cref{lemma:greatest-common-component-under-affine-transformation}, and so   \begin{equation}
        \bigwedge_{u \in \mathcal{U}} W(\mathcal{X} \mid u) = \int_{-\infty}^\infty \frac{1}{2 \varsigma \sqrt{2 \pi}} \exp \bigprn{-\frac{x^2}{2 \varsigma^2}} \diff x = \frac{1}{2}.
    \end{equation}
    Thus, the maximal coupling $\ibrc{X_u}_{u \in \mathcal{U}}$ such that $X_u \sim W(\lwildcard \mid u)$ for all $u \in \mathcal{U}$ is
    \begin{equation}
        X_u \triangleq \begin{cases}
            X^*  , & \text{if $I = 1$}  , \\
            \tilde{X}_u  , & \text{if $I = 0$}  ,
        \end{cases}
    \end{equation}
    where the random variables $I$, $X^*$, and $\ibrc{\tilde{X}_u}_{u \in \mathcal{U}}$ are sampled independently from the probability measures   \begin{equation}
        I \sim \mathsf{Bernoulli} \mprn{\frac{1}{2}}  , \quad
        X^* \sim \mathsf{Normal} \mprn{0, \sigma^2}  , \quad
        \tilde{X}_u \sim \updelta_u .
    \end{equation}
    Hence,   \begin{align}
        \circled{1} &= \E{e^{\lambda \prn{X_u^2 - X_v^2}}} \\&\stackclap{(a)}{=} \P{I = 1} \, \E{e^{\lambda \iprn{X^{*2} - X^{*2}}}} \\&\phantom{{}={}}+ \P{I = 0} \, \E{e^{\lambda \iprn{\tilde{X}_u^2 - \tilde{X}_v^2}}} \\
        &\stackclap{(b)}{=} \frac{1}{2} + \frac{1}{2} \, \E{e^{\lambda \iprn{\tilde{X}_u^2 - \tilde{X}_v^2}}} \\&\stackclap{(c)}{=} \frac{1}{2} + \frac{1}{2} \, \E{e^{\lambda \tilde{X}_u^2}} \, \E{e^{-\lambda \tilde{X}_v^2}} \\&\stackclap{(d)}{=} \frac{1}{2} + \frac{1}{2} \, e^{\lambda u^2} e^{-\lambda v^2} \\&= \frac{1 + e^{-\prn{v^2 - u^2} \lambda}}{2}  , \label{eq:upper-bound-example-totally-bounded-step-three}
    \end{align}
    where (a) holds by the law of total expectation, (b) holds because $I \sim \mathsf{Bernoulli}(1/2)$, (c) holds because $\tilde{X}_u$ and $\tilde{X}_v$ are independent, and (d) holds because $\tilde{X}_u \sim \updelta_u$ and $\tilde{X}_v \sim \updelta_v$. Next, we evaluate $\circled{2}$ and obtain   \begin{equation}
        \circled{2} \stackclap{(a)}{=} \exp \mprn{\lambda \prn{\frac{v^2 + \varsigma^2}{2} - \frac{u^2 + \varsigma^2}{2}}} = \frac{1}{\exp \mprn{-\frac{\prn{v^2 - u^2} \lambda}{2}}}  , \label{eq:upper-bound-example-totally-bounded-step-four}
    \end{equation}
    where (a) holds by \cref{eq:upper-bound-example-totally-bounded-step-one}. Combining \cref{eq:upper-bound-example-totally-bounded-step-two,eq:upper-bound-example-totally-bounded-step-three,eq:upper-bound-example-totally-bounded-step-four}, we have   \begin{align}
        &\phantom{{}={}}\E{e^{\lambda \prn{\prn{Z_u - Z_v} - \E{Z_u - Z_v}}}} \\&\stackclap{(a)}{=} \cosh \biggprn{\frac{\prn{v^2 - u^2} \lambda}{2}} \\&\stackclap{(b)}{\leq} \exp \biggprn{\frac{\prn{u^2 - v^2}^2 \lambda^2}{8}} \\&= \exp \biggprn{\frac{\prn{u + v}^2 \prn{u - v}^2 \lambda^2}{8}} \\
        &\stackclap{(c)}{\leq} \exp \biggprn{\frac{a^2 \prn{u - v}^2 \lambda^2}{2}} \\&\stackclap{(d)}{=} \exp \biggprn{\frac{\sigma^2 d_\mathcal{U}(u, v)^2 \lambda^2}{2}}
    \end{align}
    as desired, where (a) holds by the identity $\cosh(x) = (1 + e^{-2x})/(2e^{-x})$, (b) holds by the inequality $\cosh(x) \leq \exp(x^2 / 2)$, (c) holds because $\mathcal{U} = [-a, a]$, and (d) holds by definition of $d_\mathcal{U}$ and $\sigma$.

    \textbf{Measurability:} Observe that if $I = 1$, we have   \begin{equation}
        \sup_{u \in \mathcal{U}} \brc{Z_u - \E{Z_u}} = \!\!\sup_{u \in [-a, a]} \brc{X^{*2} - \frac{u^2 + \varsigma^2}{2}} = X^{*2} - \frac{\varsigma^2}{2} ,
    \end{equation}
    and if $I = 0$, we have   \begin{align}
        \sup_{u \in \mathcal{U}} \brc{Z_u - \E{Z_u}} &= \sup_{u \in [-a, a]} \brc{\tilde{X}_u^2 - \frac{u^2 + \varsigma^2}{2}} \\&= \sup_{u \in [-a, a]} \brc{u^2 - \frac{u^2 + \varsigma^2}{2}} \\&= \frac{a^2 + \varsigma^2}{2} .
    \end{align} 
    Hence, the pre-image of any measurable set ${A \subseteq \mathbb{R}}$ under $\sup_{u \in \mathcal{U}} \ibrc{Z_u - \E{Z_u}}$ is the measurable event   \begin{math}
        \iprn{\ibrc{I = 1} \cap \ibrc{X^* \in \mathcal{S}_A}} \allowbreak\cup \ibrc{I = 0}
    \end{math}
    if $(a^2 + \varsigma^2)/2 \in A$, and the measurable event   \begin{math}
        \brc{I = 1} \cap \ibrc{X^* \in \mathcal{S}_A}
    \end{math}
    otherwise, where $\mathcal{S}_A \subseteq \mathbb{R}$ is the set   \begin{math}
        \mathcal{S}_A = \ibrc{-\sqrt{z + \varsigma^2 / 2}, \sqrt{z + \varsigma^2 / 2}: z \in A, \, z \geq -\varsigma^2 / 2}
    \end{math}.
\end{proof}

\section{Doeblin Curves and Markov Chains} \label{appendix:doeblin-curves-and-markov-chains}

Lastly, we conclude with a discussion relating Doeblin curves back to the classic setting of Markov chain ergodicity. As a minor departure from our previous exclusive focus on Doeblin curves, we will also discuss contraction properties of Doeblin coefficients here. Throughout this appendix, we focus on the special case of \emph{discrete-time finite-state time-homogeneous} Markov chains, which we represent as a \emph{row} stochastic matrix $\mathbf{K} \in \mathbb{R}_\mathsf{sto}^{d \times d}$ whose entries $[\mathbf{K}]_{i,j}$ denote the probability of transitioning from state $i$ to state $j$ (which, by homogeneity, is the same at any time). We assume that norms operate on row vectors analogously to their usual behavior on column vectors.

Below, we present an overview of pertinent results on Markov chain convergence from the literature and add new results discussing how Doeblin curves relate to these existing characterizations.

\begin{proposition}[Markov Chain Properties] \label{proposition:markov-chain-properties}
    Given a discrete-time finite-state time-homogeneous Markov chain represented as a row stochastic matrix $\mathbf{K} \in \mathbb{R}_\mathsf{sto}^{d \times d}$, the following statements are equivalent:
    \begin{enumerate}[a)]
        \item \emph{(Positive spectral gap \cite{Gallager2013})} $|\lambda_2(\mathbf{K})| < 1$.
        \item \emph{(Aperiodic unichain \cite{Gallager2013})} The Markov chain represented by $\mathbf{K}$ has exactly one recurrent class, and its recurrent class is aperiodic.
        \item \emph{(Strong ergodicity \cite{Gallager2013,AnilyFedergruen1987})} The Markov chain represented by $\mathbf{K}$ converges to a fixed steady-state distribution $\boldsymbol{\pi}^* \in \mathscr{P}_{d-1}$ regardless of the initial distribution, i.e.,   \begin{math}
            \lim_{n \rightarrow \infty} \boldsymbol{\pi}_0 \mathbf{K}^n = \boldsymbol{\pi}^*
        \end{math} for all $\boldsymbol{\pi}_0 \in \mathscr{P}_{d-1}$.
        Furthermore, the convergence is exponentially fast, i.e., there exist constants $C > 0$ and $0 < \alpha < 1$ independent of $\boldsymbol{\pi}_0$ such that   \begin{math}
            \inorm{\boldsymbol{\pi}_0 \mathbf{K}^n - \boldsymbol{\pi}^*}_1 \leq C \alpha^n
        \end{math} for all $n \in \mathbb{N}$.
        \item \emph{(Weak ergodicity \cite{AnilyFedergruen1987,Seneta1981})} The rows of $\mathbf{K}^n$ equalize as time $n \rightarrow \infty$, i.e.,   \begin{math}
            \lim_{n \rightarrow \infty} \inorm{\ibkt{\mathbf{K}^n}_{\ang{i}} - \ibkt{\mathbf{K}^n}_{\ang{j}}}_\infty = 0
        \end{math} for all $i, j \in [d]$.
        \item \emph{(Doeblin characterization of weak ergodicity)} For all $n \geq d^2$, $\tau(\mathbf{K}^n) > 0$. %
        \item \emph{(Doeblin curves)}   \begin{math}
            \mathrm{F}_{\mathbf{K}^n} \iprn{t; \mathbb{R}_\mathsf{sto}^{d \times d}} < t \label{eq:markov-chain-properties}
        \end{math} for all $n \geq d^2$ and $t \in (0, 1]$.
    \end{enumerate}
\end{proposition}

\begin{proof}
    The equivalences a) through f) hold because:
    \begin{itemize}
        \item b) $\implies$ a) by \cite[Theorem 4.4.2]{Gallager2013}
        \item $\lnot \,$b) $\implies \lnot \,$a) by \cite[Section 4.4.2, p. 178]{Gallager2013}
        \item b) $\implies$ c) by \cite[Theorem 4.3.7]{Gallager2013}
        \item $\lnot \,$b) $\implies \lnot \,$c) by \cite[Section 4.3.5, p. 175]{Gallager2013}
        \item c) $\iff$ d) by \cite[p. 867]{AnilyFedergruen1987}
        \item e) $\implies$ d) by \cite[Eq. 6]{Chestnut2010} or \cite[Theorem 4.8]{Seneta1981}
    \end{itemize}
    
    By \cite[Theorem 4.8]{Seneta1981}, d) implies that there exists some $N \in \mathbb{N}$ such that $\tau(\mathbf{K}^N) > 0$; this is a weaker version of e) that lacks an explicit $d$-dependent threshold on $N$. To remedy this, we supply a proof of the threshold below.

    \textbf{Proof of b) $\implies$ e):}
    Let $r \in [d]$ be the number of states in the recurrent class of $\mathbf{K}$, and $d-r$ the number of transient states. In the graph representation of $\mathbf{K}$, for each transient state $i$ (if one exists), there must exist a path from $i$ to any state in the recurrent class of length at most $d-r$. Fix an arbitrary state $k$ in the recurrent class. Since the recurrent class is aperiodic, Wielandt's result for primitive matrices \cite[Corollary 8.5.8]{HornJohnson2013} implies that given an arbitrary integer $N\geq r^2-2r+2$, there must exist a path of length exactly $N$ from each recurrent state to $k$. Thus, for every state $i$ (either transient or recurrent), there must exist a path from $i$ to $k$ of length exactly $r^2-3r+d+2$, which implies that $\tau(\mathbf{K}^n)>0$ for all
    \begin{equation}
        n \geq d^2 \geq r^2-3r+d+2,
    \end{equation}
    which holds for all $d\geq r \geq 1$ since $d^2-d=d(d-1)\geq (r-1)(r-2) = r^2-3r+2$, as desired.
    
    Although c) $\iff$ d) is known in the literature \cite[p. 867]{AnilyFedergruen1987}, this equivalence is seldom directly argued for homogeneous Markov chains, since the distinction between strong and weak forms of ergodicity is traditionally only made in the context of inhomogeneous chains. So, for completeness, we provide a proof of c) $\iff$ d) from first principles below.
    
    \textbf{Proof of c) $\iff$ d):} Fix a time-homogeneous Markov chain $\mathbf{K} \in \mathbb{R}_\mathsf{sto}^{d \times d}$. First, we will prove the forward direction; the argument is straightforward and standard. Assume $\mathbf{K}$ is strongly ergodic, with $\boldsymbol{\pi}^*$ denoting its steady-state distribution. Then, for any $i, j \in [d]$,   \begin{align}
        &\phantom{{}={}}\lim_{n \rightarrow \infty} \norm{\bkt{\mathbf{K}^n}_{\ang{i}} - \bkt{\mathbf{K}^n}_{\ang{j}}}_\infty = \lim_{n \rightarrow \infty} \norm{\mathbf{e}_i \mathbf{K}^n - \mathbf{e}_j \mathbf{K}^n}_\infty \\&\stackclap{(a)}{\leq} \lim_{n \rightarrow \infty} \norm{\mathbf{e}_i \mathbf{K}^n - \boldsymbol{\pi}^*}_\infty + \lim_{n \rightarrow \infty} \norm{\mathbf{e}_j \mathbf{K}^n - \boldsymbol{\pi}^*}_\infty \stackclap{(b)}{=} 0  ,
    \end{align}
    where (a) holds by the triangle inequality and additivity of limits, and (b) holds by strong ergodicity. Hence, $\mathbf{K}$ is weakly ergodic, as desired.

    Next, we will prove the reverse direction. Assume $\mathbf{K}$ is weakly ergodic. Consider the sequence of powers $\ibrc{\mathbf{K}^n}_{n=1}^\infty$. Since the set of all row stochastic matrices $\mathbb{R}_\mathsf{sto}^{d \times d}$ is compact (e.g., under the Frobenius norm), there exists a subsequence $\ibrc{\mathbf{K}^{n_\ell}}_{\ell=1}^\infty$ which converges to a limit $\mathbf{K}^* \in \mathbb{R}_\mathsf{sto}^{d \times d}$, i.e.,
    \begin{equation}
        \lim_{\ell \rightarrow \infty} \frnorm{\mathbf{K}^{n_\ell} - \mathbf{K}^*} = 0 . \label{eq:ergodicity-subsequence-limit}
    \end{equation}
    Furthermore, for any $i, j \in [d]$, we have   \begin{align}
        &\phantom{{}={}} \norm{\bkt{\mathbf{K}^*}_{\ang{i}} - \bkt{\mathbf{K}^*}_{\ang{j}}}_\infty 
        \\&\stackclap{(a)}{\leq} \lim_{\ell \rightarrow \infty} \norm{\bkt{\mathbf{K}^{n_\ell}}_{\ang{i}} - \bkt{\mathbf{K}^*}_{\ang{i}}}_\infty \!\!\! + \lim_{\ell \rightarrow \infty} \norm{\bkt{\mathbf{K}^{n_\ell}}_{\ang{j}} - \bkt{\mathbf{K}^*}_{\ang{j}}}_\infty\\&\phantom{{}={}} + \lim_{\ell \rightarrow \infty} \norm{\bkt{\mathbf{K}^{n_\ell}}_{\ang{i}} - \bkt{\mathbf{K}^{n_\ell}}_{\ang{j}}}_\infty \\
        &\stackclap{(b)}{=} \lim_{\ell \rightarrow \infty} \norm{\bkt{\mathbf{K}^{n_\ell}}_{\ang{i}} - \bkt{\mathbf{K}^{n_\ell}}_{\ang{j}}}_\infty \\&\stackclap{(c)}{=} \lim_{n \rightarrow \infty} \norm{\bkt{\mathbf{K}^n}_{\ang{i}} - \bkt{\mathbf{K}^n}_{\ang{j}}}_\infty \stackclap{(d)}{=} 0  ,
    \end{align}
    where (a) holds by the triangle inequality and additivity of limits, (b) holds by \cref{eq:ergodicity-subsequence-limit}, (c) holds because the limit of any subsequence of a convergent sequence is equal to the limit of the entire sequence, and (d) holds by weak ergodicity. Hence, $\mathbf{K}^*$ has all identical rows, and so $\mathrm{rank}(\mathbf{K}^*) = 1$. Therefore,   $\iabs{\lambda_1(\mathbf{K}^*)} > 0$ and $\lambda_2(\mathbf{K}^*) = \tcdots = \lambda_d(\mathbf{K}^*) = 0$. \par %
    Moreover, by the Perron-Frobenius theorem \cite[Chapter 1]{Seneta1981}, we have $\lambda_1(\mathbf{K}^*) = 1$ and $\lambda_1(\mathbf{K}^n) = 1$ for every $n \in \mathbb{N}$.
    This, along with the continuity of eigenvalues with respect to the entries of a matrix \cite[Chapter IV, Theorem 1.1]{StewartSun1990}, yields   \begin{math}
        \lim_{\ell \rightarrow \infty} \lambda_i \mprn{\mathbf{K}^{n_\ell}} = \lambda_i(\mathbf{K}^*) = 0
    \end{math} for all $i \in \ibrc{2, \tdots, d}$.
    Since $\lambda_i(\mathbf{K}^{n_\ell}) = \lambda_i(\mathbf{K})^{n_\ell}$ for all $i \in [d]$, it follows that   \begin{math}
        \iabs{\lambda_i(\mathbf{K})} < 1
    \end{math} for all $i \in \ibrc{2, \tdots, d}$.
    Hence, consider the Jordan canonical form $\mathbf{K} = \mathbf{X}^{-1} \mathbf{J} \mathbf{X}$ \cite[Chapter 3]{HornJohnson2013}, where the Jordan matrix $\mathbf{J}$ is lower triangular with $[\mathbf{J}]_{1,1} = 1$ and all other diagonal entries strictly less than $1$ in magnitude, and the rows of $\mathbf{X}$ contain the generalized eigenvectors of $\mathbf{K}$, scaled such that $\inorm{[\mathbf{X}]_{\ang{1}}}_1 = 1$. For any initial distribution $\boldsymbol{\pi}_0 = \mathbf{c} \mathbf{X} \in \mathscr{P}_{d-1}$, we have   \begin{align}
        \lim_{n \rightarrow \infty} \boldsymbol{\pi}_0 \mathbf{K}^n &= \lim_{n \rightarrow \infty} \brc{\prn{\mathbf{c} \mathbf{X}} \prn{\mathbf{X}^{-1} \mathbf{J}^n \mathbf{X}}} = \mathbf{c} \prn{\lim_{n \rightarrow \infty} \mathbf{J}^n} \mathbf{X} \\&\stackclap{(a)}{=} \mathbf{c} \mathbf{e}_1^\top \mathbf{e}_1 \mathbf{X} = \bkt{\mathbf{c}}_1 \bkt{\mathbf{X}}_{\ang{1}} ,
    \end{align}
    where (a) holds because the powers of Jordan blocks corresponding to eigenvalues with magnitude less than $1$ vanish entry-wise in the limit. Finally, since $\inorm{\boldsymbol{\pi}_0 \mathbf{K}^n}_1 = 1$ for each $n \in \mathbb{N}$, we have $[\mathbf{c}]_1 = 1$ for any choice of $\boldsymbol{\pi}_0 \in \mathscr{P}_{d-1}$,\footnote{This occurs because $[\mathbf{X}]_{\ang{1}}$ is a point on $\mathscr{P}_{d-1}$, and the remaining rows of $\mathbf{X}$ span the direction space associated with $\mathscr{P}_{d-1}$.} and thus,   \begin{math}
        \lim_{n \rightarrow \infty} \boldsymbol{\pi}_0 \mathbf{K}^n = \bkt{\mathbf{X}}_{\ang{1}}
    \end{math}.
    Hence, $\mathbf{K}$ is strongly ergodic, as desired.

    \textbf{Proof of e) $\iff$ f):} The set of all $d \times d$ row stochastic matrices $\mathbb{R}_\mathsf{sto}^{d \times d}$ is convex, contains the identity kernel (i.e., the $d \times d$ identity matrix), and contains a constant kernel (e.g., the $d \times d$ matrix where each row is $\mathbf{e}_1$). Thus, by \cref{proposition:properties-of-doeblin-curves}, Part 2,
    \begin{align}
        \mathrm{F}_{\mathbf{K}^n} \mprn{t; \mathbb{R}_\mathsf{sto}^{d \times d}} = \rho \mprn{\mathbf{K}^n} \, t  . \label{eq:markov-chain-doeblin-curve-sharpness}
    \end{align}
    Hence, $\mathrm{F}_{\mathbf{K}^n} \iprn{t; \mathbb{R}_\mathsf{sto}^{d \times d}} < t$ holds iff $\rho(\mathbf{K}^n) < 1$, or equivalently $\tau(\mathbf{K}^n) > 0$, as desired.
\end{proof}

\cref{proposition:markov-chain-properties} states that any Markov chain $\mathbf{K} \in \mathbb{R}_\mathsf{sto}^{d \times d}$ has positive spectral gap (i.e., $|\lambda_2(\mathbf{K})| < 1$) iff $\mathrm{F}_{\mathbf{K}^n} \iprn{t; \mathbb{R}_\mathsf{sto}^{d \times d}} < t$ holds for sufficiently large $n$. We note that taking into account whether the inequality holds for \emph{sufficiently high power} $\mathbf{K}^n$, not necessarily $\mathbf{K}$ itself, is crucial to this equivalence, as $|\lambda_2(\mathbf{K})| < 1$ does not imply in general that the inequality holds for $n = 1$. As a counterexample, consider a directed cycle on $d \geq 4$ vertices with self-loops at each vertex:
\begin{equation}
    \forall i, j \in [d], \quad \bkt{\mathbf{K}}_{i,j} =\! \begin{cases}
        \frac{1}{2}  , & \text{if $j = i$}  , \\
        \frac{1}{2}  , & \text{if $j = i + 1$ or $(i, j) = (d, 1)$}  , \\
        0  , & \text{otherwise}  .
    \end{cases}
\end{equation}
Clearly, $\mathbf{K}$ has positive spectral gap because it is a unichain (by virtue of its cycle structure) and aperiodic (by virtue of its self-loops). Now, let $\mathbf{w}$ and $\mathbf{v}$ be point mass distributions at vertices at least distance $2$ apart on the cycle, e.g., $\mathbf{w} = \mathbf{e}_1$ and $\mathbf{v} = \mathbf{e}_3$. Then, by matrix multiplication, the distributions after one step of $\mathbf{K}$ are   \begin{math}
    \mathbf{w} \mathbf{K} = (1/2) \mathbf{e}_1 + (1/2) \mathbf{e}_2
\end{math} and \begin{math}
    \mathbf{v} \mathbf{K} = (1/2) \mathbf{e}_3 + (1/2) \mathbf{e}_4 .
\end{math}
It follows that   \begin{equation}
    \rho(\mathbf{K}) \stackclap{(a)}{=} \rho \bigprn{\bkt{\mathbf{w}, \mathbf{v}}} \, \rho(\mathbf{K}) \stackclap{(b)}{\geq} \rho \bigprn{\bkt{\mathbf{w} \mathbf{K}, \mathbf{v} \mathbf{K}}} \stackclap{(c)}{=} 1  ,
\end{equation}
where (a) holds because $\mathbf{w} = \mathbf{e}_1$ and $\mathbf{v} = \mathbf{e}_3$ have disjoint support and thus $\rho([\mathbf{w}, \mathbf{v}]) = 1$, (b) holds by submultiplicativity of complementary Doeblin coefficients, and (c) holds because $\mathbf{w} \mathbf{K}$ and $\mathbf{v} \mathbf{K}$ have disjoint support. By \cref{eq:markov-chain-doeblin-curve-sharpness}, it follows that $\mathrm{F}_\mathbf{K}(t; \mathbb{R}_\mathsf{sto}^{d \times d}) = t$, even though $\mathbf{K}$ has positive spectral gap.   %

Next, we present the following proposition which exactly characterizes when contraction occurs between two input distributions after one step of the Markov chain $\mathbf{K}$.\footnote{Since we consider two input distributions, the Doeblin coefficient $\rho$ reduces to TV distance.}

\begin{proposition}[Strict Contraction of Doeblin Coefficient] \label{proposition:strict-contraction-of-doeblin-coefficient}
    Let $\mathbf{K} \in \mathbb{R}_\mathsf{sto}^{d \times d}$ be a Markov chain. For any subset of states $\mathcal{S} \subseteq [d]$, let $\mathcal{N}_\mathbf{K}(\mathcal{S})$ denote the set of states reachable from $\mathcal{S}$ in one time step, i.e.,
    \begin{equation}
        \mathcal{N}_\mathbf{K}(\mathcal{S}) = \brc{j \in [d]: \exists i \in [d], \, \bkt{\mathbf{K}}_{i,j} > 0} . \label{eq:strict-contraction-neighborhood}
    \end{equation}
    Then, for any pair of distributions $\mathbf{w}, \mathbf{v} \in \mathscr{P}_{d-1}$, we have   \begin{math}
        \rho \iprn{\bkt{\mathbf{w}, \mathbf{v}} \mathbf{K}} < \rho \iprn{\bkt{\mathbf{w}, \mathbf{v}}}
    \end{math}
    iff $\mathcal{N}_\mathbf{K}(\mathcal{S}_{\mathbf{w} > \mathbf{v}}) \cap \mathcal{N}_\mathbf{K}(\mathcal{S}_{\mathbf{w} < \mathbf{v}})$ is non-empty, where we define   \begin{align}
        \mathcal{S}_{\mathbf{w} > \mathbf{v}} &= \brc{i \in [d]: \bkt{\mathbf{w}}_i > \bkt{\mathbf{v}}_i} , \\
        \mathcal{S}_{\mathbf{w} < \mathbf{v}} &= \brc{i \in [d]: \bkt{\mathbf{w}}_i < \bkt{\mathbf{v}}_i} . \label{eq:strict-contraction-index-sets}
    \end{align}
\end{proposition}

\begin{proof}
    Fix a Markov chain $\mathbf{K} \in \mathbb{R}_\mathsf{sto}^{d \times d}$ and distributions $\mathbf{w}, \mathbf{v} \in \mathscr{P}_{d-1}$. We have   \begin{align}
        \rho \bigprn{\bkt{\mathbf{w}, \mathbf{v}} \mathbf{K}} &\stackclap{(a)}{=} 1 - \sum_{j=1}^d \min \mbrc{\bkt{\mathbf{w} \mathbf{K}}_j, \bkt{\mathbf{v} \mathbf{K}}_j} \\&\stackclap{(b)}{=} 1 - \sum_{j=1}^d \min \mbrc{\sum_{i=1}^d \bkt{\mathbf{w}}_i \bkt{\mathbf{K}}_{i,j}, \sum_{i=1}^d \bkt{\mathbf{v}}_i \bkt{\mathbf{K}}_{i,j}} \\
        &\stackclap{(c)}{\leq} 1 - \sum_{j=1}^d \sum_{i=1}^d \min \mbrc{\bkt{\mathbf{w}}_i \bkt{\mathbf{K}}_{i,j}, \bkt{\mathbf{v}}_i \bkt{\mathbf{K}}_{i,j}} \\&\stackclap{(d)}{=} 1 - \sum_{i=1}^d \min \mbrc{\bkt{\mathbf{w}}_i, \bkt{\mathbf{v}}_i} \sum_{j=1}^d \bkt{\mathbf{K}}_{i,j} \\
        &\stackclap{(e)}{=} 1 - \sum_{i=1}^d \min \mbrc{\bkt{\mathbf{w}}_i, \bkt{\mathbf{v}}_i} \\&\stackclap{(f)}{=} \rho \bigprn{\bkt{\mathbf{w}, \mathbf{v}}}  , \label{eq:strict-contraction-step-one}
    \end{align}
    where (a) holds by definition of Doeblin coefficient \cref{eq:doeblin-coefficient-matrix}, (b) holds by matrix algebra, (c) holds by the identity
    \begin{equation}
        \min \mbrc{\sum_i x_i, \sum_i y_i} \geq \sum_i \min \mbrc{x_i, y_i}  , \label{eq:pushing-min-into-sum}
    \end{equation}
    (d) holds by factoring $[\mathbf{K}]_{i,j}$ outside the minimum (because $[\mathbf{K}]_{i,j} \geq 0$) and interchanging the order of summation, (e) holds because $\mathbf{K}$ is row stochastic, and (f) holds by definition of Doeblin coefficient \cref{eq:doeblin-coefficient-matrix}. The inequality in \cref{eq:pushing-min-into-sum} is strict iff $x_i > y_i$ and $x_{i'} < y_{i'}$ for some $i$ and $i'$. Hence, the inequality in step (c) of \cref{eq:strict-contraction-step-one} is strict iff there exist $i, i', j \in [d]$ such that{ }$\bkt{\mathbf{w}}_i > \bkt{\mathbf{v}}_i$ and $\bkt{\mathbf{K}}_{i,j} > 0$ (i.e., $j \in \mathcal{N}_\mathbf{K}(\mathcal{S}_{\mathbf{w} > \mathbf{v}})$), and $\bkt{\mathbf{w}}_{i'} < \bkt{\mathbf{v}}_{i'}$ and $\bkt{\mathbf{K}}_{i',j} > 0$ (i.e., $j \in \mathcal{N}_\mathbf{K}(\mathcal{S}_{\mathbf{w} < \mathbf{v}})$).\footnote{This condition for equality can also be verified by using the $\ell^1$-norm characterization of TV distance in (a), applying the triangle inequality to deduce (e), and then using known conditions for equality in the triangle inequality.}
    Clearly, this is equivalent to the condition that $\mathcal{N}_\mathbf{K}(\mathcal{S}_{\mathbf{w} > \mathbf{v}}) \cap \mathcal{N}_\mathbf{K}(\mathcal{S}_{\mathbf{w} < \mathbf{v}}) \neq \emptyset$, as desired.
\end{proof}

\cref{proposition:strict-contraction-of-doeblin-coefficient} exactly characterizes the condition on $\mathbf{K}$, $\mathbf{w}$, $\mathbf{v}$ such that strict contraction occurs in one step. We remark that two natural corollaries of this result give intuitive conditions on $\mathbf{K}$ under which strict contraction occurs for entire classes of $\mathbf{w}$ and $\mathbf{v}$.

\begin{corollary}[Strict Contraction and Gramian] \label{corollary:strict-contraction-and-gramian}
    Let $\mathbf{K} \in \mathbb{R}_\mathsf{sto}^{d \times d}$ be a Markov chain. The joint range of $\mathbf{K}$ satisfies   \begin{math}
        \mathfrak{F} \iprn{\mathbf{K}; \mathbb{R}_\mathsf{sto}^{2 \times d}} \subseteq \ibrc{(t, y) \in \ibkt{0, 1}^2: y < t} \cup \ibrc{(0, 0)}
    \end{math},
    or, equivalently, we have   \begin{math}
        \rho \iprn{\ibkt{\mathbf{w}, \mathbf{v}} \mathbf{K}} < \rho \iprn{\ibkt{\mathbf{w}, \mathbf{v}}}
    \end{math}
    for all pairs of distinct distributions $\mathbf{w}, \mathbf{v} \in \mathscr{P}_{d-1}$, iff the Gramian matrix $\mathbf{K} \mathbf{K}^\top$ is entry-wise strictly positive.
\end{corollary}

\begin{proof}
    Fix $\mathbf{K} \in \mathbb{R}_\mathsf{sto}^{d \times d}$.
    
    First, we will prove the forward direction. Assume $\mathbf{K} \mathbf{K}^\top$ is entry-wise strictly positive. Fix any distinct distributions $\mathbf{w}, \mathbf{v} \in \mathscr{P}_{d-1}$. Recall the definitions of $\mathcal{S}_{\mathbf{w} > \mathbf{v}}$ and $\mathcal{S}_{\mathbf{w} < \mathbf{v}}$ from \cref{eq:strict-contraction-index-sets}. Since probability distributions sum to $1$, both $\mathcal{S}_{\mathbf{w} > \mathbf{v}}$ and $\mathcal{S}_{\mathbf{w} < \mathbf{v}}$ are non-empty. Consider any $i \in \mathcal{S}_{\mathbf{w} > \mathbf{v}}$ and $i' \in \mathcal{S}_{\mathbf{v} < \mathbf{w}}$. By assumption, we have   \begin{math}
        [\mathbf{K} \mathbf{K}^\top]_{i,i'} = \sum_{j=1}^d [\mathbf{K}]_{i,j} [\mathbf{K}]_{i',j} > 0 ,
    \end{math}
    so there exists $j^* \in [d]$ such that $[\mathbf{K}]_{i,j^*} > 0$ and $[\mathbf{K}]_{i',j^*} > 0$. By \cref{eq:strict-contraction-neighborhood}, it holds that $j^* \in \mathcal{N}_\mathbf{K}(\mathcal{S}_{\mathbf{w} > \mathbf{v}})$ and $j^* \in \mathcal{N}_\mathbf{K}(\mathcal{S}_{\mathbf{w} < \mathbf{v}})$, and so   \begin{math}
        \mathcal{N}_\mathbf{K}(\mathcal{S}_{\mathbf{w} > \mathbf{v}}) \cap \mathcal{N}_\mathbf{K}(\mathcal{S}_{\mathbf{w} < \mathbf{v}}) \neq \emptyset
    \end{math}.
    By \cref{proposition:strict-contraction-of-doeblin-coefficient}, we have $\rho([\mathbf{w}, \mathbf{v}] \, \mathbf{K}) < \rho([\mathbf{w}, \mathbf{v}])$ as desired.

    Next, we will prove the reverse direction. Assume that $\rho([\mathbf{w}, \mathbf{v}] \, \mathbf{K}) < \rho([\mathbf{w}, \mathbf{v}])$ for any distinct distributions $\mathbf{w}, \mathbf{v} \in \mathscr{P}_{d-1}$. Fix unspecified distinct $i, i' \in [d]$. Choosing $\mathbf{w} = \mathbf{e}_i$ and $\mathbf{v} = \mathbf{e}_{i'}$ and applying \cref{proposition:strict-contraction-of-doeblin-coefficient}, we have   \begin{math}
        \mathcal{N}_\mathbf{K}(\brc{i}) \cap \mathcal{N}_\mathbf{K}(\ibrc{i'}) \neq \emptyset
    \end{math}.
    Consider any $j^* \in \mathcal{N}_\mathbf{K}(\ibrc{i}) \cap \mathcal{N}_\mathbf{K}(\ibrc{i'})$. We have   \begin{equation}
        \bkt{\mathbf{K} \mathbf{K}^\top}_{i,i'} = \sum_{j=1}^d \bkt{\mathbf{K}}_{i,j} \bkt{\mathbf{K}}_{i',j} \geq \bkt{\mathbf{K}}_{i,j^*} \bkt{\mathbf{K}}_{i',j^*} \stackclap{(a)}{>} 0 ,
    \end{equation}
    where (a) holds by definition of $\mathcal{N}_\mathbf{K}$ from \cref{eq:strict-contraction-neighborhood}, and so the off-diagonal entries of $\mathbf{K} \mathbf{K}^\top$ are positive. Furthermore, for any $i \in [d]$, we have   \begin{math}
        [\mathbf{K} \mathbf{K}^\top]_{i,i} = \sum_{j=1}^d [\mathbf{K}]_{i,j}^2 > 0 ,
    \end{math}
    where the final inequality
    holds because $\mathbf{K}$ is row stochastic, and so the diagonal entries of $\mathbf{K} \mathbf{K}^\top$ are positive. Hence, $\mathbf{K} \mathbf{K}^\top$ is entry-wise strictly positive, as desired.
\end{proof}

\begin{corollary}[Strict Contraction and Laziness] \label{corollary:strict-contraction-and-laziness}
    Let $\mathbf{K} \in \mathbb{R}_\mathsf{sto}^{d \times d}$ be a lazy Markov chain (i.e., $[\mathbf{K}]_{i,i} > 0$ for all $i \in [d]$) with positive spectral gap (i.e., $|\lambda_2(\mathbf{K})| < 1$). Then, for all pairs of distributions $\mathbf{w}, \mathbf{v} \in \mathscr{P}_{d-1}$ which differ at each entry (i.e., $[\mathbf{w}]_i \neq [\mathbf{v}]_i$ for all $i \in [d]$), we have   \begin{math}
        \rho \iprn{\ibkt{\mathbf{w}, \mathbf{v}} \mathbf{K}} < \rho \iprn{\ibkt{\mathbf{w}, \mathbf{v}}}
    \end{math}.
\end{corollary}

\begin{proof}
    Fix a Markov chain $\mathbf{K} \in \mathbb{R}_\mathsf{sto}^{d \times d}$ and distributions $\mathbf{w}, \mathbf{v} \in \mathscr{P}_{d-1}$ satisfying the conditions in \cref{corollary:strict-contraction-and-laziness}. Since $\mathbf{w}$ and $\mathbf{v}$ differ at each entry, we have $\mathcal{S}_{\mathbf{w} > \mathbf{v}}^\complement = \mathcal{S}_{\mathbf{w} < \mathbf{v}}$. Since $\mathbf{K}$ has positive spectral gap, $\mathbf{K}$ has exactly one recurrent class (by \cref{proposition:markov-chain-properties}, Part (b)), and so the directed graph representation of $\mathbf{K}$ is connected (in an undirected sense). Hence, there exists an edge between $\mathcal{S}_{\mathbf{w} > \mathbf{v}}$ and $\mathcal{S}_{\mathbf{w} > \mathbf{v}}^\complement$, i.e., for some $i \in \mathcal{S}_{\mathbf{w} > \mathbf{v}}$ and $i' \in \mathcal{S}_{\mathbf{w} < \mathbf{v}}$, we have $[\mathbf{K}]_{i,i'} > 0$ or $[\mathbf{K}]_{i',i} > 0$. Without loss of generality, assume $[\mathbf{K}]_{i,i'} > 0$. Since $\mathbf{K}$ is lazy, we have $[\mathbf{K}]_{i',i'} > 0$. By \cref{eq:strict-contraction-neighborhood}, it holds that $i' \in \mathcal{N}_\mathbf{K}(\mathcal{S}_{\mathbf{w} > \mathbf{v}})$ and $i' \in \mathcal{N}_\mathbf{K}(\mathcal{S}_{\mathbf{w} < \mathbf{v}})$, and so   \begin{math}
        \mathcal{N}_\mathbf{K}(\mathcal{S}_{\mathbf{w} > \mathbf{v}}) \cap \mathcal{N}_\mathbf{K}(\mathcal{S}_{\mathbf{w} < \mathbf{v}}) \neq \emptyset
    \end{math}.
    By \cref{proposition:strict-contraction-of-doeblin-coefficient}, we have $\rho([\mathbf{w}, \mathbf{v}] \, \mathbf{K}) < \rho([\mathbf{w}, \mathbf{v}])$ as desired.
\end{proof}

\balance
\bibliographystyle{IEEEtran}
\bibliography{arxiv}

\end{document}